\documentclass[12pt]{article}
\pdfoutput=1
\usepackage{jheppub}
\usepackage[utf8]{inputenc}
\usepackage{verbatim}
\usepackage{amsmath}
\usepackage{amssymb}
\usepackage{amsthm}
\usepackage{slashed}
\usepackage{amsfonts}

\newcommand{\be}{\begin{equation}}
\newcommand{\ee}{\end{equation}}
\newcommand{\bfig}{\begin{figure}\begin{center}}
\newcommand{\efig}{\end{center}\end{figure}}
\newcommand{\bi}{\begin{itemize}}
\newcommand{\ei}{\end{itemize}}

\newcommand{\lan}{\langle}
\newcommand{\ran}{\rangle}
\newcommand{\Tr}{\mathrm{Tr}}
\newcommand{\mO}{\mathcal{O}}

\newcommand{\wt}{\widetilde}

\newcommand{\ol}{\overline}
\newcommand{\Rd}{\mathbb{R}^{d}}
\newcommand{\Rdd}{\mathbb{R}^{d-1}}
\newcommand{\HA}{\mathcal{H}_A}
\newcommand{\HB}{\mathcal{H}_B}
\newcommand{\vx}{\vec{x}}
\newcommand{\vd}{\vec{\delta}}
\newcommand{\HG}{\mathcal{H}_G}
\newcommand{\AR}{\mathcal{A}[R]}
\newcommand{\BV}{\mathcal{B}(V)}
\newtheorem{thm}{Theorem}[section]

\newtheorem{conj}{Conjecture}
\theoremstyle{definition}
\newtheorem{mydef}{Definition}[section]

\newcommand{\CPN}{\mathbb{CP}^{N-1}}
\begin{document}
\title{Symmetries in Quantum Field Theory and Quantum Gravity}
\author[a]{Daniel Harlow}
\author[b,c]{and Hirosi Ooguri}
\affiliation[a]{Center for Theoretical Physics\\ Massachusetts Institute of Technology, Cambridge, MA 02139, USA}
\affiliation[b]{Walter Burke Institute for Theoretical Physics\\ California Institute of Technology,  Pasadena, CA 91125, USA}
\affiliation[c]{Kavli Institute for the Physics and Mathematics of the Universe (WPI)\\ University of Tokyo,
   Kashiwa, 277-8583, Japan}
\emailAdd{harlow@mit.edu, ooguri@caltech.edu}
\abstract{In this paper we use the AdS/CFT correspondence to refine and then establish a set of old conjectures about symmetries in quantum gravity.  We first show that any global symmetry, discrete or continuous, in a bulk quantum gravity theory with a CFT dual would lead to an inconsistency in that CFT, and thus that there are no bulk global symmetries in AdS/CFT.  We then argue that any ``long-range'' bulk gauge symmetry leads to a global symmetry in the boundary CFT, whose consistency requires the existence of bulk dynamical objects which transform in all finite-dimensional irreducible representations of the bulk gauge group.  We mostly assume that all internal symmetry groups are compact, but we also give a general condition on CFTs, which we expect to be true quite broadly, which implies this.  We extend all of these results to the case of higher-form symmetries. Finally we extend a recently proposed new motivation for the weak gravity conjecture to more general gauge groups, reproducing the ``convex hull condition'' of Cheung and Remmen. 

An essential point, which we dwell on at length, is precisely defining what we mean by gauge and global symmetries in the bulk and boundary.  
Quantum field theory results we meet while assembling the necessary tools include continuous global symmetries without Noether currents, new perspectives on spontaneous symmetry-breaking and 't Hooft anomalies, a new order parameter for confinement which works in the presence of fundamental quarks, a Hamiltonian lattice formulation of gauge theories with arbitrary discrete gauge groups, an extension of the Coleman-Mandula theorem to discrete symmetries, and an improved explanation of the decay $\pi^0\to\gamma \gamma$ in the standard model of particle physics. We also describe new black hole solutions of the Einstein equation in $d+1$ dimensions with horizon topology $\mathbb{T}^p\times \mathbb{S}^{d-p-1}$.}
\maketitle

\section{Introduction}
It has long been suspected that the consistency of quantum gravity places constraints on what kinds of symmetries can exist in nature \cite{Misner:1957mt}. In this paper we will be primarily interested in three such conjectural constraints \cite{Polchinski:2003bq,Banks:2010zn}: 

\begin{conj} \label{nosym}
No global symmetries can exist in a theory of quantum gravity.  
\end{conj}
\begin{conj}\label{allcharge}
If a quantum gravity theory at low energies includes a gauge theory with compact gauge group $G$, there must be physical states that transform in all finite-dimensional irreducible representations of $G$.  For example if $G=U(1)$, with allowed charges $Q=nq$ with $n \in \mathbb{Z}$, then there must be states with all such charges.
\end{conj}
\begin{conj}\label{compact}
If a quantum gravity theory at low energies includes a gauge theory with gauge group $G$, then $G$ must be compact.
\end{conj}

These conjectures are quite nontrivial, since it is easy to write down low-energy effective actions of matter coupled to gravity which violate them.  For example Einstein gravity coupled to two $U(1)$ gauge fields has a $\mathbb{Z}_2$ global symmetry exchanging the two gauge fields, and also has no matter fields which are charged under those gauge fields.  If we instead use two $\mathbb{R}$ gauge fields, then we can violate all three  at once.  Conjectures \ref{nosym}-\ref{compact} say that such effective theories cannot be obtained as the low-energy limit of a consistent theory of quantum gravity: they are in the ``swampland'' \cite{Vafa:2005ui,ArkaniHamed:2006dz,Adams:2006sv,Ooguri:2006in}.\footnote{Note however that the charged states required by conjecture \ref{allcharge} might be heavy, and in particular they might be black holes.}

The ``classic'' arguments for conjectures \ref{nosym}-\ref{compact} are based on the consistency of black hole physics.  One  argument for conjecture \ref{nosym} goes as follows \cite{Banks:2010zn}. Assume that a \textit{continuous} global symmetry exists.  There must be some object which transforms in a nontrivial representation of $G$.   Since $G$ is continuous, by combining many of these objects we can produce a black hole carrying an arbitrarily complicated representation of $G$.\footnote{More rigorously, given any faithful representation of a compact Lie group $G$, theorem \ref{levythm} below tells us that all irreducible representations of $G$ must eventually appear in tensor powers of that representation and its conjugate.  If $G$ is continuous, meaning that as a manifold it has dimension greater than zero, then there are infinitely many irreducible representations available.}  We then allow this black hole to evaporate down to some large but fixed size in Planck units: the complexity of the representation of the black hole will not decrease during this evaporation since the Hawking process depends only on the geometry and is uncorrelated with the global charge (for example if $G=U(1)$ then positive and negative charges are equally produced).  According to Bekenstein and Hawking the entropy of this black hole is given by \cite{Bekenstein:1973ur,Hawking:1974sw}
\be\label{bhS}
S_{BH}=\frac{Area}{4G_N},
\ee
but this is not nearly large enough to keep track of the arbitrarily large representation data we've stored in the black hole.  Thus either \eqref{bhS} is wrong, or the resulting object cannot be a black hole, and is instead some kind of remnant whose entropy can arbitrarily exceed \eqref{bhS}.  There are various arguments that such remnants lead to inconsistencies, see eg \cite{Susskind:1995da}, but perhaps the most compelling case against either of these possibilities is simply that they would necessarily spoil the statistical-mechanics interpretation of black hole thermodynamics first advocated in \cite{Bekenstein:1973ur}.  This interpretation has been confirmed in many examples in string theory \cite{Susskind:1993ws,Strominger:1996sh,Horowitz:1996nw,Strominger:1997eq,Benini:2015eyy,Berkowitz:2016jlq}.  

The classic argument for conjecture \ref{allcharge} is simply that once a gauge field exists, then so does the appropriate generalization of the Reissner-Nordstrom solution for any representation of the gauge group $G$.  The classic argument for conjecture \ref{compact} is that at least if $G$ were $\mathbb{R}$, the non-quantization of charge would imply a continuous infinity in the entropy of black holes in a fixed energy band, assuming that black holes of any charge exist, which again contradicts the finite Bekenstein-Hawking entropy.  Moreover non-abelian examples of noncompact continuous gauge groups are ruled out already in low-energy effective field theory since they do not have well-behaved kinetic terms (for noncompact simple Lie algebras the Lie algebra metric $\Tr\left(T_a T_b\right)$ is not positive-definite).  

These arguments for conjectures \ref{nosym}-\ref{compact} certainly have merit, but they are not completely satisfactory.  The argument for conjecture \ref{nosym} does not apply when the symmetry group is discrete, for example when $G=\mathbb{Z}_2$ then there is only one nontrivial irreducible representation, but why should continuous symmetries be special?  In arguing for conjecture \ref{allcharge}, does the existence of the Reissner-Nordstrom solution really tell us that a charged object exists? As long as it is non-extremal, this solution really describes a two-sided wormhole with zero total charge.  It therefore does not obviously tell us anything about the spectrum of charged states with one asymptotic boundary.\footnote{A common response to this complaint is that we should view the ends of the Reissner-Nordstrom wormhole as ``objects'' in their own right, which could exist even without the other end, but why should we?  It certainly does not follow from classical general relativity, and semiclassically charged black holes are always pair-produced unless we make them out of charged matter.  In \cite{Harlow:2015lma} it was argued that the question of whether or not a wormhole can be cut is a UV-sensitive one, which can be resolved only with input from a complete quantum gravity theory such as AdS/CFT, and we also take this point of view here.  In the end we agree that wormholes should always be cuttable, but this is more like a consequence of conjecture \ref{allcharge} rather than an argument for it.}  We could instead consider ``one-sided'' charged black holes made from gravitational collapse, but then we must first have charged matter to collapse: conjecture \ref{allcharge} would then already be satisfied by this charged matter, so why bother with the black hole at all?  To really make an argument for conjecture \ref{allcharge} based on charged solutions of general relativity that do not already have charged matter, we need to somehow satisfy Gauss's law with a non-trivial electric flux at infinity but no sources.  It is not possible to do this with trivial spatial topology. One possibility is to consider one-sided charged ``geons'' created by quotienting some version of the Reissner-Nordstrom wormhole by a $\mathbb{Z}_2$ isometry \cite{Louko:2004ej}, but this produces a non-orientable spacetime and/or requires that we gauge a discrete $\mathbb{Z}_2$ symmetry that flips the sign of the field strength. Depending on what kinds of matter fields exist these operations may not be allowed, for example there could be fermions which require the spacetime manifold to admit a spin structure.  Another possibility is to consider extremal Reissner-Nordstrom black holes, where the electric flux ends on a timelike singularity, but again it is not clear if this is really allowed without knowing more about the structure of quantum gravity.  Finally the argument for conjecture \ref{compact} implicitly relies on that for conjecture \ref{allcharge}, since one needs to assume that a continuous infinity of Reissner-Nordstrom wormholes implies a continuous infinity of charged black holes, and the argument also does not work if the gauge group $G$ is discrete.  We thus feel that there is considerable room still to improve our understanding of conjectures \ref{nosym}-\ref{compact}.

A more ``empirical'' approach to these conjectures is simply to observe that they seem to be true in all known string compactifications \cite{Banks:1988yz,Polchinski:2003bq,ArkaniHamed:2006dz}.  In particular there do not seem to be any discrete global symmetries.  But again this is also not particularly satisfying: this type of reasoning will never tell us \textit{why} conjectures \ref{nosym}-\ref{compact} are correct.  

The main goal of this paper is to use our best set of quantum gravity theories, those provided by the AdS/CFT correspondence, to justify conjectures \ref{nosym}-\ref{compact}.  Our arguments are partly based on those given in \cite{Harlow:2015lma} for case of $G=U(1)$, but they are more systematic.  Indeed we will for the most part use general group-theoretic language which applies equally well to continuous and discrete symmetry groups.

Roughly speaking our main results are the following:

\bi
\item[(i)] Any global symmetry in the bulk of AdS/CFT would be inconsistent with the local structure of the degrees of freedom in the CFT, so no such symmetries can exist.  
\item[(ii)] A compact global symmetry in a holographic CFT corresponds to a compact gauge symmetry in the bulk, with the same symmetry group in either description.
\item[(iii)] A holographic CFT with a compact global symmetry $G$ must have have local operators that transform in all finite-dimensional irreducible representations of $G$. These are then dual to objects in the bulk charged under all representations of $G$.   
\item[(iv)] There is a simple condition on the set of CFTs, which we believe holds in all CFTs with discrete spectrum and a unique stress tensor, which requires the full internal global symmetry group of that CFT to be compact.
\ei

There are several problems with these results as stated: the most obvious is that we have not said what we mean by gauge and global symmetries.  For example in any quantum field theory, the projection operator onto the $42$nd eigenstate of the Hamiltonian is a hermitian operator that commutes with the Hamiltonian.  Does this mean it generates a symmetry?  Should it have a Noether current? Do we expect it to correspond to a gauge symmetry in the bulk?   Moreover aren't gauge symmetries just redundancies of description?  How can something which is unphysical be dual to something which is physical?  What if there is a bulk gauge theory which is in a confining and/or Higgs phase? Is it still dual to a global symmetry in the CFT?  What precisely would we mean by a global symmetry of a gravitational theory if one existed?  Resolving these questions will be our first order of business, and will require careful consideration of some deep issues in quantum field theory and quantum gravity.  Our main innovation is perhaps in introducing the notion of ``long-range gauge symmetry'' in section \ref{gaugesec}, which formalizes the idea of a weakly-coupled gauge field.  It also gives a new order-parameter for confinement in the presence of fundamental quarks, which could be useful in many circumstances. Roughly speaking we use the presence of a global symmetry in the dual CFT to diagnose the phase of a gauge theory in the bulk, but we strip the holography out of this and give a strictly bulk definition which makes sense even if there is no gravity.  Also in section \ref{globalsec} we discuss the validity of Noether's theorem at some length, giving examples of quantum field theories with continuous global symmetries that do not have Noether currents, and explaining both why such examples are possible and why they do not affect our later arguments for points (i-iv).  We also point out a connection between anomalies and Noether's theorem, which we use to clarify the usual discussion of pion physics in the standard model of particle physics.

\bfig
\includegraphics[height=5cm]{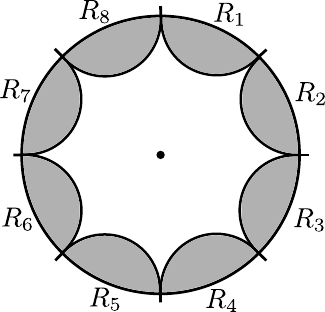}
\caption{A bulk time slice viewed from above, with the boundary timeslice $\Sigma$ split up into disjoint spatial regions $R_i$.  We've shaded the entanglement wedge of each $R_i$ grey, and the point in the center lies in none of these entanglement wedges.}\label{regionsfig}
\efig
The precise formulations of and arguments for (i-iv) are presented in sections \ref{symsec}-\ref{compsec}, and are actually quite simple once we have all the terminology straight.  To give a flavor of our methods, we here sketch our arguments for points (i) and (iii) for the special case of $G=U(1)$ (point (ii) ends up being basically equivalent to point (iii) once the relevant definitions are in place, and our argument for (iv) is simple  and self-contained enough that we just present it in section \ref{compsec}).  Indeed say that we had a $U(1)$ global symmetry in the bulk: we would then also have a $U(1)$ global symmetry in the boundary theory.  By Noether's theorem, this would be generated by a conserved current $J_\mu$.  The usual argument from here is to simply observe that this current is dual to a dynamical gauge field in the bulk \cite{Witten:1998qj}, contradicting our assumption that the symmetry was global.  This argument however fails for discrete symmetries: an argument which generalizes better to arbitrary symmetry groups is as follows. Split a spatial slice $\Sigma$ of the boundary into a disjoint set of small regions $R_i$, as shown in figure \ref{regionsfig}.  We can write the symmetry generator which rotates by an angle $\theta$ as
\be\label{Uprod}
U(\theta,\Sigma)\equiv e^{i\theta \int_\Sigma *J}=\prod_i e^{i\theta \int_{R_i}*J}.
\ee
Now since we have assumed the existence of a nontrivial bulk global symmetry, there must be a localized object that is charged under this symmetry.  Moreover there must be a charged operator $\phi^\dagger$ that creates it, obeying
\be\label{Uphi}
U^\dagger(\theta,\Sigma)\phi U(\theta,\Sigma)=e^{iq\theta}\phi,
\ee
where $q$ is the charge of the object.  

But now there is a problem: for small enough regions $R_i$, \eqref{Uprod} and \eqref{Uphi} are inconsistent. Roughly speaking this is because the finite spatial support of the operators $e^{i\theta \int_{R_i}*J}$ ensures that from the bulk point of view they are localized ``near the boundary'', and thus by bulk causality must commute with the operator $\phi$ when it is located near the center of the bulk, as in figure \ref{regionsfig}.  We can formalize this by noting that we can arrange for the operator $\phi$ to be in the complement of the ``entanglement wedge'' of each of the $R_i$'s, which is the natural bulk subregion dual to $R_i$ \cite{Czech:2012bh,Wall:2012uf,Headrick:2014cta,Dong:2016eik}.  This means that within a ``code subspace'' of sufficiently semiclassical states, $\phi$ can be represented in the CFT with spatial support only on the complement of any particular $R_i$, and thus within this subspace must commute with all of the $e^{i\theta \int_{R_i}*J}$ \cite{Almheiri:2014lwa,Harlow:2016vwg}.\footnote{This argument is complicated by the fact that bulk local operators do not really exist, since they must be ``dressed'' by Wilson lines, etc, to make them invariant under bulk diffeomorphisms and internal gauge symmetries.  But this dressing must also commute with our assumed global symmetry, since otherwise that symmetry would have to be gauged as well.  We will discuss this further in section \ref{symsec} below when we define what we mean by a global symmetry in gravity.}  But then satisfying \eqref{Uphi} is impossible, so there must not have been such a bulk global symmetry in the first place.  The key input in this argument was Noether's theorem, which as we explain more below is basically a consequence of the local structure of the boundary CFT, and our general argument for arbitrary symmetry groups will rely on a generalization of that theorem (hence our need to treat that theorem carefully in section \ref{globalsec}). 

Our argument for point (iii) proceeds on similar lines.  Following \cite{Harlow:2015lma} we consider the algebra of a Wilson line in the minimal-charge representation of $U(1)$ threading the AdS-Schwarzschild geometry from one boundary to the other (see figure \ref{threadfig} below) with the exponential of the integrated electric flux over one of the spatial boundaries
\be
e^{-\frac{i\theta}{q^2}\int\star F}e^{i\int A} e^{\frac{i\theta}{q^2}\int\star F}=e^{i\theta}e^{i\int A}.
\ee  
The locality of the boundary CFT implies that this electric flux is an operator with nontrivial support only on one of the CFTs, and its algebra with the Wilson line is apparently nontrivial for all $\theta\in (0,2\pi)$.  But this is only possible if a single copy of the CFT has states of minimal charge, since otherwise there would be a $0<\theta<2\pi$ for which the exponential of the integrated flux would be trivial and thus have to act trivially on the Wilson line.  For example if there were only even charges, so that $\frac{1}{q^2}\int \star F=2n$ in all states, then we would have $e^{\frac{i\pi}{q^2}\int\star F}=1$.  Thus all charges must be present.

To ease the presentation we will first establish (i-iv) only for internal global symmetries, which send all operators at a point to other operators at the same point, and wait until section \ref{bigdiffsec} to discuss spacetime global symmetries such as boosts and rotations.  In that section we also give a discrete generalization of the Coleman-Mandula theorem.  In section \ref{psec} we will then show that analogous conjectures also hold for higher-form symmetries, which we review for the convenience of the reader.  The arguments for spacetime and higher-form symmetries are mostly the same as for ordinary internal global symmetries, but several interesting new subtleties arise.  The higher-form versions of the conjectures have some interesting interplay with the original conjectures, which we discuss.

Finally in section \ref{wgcsec} we briefly consider the ``weak gravity conjecture'' of \cite{ArkaniHamed:2006dz}.  In \cite{Harlow:2015lma} it was pointed out that arguments similar to those we use in proving (i-iv) motivate the idea that any bulk gauge field is emergent, and it was shown that a simple model of such an emergent gauge field, the $\mathbb{CP}^{N-1}$ $\sigma$-model of \cite{D'Adda:1978uc,Witten:1978bc}, automatically obeys a version of the weak gravity conjecture.  We will show that this argument can be generalized to gauge groups other than $U(1)$, and in particular for gauge group $U(1)^k$ reproduces the rather nontrivial ``convex hull condition'' introduced in \cite{Cheung:2014vva}.  We view this as evidence that the ``emergence'' explanation of the weak gravity conjecture is on the right track, although we are unfortunately not able to resolve the long-standing debate over what the precise version of the conjecture should be \cite{ArkaniHamed:2006dz,Heidenreich:2015nta,Heidenreich:2016aqi}. 

Various technical results and reviews are presented in the appendices, and may be referred to as needed.

It is worth discussing what our results do not exclude.  The most important thing they do not exclude is \textit{approximate} global symmetries in quantum gravity.  Indeed these are quite common in string theory, and arise basically anytime that the low-energy effective action for the appropriate light degrees of freedom does not have relevant or marginal terms which break a possible global symmetry.  For example even in the standard model this happens with $B-L$ symmetry ($B$ and $L$ separately are broken by anomalies).  Our arguments will only exclude bulk global symmetries which are good symmetries acting on the entire Hilbert space of quantum gravity, including black hole states.  In contrast, approximate symmetries which emerge in the way just described are good only in some low-energy subspace.  It is very important for phenomenology to understand how approximate such global symmetries can be (see e.g. \cite{Kamionkowski:1992mf}), for example are there lower bounds on the sizes of the coefficients of operators which violate them in the low energy effective action?  We will not answer this question here, but we view it as ripe for future study.  

A second restriction on our results is that they apply only in theories of quantum gravity which are holographic.  In fewer than four spacetime dimensions there are known examples of quantum gravity theories which are precisely formulated using local gravitational path integrals, with the string worldsheet being an especially simple example.  There is no obstruction to such theories having global symmetries: indeed in the string worldsheet theory target space isometries and worldsheet parity give examples of internal and spacetime global symmetries.  In this context it is interesting to note that in fact several of our arguments as stated work only for at least three (bulk) spacetime dimensions.  For example the situation in figure \ref{regionsfig} requires spatial locality in the boundary theory.  We believe however that it is the absence of holography which is the real culprit, for example the oriented version of pure three-dimensional Einstein gravity has spatial reflection and time reversal as global symmetries even though our arguments would have applied there had it been holographic.  More discussion on how these theories avoid being holographic is given in \cite{Harlow:2018tqv}, along with further references.

Finally we apologize for the length of this paper, which is the result of our efforts to be careful about the many subtleties involved in what at heart are relatively simple arguments. We have done our best to structure the paper in a modular way, and we encourage readers to skip to whichever subjects they find interesting without feeling the need to read all intervening material.  To aid this process, we have included markers in sections \ref{globalsec} and \ref{gaugesec} to indicate which material is essential in getting to our arguments for conjectures \ref{nosym}-\ref{compact}: one good strategy might be to read only the definitions in the beginnings of these sections and then jump straight to section \ref{symsec}.  Sections \ref{completenesssec} and \ref{compsec} are more or less independent, and section \ref{wgcsec} is especially so.  Obviously the appendices are only there for those who want them.  A short overview of our arguments is also available in \cite{hosmall}.

\subsection{Notation}\label{notsec}
In this paper we discuss quantum field theory at a higher level of rigor than is usual, but still not at a level that would satisfy a mathematician.  In particular we will \textit{not} give a formal set of axioms which defines quantum field theory.  This is unavoidable, since there is currently no such set of axioms which is both necessary and sufficient to capture the full range of examples of interest, but it puts us in the awkward position of ``proving'' statements about objects which we have not defined.  To make this less piecemeal, we here state a few basic ideas which we expect to be part of any reasonable definition of quantum field theory.  
\bi
\item We will for the most part be interested in quantum field theories on Lorentzian manifolds of the form $\Sigma\times \mathbb{R}$, where $\Sigma$ is some spatial manifold and $\mathbb{R}$ is time.  We will view the metric $g_{\mu\nu}$ on $\Sigma \times \mathbb{R}$ as a  background gravitational field.  A given quantum field theory may or may not make sense on a specific choice of $\Sigma$ and $g_{\mu\nu}$, but  for each choice where it does there is a Hilbert space and a (possibly time-dependent) Hamiltonian.  
\item For any subregion $R$ of any Cauchy slice $\Sigma$, there is an associated von Neumann algebra $\mathcal{A}[R]$ acting on this Hilbert space \cite{Haag:1992hx}.  Intuitively one should think of $\mathcal{A}[R]$ as the algebra of operators localized in the domain of dependence $D[R]$ of $R$.  We will not attempt to list all of the properties these operator algebras should obey, but two essential ones are that bosonic/fermionic operators in spacelike-separated regions should commute/anticommute, and that $\mathcal{A}[R]\subset \mathcal{A}[R']$ if $R\subset D[R']$. 
\item There are a set of operator-valued distributions, conventionally just called local operators, with the property that integrating such a local operator against a smooth test function with support only in $D[R]$ produces an element of $\mathcal{A}[R]$.\footnote{This isn't quite correct, because the operator we obtain this way might not be bounded, while elements of von Neumann algebras are bounded.  So what we should really do is take the hermitian and anti-hermitian parts of this smeared operator, and then either exponentiate them or use their spectral projection operators to get ``honest'' elements of $\mathcal{A}[R]$.}
\item More generally one can have surface operators, which are operator-valued distributions localized to a submanifold (possibly with boundary) of $\Sigma\times \mathbb{R}$ of non-maximal codimension.  These again can be smeared to obtain elements of $\mathcal{A}[R]$ provided that the support of the smearing lives only in $D[R]$.  
\item There is a local operator transforming in the symmetric tensor representation of the Lorentz group,  the stress tensor $T_{\mu\nu}$, which is covariantly conserved and has the property that any continuous isometry with Killing vector $\xi^\mu$ is generated on the Hilbert space by the $T_{\mu\nu}\xi^\nu$.  Its insertion into time-ordered expectation values is defined by the derivative of those expectation values with respect to the background metric:
\be
\lan T \mO_1(x_1,g)\ldots \mO_n(x_n,g) T^{\mu\nu}(x)\ran_g\equiv -i\frac{2}{\sqrt{-g(x)}} \frac{\delta}{\delta g_{\mu\nu}(x)}\lan T \mO_1(x_1,g)\ldots \mO_n(x_n,g)\ran_g.
\ee
Note that the derivative with respect to the metric can act on any metric-dependence in the operators $\mO_i(x_i,g)$, leading potentially to contact terms.
\ei
We want to be clear that this is not a complete list of axioms.  For example there should be axioms which imply that the local and surface operators generate the full operator algebra, and also that the vacuum cannot be annihilated by operators with compact support.  We have not included such axioms not because they are not important, but rather because we are not sure what their final forms will be and we do not want to imply that there are not additional axioms we don't know about.  

We emphasize that in this paper the word ``operator'' will \textit{always} means a map from a Hilbert space to itself.  Although this may seem like it should not need any explanation, it is becoming common to see the word used in situations where this is not the case.  For example one sometimes sees a Wilson loop wrapping a temporal circle called an operator, when more precisely it should be interpreted as a modification of the theory which changes both the Hilbert space and the Hamiltonian.  This tendency has arisen from an alternative axiomatic trend in quantum field theory which is based on formal path integrals on general manifolds, not necessarily of the form $\Sigma\times \mathbb{R}$, in which arbitrary functionals of the fields can be inserted, and one downplays any Hilbert space interpretation of the result. This approach has the advantage of being covariant, but the disadvantage of being tied to the Lagrangian formalism.  One can escape this reliance on having a Lagrangian by simply \textit{defining} a quantum field theory to be the list of all possible insertions and their expectation values on all possible backgrounds, but this surely will not be the most efficient way of encoding this information.  In particular such a definition will not include a priori the constraints that come from insisting that such expectation values \textit{do} have a Hilbert space interpretation when appropriate, in which many insertions do correspond to actual operators, so this needs to be imposed by hand.  In this paper the operator algebra is essential, so we will primarily use the algebraic approach outlined in the above bullet points.   We will however also occasionally use the formal path integral insertion point of view, especially in Lagrangian examples where it is most natural.

We will make frequent use of differential forms. There is still no universally standard convention for the basic operations on these, so we here describe ours. They coincide with those in \cite{Polchinski:1998rr} except for the sign of the Hodge star, which differs by a factor of $(-1)^{p(d-p)}$ and instead agrees with, eg, \cite{Nakahara:2003nw,Carroll:2004st}.  Differential forms are completely antisymmetric tensors, whose components thus obey
\be
\omega_{\mu_1\ldots \mu_p}=\omega_{[\mu_1\ldots \mu_p]},
\ee
where the brackets on the right-hand side denote a signed average over permutations of the indices:
\be
T_{[\mu_1\ldots \mu_p]}=\frac{1}{p!}\sum_{\pi\in S_p}s_\pi T_{\mu_{\pi(1)}\ldots \mu_{\pi(p)}},
\ee
where $S_p$ denotes the symmetric group on $p$ elements and $s_\pi$ is one if $\pi$ is even and minus one if $\pi$ is odd.
The wedge product of $\omega$ a $p$-form and $\sigma$ a $q$-form is defined as
\be
\left(\omega\wedge \sigma\right)_{\mu_1\ldots \mu_p\nu_1\ldots \nu_q}=\frac{(p+q)!}{p!q!}\omega_{[\mu_1\ldots \mu_p}\sigma_{\nu_1\ldots \nu_q]},
\ee
and the exterior derivative of $\omega$ is
\be
\left(d\omega\right)_{\mu_0\mu_1\ldots \mu_p}=(p+1)\partial_{[\mu_0}\omega_{\mu_1\ldots \mu_p]}.
\ee
The completely antisymmetric symbol $\hat{\epsilon}$ in $d$ dimensions is defined as
\be
\hat{\epsilon}=dx^1\wedge dx^2\wedge\ldots\wedge dx^d,  
\ee
while the $\epsilon$ tensor is defined as
\be
\epsilon=\sqrt{|g|}\hat{\epsilon}.  
\ee
In particular note that in Lorentzian signature we have $\epsilon^{0\ldots d-1}=-\frac{1}{\sqrt{|g|}}$.\footnote{We are of course using the vastly superior ``mostly-plus'' signature for the metric.} The integral of a $d$-form $\omega$ over a $d$-dimensional manifold is defined as
\be\label{intdef}
\int_M \omega=\frac{(-1)^s}{d!}\int d^dx\sqrt{|g|}\epsilon^{\mu_1\ldots \mu_d}\omega_{\mu_1\ldots\mu_d},
\ee
where $s$ is zero in Euclidean signature and one in Lorentzian signature. Contrary to appearances, the right hand side of \eqref{intdef} depends neither on the metric nor the signature, and moreover if $N$ is a $d+1$ manifold with boundary then we have Stokes theorem
\be
\int_N d\omega=\int_{\partial N}\omega.
\ee
Finally the Hodge star operation mapping a $p$-form to a $d-p$ form is defined as
\be
\left(\star \omega\right)_{\mu_{1}\ldots \mu_{d-p}}=\frac{1}{p!}\epsilon^{\nu_1\ldots \nu_p}_{\phantom{\nu_1\ldots \nu_p}\mu_1\ldots \mu_{d-p}}\omega_{\nu_1\ldots\nu_p}.
\ee
A few useful identities, with $\omega$ again a $p$-form and $\sigma$ a $q$-form, are
\begin{align}\nonumber
\omega\wedge\sigma&=(-1)^{pq}\sigma\wedge\omega\\\nonumber
d(\omega\wedge\sigma)&=d\omega\wedge\sigma+(-1)^p\omega\wedge d\sigma\\\nonumber
\epsilon_{\mu_1\ldots \mu_d}\epsilon^{\mu_1\ldots \mu_d}&=(-1)^s d!\\
\star\star\omega&=(-1)^{p(d-p)+s}\omega.
\end{align}

We will occasionally use Dirac fermions, for which we take the $\gamma$-matrices to obey
\be
\{\gamma^\mu,\gamma^\nu\}=2g^{\mu\nu} 
\ee
and define the Dirac conjugate to be
\be
\ol{\psi}=\psi^\dagger\gamma^0.
\ee
In even spacetime dimensions we define the chirality operator to be
\be
\gamma^{d+1}=i^{-d/2}\gamma^0\ldots \gamma^{d-1},
\ee
which e.g. is equal to $+1$ on left-moving spinors for $d=2$ and $+1$ on left-handed spinors for $d=4$.

In Yang-Mills theory we take the gauge field $A^a_\mu$ to be real, and the matrix generators $T_a$ of any representation of a compact Lie algebra to be hermitian.  The structure constants $C^c_{\phantom{c}ab}$ are defined via $[T_a,T_b]=iC^c_{\phantom{c}ab}T_c$,  The covariant derivative is $D_\mu=\partial_\mu-iA^a_\mu T_a$.  For logical clarity we will maintain a distinction between lowered indices in the adjoint representation and raised indices in its inverse-transpose, even though in the compact case these representations are unitarily equivalent.

We always assume that any group we discuss is a Lie group, meaning that the group is a smooth manifold and multiplication and inversion are smooth maps.  We have found that physicists are sometimes surprised to learn that this definition includes discrete groups such as $SL(2,\mathbb{Z})$ and $\mathbb{Z}_n$, which are zero-dimensional Lie groups.  In particular any finite group is a compact Lie group with the discrete topology.  Following standard physics parlance, we will refer to Lie groups with dimension zero as ``discrete'' and Lie groups with dimension greater than zero as ``continuous'', but we emphasize that multiplication and inversion are continuous (and in fact smooth) regardless of the dimension.   We throughout adopt a convention that representations of a Lie group on a Hilbert space must be continuous, so when we encounter homomorphisms from $G$ into the set of linear operators on Hilbert space which are not necessarily continuous we will just refer to them as homomorphisms (recall that a map $f$ from one group to another is a homomorphism if $f(g_1)f(g_2)=f(g_1g_2)$ for all $g_1,g_2$).  In appendix \ref{groupapp} we explain our group theory conventions in more detail, and briefly review those aspects of the theory of Lie groups and their representations which are necessary for our arguments. The results are mostly standard but some may not be familiar to all physics readers.

Finally we will always assume that in any CFT which we are discussing, the vacuum on $\mathbb{S}^{d-1}$ is normalizable and we can therefore use the state-operator correspondence.  We view this as necessary to produce reasonable low-energy particle physics in the dual theory of asymptotically-AdS quantum gravity.  
  
\section{Global symmetry}\label{globalsec}
What is a symmetry in quantum mechanics?  The definition most of us learn as undergraduates is that a system with Hilbert space $\mathcal{H}$ and Hamiltonian $H$ has a symmetry with group $G$ if there exist a set of distinct unitary operators $U(g)$ on $\mathcal{H}$, labeled by elements $g\in G$, which respect the group multiplication\footnote{One occasionally also encounters the more general multiplication law $U(g)U(g')=e^{i\alpha(g,g')}U(gg')$, which is described by saying that the symmetry is represented projectively on the Hilbert space. This possibility does not seem to be realized in an interesting way in quantum field theory on $\mathbb{R}^d$, we explain why in appendix \ref{projapp}.}
\be\label{repeq}
U(g)U(g')=U(gg'),
\ee
and which all commute with $H$. More abstractly, there is a faithful homomorphism $U$ from $G$ into the set of unitary operators on $\mathcal{H}$, such that $U(g)$ commutes with $H$ for any $g\in G$.  This definition however is deficient in two respects:
\bi
\item It is not general enough to include spacetime symmetries.  For example Lorentz boosts and time-reversal both do not commute with $H$, and the latter is represented with an antiunitary operator instead of a unitary one.  
\item In quantum field theory it is too general, since it includes operations which do not respect the local structure of the theory.  For example consider the ``$U(1)$ symmetry'' generated by the projection onto the 42nd eigenstate of $H$: this commutes with $H$, but acts very non-locally.  
\ei 
In this paper we will not discuss spacetime symmetries until section \ref{bigdiffsec}, so the first point is currently no trouble.  The second however is a serious problem, since in quantum field theory the symmetries which are interesting seem to always be those which respect locality. We therefore propose a definition of what it means to have a global symmetry in quantum field theory:\footnote{The idea of a non-Lagrangian definition of global symmetry along these lines goes back at least to \cite{Doplicher:1969tk,Doplicher:1969kp}, although those authors did not include condition (d) (neutrality of the stress tensor).    A Euclidean definition related to this one appeared more recently in \cite{Gaiotto:2014kfa}, but condition (c) (faithfulness) was not included, and the spacetime was not restricted to $\Rd$, as it must be if we wish global symmetries with gravitational 't Hooft anomalies to be included.  We comment further on the definition of \cite{Gaiotto:2014kfa} at the end of this subsection.  Also note that definition \ref{globaldef} applies only to quantum field theories, we give a modified definition for gravitational theories in section \ref{symsec} below.}
\begin{mydef}\label{globaldef}
A Lorentz-invariant quantum field theory in $d$ spacetime dimensions has a \textit{global symmetry with symmetry group G} if the following are true:
\bi
\item[(a)] If we study the theory on the spacetime manifold $\Rd$ with flat metric, with flat time slices $\Sigma_t\cong\Rdd$, then for each time slice $\Sigma_t$ there is a unitary homomorphism $U(g,\Sigma_t)$, not necessarily continuous, from $G$ to the set of unitary operators on the Hilbert space. 
\item[(b)] For any $g\in G$ and $R\subset \Sigma_t$, we have
\be
U^\dagger(g,\Sigma_t)\mathcal{A}[R]U(g,\Sigma_t)=\mathcal{A}[R],
\ee
where $\mathcal{A}[R]$ is the algebra of operators in $D[R]$.  Moreover if $R$ is bounded as a spatial region, then the map $f_U:G\times \AR\to\AR$ defined by $f(g,\mO)=U^\dagger(g,\Sigma_t)\mO U(g,\Sigma_t)$ has the property that its restriction to any uniformly bounded subset of $\AR$ is jointly continuous in the strong operator topology (see appendix \ref{contapp} for definitions of these terms, although we encourage most readers not to worry too much about continuity).    
\item[(c)] For any $g\in G$ not equal to the identity, there exists some local operator $\mO$ for which 
\be
U^\dagger(g,\Sigma_t)\mO(x)U(g,\Sigma_t)\neq \mO(x).
\ee  
\item[(d)]  For any $g\in G$ and $x\in \Rd$, we have
\be
U^\dagger(g,\Sigma_t)T_{\mu\nu}(x)U(g,\Sigma_t)=T_{\mu\nu}(x),
\ee
where $T_{\mu\nu}$ is the stress tensor of the theory.
\ei
\end{mydef}
We first observe that condition $(d)$ tells us that the $U(g,\Sigma_t)$ commute with the Hamiltonian and thus are independent of $t$, so from now on we will just call them $U(g,\Sigma)$.  In fact condition (d) tells us something much stronger, it tells us that for any $g\in G$, $U(g,\Sigma)$ is unchanged by \textit{arbitrary} continuous deformations of $\Sigma$.  It is therefore sometimes said that the $U(g,\Sigma)$ are topological operators.  Condition (b) tells us that the $U(g,\Sigma)$ give a linear action of $G$ on the set of local operators at each point, and moreover condition (d) tells us that this linear action can be taken to be identical at each point in $\Rd$.  Indeed if we choose a basis $\mO_n(0)$ for the set of local operators at the origin, we can use spacetime translations to extend this to a basis $\mO_n(x)$ at each point in $\Rd$.  We then have
\be\label{Dmap}
\mO'_n(x)\equiv U^\dagger(g,\Sigma)\mO_n(x)U(g,\Sigma)=\sum_m D_{nm}(g)\mO_m(x),
\ee
where $D(g)$ is independent of $x$.   Condition (c) tells us that $D(g)$ is nontrivial for all $g$ except the identity.  

We have so far not referred to $U(g,\Sigma)$ and $D(g)$ as representations of $G$.  The reason is that in our conventions any Lie group representation is required to be continuous (see appendix \ref{groupapp}), while we did not require $U(g,\Sigma)$ to be continuous and we required $D$ to be continuous in the strong operator topology only on uniformly-bounded subsets of $\AR$.  We have adopted only these relatively weak requirements because we want our definition of global symmetry to apply to spontaneously-broken global symmetries, and we will see soon that $U(g,\Sigma)$ is not necessarily continuous for a symmetry which is spontaneously broken. For unbroken symmetries however, meaning symmetries for which there is a ground state on which they act trivially, we show in appendix \ref{contapp} that the continuity requirement in condition (b) of definition \ref{globaldef} implies that $U(g,\Sigma)$ is indeed continuous, and thus gives a representation of $G$ on the Hilbert space.  Moreover we also show that in this case $D$ is continuous without any domain restriction in a different topology on $\AR$, which is defined by the two-point function in the ground state.  Thus in this topology $D$ does give a representation of $G$ on the set of local operators: in fact it is a unitary representation since the set of states obtained by acting on the invariant vacuum with $\mO_n(x)$ (smeared against a smooth test function of compact support) will transform in the inverse-transpose representation of $D$, which therefore must be unitary since $U(g,\Sigma)$ is.  We relegate further discussion of operator continuity to appendix \ref{contapp}, where we also give more motivation for the continuity assumption in condition (b).

To get some intuition for definition \ref{globaldef}, let's consider a few simple examples.  One example is the $\mathbb{Z}_2$ symmetry $\phi'=-\phi$ of the three dimensional real scalar theory with Lagrangian
\be\label{phi4}
S=-\frac{1}{2}\int d^3 x \left(\partial^\mu \phi \partial_\mu \phi+m^2 \phi^2+\frac{\lambda}{6}\phi^4\right).
\ee
Another example is the $U(N)$ symmetry $\phi'_i=\sum_j U_{ij}\phi_j$ of the three-dimensional theory of $N$ complex scalars $\phi_i$ with Lagrangian 
\be
S=-\int d^3 x \left(\partial^\mu \phi^*_i \partial_\mu \phi_i+m^2\phi^*_i \phi_i+\frac{\lambda}{6}(\phi^*_i \phi_i)^2\right).
\ee   
A more nontrivial example is the $U(1)$ symmetry generated by $B-L$, with $B$ baryon number and $L$ lepton number, in the standard model of particle of physics (without gravity). 

An example of something which is not included is the $U(1)$ gauge symmetry of quantum electrodynamics.  There are no local operators which are charged under it, contrary to (c), and in fact if we study the theory on a compact spatial manifold without boundary then the gauge symmetry acts trivially on the Hilbert space.  We discuss this in much more detail in section \ref{gaugesec}.  Another thing which is not included is the ``$\mathbb{Z}_N$ center symmetry'' of pure Yang-Mills theory with gauge group $SU(N)$ \cite{Polyakov:1978vu,tHooft:1977nqb}.  This is a symmetry under which only line operators are charged, so again it does not obey (c).  The modern understanding of center symmetry is that it is really a ``one-form symmetry'' in the sense of \cite{Gaiotto:2014kfa}, so we postpone further discussion to section \ref{psec} below.  As already mentioned, spacetime symmetries are also not included.  In a similar vein, the higher Kac-Moody symmetries in $1+1$ dimensional current algebra are also not included, since they have a nontrivial algebra with the stress tensor.  

Something which \textit{is} included is a global symmetry with an 't Hooft anomaly, such as the chiral phase rotation $\psi'=e^{i\gamma^5\theta}\psi$ of a massless Dirac Fermion in $3+1$ dimensions
\be\label{dirac4}
S=-i\int d^{4}x \ol{\psi}\slashed{\partial}\psi.
\ee  
This symmetry is broken if we turn on a background nonchiral $U(1)$ gauge field with $\int d^d x \sqrt{-g}F_{\alpha\beta}F_{\mu\nu} \epsilon^{\alpha\beta\mu\nu}\neq 0$, or a background metric with $\int d^d x \sqrt{-g}\epsilon^{\alpha\beta\mu\nu}R_{\alpha\beta}^{\phantom{\alpha\beta}\gamma\delta}R_{\mu\nu\gamma\delta}\neq 0$, but in our definition \ref{globaldef} we have turned on no background fields of any kind.\footnote{These particular 't Hooft anomalies cannot destroy the symmetry if the spacetime topology is $\mathbb{R}^4$ and the background fields vanish at infinity, since the integrals in question always vanish for topological reasons, but there are other 't Hooft anomalies which can.}  We will discuss 't Hooft anomalies in more detail in subsections \ref{anomalysec}-\ref{anomalysec2} below, but we note now that for applications to AdS/CFT it will be very convenient to introduce a notion of when a global symmetry extends to a more general spatial geometry $\Sigma$:

\begin{mydef}\label{symextend}
A global symmetry of a quantum field theory is \textit{preserved on a spatial geometry $\Sigma$} if, after quantizing the theory on $\Sigma$, there is a homomorphism $U(g,\Sigma)$ from $G$ into the set of unitary operators whose action by conjugation preserves the local algebras $\mathcal{A}[R]$, with the same continuity requirement as in definition \ref{globaldef}, as well as a basis $\mO_n(x)$ for the local operators at each point $x\in \Sigma\times\mathbb{R}$, such that $U(g,\Sigma)$ acts on the $\mO_n(x)$ with the same linear map $D$ that appeared in eq. \eqref{Dmap} for the theory on $\Rd$.\footnote{In general there are ambiguities in how to extend a flat space local operator to curved space, arising from the possibility of adding multiples of the curvature tensor.  Our $\mO_n(x)$ should be extensions of their flat space analogues up to these ambiguities, and our requirement that \eqref{Dmap} continues to hold on $\Sigma\times \mathbb{R}$ restricts them.} In particular this action is still faithful and preserves the stress tensor.
\end{mydef}
The $\Sigma$ we will predominantly consider is the sphere $\mathbb{S}^{d-1}$ with a round metric; for conformal field theories we will argue below that any global symmetry is preserved on this geometry since it is conformally flat.  In fact in this case $U(g,\Sigma)$ and $D(g)$ are equivalent due to the state-operator correspondence.  We postpone further discussion of which global symmetries are preserved in the presence of a background gauge field to section \ref{anomalysec}.

If the volume of $\Sigma$ is infinite, such as for $\Sigma=\Rdd$, we need to consider the possibility of spontaneous symmetry breaking.  It is sometimes said that if a global symmetry is spontaneously broken, the symmetry operators $U(g,\Sigma)$ do not exist (see eg a comment in section 10.4 of \cite{Weinberg:1995mt}).  Our point of view will be that in this situation we take the Hilbert space on $\Sigma$ to include a special kind of direct sum over the superselection sectors associated to any degenerate vacua, in which case the $U(g,\Sigma)$ do exist, and there are local operators which are charged under them as in eq. \eqref{Dmap}.\footnote{It is important here that our definition \ref{globaldef} excludes things like the higher Kac-Moody symmetries of 2D current algebra which do not commute with the stress tensor: these do not lead to degenerate vacua or superselection sectors even though the vacuum is not invariant.}  Our direct sum is special because we choose a nonstandard inner product on the vacuum space: if $b$ is the set of order parameters which label the degenerate vacua $|b\ran$, then we take
\be\label{ssbip}
\lan b|b'\ran=\begin{cases} 1 & b=b'\\ 0 & b\neq b' \end{cases}
\ee 
even if the order parameters are continuous.  For each $b$ there is a superselection sector spanned by states of the form
\be\label{ssstates}
\mO_{1}(x_1)\ldots \mO_{m}(x_m)|b\ran,
\ee
where the $\mO_{n}$ are local operators, each transforming in a represention $D_n$ of $G$.\footnote{In the presence of a ``long range gauge symmetry with dynamical charges'', introduced in definition \ref{gaugedef} below, we should also allow the $\mO_n$ to be line operators connecting infinity to itself or to a charged operator in the interior of $\Sigma$.} The full Hilbert space is then obtained from countable superpositions of such states which are normalizable in the inner product \eqref{ssbip}. States in different superselection sectors are always orthogonal.  The symmetry operators act as
\be\label{ssbUdef}
U(g)\mO_{1}(x_1)\ldots \mO_{m}(x_m)|b\ran=D_1(g^{-1})\mO_1(x_1)\ldots D_m(g^{-1})\mO_m(x_m)|gb\ran,
\ee 
which is clearly well-defined.  The infrared divergences which appear in perturbative computations of the matrix elements of the spontaneously broken charges, sometimes used to argue that $U(g,\Sigma)$ does not exist, are here properly interpreted as ensuring that $U(g,\Sigma)$ has zero matrix element between any two states in the same superselection sector.  These divergences do however also imply that when the symmetry which is spontaneously broken is continuous, meaning $G$ has positive dimension as a Lie group, then $U(g,\Sigma)$ is \textit{not} continuous as a map from $G$ to the set of unitary operators: no matter how close $g$ is to the identity, if it is not actually the identity then acting with $U(g,\Sigma)$ on any state $|\psi\ran$ in a given superselection sector gives another state which is orthogonal to $|\psi\ran$.  By contrast we do expect the action of the symmetry by conjugation on $\AR$ for bounded regions to be as continuous as it is in the unbroken case, since that action should not depend on whether or not the volume of $\Sigma$ is finite or infinite.  Thus we see that the continuity properties required in definition \ref{globaldef} are consistent with spontaneous symmetry breaking, which is therefore included (see appendix \ref{contapp} for more discussion of continuity).  In what follows we will mostly discuss unbroken global symmetries, since we will only consider compact $\Sigma$ in the boundary CFT, but we will argue that the global symmetries which are forbidden in the bulk include spontaneously broken ones (spontaneous global symmetry breaking is possible for quantum field theories in $AdS$ \cite{Inami:1985dj}, so ruling it out is nontrivial).  

\bfig
\includegraphics[height=4.5cm]{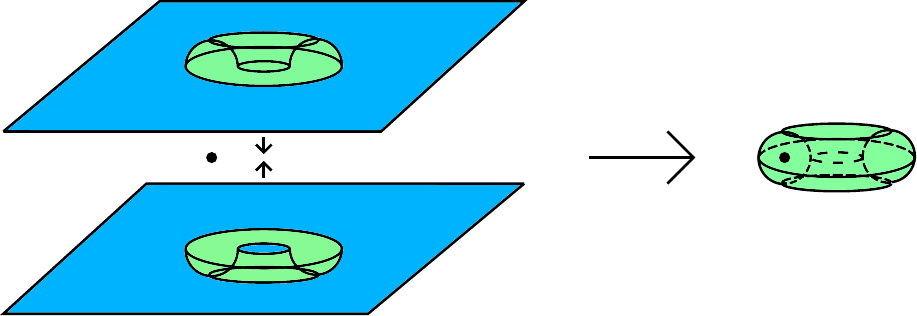}
\caption{Constructing a symmetry insertion on a torus in the path integral of a QFT on a spacetime that is topologically $\mathbb{R}^3$: the ``upper'' operator on the left hand side is a deformation of $U^\dagger(g,\mathbb{R}^2)$, while the ``lower'' operator is a deformation of $U(g,\mathbb{R}^2)$.  If we bring them together the blue sections cancel, leaving the green torus.  Since the $U(g,\mathbb{R}^2)$ commute with $T_{\mu\nu}$ they are topological, so it does not matter where we join them.   If there are no charged insertions inside the torus then we can further collapse it to nothing, while if a charged operator is inserted inside the torus, say an operator $\mO$ at the black dot in the figure, then the joint insertion amounts to inserting $U^\dagger(g,\mathbb{R}^{2})\mO U(g,\mathbb{R}^2)=D(g)\mO$ into the path integral.}\label{gluingfig}
\efig
Finally we note that in \cite{Gaiotto:2014kfa}, symmetries were defined not as operators on the Hilbert space associated to a Cauchy slice $\Sigma$, but instead as formal path integral insertions which should make sense on any codimension-one closed oriented submanifold.\footnote{We here adhere to the terminology explained in the introduction: ``path integral insertions'' are defined without reference to a Hilbert space formalism.  They can be sometimes be given Hilbert space interpretations as operators, and we will use that term only when an insertion can and is being given such an interpretation.}  We here briefly comment on how this relates to our definition \ref{globaldef}.  The basic idea is illustrated in figure \ref{gluingfig}: we can assemble such an insertion by using two of our $U(g,\Sigma)$ operators to surround whatever the surface in question encloses.  Instead of defining a single operator of the theory quantized on $\Sigma$, this instead defines a family of such operators, obtained by conjugating whatever operators are inserted in the interior of the surface by the symmetry.  In appendix \ref{closedsubapp} we explain in more detail how the construction of figure \ref{gluingfig} can be extended to any closed oriented codimension-one submanifold in $\mathbb{R}^d$.  

\subsection{Splittability}\label{splitsec}
When a global symmetry in quantum field theory is continuous, meaning that the symmetry group $G$ has dimension greater than zero as a Lie group, we usually expect the existence of a set of conserved currents $J_a^\mu$ transforming in the adjoint representation of $G$.  For Lagrangian theories this seems to follow from a local version of Noether's theorem \cite{Weinberg:1995mt,Polchinski:1998rq}. Indeed say that we define a continuous symmetry as a continuous family of local changes of variables 
\be
\phi'_i(x)=\phi_i(x)+\epsilon^a f_{a,i}(\phi(x),\partial \phi(x),\ldots)+O(\epsilon^2)
\ee
that leave the product of the path integral measure and action invariant
\be
\mathcal{D}\phi'e^{iS[\phi']}=\mathcal{D}\phi e^{iS[\phi]}.
\ee
If we now allow the group coordinates $\epsilon^a$ to be position dependent, then by locality we have
\be\label{currentdef}
\mathcal{D}\phi'e^{iS[\phi']}=\mathcal{D}\phi e^{iS[\phi]-i\int d^dx\sqrt{-g} J^\mu_a \partial_\mu \epsilon^a+O(\epsilon^2)}=\mathcal{D}\phi e^{iS[\phi]+i\int d^dx\sqrt{-g} \epsilon^a\nabla_\mu J^\mu_a+O(\epsilon^2)}
\ee
for some nonzero local functional $J^\mu_a$ of the fields.  In the second equality we have taken $\epsilon^a$ to vanish at any boundaries of the spacetime, justifying an integration by parts.  Integrating both sides of this equation over field space, and changing variables on the left hand side, we then find
\be
\int \mathcal{D}\phi e^{iS[\phi]}=\int\mathcal{D}\phi e^{iS[\phi]+i\int d^dx\sqrt{-g} \epsilon^a\nabla_\mu J^\mu_a+O(\epsilon^2)}
\ee
for arbitrary $\epsilon^a$,  which is possible only if $\nabla_\mu J^\mu_a=0$ as an operator equation so this establishes the existence of a conserved current.  

So far however no satisfactory non-Lagrangian formulation of this theorem has been found, nevermind proven.  There is however an obvious guess for what such a theorem might say:
\begin{conj}\textbf{Naive Noether Conjecture:} Any quantum field theory with a continuous global symmetry, as defined via definition \ref{globaldef}, has a conserved current whose integral infinitesimally generates that symmetry.\label{naiveN}
\end{conj}
No proof of this conjecture has ever been given, and in fact this is for a good reason: there are quantum field theories, and even Lagrangian quantum field theories, where this conjecture is false!  But is there something strange about these theories?  And moreover is there something analogous to the existence of Noether currents for discrete symmetries?  In this subsection and the following one we discuss these questions in some detail.\footnote{Readers who are primarily interested in quantum gravity may wish to simply take it on faith that the splittability we define momentarily holds for any global symmetry and proceed to subsection \ref{backgroundsec}, since the ensuing discussion is perhaps primarily of interest to quantum field theory experts.  A similar signpost there will suggest further omissions for casual readers.}

We begin with a definition:\footnote{The idea of this definition goes back to \cite{Doplicher:1982cv,Doplicher:1983if,Buchholz:1985ii}, although they didn't give it a name.}
\begin{mydef}\label{splitdef}
A global symmetry of a quantum field theory which is preserved on a spacetime $\mathbb{R}\times\Sigma$ is \textit{splittable on} $\Sigma$ if for every open spatial subregion $R\subset \Sigma$ and every $g\in G$ there is a unitary operator $U(g,R)$ such that we have 
\be\label{split1}
U^{\dagger}(g,R)\mO U(g,R)=\begin{cases} U^\dagger(g,\Sigma) \mO U(g,\Sigma) & \forall\mO\in \mathcal{A}[R]\\ \mO & \forall\mO \in \mathcal{A}[\mathrm{Int}(\Sigma-R)]\end{cases}.
\ee
We leave arbitrary how the $U(g,R)$ act on operators which are neither in $\mathcal{A}[R]$ nor $\mathcal{A}[\mathrm{Int}(\Sigma-R)]$, and in particular we do not restrict how they act on operators localized right on the boundary of $R$.  We however can and will always arrange that if $R_i$ are a finite disjoint set of open subregions of $\Sigma$ whose boundaries do not intersect, then
\be\label{split2}
\prod_i U(g,R_i)=U(g,\cup_i R_i).
\ee 
\end{mydef}   
\vspace{.5cm}
This definition is related to Noether currents as follows: if $J^\mu_a$ is a current for a global symmetry, with $G$ a compact connected Lie group, then since for any such group the exponential map is surjective, we can define operators
\be\label{currentU}
U\left(e^{i\epsilon^aT_a},R\right)\equiv e^{i\epsilon^a\int_R d^{d-1}x \sqrt{\gamma}n_\mu J^\mu_a}= e^{i\epsilon^a\int_R \star J_a},
\ee
which clearly obey the criteria \eqref{split1}, \eqref{split2}.  Thus a compact connected global symmetry with a Noether current is always splittable on any $\Sigma$ for which it is preserved.  Splittability however also can apply to discrete symmetries: for example in the Ising model, $U(-1,R)$ is the operator which flips all the spins in region $R$ and does nothing in the complement of $R$.  We have left what happens at the edges of the regions arbitrary because in quantum field theory it will be UV-sensitive, or in other words it will depend on precisely how we regulate the $U(g,R)$ at the edges.\footnote{To really get something well-defined in the continuum, we should fatten the location of the ambiguity in each $U(g,R)$ to a small open neighborhood of $\partial R$: this is what was done in \cite{Doplicher:1982cv,Doplicher:1983if,Buchholz:1985ii}, but to lighten the notation we will keep this implicit.} 

It is clear that if we can show that all global symmetries are splittable, we will have proven at least some kind of abstract version of Noether's theorem.  In fact this is precisely the context in which the notion of splittability was first introduced in the algebraic quantum field theory community \cite{Doplicher:1982cv,Doplicher:1983if,Buchholz:1985ii}.  We now revisit this issue from a more modern point of view.  We'll begin by giving a lattice  argument that all global symmetries are splittable, to help us identify the relevant issues for the continuum discussion that follows. We phrase this argument as a theorem, which shows that for finite tensor product systems, a unitary operator which acts locally on all local operators must itself be built out of local unitary operators:
\begin{thm}\label{latticethm}
Let $\mathcal{H}$ be a finite-dimensional Hilbert space that tensor factorizes as $\mathcal{H}=\otimes_i\mathcal{H}_i$, and let $U$ be a unitary operator on $\mathcal{H}$ with the property that for any tensor factor $\mathcal{H}_i$ and any operator $\mO_i$ which acts nontrivially only on $\mathcal{H}_i$, $\mO'_i\equiv U^\dagger \mO_i U$ also acts nontrivially only on $\mathcal{H}_i$.  Then $U=\prod_i U_i$, where each $U_i$ acts nontrivially only on $\mathcal{H}_i$.
\end{thm} 

\bfig
\includegraphics[height=3cm]{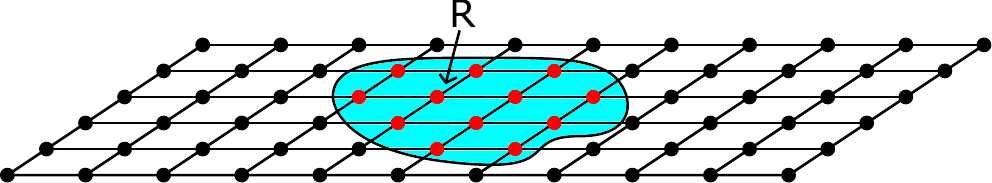}
\caption{Splittability of any global symmetry for a lattice theory.  Here each dot is a spin, so a spatial region $R$, shaded blue, corresponds to a subset of the spins, shaded red.  To produce a localized symmetry operator we take the product over the $U_i(g)$ associated to the red spins.}\label{spinfig}
\efig  
There is a nice ``information-theoretic'' proof of this theorem, but since the method is a bit far from the rest of this paper we relegate it to appendix \ref{splitapp}.  To see how this theorem relates to splittability, consider a spin system whose Hilbert space is the tensor product of a bunch of individual spins.  We can imagine the spins are arranged in a lattice, as in figure \ref{spinfig}.  By theorem \ref{latticethm}, any symmetry operator $U(g,\Sigma)$ which acts locally on the spins can be decomposed as $U(g,\Sigma)=\prod_i U_i(g)$, with $i$ labelling the spins and $U_i(g)$ acting nontrivially only on spin $i$.  So then we may simply define
\be
U(g,R)\equiv \prod_{i\in R}U_i(g),
\ee
which clearly has the property that it acts in the same way as $U(g,\Sigma)$ on operators with support only in $R$, while it acts trivially on operators with support only on the complement of $R$.  In figure \ref{spinfig}, the included tensor factors live at the red dots.  At least to the extent that this lattice model is a good model for quantum field theory, we should expect all symmetries to be splittable.  

In attempting to generalize theorem \ref{latticethm} to continuum quantum field theory, we immediately encounter the problem that the Hilbert space of a quantum field theory never has the tensor product structure assumed in theorem \ref{latticethm}: any finite-energy state will have an infinite amount of spatial entanglement between the fields in a region $R$ and those in its complement $\Sigma-R$.  This may seem decisive against proving the splittability of global symmetries along these lines, but in fact there is a standard axiom in algebraic quantum field theory which allows this lattice argument to be generalized to the continuum.  This axiom gives a clever way to extend the notion of a tensor product structure of the Hilbert space to continuum quantum field theory, and is given as follows \cite{Buchholz:1973bk,Buchholz:1986dy,Fewster:2016mzz}:
\begin{mydef}
A quantum field theory is said to have the \textit{split property on} $\Sigma$ if for any two open regions of bounded size $R$, $R'\subset \Sigma$ which obey $\mathrm{Closure}[R]\subset\mathrm{Interior}[R']$, there exists a von Neumann algebra $\mathcal{N}$, which is a type I factor, such that
\be
\mathcal{A}[R]\subset \mathcal{N} \subset \mathcal{A}[R'].
\ee
Here $\mathcal{A}[R]$, $\mathcal{A}[R']$ are the algebras of operators in $R$ and $R'$ respectively. 
\end{mydef}
A type I factor algebra, which is a von Neumann algebra with trivial center and containing a minimal projection, is always isomorphic to the set of all the operators on some Hilbert space (see eg. \cite{Jones}), so we can view the split property as saying that, although the Hilbert space does not factorize based on spatial regions (in fact the algebra $\mathcal{A}[R]$ is expected to be type III for any nontrivial $R$), by gradually ``thinning out'' the algebra between $R$ and $R'$ we can find a tensor factor whose operator algebra contains all the (bounded) operators on $R$ and none of the operators on the complement of $R'$.  Given a quantum field theory obeying the split property on $\Sigma$, it can be argued fairly straighforwardly that any global symmetry is splittable on $\Sigma$ \cite{Doplicher:1982cv,Doplicher:1983if,Buchholz:1985ii}, basically along the lines of theorem \ref{latticethm}.

Is the split property actually true in quantum field theory?  It has been shown explicitly in various free theories with $\Sigma=\Rdd$ \cite{Buchholz:1973bk,DAntoni:1983qsc,Buchholz:1987dr}, and also in certain interacting theories with $\Sigma=\mathbb{R}$ \cite{Summers:1982np}, and there are general arguments for it based on the notion that the energy spectrum of the theory quantized on $\Sigma=\mathbb{R}^{d-1}$ should be ``well-behaved'' in a technical sense which is called \textit{nuclearity} \cite{Buchholz:1986dy,Buchholz:1986bg}.  We are not aware of any quantum field theory that does not obey the split property on $\Sigma=\mathbb{R}^{d-1}$.  The situation is more subtle for quantum field theories on manifolds with nontrivial topology, we will see in the following section that there are reasonable quantum field theories which do \textit{not} obey the split property on more complicated spatial topologies.  And moreover we will see that in these theories we can indeed have symmetries which are not splittable on those topologies!  It may seem that a failure of splittability on nontrivial manifolds is of relatively obscure technical interest, but we emphasize that if the symmetry group is continuous, then this must imply the non-existence of a Noether current; if one existed we could use it to construct $U(g,R)$ for any region $R$ on any spatial manifold $\Sigma$ using equation \eqref{currentU}.  We believe that these observations are unknown in the algebraic quantum field theory literature, which has focused almost exclusively on spatial $\Rdd$ (see however \cite{Buchholz:2015epa,Buchholz:2016yqp,Buchholz:2018npi} for recent work which is somewhat related).   

Splittability on spatial $\Rdd$ is not quite sufficient for our purposes in AdS/CFT, where we will want to use it on spatial $ \mathbb{S}^{d-1}$.  We have not attempted to prove this splittability using the energetic arguments of \cite{Buchholz:1986dy,Buchholz:1986bg}, but based on our study of examples we expect that it should follow for $d>2$ from splittability on spatial $\Rdd$.  In conformal field theory however we can do better: there for $d\geq 2$ we can argue that a symmetry which is splittable on spatial $\mathbb{R}^{d-1}$ must always be splittable on $\mathbb{S}^{d-1}$.  This is because we can use the state-operator correspondence to explicitly define the matrix elements of $U(g,R)$ on $\mathbb{S}^{d-1}$ in terms of its matrix elements on $\mathbb{R}^{d-1}$.  This will be enough for our quantum gravity arguments below, but as splittability and Noether's theorem are interesting on their own as issues in quantum field theory, we will now study them a bit further, focusing on the question of what modification of the naive Noether conjecture \eqref{naiveN} would be necessary to obtain a true statement with no counterexamples.  We aim to motivate a general picture where non-pathological quantum field theories which do not obey the split property on some spatial manifold $\Sigma$ should be deformable to ones that do obey it for any $\Sigma$ by adding a finite number of arbitrarily massive degrees of freedom, and that in such theories the Noether conjecture should hold.

\subsection{Unsplittable theories and continuous symmetries without currents}\label{splitexsec}
How might we obtain a quantum field theory that does not obey the split property?  Any theory which is obtained from a lattice theory with a tensor product structure, like that in figure \ref{spinfig}, seems likely to obey the split property in the continuum limit.  But what if even in the lattice theory we do not have this tensor product structure?  For example we could have a theory whose Hilbert space is obtained by imposing local constraints on a tensor product theory, e.g. a lattice gauge theory.  We do not have a complete understanding of which lattice theories have continuum limits obeying the split property and which do not, nor for that matter do we expect that all continuum QFTs have lattice formulations, but with this motivation we can construct a few examples of unsplittable symmetries which clarify the issue and motivate the general picture we conjectured at the end of the previous subsection.  These examples may seem contrived, since they rely on noncompact gauge groups and/or decoupled free theories. In subsection \ref{anomsplit} we will give two interacting examples based on the ABJ anomaly, which basically work in the same way as our examples here. Unsplittable discrete global symmetries are easily obtained in theories with compact gauge group, we will already meet one in this subsection, but a noncompact gauge group seems hard to avoid if we want to produce an unsplittable continuous global symmetry.  We will comment on why this is so at the end of this subsection.

The simplest gauge theory with a continuous global symmetry is a pure gauge theory with gauge group $\mathbb{R}\times\mathbb{R}$:
\be
S=-\frac{1}{4}\int_M d^d x \sqrt{-g}F_{a\mu\nu}F_b^{\mu\nu}\delta^{ab}=-\frac{1}{2}\int_M F_a\wedge \star F_b \delta^{ab}. 
\ee
Here $a, b=1,2$, and there is a $U(1)$ global symmetry which rotates the two gauge fields into each other.  This theory provably obeys the split property on $\mathbb{R}^d$ \cite{Buchholz:1987dr}, but we will see that it does not on more general manifolds and moreover we will see that this symmetry is itself not splittable on those manifolds.  There must therefore be something wrong with the Noether current for this symmetry.  The Noether procedure outlined around equation \eqref{currentdef} gives a Noether current which in differential form notation is
\be\label{Jeq}
\star J=\epsilon^{ab} A_a\wedge \star F_b,
\ee
with
\be
\epsilon^{ab}=\begin{pmatrix} 0 & 1 \\ -1 & 0\end{pmatrix}.
\ee
We see however that under a gauge transformation
\be
A_a'=A_a+d\lambda_a,
\ee
we have 
\be\label{Jtransform}
\star J'=\star J+\epsilon^{ab}d\lambda_a\wedge\star F_b =\star J+d\left(\epsilon^{ab} \lambda_a\star F_b\right),
\ee
where in the second equality we have used the equation of motion $d\star F_a=0$.  The current constructed by the Noether procedure is not gauge-invariant!  It is however gauge-invariant up to a total derivative, so if we integrate it over a closed manifold $\Sigma$ we get a well-defined charge
\be
Q(\Sigma)\equiv \int_\Sigma \star J.
\ee
The gauge non-invariance of $J$ is a potential obstruction to any attempt to define localized symmetry operators $U(g,R)$. For example if we define a localized charge
\be
Q(R)\equiv \int_R \star J,
\ee
then apparently we have the gauge transformation
\be\label{Qinv}
Q(R)'=Q(R)+\epsilon^{ab}\int_{\partial R}\lambda_a\star F_b .
\ee
How are we to reconcile this with the known splittability \cite{Buchholz:1987dr} of this theory on $\mathbb{R}^d$?

One useful observation is that, although $Q(R)$ is not gauge invariant, its gauge non-invariance is restricted to an operator supported only at $\partial R$.  Our definition of splittability left it ambiguous how $Q(R)$ should act on operators right at $\partial R$, so we might hope that we can modify $Q(R)$ by a gauge non-invariant boundary operator in just such a way that we cancel the gauge non-invariance in equation \eqref{Qinv}.  We now argue that indeed this can be done provided that the boundary is connected, and more generally that it can be done provided that each connected component of the boundary is itself a boundary. Let us first consider the case where $\partial R$ is connected.  We may then define the non-local operator
\be\label{Idef}
I_a(x)\equiv \int_{\gamma_{x,x_0}}A_a,
\ee
where for each $x\in \partial R$ we have arbitrarily chosen a curve $\gamma_{x,x_0}$ in $\partial R$ which connects that point to a fixed reference point $x_0$.  This operator has gauge transformation
\be
I'_a=I_a+\lambda_a(x)-\lambda_a(x_0).
\ee
We may then easily see that the ``doubly-nonlocal'' boundary operator
\be
C[\partial R]\equiv \epsilon^{ab}\int_{\partial R}I_a\star F_b
\ee
has gauge transformation
\be\label{Cinv}
C'[\partial R]=C[\partial R]+\epsilon^{ab}\int_{\partial R} \lambda_a \star F_b ,
\ee
where we've used that
\be\label{starFzero}\epsilon^{ab}
\int_{\partial R}\lambda_a(x_0)\star F_b=\lambda_a(x_0)\epsilon^{ab}\int_R d \star F_b=0.
\ee
But \eqref{Cinv} is precisely what we need to cancel the gauge transformation in \eqref{Qinv}, so apparently the quantity
\be\label{tQeq}
\wt{Q}(R)\equiv Q(R)-C[\partial R]
\ee
is gauge invariant!  We may then define 
\be\label{giUR}
U(\theta,R)\equiv e^{i\theta \wt{Q}(R)},
\ee
which give a set of local symmetry generators which split the symmetry.  More generally, if each connected component of the boundary is itself a boundary, we can pick an $x_0$ for each component and \eqref{starFzero} will hold component by component.  In particular if $M$ has the property that \textit{every} closed $d-2$ manifold is the boundary of some $d-1$ manifold, or in other words the homology group $H_{d-2}(M)$ is trivial, then this symmetry will be splittable for any choice of $R$.  This is indeed the case for $\mathbb{R}^d$, so there is no tension with the proof of the split property there.\footnote{This is a bit subtle for $d=2$, since in order for a single point to be a boundary it needs to be attached to a line which goes off to infinity.}  Note also that for $M=\mathbb{R}\times\mathbb{S}^{d-1}$, which is our case of primary interest, we have $H_{d-2}(M)=0$ for $d>2$.

\bfig
\includegraphics[height=5cm]{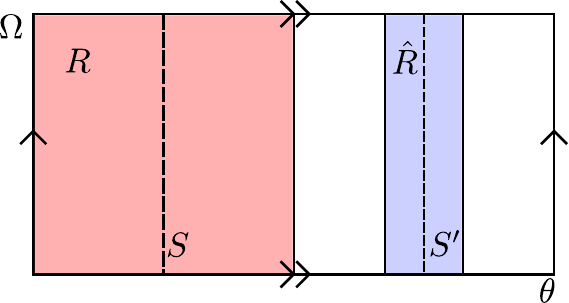}
\caption{A counterexample to the split property: electrodynamics on a spatial torus.  The flux operator through $S$ is equal to the flux operator through $S'$, but they live in spacelike-separated regions $R$ and $\hat{R}$. }\label{torusfig}
\efig
The reader may wonder why we did not first attempt to ``improve'' the current \eqref{Jeq}, by adding to $\star J$ a local gauge non-invariant total derivative whose gauge transformation would cancel the non-invariance of $\star J$.  It is easy to see however that there is no candidate which will succeed: such a term would need to have a gauge transformation involving $\lambda_a$ without any derivatives, but no local polynomial function of $A$ and $F$, or their derivatives, will have this property.  This indeed happens for a good reason: on more complicated manifolds this theory does not obey the split property, and the symmetry we have been considering is not splittable!  For concreteness consider quantizing this theory on spatial manifold $\Sigma =\mathbb{S}^1\times \mathbb{S}^{d-2}$, parametrized by $(\theta,\Omega)$, and consider the region $R$ given by $0<\theta<\pi/2$.  See figure \ref{torusfig} for the setup for $d=3$.  The algebra of this region includes the electric flux operator 
\be
\Phi_a(S)=\int_{S} \star F_a,
\ee 
where $S$ is the spatial $\mathbb{S}^{d-2}$ at $\theta=\pi/4$.  $\Phi_a(S)$ is a nontrivial operator since it does not commute with a Wilson loop that wraps the $\mathbb{S}^1$.  But in fact by Gauss's law, $d\star F_a=0$, $\Phi_a(S)$ depends only on the homology class of $S$: in particular since $S$ is homologous to the spatial $\mathbb{S}^{d-2}$ at $\theta=3\pi/4$, which we'll call $S'$, $\Phi_a(S)$ is also in the algebra of a region $\hat{R}$ which is spacelike-separated from $R$ (see figure \ref{torusfig}).   Therefore $\Phi_a(S)$ must commute with all elements of $\mathcal{A}[R]$, and thus must be in the center of $\mathcal{A}[R]$.  Now say that the split property held: for any region $R'$ whose interior contains the closure of $R$, we should be able to have the algebraic inclusion
\be\label{inclusion1}
\mathcal{A}[R]\subset \mathcal{N}\subset\mathcal{A}[R']
\ee   
with $\mathcal{N}$ some type I factor.  In particular consider $R'$ to be defined by $-\epsilon<\theta<\pi/2+\epsilon$ with say $\epsilon=.01$.  $\Phi_a(S)$ is an element of $\mathcal{A}[R]$, and thus an element of $\mathcal{N}$.  But since $R'$ is spacelike-separated from $\hat{R}$, $\Phi_a(S)$ is also in the center of $\mathcal{A}[R']$, and therefore by \eqref{inclusion1} must commute with everything in $\mathcal{N}$.  But since $\Phi_a(S)$ is nontrivial, this contradicts the notion that $\mathcal{N}$ is a type I factor: any factor has trivial center by definition.  Thus we cannot have \eqref{inclusion1}, so the split property fails.

A few comments are in order here.  First of all this argument for non-splittability holds also for pure $U(1)$ gauge theory, which thus also does not obey the split property on general manifolds.  Second, the trouble we found is consistent with our inability to define $U(g,R)$ for regions where $\partial R$ has connected components which are not themselves boundaries: indeed it is precisely such components which allow $\Phi_a(S)$ to be nontrivial.  Third, we note that not only does the split property \eqref{inclusion1} fail, it is clear that for the $\mathbb{R}\times\mathbb{R}$ gauge theory the $U(1)$ symmetry rotating the gauge fields really cannot be splittable on this geometry in the sense of definition \ref{splitdef}.  For if it were, then $U(g,R)$ would have to act nontrivially on the $a$ index of $\Phi_a(S)$, but this is impossible since $\Phi_a(S)\in \mathcal{A}[\hat{R}]$.  Therefore it indeed must be the case that no gauge-invariant current exists.  Finally we note that, although we had to go to nontrivial spatial topology to see a break down of splittability, this breakdown actually has an avatar even in the theory on spatial $\Rdd$.  Consider a circular Wilson loop in $\Rd$, which is surrounded by a surface with topology $\mathbb{S}^1\times \mathbb{S}^{d-2}$ on which we put a symmetry insertion, constructed as in figure \ref{gluingfig}.  For $d=3$, this would amount to routing a Wilson loop through the ``bagel'' which is bounded by the torus in figure \ref{gluingfig}.  This surface insertion \textit{is} splittable into the two pieces shown in figure \ref{gluingfig}, but it is \textit{not} splittable into two ``handles'' such as the shaded red region in figure \ref{torusfig} and its complement.  This non-splittability has no interpretation as an operator statement in the Hilbert space on $\Rdd$, but it is a nontrivial statement about the insertion.  

The reader may worry that this example of a non-splittable global symmetry is pathological since it has a noncompact gauge group.  But we note that all the same arguments apply to the $\mathbb{Z}_2$ global symmetry of a pure gauge theory with gauge group $U(1)\times U(1)$.\footnote{The reason that this theory no longer has a continuous global symmetry mixing the two gauge fields is that such a symmetry would not act locally on the Wilson loops, since it wouldn't respect charge quantization.  It therefore would violate part (b) of definition \ref{globaldef}, since it would map the Wilson loop out of $\mathcal{A}[R]$, where $R$ is a thin tube containing the Wilson loop.}  We no longer expect a current, but we still have a symmetry operator
\be
U(-1,\Sigma)\equiv e^{i\pi \epsilon^{ab}\int_\Sigma A_a \wedge\star F_b}
\ee
under which the exponentiated $U(1)$ electric flux 
\be
L_a(\theta,S) \equiv e^{i\theta \Phi_a(S)}
\ee
transforms via $L_1(\theta,S)\leftrightarrow L_2(\theta,S)$.  $U(-1,\Sigma)$ can still be split when $M$ has vanishing $H_{d-2}(M)$, but on $\Sigma=\mathbb{S}^1\times \mathbb{S}^{d-2}$ it cannot be split for the same reason as in the noncompact case: any $U(-1,R)$ on a spatial $\mathbb{S}^1\times \mathbb{S}^{d-2}$ would have to act both trivially and nontrivially on $L_a(\theta,S)=L_a(\theta,S')$. Unfortunately this is no longer a counterexample to the naive Noether conjecture \eqref{naiveN}, since the global symmetry is now discrete.  It also still involves two decoupled free theories: we can remove one of them if we instead consider the discrete symmetry $A'=-A$, also called charge conjugation, of one $U(1)$ gauge field, which is also not splittable for the same reasons.  We give an example which is not free in subsection \ref{anomsplit}.

\bfig
\includegraphics[height=4cm]{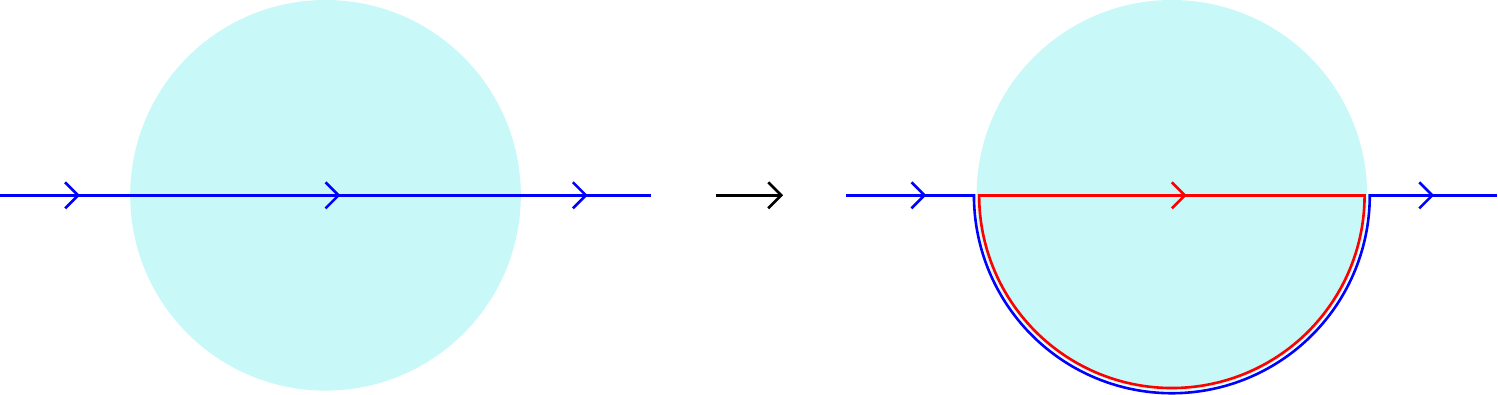}
\caption{Re-routing unbreakable lines.  Here we have a symmetry exchanging blue and red lines, and we can arrange for it to act locally in the shaded region by rerouting the blue line around the boundary of the region.  This is not possible however when the region has multiple boundary components which are not contractible, for example as in figure \ref{torusfig}.}\label{reroutefig}
\efig
The source of trouble in all these examples (and those of section \ref{anomsplit}) is that there are ``unbreakable lines'': line operators, here Wilson lines, which cannot have endpoints on local operators carrying gauge charge since none exist.  In more modern language, there is an exact one-form symmetry under which these lines are charged (we will discuss $p$-form symmetries in more detail in section \ref{psec}).  This notion of unbreakable lines gives us a new geometric interpretation of what our boundary modification \eqref{tQeq} of the charge in the $\mathbb{R}\times\mathbb{R}$ gauge theory (or the corresponding modification in the $U(1)\times U(1)$ gauge theory) is doing: it enables us to ``re-route'' Wilson lines around the boundary in a manner consistent with the unbreakable nature of the lines.  We illustrate this in figure \ref{reroutefig}.  The breakdown of splittability on manifolds with nontrivial $H_{d-2}(M)$ can then be understood as arising from an inability to perform this re-routing. 

It is interesting to consider to what extent the validity of the split property is a ``UV-sensitive'' property of a quantum field theory.  As a concrete example, we point out that our $U(1)\times U(1)$ gauge theory in $d$ spacetime dimensions can be obtained as the IR limit of two copies of a lattice version of the $\mathbb{CP}^{N-1}$ nonlinear $\sigma$-model \cite{Harlow:2015lma}.  This lattice theory has precisely the tensor product Hilbert structure shown in figure \ref{spinfig}, so we might expect that it should obey the split property.  So how did we get a theory in the IR that does not?  In fact what happened is that this lattice theory also has massive charged particles, whose masses can be small compared to the lattice energy but large compared to any other IR scale.  Once these massive charged particles are included, the Wilson lines are no longer unbreakable and a new possibility for constructing localized symmetry operators arises where we snip the ends of the Wilson lines using the charges.  We illustrate this in figure \ref{insertfig}.  This is possible no matter how heavy the charges are, and we only need a finite number of them.  So apparently our $U(1)\times U(1)$ counterexample to splittability can be fixed with a simple UV modification: we just add some heavy charges. This modification necessarily destroys the one-form symmetry which prevented the Wilson lines from being broken.  A similar fix does not seem to be possible for the $\mathbb{R}\times\mathbb{R}$ theory, which we after all expect to be more pathological.  Essentially the problem there is that the unbreakable lines are ``infinitely generated'': since the Wilson line can carry any real charge, cutting all these lines with a finite number of heavy fields is too much to ask for.  In one-form symmetry language, the one-form symmetry is noncompact.  More generally, we conjecture that in theories where the only topological surface operators are compact $p$-form symmetries, a finite UV modification which restores the split property on any manifold should be always be possible.

\bfig
\includegraphics[height=4cm]{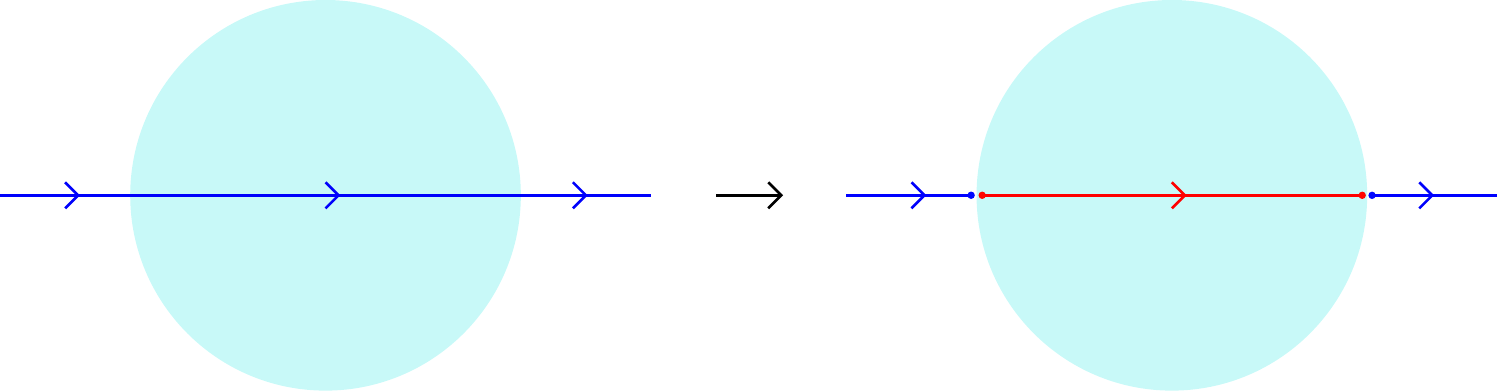}
\caption{Exchanging breakable line segments using charges.}\label{insertfig}
\efig
Finally we return to the question of when the naive Noether conjecture holds.  It is interesting to consider what happens if we try to extract a gauge-invariant current from the gauge-invariant $U(g,R)$ constructed in \eqref{giUR} in the $\mathbb{R}\times \mathbb{R}$ gauge theory on $\mathbb{R}^d$.  The obvious way to do this is to take $g\to 1$ and $R$ to be perturbatively small, and then attempt to extract $J^0$ from the part of $\log U(g,R)$ that scales with the volume of $R$ (see \cite{Carpi:1999yz,Morsella:2008pw} for rigorous attempts to do this in a few simpler theories). But this procedure actually fails in our example due to a non-decoupling of the boundary modification in this limit: this is why the algebraic quantum field theory literature was never able to actually extract a current from their $U(g,R)$, even though they assumed the split property on $\mathbb{R}^d$ \cite{Buchholz:1985ii}.  This failure arises in the following way: taking the exterior derivative of $I_a(x)$ with respect to $x^\mu$, we have
\be
dI_a=A_{a}+P_{a},
\ee
with
\be
P_{a,\mu}\equiv \int_{\gamma_{x,x_0}}ds \frac{d\gamma^\alpha}{ds}F_{a,\alpha\beta}v^\beta_\mu(s),
\ee
where $v^\beta_\mu$(s) is a somewhat unsual object with a tensor index $\beta$ at point $\gamma_{x,x_0}(s)$ and a tensor index $\mu$ at point $x$, which keeps track of how the curve $\gamma_{x,x_0}$ varies with $x$. We therefore have
\begin{align}\nonumber
\wt{Q}(R)=&\epsilon^{ab}\int_R \left(A_a\wedge \star F_b-d(I_a\wedge \star F_b)\right)\\\nonumber
=&\epsilon^{ab}\int_R \left(A_a\wedge \star F_b-dI_a\wedge \star F_b\right)\\
=&-\epsilon^{ab}\int_R P_{a}\wedge \star F_b,
\end{align}
which scales to zero faster than the volume as we shrink $R$.

We are thus led to the following suggestion: perhaps if we restrict to quantum field theories  which obey the split property on \textit{any} manifold, it is actually possible to construct a Noether current for any continuous global symmetry.  The boundary action in figure \ref{insertfig} seems less severe to us than the boundary action in figure \ref{reroutefig}, so we are optimistic that one might be able to show the necessarily decoupling.  More generally we expect that what is really needed is just that some UV modification of the theory is possible which restores the split property on all manifolds: the existence of the current cannot depend on such modifications since it is an object in the IR theory.  We therefore expect that the naive Noether conjecture should hold provided that all topological surface operators are associated to compact $p$-form global symmetries.  This then would explain why we have only been able to find counterexamples with noncompact gauge groups: it is only these which can lead to noncompact higher-form symmetries.  We view this line of thought as a promising avenue for at long last giving an abstract formulation of Noether's theorem, but we will not attempt this here.

\subsection{Background gauge fields}\label{backgroundsec}
Given a quantum field theory with a global symmetry, a natural operation to consider is turning on a background gauge field for that global symmetry.  One example of this which we have already discussed is studying the theory on a nontrivial spacetime geometry $\mathbb{R}\times\Sigma$, which can be interpreted as turning on a background gauge field for Poincare symmetry, a spacetime global symmetry.  We now discuss background gauge fields for internal global symmetries.\footnote{This section, and the following two, can be viewed as a further side discussion.  Holography-minded readers who are simply willing to accept that all CFT global symmetries are preserved on $\mathbb{R}\times\mathbb{S}^{d-1}$, and that it is possible to turn on topologically-nontrivial background gauge fields for global symmetries, may wish to skip ahead to section \ref{gaugesec}.}   We will see immediately that turning on a background gauge field for a continuous symmetry requires us to assume that a Noether current exists, which then implies that the symmetry must be splittable.  A condition slightly weaker than splittability might be sufficient for turning on a background gauge field for a discrete symmetry, but for simplicity we will just assume splittability regardless; after all we have just argued that in reasonable quantum field theories we can always achieve it by a short-distance modification of the theory. 

For a continuous global symmetry group $G$ with a set of Noether currents $J^\mu_a$, one way to turn on a background gauge field is to add to the action a term of the form
\be\label{connection1pot}
\delta S=\int_M d^d x \sqrt{-g} A_\mu^a(x)\left(J^\mu_a(x)+\ldots\right)=\int_M A^a\wedge \left(\star J_a+\ldots\right),
\ee
where the background gauge field $A_\mu^a(x)$ is an arbitrary real one-form with an index $a$, whose range equals the dimensionality of the Lie algebra $\mathfrak{g}$ of $G$.  ``$\ldots$'' denotes local terms that are higher order in $A_\mu^a$.  As in our discussion of extending flat-space operators to curved space, there is in general some ambiguity in how we choose these higher order terms.  Given such a choice however, we may then define an extension of the Noether current in the presence of a background gauge field:
\be\label{Jdef}
\wt{J}^{\mu}_a(x)\equiv \frac{\delta\left(\delta S\right)}{\delta A^a_\mu(x)}=J^\mu_a(x)+\ldots.
\ee
We can restate this procedure in a non-Lagrangian way as a definition of a new set of ``unnormalized expectation values in the presence of $A_\mu^a$'', given by\footnote{Here $T$ denotes time-ordering and $\lan \cdot\ran$ denotes the expectation value in the vacuum state of the undeformed theory on $M=\mathbb{R}\times \Sigma$.  In general an $i\epsilon$ prescription is necessary to get a well-defined expectation value.}
\be
\lan T \mO_1\ldots \mO_n\ran_A\equiv \lan T\mO_1\ldots \mO_n e^{i\delta S}\ran.
\ee
We will be especially interested in the unnormalized expectation value of the unit operator, usually called the partition function in the presence of the background gauge field $A$:
\be\label{Zdef}
Z[A]\equiv \lan 1\ran_A=\lan Te^{i\int_M d^d x \sqrt{-g} A_\mu^a(x)\left(J^\mu_a(x)+\ldots\right)}\ran.
\ee
It should be understood here that if we view $Z$ as a map to the complex numbers, its domain allows background gauge fields for all (internal) global symmetries of the theory.  We note also that it is often convenient to consider the formal Euclidean path integral version of this quantity, 
\be
Z[A]\equiv \lan e^{-\int_M d^d x \sqrt{-g} A_\mu^a(x)\left(J^\mu_a(x)+\ldots\right)}\ran,
\ee
where now $M$ is any Riemannian manifold, perhaps requiring a spin (or pin) structure if the theory has fermionic operators. 

Background gauge fields of the form \eqref{connection1pot} are not the most general kind of background gauge fields.  In particular if $G$ is discrete, then \eqref{connection1pot} is nonsensical.  The modern notion of a gauge field configuration is formalized as a \textit{connection on a principal bundle}.  The basic idea is that we cover the spacetime manifold $M$ with a collection of open patches $U_i$, on each of which we define a ``local gauge potential'', $A_{i,\mu}$, which is a one-form taking values in the Lie algebra $\mathfrak{g}$ of $G$.  If there is a single $U$ covering all of $M$, then we revert to \eqref{connection1pot}, where $A_\mu=A_\mu^a T_a$ with  $T_a$ some basis for $\mathfrak{g}$.  We then demand that for all intersections $U_i\cap U_j$, there exist ``transition functions''
\be
g_{ij}:U_i\cap U_j\rightarrow G,
\ee
obeying
\begin{align}\nonumber
g_{ji}&=g_{ij}^{-1}\\
g_{ij}g_{jk}|_{U_i\cap U_j \cap U_k}&=g_{ik}|_{U_i\cap U_j \cap U_k},\label{gtriple}
\end{align}
such that for any $i, j$ we have\footnote{If $G$ is a matrix group then this equation makes sense as written, otherwise we define $g_{ij}A_{j,\mu}(x)g_{ij}^{-1}$ to be the pushforward of $A_{j,\mu}(x)$, viewed as a vector field on $G$, by the adjoint map $Ad_g:h\mapsto g h g^{-1}$, and we define $ -\partial_\mu g_{ij}g_{ij}^{-1}$ to be the pullback by $g_{ij}^{-1}:U_{i}\cap U_j\mapsto G$ of the Maurer-Cartan form on $G$.}
\be\label{gtrans}
A_{i,\mu}=g_{ij}A_{j,\mu}g_{ij}^{-1}-i \partial_\mu g_{ij}g_{ij}^{-1}
\ee
in $U_i\cap U_j$.  For a discrete group we must have $A_{i,\mu}=0$ in all patches, so the data of the background gauge field is just the transition functions $g_{ij}$.

Two such collections of patches and local gauge potentials, $\left(U'_{i'},A'_{i',\mu}\right)$ and $\left(U_{i},A_{i,\mu}\right)$,  are said to be \textit{gauge equivalent} if their union is ``compatible'' in the sense that there exist an additional set of transition functions $g_{ij'}$ such that together with the $g_{ij}$ and $g_{i'j'}$ they obey \eqref{gtriple}, \eqref{gtrans} for all $ij$, $ij'$, $i'j'$ pairs.  An interesting special case of such an equivalence arises when we take the $U_i$ and $U_{i'}$ to coincide, in which case gauge equivalence means the existence of a set of local gauge transformations
\be
g_i:U_i\rightarrow G
\ee
such that
\begin{align}\nonumber
A'_{i\mu}&=g_i A_{i,\mu}g_i^{-1}-i \partial_\mu g_{i}g_i^{-1}\\
g'_{ij}&=g_ig_{ij}g_{j}^{-1}.\label{patchtransform}
\end{align}
This special case is important because in fact any fixed set of contractible $U_i$ which cover $M$ are sufficient to construct a representative of every equivalence class of background gauge fields on $M$ by choosing appropriate $g_{ij}$ and $A_{i,\mu}$.\footnote{This statement is not obvious, it follows from a nontrivial theorem that there can be no nontrivial fiber bundle over a contractible base \cite{steenrod1951topology}.}  In mathematical terms the transition functions $g_{ij}$ modulo gauge equivalence define a principal $G$ bundle over $M$, while the local gauge potentials $A_{i,\mu}$ modulo gauge equivalence define a connection on that bundle.  A background gauge field which is gauge equivalent to one defined using a single patch $U=M$ is called \textit{topologically trivial}. 

\bfig
\includegraphics[height=4cm]{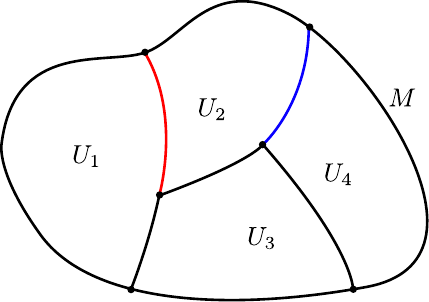}
\caption{Tiling the spacetime manifold with contractible patches, in order to turn on a general background gauge field.  The intersections $C_{ij}^{\{m\}}$ described in the text are the line segments between the dots, for example $C_{12}^{\{1\}}$ is shaded red and $C_{24}^{\{1\}}$ is shaded blue.}\label{tilefig}
\efig
Turning on a general background gauge field, possibly topologically nontrivial, for an internal global symmetry is a delicate process.  We are not aware of a standard discussion of how to do this for general $G$ in the literature, the closest we found is some comments in \cite{Gaiotto:2014kfa}.  Here we give a somewhat heuristic picture of how this can be done, expanding on the comments in \cite{Gaiotto:2014kfa}.  The basic idea is to cover $M$ with contractible patches, and then ``shrink'' the patches so that they give a tiling of $M$ via a set of \textit{closed} $U_i$ which overlap only at their boundaries.  This is illustrated in figure \ref{tilefig}.  We then define the partition function in the presence of a background gauge field (for simplicity giving the formula in Euclidean signature to avoid issues of time-ordering)
\be\label{genback}
Z[A]\equiv \lan e^{-\sum_i \int_{U_i}d^dx \sqrt{g}A_{i\mu}^a \left(J_a^\mu+\ldots\right)}\prod_{(ij)}\wt{U}_{ij}\ran,
\ee
where $A$ now stands in for the collection $\left(U_i,A_{i\mu}\right)$, $(ij)$ counts each $ij$ pair once, and the ``transition unitaries'' $\wt{U}_{ij}$ are defined via the following procedure.  First split each intersection $U_i\cap U_j$ into its connected components $C_{ij}^{\{m\}}$, on each of which we can write $g_{ij}$ as the product of a constant map $g_{ij}^{\{m\}}$ and a map whose target space is the identity component of $G$:
\be
g_{ij}(x)|_{C_{ij}^m}=g_{ij}^{\{m\}} e^{i\epsilon^a(x)T_a}.
\ee
We then define
\be\label{genback1}
\wt{U}_{ij}=\prod_m U\left(g_{ij}^{\{m\}},C_{ij}^{\{m\}}\right)\exp{\left(i \int_{C_{ij}^{\{m\}}}\epsilon^a \star J_a\right)},
\ee
where $U(g,R)$ is the codimension-one surface with boundary insertion guaranteed to exist by splittability of the global symmetry (which we here assume), and the normal vector used in defining the orientation of $C_{ij}^{\{m\}}$ is chosen to point from $i$ to $j$.\footnote{For continuous global symmetries, splittability is clearly necessary to turn on a background gauge field since a current is.  For discrete global symmetries it does not seem to be: a weaker sufficient assumption is that the junctions in figure \ref{tilefig} exist.  This follows from splittability, but is not obviously equivalent to it: due to the triple overlap condition \eqref{gtriple}, we only need junctions where the product of the $g_{ij}$ around the junction is the identity.}  The ambiguity of $U(g,R)$ at $\partial R$ means that there may be some ambiguity at the dots in figure \ref{tilefig}.   As a simple example of turning on a topologically nontrivial background gauge field, consider a theory with a $\mathbb{Z}_2$ global symmetry on the Euclidean spacetime manifold $\mathbb{S}^1\times \mathbb{R}^{d-1}$.  We can define a partition function in a nontrivial background gauge field for which there is a $-1$ holonomy around the $\mathbb{S}^1$ by evaluating the Euclidean path integral
\be
Z[A]=\lan U\left(-1,R_\theta\right)\ran,
\ee
where $R_\theta$ denotes the codimension-one submanifold at fixed angle $\theta$ on $\mathbb{S}^1$.

It is interesting to ask what happens to correlation functions of charged operators in the background defined by eq. \eqref{genback}: instead of being continuous functions on $M$, as we move from $U_i$ to $U_j$ they encounter $\wt{U}_{ij}$ and thus jump via\footnote{Note here that $i$ and $j$ label patches, the indices for the matrix multiplication in equation \eqref{Dmap} are here suppressed.}
\be
\mO_i=D\left(g_{ij}\right) \mO_j.
\ee
Geometrically this is described by saying that the operators are \textit{sections} of a vector bundle associated to the principal bundle defined by the $g_{ij}$.

\subsection{'t Hooft anomalies}\label{anomalysec}
We have now defined the partition function $Z[A]$ of a quantum field theory with a global symmetry group in the presence of an arbitrary background gauge field.  But there were two potential sources of ambiguity in this definition: the choice of higher order terms in equation \eqref{connection1pot}, and the choice of how the intersections of boundaries in \eqref{genback1}, shown as dots in figure \ref{tilefig}, are regulated.  It would be nice to have some sort of principle to restrict these choices, and in fact there is a very natural choice: we can try to arrange so that the partition function $Z[A]$ depends only on the gauge equivalence class of the background gauge fields $A$, not on their patch-wise construction.  It turns out however that sometimes this is not possible \cite{tHooft:1979rat,Frishman:1980dq,Coleman:1982yg,AlvarezGaume:1983ig}:\footnote{The term ``'t Hooft anomaly'' is a modern invention \cite{Kapustin:2014lwa}, to distinguish 't Hooft anomalies from related phenomena which arise when we attempt to make some of the background gauge fields dynamical in a theory with an 't Hooft anomaly \cite{Adler:1969gk,Bell:1969ts,Gross:1972pv}. 't Hooft has also done famous work with these related phenomena \cite{tHooft:1976rip}, so the name is a bit unfortunate.}
\begin{mydef}
A quantum field theory has an \textit{'t Hooft anomaly} if there is no choice of higher order terms in equation \eqref{connection1pot} and regulation of boundary intersections in equation \eqref{genback1} such that $Z[A]$ is a gauge-invariant functional of the background gauge fields for all global symmetries.
\end{mydef}
In this definition we also allow $A$ to include background gauge fields for spacetime symmetries, namely studying the theory on a nontrivial spacetime manifold $M$ with a nontrivial metric $g$. We can cast 't Hooft anomalies in a more conventional light when $G$ is continuous by considering the effect of infinitesimal local gauge transformations $A'_{i\mu}=A_{i\mu}+D_\mu \epsilon_i(x)$ on the partition function \eqref{genback}.  Choosing $\epsilon_i$ to vanish at the boundary of $U_i$, we see that invariance of $Z[A]$ requires
\be\label{DJeq}
D_\mu \wt{J}^{\mu}_a\equiv\partial_\mu \wt{J}^\mu_a+C^{b}_{\phantom{b}ac}A^c_\mu \wt{J}_b^\mu=0,
\ee
where $\wt{J}_a^\mu$ was defined in \eqref{Jdef}.  Moreover if $\epsilon_i$ does not vanish at the boundary of $U_i$ then it will combine with the gauge transformation of $\wt{U}_{ij}$ such that $Z[A]$ is still gauge-invariant, at least up to possible issues at the edges.  Thus \eqref{DJeq} is a necessary condition to avoid an 't Hooft anomaly.\footnote{That it is not sufficient can be seen by the existence of ``non-infinitesimal'' 't Hooft anomalies such as those in discrete symmetries or the Witten anomaly in the $SU(2)$ global symmetry of an odd number of Majorana doublets \cite{Witten:1982fp}.}  

We emphasize that the presence of an 't Hooft anomaly is \textit{not} an inconsistency of a quantum field theory; there are many respectable quantum field theories with 't Hooft anomalies. For example consider the chiral anomaly of a free complex Dirac Fermion in $1+1$ dimensional Minkowski space:
\be
S=-i\int d^2 x \ol{\psi}\slashed{\partial}\psi.
\ee
  This theory has two $U(1)$ global symmetries, $\psi'=e^{i\theta}\psi$ and $\psi'=e^{i\theta \gamma^3}\psi$, with conserved currents\footnote{Note that $v$ and $p$ here are labels, not indices.  They stand for ``vector'' and ``pseudovector''.  }
\begin{align}
J_v^\mu&=-\ol{\psi}\gamma^\mu \psi\\
J_p^\mu&=-\ol{\psi}\gamma^\mu \gamma^3\psi,
\end{align}
and we can easily turn on background gauge fields for both:
\be
S=-i\int d^2 x \ol{\psi}\gamma^\mu \left(\partial_\mu-iA_\mu^v-iA_\mu^p\gamma^3\right)\psi.
\ee
A simple Feynman diagram calculation shows that, using dimensional regularization, these currents obey
\begin{align}\nonumber
\partial_\mu J^{\mu}_v&=\partial_\mu \wt{J}^{\mu}_v=0\\
\partial_\mu J^\mu_p&=\partial_\mu \wt{J}^\mu_p=-\frac{1}{2\pi}\epsilon^{\mu\nu}F^v_{\mu\nu},\label{dimregJ}
\end{align}
where $\epsilon^{\mu\nu}$ is antisymmetric with $\epsilon^{01}=-1$ (since we are in Lorentzian signature), $F^v_{\mu\nu}=\partial_\mu A^v_\nu-\partial_\nu A^v_\mu$, and we have used that there is no distinction between tilded and untilded currents since the action is linear in $A^v$ and $A^p$.
This nonconversation of $\wt{J}^\mu_p$ could be removed by modifying the action to include a term 
\be\label{contact}
\delta S=-\frac{1}{\pi}\int d^2 x \epsilon^{\mu\nu}A^v_\mu A^p_\nu,
\ee
which is an example of changing the $\ldots$ terms in \eqref{connection1pot}, but now the current conservation equations become
\begin{align}\nonumber
\partial_\mu\wt{J}^\mu_v&=\partial_\mu \left(J_v^\mu-\frac{1}{\pi}\epsilon^{\mu\nu}A^p_\nu\right)=-\frac{1}{2\pi}\epsilon^{\mu\nu}F^p_{\mu\nu}\\\label{altJcon}
\partial_\mu\wt{J}^\mu_p&=\partial_\mu \left(J_p^\mu+\frac{1}{\pi}\epsilon^{\mu\nu}A^v_\nu\right)=0,
\end{align}
so we have saved $J_p$ only at the expense of $J_v$.  Thus this theory has an 't Hooft anomaly: in the presence of background gauge fields we cannot maintain the gauge invariance of the partition function.  Note that when $A^v=A^p=0$, our modification \eqref{contact} does not affect correlation functions at finite separation but it does change the contact terms in the two-point functions of the currents; this is one manifestation of the ``short-distance'' nature of 't Hooft anomalies.  Different choices of regulator lead to different results for these contact terms, and indeed the contact terms for two different regulators will differ only by what is obtained by adding some local term such as \eqref{contact} to the action \cite{Closset:2012vp}.  If we stick to our original choice of dimensional regularization, which led to \eqref{dimregJ}, then from \eqref{Zdef} we see that the partition function transforms in the following manner:
\be\label{2dZ}
Z[A^v+d\Lambda^v,A^p+d\Lambda^p]=e^{\frac{i}{2\pi}\int d^2 x \Lambda^p \epsilon^{\mu\nu}F^v_{\mu\nu}}Z[A^v,A^p].
\ee

't Hooft anomalies have many important implications.  Perhaps the most obvious is that in a theory with an 't Hooft anomaly, it is not possible to consistently make all of the background gauge fields dynamical \cite{Gross:1972pv}.  This would be accomplished by integrating $Z[A]$ over gauge field configurations, perhaps weighted by additional gauge-invariant local terms, but if $Z[A]$ is not gauge-invariant then this leads to real inconsistencies such as violations of unitarity.  For example in the standard model of particle physics, since we want to introduce dynamical gauge fields for the $\left(SU(3)\times SU(2)\times U(1)\right)/\mathbb{Z}_6$ global symmetry of the ``un-gauged'' theory of quarks and leptons, it is essential that there is no 't Hooft anomaly in this symmetry \cite{Weinberg:1996kr}.\footnote{The gauge group of the standard model is most conservatively taken to be $\left(SU(3)\times SU(2)\times U(1)\right)/\mathbb{Z}_6$, since this is the group which acts faithfully on the known quarks and leptons.  This is not widely appreciated, but the logic is similar to that by which we assume that the gauge group of electromagnetism is $U(1)$ instead of $\mathbb{R}$: otherwise the observed quantization of charge would look like a conspiracy.  Future discoveries of more charged particles in new representations could change this situation however, so one can also say that we do not yet really know the gauge group of the standard model (see \cite{Tong:2017oea} for a recent discussion that takes this point of view).  We discuss this more in section \ref{gaugetopsec} below.}  

A less severe consequence of 't Hooft anomalies is that in the presence of background gauge fields, a global symmetry may be broken even if the currents for those background gauge fields are neutral under the symmetry \cite{tHooft:1976rip}.  For example $J_v$ and $J_p$ are both neutral under both of the global symmetries they generate, but nonetheless \eqref{dimregJ} tells us that $J_p$ is not conserved in the presence of a background gauge field for $J_v$.  We can rewrite \eqref{dimregJ} using differential forms as
\be
d\star J_p=\frac{1}{\pi} F^v,
\ee
which seems to immediately imply that the vector $U(1)$ charge
\be
Q_p\equiv \int_\Sigma \star J_p
\ee
is not conserved in the presence of this background field.  The truth however is more complicated: locally we have $F^v=d A^v$, so the quantity
\be\label{modQ}
\hat{Q}_p\equiv \int_\Sigma \left(\star J_p -\frac{1}{\pi}A^v\right)
\ee 
acts in the same way on all operators but appears to be conserved.  Indeed $\star J_p-\frac{1}{\pi}A^v$ is precisely the alternative current $\wt{J}_p$ which appeared in \eqref{altJcon}, and which was indeed covariantly conserved.  It is not mutually local with $J_v$ unless we similarly modify $J_v$ as in \eqref{altJcon}, which would lead to a nonconservation of $J_v$, but it might not seem like there is any problem with the charge $\hat{Q}_p$ defined by equation \eqref{modQ}.  In fact there is a problem with $\hat{Q}_p$, but it does not appear until we allow $A^v$ to be topologically nontrivial \cite{tHooft:1976rip,tHooft:1986ooh}.  First recall that boundary conditions which require all gauge-invariant operators to go to zero at infinity in $\mathbb{R}^2$ allow us to interpret the spacetime as being topologically $\mathbb{S}^2$, which can support topologically nontrivial $U(1)$ gauge field configurations \cite{Coleman:1978ae}.  One family of such configurations is the Wu-Yang monopoles \cite{Wu:1975es}
\begin{align}\nonumber
A_N&=\frac{n}{2}(1-\cos \theta)d\phi  \qquad 0\leq\theta\leq \pi/2\\
A_S&=-\frac{n}{2}(1+\cos \theta)d\phi \qquad \pi/2 \leq \theta \leq \pi,\label{wymonopole}
\end{align}
where the ``northern'' and ``southern'' patches are related at the equator by the transition function 
\be\label{monopoleswitch}
g_{NS}=e^{in\phi}
\ee 
as in equation \eqref{gtrans}.  $n$ is required to be an integer in order for $g_{NS}:\mathbb{S}^1\to U(1)$ to be a smooth map: it counts the number of magnetic flux units through $\mathbb{S}^2$.  The key point is then that if we turn on a Wu-Yang monopole background for $A^v$, the charge $\hat{Q}_p$ really needs to be defined separately in the northern and southern patches. The transformation \eqref{monopoleswitch} then leads to a nonconservation
\be
\hat{Q}_{p,N}=\hat{Q}_{p,S}-2n
\ee 
as we move the charge operator from the southern to the northern hemisphere.  The symmetry operator
\be
U\left(e^{i\theta},\mathbb{S}^2\right)\equiv e^{i\theta \hat{Q}_p}
\ee
is therefore not conserved, violating condition (d) of our definition \ref{globaldef}, so the $U(1)$ pseudovector symmetry has indeed been explicitly broken by the background gauge field for the $U(1)$ vector symmetry.\footnote{This c-number nonconservation of $\hat{Q}_p$ may seem innocuous, but it has real consequences for the selection rules obeyed by correlation functions.  Indeed a vacuum expectation value in this background of a product of operators charged under the pseudovector $U(1)$ symmetry will vanish unless the sum of their charges is equal to $2n$, while this sum would have needed to be zero to get a nonvanishing expectation value if the symmetry had been preserved.}  Moreover note that if we make the vector gauge field $A^v$ dynamical, these configurations will be unavoidable and the pseudovector symmetry will be broken altogether: this is a two-dimensional analogue of 't Hooft's famous discovery that instantons destroy the apparent axial isospin symmetry $u'=e^{i\theta\gamma_5}u$, $d'=e^{i\theta \gamma_5 }d$ of massless quantum chromodynamics, as well as the independent baryon and lepton number symmetries of the standard model of particle physics  \cite{tHooft:1976rip,tHooft:1986ooh} ($B-L$ is still a symmetry).

In this paper our primary concern with 't Hooft anomalies is that we need to make sure that our discussion of CFT global symmetries is not corrupted by the fact that we mostly work on the spacetime $\mathbb{R}\times\mathbb{S}^{d-1}$, with a round metric on the spatial $\mathbb{S}^{d-1}$.  We can view this metric as a background gauge field for the CFT stress tensor, so we are asking if there can be 't Hooft anomalies where this background gauge field spoils the CFT global symmetries we consider.  It is certainly possible for a background metric to spoil a global symmetry, for example a single Dirac fermion in $(3+1)$  dimensions has a $U(1)$ global symmetry with current
\be
J^\mu_p=-\ol{\psi}\gamma^\mu \gamma^5 \psi,
\ee
which obeys (assuming we regulate to preserve conservation of the stress tensor) \cite{Delbourgo:1972xb}
\be
\nabla_\mu J_p^\mu\propto \epsilon^{\mu\nu\alpha\beta}R_{\mu\nu\sigma\rho}R_{\alpha\beta}^{\phantom{\alpha\beta}\sigma\rho}.\label{gravanom}
\ee
It is easy to see however that this particular anomaly vanishes on $\mathbb{R}\times\mathbb{S}^{d-1}$, or more generally on $\mathbb{R}\times \Sigma$ for any $\Sigma$ provided that the spatial metric on $\Sigma$ is time-independent and there are no cross terms.  In fact at least for $\mathbb{R}\times\mathbb{S}^{d-1}$, this observation holds for any global symmetry in any conformal field theory.    This follows because Euclidean $\mathbb{R}\times \mathbb{S}^{d-1}$ is Weyl equivalent to Euclidean $\mathbb{R}^d$, via
\be
d\tau^2+d\Omega_{d-1}^2=\frac{1}{r^2}\left(dr^2+r^2d\Omega_{d-1}^2\right)
\ee
with $r=e^\tau$.  We can then simply \textit{define} the CFT on $\mathbb{R}\times \mathbb{S}^{d-1}$ via the Weyl transformation\footnote{We thank Z. Komargodski for a useful discussion of this definition, see some relevant comments in \cite{Gomis:2015yaa}.  In particular note that we may not be able to arrange for this equation to hold at coincident points, but our argument does not require it to.} 
\be\label{weyl}
\lan \mO_1(x_1)\ldots \mO_n(x_n)\ran_{e^{2\omega}g_{\mu\nu}}=e^{-\Delta_1\omega(x_1)-\ldots - \Delta_n \omega(x_n)+A[g,\omega]}\lan\mO_1(x_1)\ldots \mO_n(x_n)\ran_{g_{\mu\nu}}.
\ee
Here the $\mO_i$ are primary operators at distinct points $x_i$; this equation reflects that we have renormalized them to be Weyl tensors.  $\lan\cdot\ran_{g_{\mu\nu}}$ denotes the Euclidean path integral with background metric $g_{\mu\nu}$, and the factor $A[g,\omega]$ represents the standard 't Hooft anomaly in Weyl symmetry.  For example in a $1+1$ dimensional  CFT with Virasoro central charge $c$, we have \cite{Polchinski:1998rq}
\be
A[g,\omega]=\frac{c}{24\pi}\int d^2 x \sqrt{g}\left(\omega R+g^{\mu\nu}\partial_\mu\omega \partial_\nu\omega \right).
\ee
The correlation functions on $\mathbb{R}\times \mathbb{S}^{d-1}$ thus obey all the same selection rules from global symmetries that they do in flat space, with the symmetry operators $U(g,\mathbb{S}^{d-1})$ defined to act on local operators using the same matrix \eqref{Dmap} as in flat space (the Weyl anomaly does not spoil this since it is a $c$-number).  These statements are preserved under analytic continuation to Lorentzian signature, so therefore no global symmetry in a CFT can be violated purely by putting the theory onto Lorentzian $\mathbb{R}\times \mathbb{S}^{d-1}$.

\subsection{ABJ anomalies and splittability}\label{anomsplit}
't Hooft anomalies can be used to generate additional examples of unsplittable symmetries in quantum field theory.  In particular we can generate counterexamples to the naive Noether conjecture which do not rely on free or decoupled theories, and which are thus perhaps of more physical interest.\footnote{To avoid confusion we  emphasize here that the presence of an 't Hooft anomaly in a symmetry does not \textit{imply} that that symmetry is unsplittable.  For example the $U(N)$ global symmetry we describe momentarily has an 't Hooft anomaly, but it has a perfectly good set of Noether currents \eqref{4dJ} and is therefore splittable on any manifold we like.  In condensed matter language, splittability of a symmetry is a different question from whether or not the symmetry is ``on-site''.  Our unsplittable symmetries do not arise until we make some subset of the background gauge fields dynamical.}  The two examples of unsplittable symmetries that we will discuss here arise from the 3+1 dimensional version of the chiral anomaly we discussed in the previous section \cite{Bell:1969ts,Adler:1969gk}.  We will also use this anomaly in the next subsection, so we first briefly recall how it works in some generality.\footnote{This is of course textbook material, we apologize for presenting it in some detail nonetheless.  We have found the textbook treatments of this subject to be unclear at best, and our perspective has some novelty.  Readers who make it to the end of this subsection will be rewarded with an improved interpretation of the venerable process $\pi^0\rightarrow \gamma \, \gamma$ in the standard model of particle physics.}

Consider the theory of $N$ free left-handed Weyl fermions $\psi_i$, with Lagrangian
\be
\mathcal{L}=-i\sum_{i=1}^N \ol{\psi}_i \slashed{\partial}P_L\psi_i,
\ee
where 
\be
P_L\equiv \frac{1+\gamma^5}{2}.
\ee
There is a $U(N)$ global symmetry rotating the $\psi_i$ amongst each other which has an 't Hooft anomaly. The currents for this symmetry are
\be\label{4dJ}
J_a^\mu=-\sum_{ij}\ol{\psi}_i \left(\gamma^\mu P_L\otimes (T_a)_{ij}\right)\psi_j,
\ee
where $(T_a)_{ij}$ are the Lie algebra matrices of $U(N)$, and if we regulate this theory in a way that treats all these currents equally then in the presence of background gauge fields $A_\mu^a$ we have the anomalous current conservation equation \cite{Weinberg:1996kr}\footnote{There are sign errors in the derivation of \eqref{4danom} in \cite{Weinberg:1996kr}, but since there are an even number the final result is correct.  Our final sign is the same as that in \cite{Weinberg:1996kr} even though our currents \eqref{4dJ} differ from his by a sign, because we have taken $\epsilon^{0123}=-1$.}
\be\label{4danom}
D_\mu J^\mu_a=-\frac{D_{abc}}{24\pi^2}\epsilon^{\lambda\rho\sigma\nu}\partial_\lambda A^b_\rho \partial_\sigma A_\nu^c+\ldots,
\ee
where 
\be\label{Ddef}
D_{abc}\equiv \frac{1}{2}\Tr\left(\{T_a,T_b\}T_c\right),
\ee
and ``$\ldots$'' denotes higher order terms in $A$ which can be determined by symmetry and the Wess-Zumino consistency conditions \cite{Weinberg:1996kr}.  We can then play the game of adding local terms to the action, analogous to eq. \eqref{contact} above, to see how much of the $U(N)$ symmetry we can restore.  The $D_{abc}$ are in general not zero, and it is not hard to see that we will not be able to restore the full $U(N)$ symmetry in the presence of arbitrary background gauge fields, hence the 't Hooft anomaly.  It does turn out however that for any triple of distinct currents with $D_{abc}\neq 0$, we can arrange so that only one of them has an anomalous contribution to its conservation equation from background gauge fields for other two. For triples where two of the currents are identical and $D_{aab}\neq 0$, we can pick whether $J_a^\mu$ gets an anomalous contribution to its conservation equation from $A^a_\mu$ and $A^b_\mu$ or $J_b^\mu$ gets an anomalous contribution to its conservation equation from $A^a_\mu$ and $A^a_\mu$. For triples where all three currents are identical and $D_{aaa}\neq 0$, there is no hope and $J_a^\mu$ cannot be conserved in the presence of a background gauge field for itself.  These choices can be made independently for each triple, since they correspond to adding different local terms to the action.  

The original example of the four-dimensional chiral anomaly is in the theory of a free massless Dirac fermion, with Lagrangian \eqref{dirac4}.  As in two dimensions, in $\mathbb{R}^4$ with no background fields this theory has two conserved currents:
\begin{align}\nonumber
J_v^\mu&\equiv-\ol{\psi}\gamma^\mu \psi\\
J_p^\mu&\equiv -\ol{\psi}\gamma^\mu \gamma^5 \psi.
\end{align}
We can view this Dirac fermion as two left-handed Weyl fermions, in 
which case the anomaly coefficients \eqref{Ddef} are given by $D_{vvv}=D_{vpp}=0$, $D_{vvp}=D_{ppp}=2$.  We will consider only background gauge fields for $J_v^\mu$, so the only relevant anomaly coefficient is $D_{vvp}$.  Since we will want to make these gauge fields dynamical, for consistency we must add local terms to the action to modify \eqref{4danom} so that $J_v^\mu$ is conserved. After doing so, we arrive at the standard ABJ anomaly \cite{Bell:1969ts,Adler:1969gk}
\be
\partial_\mu J^\mu_p=-\frac{1}{16\pi^2}\epsilon^{\mu\nu\alpha\beta}F_{\mu\nu}^v F_{\alpha\beta}^v,
\ee
or in differential form notation
\be\label{4dformanom}
d\star J_p=\frac{1}{4\pi^2}F^v\wedge F^v.
\ee

So far this is all similar to what happened with the chiral anomaly in $1+1$ dimensions, but now an interesting difference arises: in $3+1$ dimensions we claim that, despite the 't Hooft anomaly \eqref{4dformanom}, chiral symmetry is preserved in the presence of any background $A^v$ gauge field on $\mathbb{R}^4$!  The reason is simple: there are no topologically non-trivial $U(1)$ gauge field configurations on $\mathbb{S}^4$ (and thus $\mathbb{R}^4$), unlike $\mathbb{S}^2$ (and $\mathbb{R}^2$) where there are,
so the ``improved'' current
\be
\star\hat{J}_p\equiv \star J_p-\frac{1}{4\pi^2}A^v\wedge F^v
\ee
integrates to an ``improved'' charge
\be\label{chiralcharge}
\hat{Q}_p\equiv \int_{\mathbb{R}^3}\star \hat{J}_p,
\ee
which acts in the same way as $\int_\Sigma \star J_p$ on all local operators, but is conserved on $\mathbb{R}^4$ for any background gauge field $A^v$.  

At first we might therefore think that this chiral symmetry will persist even if we now make $A^v$ dynamical.  We will now see however that the truth is more subtle.  Once $A^v$ is dynamical, the charge \eqref{chiralcharge} will indeed continue to exist as a gauge-invariant operator (this is because there are no topologically non-trivial gauge transformations on $\mathbb{S}^3$ since $\pi_3(U(1))=0$), and it will commute with the stress tensor.  Moreover it manifestly seems to act locally on local operators, so it seems we have satisfied all of the criteria of definition \ref{globaldef} for a global symmetry.  In fact however the charge \eqref{chiralcharge} fails condition (2) of definition \ref{globaldef}: it does not preserve the local algebra $\mathcal{A}[R]$ for all regions $R\subset \mathbb{R}^3$.  The problem is the following: now that the gauge field is dynamical, we need to check if the charge \eqref{chiralcharge} acts locally on the new operators we can construct from it.\footnote{We thank Edward Witten for pointing out that the electromagnetic part of this charge has a simple interpretation: in free Maxwell theory it is proportional to the helicity.  Thus conservation of $\hat{Q}_p$ says that although chiral symmetry is explicitly broken, the chiral charge plus a multiple of the helicity is conserved.}  This will obviously be the case for operators which are locally constructed out of $A^v$, such as the field strength $F^v$ and the Wilson loops $e^{in\int_C A^v}$, but since the gauge group is $U(1)$ we also need to check if it acts locally on 't Hooft loops. We will now show that it doesn't.  

't Hooft loops are an additional set of line operators in $U(1)$ gauge theory in four spacetime dimensions, defined by removing a narrow tube out of the path integral around the closed line $C$ where the operator will be defined and imposing certain boundary conditions. This tube has boundary $\mathbb{S}^2\times \mathbb{S}^1$, and the 't Hooft loop is defined by requiring that at this boundary the gauge field on $\mathbb{S}^2$ is given by the Wu-Yang monopole \eqref{wymonopole} \cite{Kapustin:2005py}.  Since this may seem a bit abstract, we note that in free $U(1)$ gauge theory an 't Hooft loop on a contractible curve $C=\partial D$ can also be represented as
\be\label{thooftfree}
T_n(C)\equiv e^{\frac{2\pi in}{q^2}\int_D \star F}.
\ee
This may not look like a loop operator, but we note the obvious analogy to the Wilson loop:
\be\label{wilson}
W_m(C)\equiv e^{im\int_{\partial D} A}=e^{i m \int_D F}.
\ee
Indeed $n$ and $m$ must be integers precisely so that these two lines are mutually local, meaning that they commute at spacelike separation even if they are linked in space (this is one way of understanding Dirac quantization). 

The action of the charge \eqref{chiralcharge} on an 't Hooft line can be computed in several ways.  In free $U(1)$ Maxwell theory we may simply study the commutator of \eqref{chiralcharge} and \eqref{thooftfree}, which shows without too much difficulty that the would-be symmetry generated by \eqref{chiralcharge} mixes the 't Hooft line with an improperly quantized Wilson loop, which then must be understood as a surface operator on a disc $D$, as in the second equality in \eqref{wilson}, rather than a line operator on $\partial D$.  We will instead obtain this result using the boundary-condition definition of $T_n(C)$, since this argument will be correct also in interacting theories such as the one we are studying.  If we view the 't Hooft line as an insertion into the Euclidean path integral on $S^4$, we can compute the action of chiral symmetry on it by surrounding it with a symmetry insertion on $\mathbb{S}^2\times \mathbb{S}^1$, constructed as in figure \ref{gluingfig} by approaching the line from above and below by symmetry operators on $\mathbb{S}^3$.  If we remove the small tube $B^3\times \mathbb{S}^1$ surrounding the line from $\mathbb{S}^4$, the remaining space has topology $\mathbb{S}^2\times B^2$ (these are glued at their mutual boundary $\mathbb{S}^2\times \mathbb{S}^1$).  This space allows nontrivial $U(1)$ bundles, since we can put a Wu-Yang monopole on the $\mathbb{S}^2$, and indeed the boundary condition from the 't Hooft line tells us that we must do so.  We therefore need to split the remaining space into ``northern'' and ``southern'' regions with topology $B^2\times B^2$.  The are glued together on a spatial region $\mathbb{S}^1\times B^2$, which is the $3+1$ dimensional version of the shaded blue regions in figure \ref{gluingfig}.  The gauge fields in the two regions differ by
\be
A_N^v=A_S^v+n d\phi,
\ee
where $\phi$ is the angular coordinate on the $\mathbb{S}^1$ and $n$ is the strength of the 't Hooft line, so the difference in the charge approached from above or below contains a term 
\begin{align}\nonumber
\hat{Q}_{p,N}-\hat{Q}_{p,S}&\supset-\frac{n N_f}{4\pi^2}\int_{\mathbb{S}^1\times B^2}d\phi\wedge F^v\\
&=-\frac{nN_f}{2\pi}\int_{B^2}F^v.
\end{align}
Here for later convenience we have generalized to $N_f$ Dirac fermions instead of just one, so now $D_{pvv}=2N_f$, and in evaluating the integral we have used that $\int_{B^2} F^v$ is independent of $\phi$. The other terms in $\hat{Q}_{p,N}-\hat{Q}_{p,S}$ are integrals over the upper and lower pieces of the $\mathbb{S}^2\times \mathbb{S}^1$ surrounding the loop, and are those localized near it. Therefore we see that the symmetry transformed operator
\be
T_n'(C)=e^{-i\theta \hat{Q}}T_n(C) e^{i\theta \hat{Q}}
\ee
includes a potentially nonlocal factor
\be
e^{i\frac{n N_f }{2\pi}\theta\int_{B_2} F^v},
\ee
where $B_2$ is any disc whose boundary is $C$.  If $\frac{n N_f \theta}{2\pi}$ is an integer then this will be a Wilson loop on $C$ written as in \eqref{wilson}, but otherwise this will be a disc operator with nontrivial support throughout $B^2$.  Therefore we see that only the $\mathbb{Z}_{N_f}$ subgroup of the $U(1)$ chiral symmetry generated by $\hat{Q}_p$ actually gives a good global symmetry which acts locally on 't Hooft lines.  

What then does this have to do with splittability?  We will now argue that this remaining $\mathbb{Z}_{N_f}$ symmetry is not splittable on $\Sigma=\mathbb{S}^2\times\mathbb{S}^1$, giving us another example of an unsplittable symmetry.  The analysis is quite similar to our discussion of the $\mathbb{R}\times\mathbb{R}$ gauge theory in section \ref{splitexsec}, so we will be brief.  The basic point is that our ``improved'' charge $\hat{Q}_p$ is not gauge invariant if we restrict it to a spatial subregion $R$.  Indeed if we define
\be
\hat{Q}_p(R)\equiv \int_R\left(\star J_p-\frac{N_f}{4\pi^2}A^v\wedge F^v\right)
\ee 
we have the gauge transformation
\be
\hat{Q}'_p(R)=\hat{Q}_p(R)-\frac{N_f}{4\pi^2}\int_{\partial R}\lambda^v F^v.
\ee
We encourage the reader to compare this equation to equation \eqref{Qinv}: they are almost identical except that we have gotten rid of some indices and exchanged $F$ and $\star F$.  Therefore on $\Rd$, or more generally on any spacetime manifold $M$ with $H_{d-2}(M)=0$, we can define an ``further improved'' localized charge
\be
\wt{\hat{Q}}_p(R)\equiv \hat{Q}_p(R)+\frac{N_f}{4\pi^2}\int_{\partial R}I F,
\ee
which is gauge-invariant and which will act in the same way on operators in $R$ and its complement once we exponentiate to get an element of $\mathbb{Z}_{N_f}$.  Here $I$ is a Wilson line integrated from a reference point $x_0$ on each connected component of $\partial R$ to the integration point $x$, as in equation \eqref{Idef}.  As before, this gauge invariance requires each connected component of the boundary to be contractible, since otherwise there could be components where $\int F \neq 0$.  Inspired by our discussion of the $\mathbb{R}\times \mathbb{R}$ theory, we may then study this theory on $\Sigma=\mathbb{S}^2\times \mathbb{S}^1$.  We may then run the same argument before, with $\int_{S^2} F$ replacing $\int_{S^2} \star F$, to conclude that the split property does not hold and the $\mathbb{Z}_{N_f}$ global symmetry is not splittable.  The unbreakable line operators which are to blame are now the 't Hooft lines.  

Finally we observe that we can use a similar mechanism to generate another example of a quantum field theory with a continuous global symmetry that has no Noether current; this time the theory will not be free.  The idea is simple: we consider exactly the same theory we have been discussing so far, but now we take the gauge group to be $\mathbb{R}$ instead of $U(1)$.  There are no longer 't Hooft lines, so the full $U(1)$ chiral symmetry is now preserved.  Moreover since we now have $\int_S F=0$ for any submanifold $S$ whatsoever, this symmetry is splittable on \textit{any} manifold.  But it nonetheless doesn't have a Noether current!\footnote{Although chiral symmetry is now splittable on any manifold, the theory with gauge group $\mathbb{R}$ still does not obey the split property on $\mathbb{S}^2\times \mathbb{S}^1$; the unbreakable lines are now the Wilson lines of fractional charge.  It thus is not a counterexample to our conjecture that theories which obey the split property on all manifolds should obey the Noether conjecture.}  Why not?  Because if it did, then we could use this Noether current in the case with gauge group $U(1)$ as well, since the set of local operators for the $U(1)$ gauge theory and the $\mathbb{R}$ gauge theory are exactly the same (more on this in section \ref{gaugetopsec} below), and this would contradict the fact that the $\mathbb{Z}_{N_F}$ subgroup of chiral symmetry which is preserved in the $U(1)$ case is not splittable on $\mathbb{S}^2\times \mathbb{S}^1$.  We find this to be quite remarkable: the existence of a Noether current is obstructed by features of a \textit{different} quantum field theory!  Moreover in that theory, with gauge group $U(1)$, we have another remarkable feature: all correlation functions not involving 't Hooft lines, and all scattering matrix elements not involving magnetic monopoles (if there are any) obey with complete precision the selection rules of a $U(1)$ global symmetry, \textit{even though no such symmetry exists}.

This analysis has interesting implications for the interpretation of the decay $\pi^0\rightarrow \gamma \, \gamma$ in the standard model of particle physics.  The traditional explanation of this decay is that the symmetry $u'=e^{i\theta \gamma^5}u$, $d'=e^{-i\theta \gamma^5}d$ of QCD with massless up and down quarks is explicitly broken by electromagnetism due to the anomaly \eqref{4dformanom}, see eg \cite{Weinberg:1996kr}, but we at least were never satisfied with this explanation for the following reason: if the symmetry is explicitly broken by the anomaly, why does it have a Goldstone boson (the $\pi^0$) in the first place?  Shouldn't explicit breaking of the symmetry give a mass to the $\pi^0$ even when the up and down quarks are massless?  The resolution of this puzzle is the following: we can choose to interpret the gauge group of electromagnetism as $\mathbb{R}$, and if we do then we indeed have a genuine $U(1)$ global symmetry generated by a charge $\hat{Q}_p=\int_{\mathbb{R}^3}\star\hat{J}_p$, with $\hat{J}_p$ now defined by\footnote{The anomaly coefficient is $D_{pvv}$ is still two, since $3\left(2\left(\frac{2}{3}\right)^2-2\left(\frac{1}{3}\right)^2\right)=2$.}
\begin{align}\nonumber
J_p^\mu&\equiv-\ol{u}\gamma^\mu\gamma^5u+\ol{d}\gamma^\mu \gamma^5 d\\
\star\hat{J}_p&\equiv\star J_p-\frac{1}{4\pi^2}A^v\wedge F^v\label{UVcurrent}.
\end{align}
This symmetry is spontaneously broken by the dynamics of QCD, and so it has a Goldstone boson, the $\pi^0$.  This is clear in the effective action for the pion, 
\be\label{piact}
S=-\int_{\mathbb{R}^4}\left(\frac{1}{2}d \pi^0\wedge\star d\pi^0+\frac{1}{4\pi^2}\frac{\pi^0}{f_\pi}F\wedge F\right),
\ee
which has a global symmetry $\pi^{0\,\prime}=\pi^0+f_\pi \epsilon$.  The Noether current for this symmetry that we can derive from this low-energy action, as in  \eqref{currentdef}, is
\be
\star \hat{J}_p=f_\pi \star d\pi^0-\frac{1}{4\pi^2}A^v\wedge F^v,
\ee
which indeed is not gauge-invariant, and in precisely the same way as the ``UV'' description \eqref{UVcurrent} of this current. Thus the $\pi^0$ is indeed the Goldstone boson of a perfectly good global symmetry, it just isn't quite the putative global symmetry we started with. The explanation of its ``surprisingly large'' decay rate is \textit{not} that the symmetry for which it is the Goldstone boson is explicitly broken by the anomaly, instead it is that this symmetry does not have (or need) a gauge-invariant Noether current: it is a counterexample to the naive Noether conjecture \ref{naiveN}, and this is what allows the second term in the action \eqref{piact}.\footnote{We remind the reader that this second term is what leads to the decay $\pi_0\to\gamma \gamma$ once quark masses are added, when $m_u=m_d=0$ this decay is not allowed kinematically but we can use the coefficient of $\pi^0 F \wedge F$ as a stand-in.}  We may then observe that if we revert to viewing the gauge group of electromagnetism as $U(1)$, none of the above conclusions can change so they must be true there as well even though our improved chiral symmetry charge $\hat{Q}_p$ now acts badly on 't Hooft lines.  It is instructive to compare this to another possible global symmetry of QCD with two massless quarks, $u'=e^{i\theta\gamma^5}u$, $d'=e^{i\theta \gamma^5}d$.  Prior to gauging $SU(3)$, this is indeed a global symmetry, with an 't Hooft anomaly $d\star J\propto G\wedge G$ where $G$ is the background gluon field.  Once the gluons are dynamical, this anomaly causes instantons to break this symmetry explicitly, just as monopoles did for $1+1$ dimensional chiral symmetry in the previous subsection, and the ``would-be Goldstone boson'', the $\eta'$, is indeed massive \cite{tHooft:1976rip}.  The distinction between the two cases arises because $\pi_3(U(1))=0$ while $\pi_3(SU(3))=\mathbb{Z}$.

\subsection{Towards a classification of 't Hooft anomalies}\label{anomalysec2}
We have discussed background gauge fields and 't Hooft anomalies at some length now, and we already have everything we need for our AdS/CFT arguments in the following sections.  't Hooft anomalies are such a hot topic these days however that we feel it appropriate to make a few more comments which may be of more general interest.  These comments are motivated by occasional statements we have heard that the classification of SPT phases in \cite{Chen:2011pg} based on the machinery of \cite{Dijkgraaf:1989pz}, together with some mathematical results from \cite{Freed:2014iua,Freed:2014eja,Freed:2016rqq} (see also \cite{Kapustin:2014zva}), result in a classification of 't Hooft anomalies for internal symmetries.  We argue here that the truth is more subtle, pointing out several gaps in this would-be argument.  Two of these gaps lead to explicit counterexamples to the putative classification, and thus require additional assumptions to exclude them.   A third gap we suspect can be filled, and we suggest a strategy for doing so.   The gaps are the following:
\bi
\item Not all 't Hooft anomalies act by multiplying the partition function by a $c$-number.
\item Not all 't Hooft anomalies which act by a $c$-number have that $c$-number be a phase.
\item Even when the 't Hooft anomaly is phase-valued, it has not been shown that this phase can always be canceled by the gauge transformation of the classical action of a topological gauge theory in $d+1$ dimensions.  
\ei
What is really attempted in \cite{Dijkgraaf:1989pz,Chen:2011pg,Freed:2014iua,Freed:2014eja,Freed:2016rqq} is a classification of such $(d+1)$-dimensional classical topological gauge actions, so until these gaps are better understood it is not correct to say that 't Hooft anomalies have been classified.  In the rest of this section we discuss these questions in more detail; along the way we will also point out an obstruction to generalizing the topological analysis of chiral 't Hooft anomalies in \cite{AlvarezGaume:1983cs} to more general 't Hooft anomalies.  As this work was being completed, \cite{Cordova:2018cvg,Benini:2018reh} appeared, which study the first of the phenomena we mention here, operator-valued 't Hooft anomalies, in much more detail; we direct the reader there for more on this phenomenon.\footnote{In those papers the authors introduce new background gauge fields, which are in general higher-form fields, and then modify the definition of ``gauge transformation'' to include transformations of these new background gauge fields which are designed to cancel the operator-valued anomalies of the type we point out here.  They then prefer to use the terminology ``$n$-group global symmetry'' instead of ``operator-valued 't Hooft anomaly''.  In this language, $c$-number 't Hooft anomalies in $d$ spacetime dimensions are ``$d$-group global symmetries''.  We'll stick with ``'t Hooft anomaly'' here since we've been using it so far, but in the long run getting rid of the word ``anomaly'' in this context is probably a good idea.}

We begin by noting that all examples of 't Hooft anomalies that we discussed in the previous section have the special property that, although the partition function is not gauge invariant, this non-invariance is realized as a multiplication by a $c$-number functional of the background gauge fields and the gauge transformation (see equations \eqref{2dZ} and \eqref{weyl}).  It is not hard see however that more general 't Hooft anomalies are possible; we will call them \textit{operator-valued 't Hooft anomalies}.  They have appeared in some form already in \cite{Kapustin:2014zva}, but the example we give here should be more accessible to most physicists.  It is a chiral fermion theory in $3+1$ dimensions, with an $SU(2)$ global symmetry and a $U(1)$ gauge symmetry.  The matter fields consist of eight left-handed fermions, grouped into two $SU(2)$ doublets with $U(1)$ charge $+1$, and four $SU(2)$ singlets with $U(1)$ charge $-1$.\footnote{We have doubled the matter content of what might seem like the simplest example, to avoid an additional Witten anomaly in the $SU(2)$ symmetry \cite{Witten:1982fp} which may distract some readers.}  We can view both of these symmetries as subgroups of the $U(8)$ symmetry generated by the currents \eqref{4dJ}, but the rest of this $U(8)$ may or may not be broken by other interactions we won't discuss explicitly. Since the $U(1)$ symmetry is gauged, its current must be conserved to avoid inconsistencies. And indeed,
\be
D_{U(1)U(1)U(1)}=4(+1)^3+4(-1)^3=0.
\ee
This $U(1)$ conservation is also not broken by the gravitational anomaly \eqref{gravanom}, since $4(+1)+4(-1)=0$.  If we use indices $a$, $b$, etc to denote $SU(2)$ generators, with $T_a$ taken to be the Pauli matrices divided by 2, then we see that
\begin{align}\nonumber
D_{abc}&=0\\
D_{ab U(1)}&=2\Tr(T_a T_b)=\delta_{ab}.
\end{align}
Thus in the presence of a background $SU(2)$ gauge field, since we must preserve the conservation of the $U(1)$ current, we have no choice but to allow the $SU(2)$ global currents not to be conserved.  After adding an appropriate local term to the action to ensure conservation of the $U(1)$ current, we find that the $SU(2)$ currents obey
\be
D_\mu J_a^\mu=-\frac{1}{32\pi^2}\delta_{ab}\epsilon^{\lambda\rho\sigma\nu}\partial_\lambda A^b_\rho F_{\sigma\nu}^{U(1)}+\ldots
\ee
The key point here is that if the background $SU(2)$ gauge field $A_\mu^a$ is zero, the $SU(2)$ current is conserved.  So this theory indeed has $SU(2)$ global symmetry.  But once we turn on this background gauge field, the right hand side of the current conservation involves a dynamical operator, $F_{\sigma\nu}^{U(1)}$.  Thus, unlike in the 't Hooft anomalies we have considered so far, the partition function does not transform by a $c$-number rescaling under a background $SU(2)$ gauge transformation.  In such a situation we cannot cancel the anomaly by the gauge transformation of the classical action of a topological gauge theory in $d+1$ dimensions, essentially because that action would already need to contain a dynamical $U(1)$ gauge field.

Of course nothing stops us from simply restricting discussion to $c$-number 't Hooft anomalies.  In fact in the classification program based on \cite{Dijkgraaf:1989pz,Chen:2011pg,Freed:2014iua,Freed:2014eja,Freed:2016rqq}, it is further assumed that the $c$-number involved is always a phase.  This is certainly true for the $1+1$ and $3+1$ dimensional chiral anomalies \eqref{2dZ}, \eqref{4danom}, and more generally it is a rather standard property of chiral anomalies \cite{AlvarezGaume:1983ig}.  But again it is not always true, and in fact we have already met a counterexample: in Euclidean signature the Weyl anomaly \eqref{weyl} is real.  And indeed there has so far been no success in trying to cancel the Weyl anomaly with the gauge transformation of a topological action living in $d+1$ dimensions.\footnote{The Weyl anomaly \textit{can} be cancelled by a non-unitary gravitational action, one way to see this is that we know the ``right sign'' Einstein-Hilbert action can reproduce the Weyl anomaly in AdS/CFT \cite{Henningson:1998gx}, so the ``wrong sign'' Einstein-Hilbert action can cancel it.  It is not clear however whether such actions can be classified by some generalization of the machinery of \cite{Dijkgraaf:1989pz,Chen:2011pg,Freed:2014iua,Freed:2014eja,Freed:2016rqq}.} 

Nevertheless we can still proceed by further restricting to 't Hooft anomalies where under background gauge transformations the partition function is only multiplied by a phase.  We now give a general formulation of this problem.  As above will use the symbol $A$ to jointly denote a collection of $A_i$ and the $g_{ij}$ which glue them together, we will use the symbol $g$ to denote the collection of $g_i$ under which the $A_i$ and $g_{ij}$ transform via \eqref{patchtransform}, and we will write the action of $g$ on $A$ as $gA$.  This $A$ will include background gauge fields for all global symmetries, both continuous and discrete.  A phase-valued 't Hooft anomaly then says that the partition function of the theory as a functional of these background gauge fields obeys
\be\label{alphadef}
Z[gA]=e^{i\alpha(A,g)}Z[A].
\ee
Moreover it says that this phase cannot be removed by redefining $Z[A]$ by a local functional $\beta(A)$, via
\be
Z'[A]\equiv e^{i \beta(A)}Z[A].
\ee
Such a redefinition induces a transformation
\be
\alpha'(A,g)=\alpha(A,g)+\beta(gA)-\beta(A) \qquad\mathrm{mod} \, 2\pi,
\ee  
so we will have an anomaly if and only if there is no $\beta(A)$ such that
\be\label{trivialch}
\alpha(A,g)=\beta(A)-\beta(gA)\qquad\mathrm{mod} \, 2\pi.
\ee
The task of classifying possible phase-valued 't Hooft anomalies is thus the task of classifying phases $\alpha(A,g)$ modulo local functionals $\beta(A)$, which is a kind of exotic group cohomology.  We emphasize however that this is \textit{not} the group cohomology studied in \cite{Chen:2011pg}; we will comment on the relationship below.  

This group cohomology problem has an interesting relationship to the topology of fiber bundles \cite{AlvarezGaume:1983cs}.  This relationship works as follows.  Consider the space $\mathcal{G}$ of gauge transformations $g$ and the space $\mathcal{A}$ of gauge field configurations $A$.  We can view the partition function as a map
\be
Z:\mathcal{A}\rightarrow \mathbb{C},
\ee
or equivalently as a section of the trivial complex line bundle
\be
E\equiv \mathcal{A}\times \mathbb{C}.
\ee
We can then define an equivalence relation on $E$ via
\be\label{eqrel}
(A,z)\sim (gA,e^{i\alpha(A,g)}z),
\ee
and then construct a new bundle
\be
\wt{E}\equiv E/\sim,
\ee
which is a possibly nontrivial complex line bundle over  $\mathcal{A}/\mathcal{G}$, the set of gauge-equivalent classes of gauge field configurations.   
In fact the transformation \eqref{alphadef} tells us that we can also view the partition function $Z$ as a section of $\wt{E}$.  The interesting statement is then the following: if $\wt{E}$ is a nontrivial bundle, then $Z$ has a genuine 't Hooft anomaly.  The proof is simple: say that $Z$ did \textit{not} have an 't Hooft anomaly.  Then there must exist a local functional $\beta(A)$ obeying \eqref{trivialch}.  We may then consider a coordinate transformation on the bundle $E$ given by
\begin{align}\nonumber
A'&=A\\
z'&=e^{i \beta(A)}z,\label{bundletransform}
\end{align}
under which the equivalence relation \eqref{eqrel} becomes
\be\label{trivialeq}
(A,z')\sim (gA,z').
\ee
But this immediately tells us that
\be
\wt{E}=\mathcal{A}/\mathcal{G}\times \mathbb{C},
\ee
so $\wt{E}$ is trivial.  This argument shows that nontrivial line bundles over $\mathcal{A}/\mathcal{G}$ are related to potential 't Hooft anomalies.  And in fact in \cite{AlvarezGaume:1983cs} it was shown that indeed the partition function relevant for the 3+1 dimensional chiral anomaly \eqref{4danom} is a section of a nontrivial line bundle over $\mathcal{A}/\mathcal{G}$. Fiber bundle topology is an extremely well-studied subject, so this result seems to suggest that the relevant technology could be used to study general phase-valued 't Hooft anomalies. 

Unfortunately however there is a major problem in attempting to use the argument of the previous paragraph to classify 't Hooft anomalies.  This is that the result is \textit{not} an if and only if result.  We showed that a nontrivial bundle implies an anomaly, but we did not show that a trivial bundle implies no anomaly!  The problem lies with the coordinate transformation \eqref{bundletransform}. In doing this transformation, we used a $\beta(A)$ which was a local functional of $A$. But in trying to decide whether or not $\wt{E}$ is trivial, there is no such restriction on what coordinate transformations we may do: if we can achieve \eqref{trivialeq} with a nonlocal $\beta$, then the bundle is trivial even though there might still be an 't Hooft anomaly.  This observation leads immediately to a puzzle: if we allow $\beta$ to be nonlocal, then doesn't the logarithm of \eqref{alphadef} immediately tell us that $\beta(A)\equiv i \log Z(A)$ gives a nonlocal coordinate transformation which trivializes $\wt{E}$?  And if so then how were the authors of \cite{AlvarezGaume:1983cs} able to get a nontrivial bundle $\wt{E}$?  The resolution of this puzzle is that the problem with this $\beta$ is \textit{not} that it is nonlocal, it is that $Z[A]$, which for them was the square root of the determinant of a Dirac operator, has zeros at certain special values of $A$.  So then $i\log Z$ is not well-defined at those values, which prevents it from defining a good coordinate transformation on $E$.  

How then might we proceed in our goal to classify possible phase-valued 't Hooft anomalies?  In fact we have already stated the mathematical problem: we need to classify phases $\alpha$ modulo local functionals $\beta$.  The natural idea suggested by the topological arguments of the previous two paragraphs is to recast this as a generalization of the notion of a complex line bundle over $\mathcal{A}$, where only \textit{local} functionals of $A$ are allowed in coordinate transformations.  We will not attempt this here, but we have already mentioned several times the standard conjecture for what the answer is: any solution of this problem is always obtainable from the classical action of some topological gauge theory in $d+1$ dimensions \cite{Stora:1983ct,Zumino:1983ew,Manes:1985df,Faddeev:1985iz}.  Indeed the validity of this conjecture is taken as the starting point of the work of \cite{Freed:2014iua,Freed:2014eja,Freed:2016rqq}.  Let's review how this works for the $1+1$ dimensional chiral anomaly \eqref{2dZ}, which we can now describe as 
\be
\alpha(A^v,A^p;\Lambda^v,\Lambda^p)=-\frac{1}{\pi}\int_{\partial N} \Lambda^p F^v.
\ee
Here we have switched to differential form notation and assumed for simplicity that our spacetime manifold $M$ is the boundary of some three-dimensional manifold $N$.  The key observation is that the three-dimensional Chern-Simons-like action,
\be
S_3[A^v,A^p]\equiv\frac{4}{4\pi}\int_N A^p\wedge F^v,
\ee
has gauge transformation
\begin{align}
S_3[A^v+d\Lambda^v,A^p+d\Lambda^p]&=S_3[A^v,A^p]+\frac{1}{\pi}\int_N d(\Lambda^p F^v)\\
&=S_3[A^v,A^p]+\frac{1}{\pi}\int_{\partial N}\Lambda^p F^v,
\end{align}
so if we take the three-dimensional gauge fields to be extensions of the two-dimensional ones then the functional 
\be
\hat{Z}[A]\equiv Z[A]e^{iS_3(A)}
\ee
is gauge-invariant.  So although the anomaly cannot be canceled by a local term in $(1+1)$ dimensions, it \textit{can} be canceled by a local term in $(2+1)$ dimensions!  A similar construction is possible for the $(3+1)$ dimensional chiral anomaly, based on a five dimensional Chern-Simons-like action \cite{Stora:1983ct,Zumino:1983ew,Manes:1985df,Faddeev:1985iz}.  But now we come to the key question: is this relationship with $d+1$ dimensional topological actions a coincidence, or is it intrinsic to the nature of 't Hooft anomalies?  The conjecture just mentioned says that it is intrinsic, and certainly the fact that so far every phase-valued 't Hooft anomaly to be discovered fits into this framework speaks powerfully in favor of this conjecture.  But can it be proven?  We believe that the answer is yes.  One reason is that for infinitesimal anomalies it has indeed already been proven, by a careful study of the cohomology of the BRST operator \cite{Manes:1985df,DuboisViolette:1985jb,Brandt:1989gy,Dixon:1991wi,DuboisViolette:1992ye}.  But more generally the reason we believe so is that both sides of the conjecture can be precisely formulated as statements about group cohomology: the general classification of 't Hooft anomalies outlined below equation \eqref{alphadef} casts the question directly as a group cohomology problem, and the classification of topological actions studied in \cite{Dijkgraaf:1989pz,Chen:2011pg,Freed:2014iua,Freed:2014eja,Freed:2016rqq} essentially proceeds by reformulating the question again as a group cohomology problem.  In both cases the objects which appear or more or less the same: we need to define local functionals of background gauge fields and then study how they transform under gauge transformations, with appropriate identifications.  Given the strong experimental evidence for the conjecture, together with this plausible mathematical formulation, we expect that a proof is possible.  We will not however attempt it here.  
 
\section{Gauge symmetry}\label{gaugesec}
We now turn to the topic of gauge symmetry.  Gauge symmetry is ubiquitous in physics.  Our understanding of particle physics, general relativity, string theory, the fractional quantum hall effect, superconductivity, and more all rely on it.  And yet, paradoxically, we also say that ``gauge symmetry is merely a redundancy of description.''  How can a redundancy of description be so powerful?  In AdS/CFT this paradoxical situation is acutely instantiated  by the well-known adage ``a gauge symmetry in the bulk is dual to a global symmetry in the boundary.''  In the words of the master \cite{Witten:1998qj}, ``suppose the AdS theory has a gauge group $G$, \ldots Then in the scenario of (Maldacena), the group $G$ is a global symmetry of the conformal field theory on the boundary.'' How can a mere redundancy of description be dual to something as substantial as a global symmetry?  

In this section we develop machinery to address this question, introducing a notion of ``long-range gauge symmetry'' that we will eventually argue is really what should be understood as the gravity dual of a global symmetry.  To aid with intuition, we illustrate our definition using a general formulation of Hamiltonian lattice gauge theory for arbitrary compact gauge group $G$.  We then make some comments on the meaning of the topology of the gauge group and briefly discuss the possibility of nontrivial mixing between gauge and global symmetries.

\subsection{Definitions}
Roughly speaking, the traditional definition of a gauge symmetry in quantum field theory is that it is obtained by ``gauging'' a global symmetry, meaning that we begin with a quantum field theory with a global symmetry, introduce background gauge fields for that symmetry as in section \ref{backgroundsec}, and then make them dynamical by summing over them in the path integral (this procedure makes sense even if the theory to be gauged does not have a Lagrangian).  This definition is not quite consistent with our definition \ref{globaldef} of global symmetry however: there we required that global symmetries act faithfully on the set of local operators, while for gauge symmetries there should be no such requirement (otherwise we would exclude e.g. free Maxwell theory).\footnote{This statement applies in quantum field theory. One of our main goals in this paper is to establish conjecture \ref{allcharge}, which says that in quantum gravity there \textit{is} such a requirement!}  So in our language, the way to phrase this definition is to interpret the full (internal) global symmetry group $G$, which does act faithfully on the local operators, as the quotient of a possibly-larger ``extended'' symmetry group $\hat{G}$, which acts on the local operators in a not-necessarily faithful representation, by the kernel of that representation.  $\hat{G}$ is far from unique, but whichever choice we make we then choose a normal subgroup $H\subset \hat{G}$, and introduce background gauge fields for it.  We then check whether or not any 't Hooft anomalies prevent us from arranging for the partition function to depend only on the gauge equivalence classes of these background gauge fields: if not, then we may at last make them dynamical.\footnote{The question of what the global symmetry group is after doing this procedure is a very delicate one, involving not only the group-theoretic structure of how $H$ and $G$ fit into $\hat{G}$, but also the effects of any 't Hooft anomalies in $\hat{G}$ which might be activated (see \cite{Tachikawa:2017gyf} for one recent discussion).  We will not explore this question further except for a brief discussion in section \ref{mixsec} below, but we view it as ripe for additional work.}  

Although this definition is completely standard, it has the very serious problem that the same abstract quantum field theory can be obtained in this manner by gauging inequivalent extended global symmetry subgroups $H$ of inequivalent abstract quantum field theories.  For example the $U(1)$ Maxwell theory in 2+1 dimensions has an equivalent formulation as a free compact scalar with no gauge fields at all.  Much more nontrivially, the $\mathcal{N}=4$ super Yang-Mills theory with gauge group $SU(N)$ and gauge coupling $g$ is equivalent as an abstract quantum field theory to the $\mathcal{N}=4$ super Yang-Mills theory with coupling $\frac{4\pi}{g}$ and gauge group $SU(N)/\mathbb{Z}_N$ by $S$-duality \cite{Montonen:1977sn,Sen:1994fa,Vafa:1994tf,Kapustin:2005py}.  Given examples like these, it seems clear that there is no unique answer to the question ``what is the gauge group of abstract quantum field theory $X$?''  This is to be distinguished from the case of global symmetry, where the analogous question indeed has a unique answer given by definition \ref{globaldef}.

That said, there are certainly unambiguous physical phenomena which we typically \textit{associate} with gauge symmetry, such as massless gauge bosons, loop operators whose vacuum expectation values obey an area law scaling, asymptotic symmetry groups, and certain topological field theories such as the $\mathbb{Z}_2$ gauge theory that describes superconductivity.  The second of these has recently been formalized into the abstract notion of an unbroken one-form global symmetry \cite{Gaiotto:2014kfa}, which we will discuss more in section \ref{psec} below: it gives one way of defining confinement abstractly.  The others are all associated to gauge theories in what \cite{Fradkin:1978dv} called a ``free charge phase'': this means a phase which allows charged states of finite energy in infinite volume  (see also \cite{Banks:1979fi,Alford:1990fc} for related discussion).  For continuous gauge groups this is usually called a Coulomb phase, while for discrete gauge groups (or continuous gauge groups in $2+1$ dimensions with Chern-Simons terms) it is sometimes called a topological phase.  In \cite{Fradkin:1978dv} the notion of a free charge phase was introduced in the context of lattice gauge theory, which is a specific presentation of a quantum field theory.  As we just discussed, different lattice gauge theories might flow to the same abstract quantum field theory in the infrared.  But in fact the notion of a free charge phase can be rephrased using only abstract notions, which thus frees it of such ambiguities.   We now formalize this as a new definition:\footnote{In this paper we are primarily interested in spacetimes which are asymptotically-flat or asymptotically-$AdS$. This definition may need further refinement for more complicated spatial manifolds $\Sigma$, but for our purposes it is good enough.}
\begin{mydef}\label{gaugedef}
A quantum field theory on an infinite-volume spatial manifold $\Sigma$, with asymptotic boundary $\partial \Sigma$ and boundary conditions such that in any state the energy density vanishes as we approach $\partial \Sigma$, has a \textit{long-range gauge symmetry with gauge group $G$} (here $G$ is assumed compact) if the following are true:
\bi
\item[(1)] For each closed spatial curve $C$ in the interior of $\Sigma$, there exist a set of directed line operators $W_\alpha(C)$, the \textit{Wilson loops}, which are labeled by the finite-dimensional irreducible representations $\alpha$ of $G$.   Moreover for any curve $C$ which starts and ends at $\partial \Sigma$ there are a set of \textit{Wilson lines} $W_{\alpha,ij}(C)$, where $i$ and $j$ run over a range given by the representation dimension $d_\alpha$.   The orientations of Wilson loops and lines can be flipped via
\begin{align}\nonumber
W_\alpha(-C)&=W_\alpha^\dagger(C)\\
W_{\alpha,ij}(-C)&=W_{\alpha,ij}^\dagger(C),
\end{align}
where in the second of these ``$\dagger$'' denotes the adjoint operation on Hilbert space together with an exchange of the $ij$ indices, and the Wilson lines obey
\be
\sum_k W_{\alpha,ik}(-C)W_{\alpha,kj}(C)=\delta_{ij}.
\ee
A Wilson line can be turned into a Wilson loop by bringing the endpoints of $C$ together, tracing over $ij$, and then deforming $C$ into a closed loop in the interior of $\Sigma$.

\item[(2)] For every subregion $R$ of $\partial \Sigma$, and every $g\in G$, there is a unitary operator $U(g,R)$ on the Hilbert space which commutes with all operators supported only in the interior of $\Sigma$, and also with their boundary limits provided they have no support in $\partial R$, but which acts on any Wilson line $W_\alpha$ starting at point $x\in \partial\Sigma$ and ending at point $y\in \partial\Sigma$ as
\begin{align}
U^\dagger(g,R) W_{\alpha} U(g,R)=
\begin{cases}
D_{\alpha}(g)W_{\alpha} D_{\alpha}(g^{-1})& x,y\in R\\
W_{\alpha} D_{\alpha}(g^{-1})& x\in R,y\notin R\\
D_\alpha(g)W_\alpha & x\notin R, y\in R\\
W_\alpha & x,y\notin R
\end{cases},\label{Walg}
\end{align}
where we have suppressed the $ij$ representation indices using matrix notation.  When $R$ is a connected component of $\partial \Sigma$, we will refer to the $U(g,R)$ as \textit{asymptotic symmetry operators}. This name is justified by the observation that we have $[H,U(g,R)]=0$, since $H=\int_\Sigma d^{d-1}x\sqrt{g}\,T_{00}$ and $T_{00}$ is always either an operator in the interior of $\Sigma$ or the boundary limit of one.  In correlation functions the asymptotic symmetry operators will be topological except when they meet the endpoint of a Wilson line.  For arbitrary $R$ we will call the $U(g,R)$ the \textit{localized asymptotic symmetry operators}: these will be topological under deformations which in addition to not crossing Wilson line endpoints also fix $\partial R$.\footnote{Note that we are including the gauge-symmetry version of splittability in this definition. A weaker definition would ask for the $U(g,R)$ only when $R$ is a connected component of $\Sigma$, but we find our definition more convenient since it ensures that the $W_\alpha$ are nontrivial even if $\partial \Sigma$ has only one connected component, which otherwise we would need to implement with additional axioms.  We don't know of any examples of ``unsplittable long-range gauge symmetries'' which we would exclude this way.}

\item[(3)] The ground state is invariant under $U(g,\partial \Sigma)$, and moreover the theory allows finite-energy charged states in the sense that if we deform the Hamiltonian and the Hilbert space to include a background charge in representation $\alpha$ sitting at some definite point in space, there are states of finite energy which transform in that representation under $U(g,\partial \Sigma)$.  This Hilbert space and Hamiltonian are defined by the insertion of a temporal Wilson line in representation $\alpha$ into the path integral, we explain how to do this in detail for lattice gauge theory in the following subsection.  In AdS (our primary interest) there is a very concrete test: in the Euclidean thermal AdS space with metric
\be
ds^2=(1+r^2)d\tau^2+\frac{dr^2}{1+r^2}+r^2 d\Omega_{d-2}^2,
\ee
with $\tau$ periodicity $\beta$,
we study the quantity
\be
Z_\alpha(g,\beta)\equiv \lan W_\alpha(\mathbb{S}^1)U(g,\mathbb{S}^{d-2})\ran,
\ee
where the Wilson line is at $r=0$ and wraps the temporal circle, while the $\mathbb{S}^{d-2}$ is at spatial infinity.  This quantity has the interpretation of inserting the asymptotic symmetry operator $U(g,\mathbb{S}^{d-2})$ into the thermal trace over the modified Hilbert space with a background charge at $r=0$ in representation $\alpha$.
We then require that
\be\label{adscond}
\int dg \chi^*_\alpha(g)Z_\alpha(g,\beta)>0
\ee
for any $\alpha$ and large but finite $\beta$, where $dg$ is the Haar measure on $G$ and $\chi_\alpha(g)\equiv \Tr\left(D_\alpha(g)\right)$ is the character function on $G$ for representation $\alpha$.  By Schur orthogonality (see theorem \ref{schurthm}) this integral (or sum if $G$ is discrete) inserts a projection onto states in representation $\alpha$ in the thermal trace, so \eqref{adscond} is precisely requiring that there are such states with finite energy.\footnote{This test is more delicate in Minkowski space, since the thermal partition function is infrared divergent. One way to deal with this is to use AdS as a regulator, and then say that a Minkowski space theory obeys condition (3) if it does in AdS for any sufficiently large AdS radius.}
\ei 
\end{mydef}
This definition may seem like a lot to unpack, and indeed we will spend the rest of the section doing so.  We will motivate it in detail from a lattice point of view starting in the next subsection, but a few examples are in order now.  

The most obvious example is free Maxwell theory in Minkowski space with $d\geq 4$ spacetime dimensions, with action
\be\label{maxwell}
S=-\frac{1}{2q^2}\int F\wedge \star F.
\ee  
If we regulate space at some large radius, the variation of this action has a boundary term
\be\label{maxvar}
-q^{-2}\int_{\partial M} \delta A \wedge \star F,
\ee
which we can satisfy by choosing boundary conditions where the pullback of $A$ to $\partial M$ vanishes.  These boundary conditions are preserved only by gauge transformations which approach a constant on $\partial M$, and to obtain a theory where non-vanishing electric charge is possible we will quotient only by gauge transformations where this constant also vanishes: the transformations where it does not are the asymptotic symmetries.\footnote{These boundary conditions are the natural ones for a gauge field in $AdS$. In $3+1$ dimensional Minkowski space they are less natural because they set the magnetic flux density to zero at spatial infinity, and thus violate cluster decomposition if there are magnetic monopoles.  We can restore cluster decomposition by a Hilbert space direct sum over magnetic flux configurations, after which the long range gauge symmetry will actually be $U(1)\times U(1)$ since both Wilson and 't Hooft lines will be able to end at infinity.  Since our primary interest is $AdS$, we stick to the sector of vanishing magnetic flux, in which case only Wilson lines can end at infinity and the long-range gauge group is $U(1)$.}  The representations of $U(1)$ are labeled by integer charges, and the Wilson loops and lines  have the form
\be
W_n(C)=e^{in\int_C A+\ldots}.
\ee
Here ``$\ldots$'' represents a term proportional to the length of $C$ in cutoff units, with a coefficient which is chosen so that the expectation value of $W_n(C)$ is finite when $C$ has finite size in the continuum.  The localized asymptotic symmetry operators $U(g,R)$ are given by
\be
U(e^{i\theta},R)=\exp\left[\frac{i\theta}{q^2} \int_{R} \star F\right],
\ee
which is just the exponentiated electric flux through $R$ at spatial infinity.  With our choice of boundary conditions the Wilson lines are allowed to end at spatial infinity, and it is easy to see that together with the localized asymptotic symmetry operators they obey (1-2) from definition \ref{gaugedef}. Moreover since for $d\geq 4$ the electrostatic energy of a smeared point charge is finite they will also obey condition (3).  By contrast for $d=2,3$ the electrostatic energy of a (smeared) point charge is infinite, linearly for $d=2$ and logarithmically for $d=3$, so condition (3) will not be satisfied.\footnote{In $d=3$ this logarithmic divergence is sometimes confused by Polyakov's old observation  that in $U(1)$ lattice gauge theory in $2+1$ dimensions there are no photons and external charges feel a \textit{linear} potential \cite{Polyakov:1975rs}.  This however is an artifact of the lattice, the continuum $U(1)$ Maxwell theory has free photons and a logarithmic potential between external charges.  Condensed matter physicists sometimes give this continuum theory the rather silly name ``noncompact $U(1)$ Maxwell theory'', but $U(1)$ is (of course) still compact.  We \textit{could} study Maxwell theory with gauge group $\mathbb{R}$, but that is something different (see subsection \ref{gaugetopsec} below for more on the meaning of the topology of the gauge group).} Thus for $d=2,3$,  Maxwell theory does not have a long-range $U(1)$ gauge symmetry, while for $d\geq 4$ it does.

The statement that there is no long-range gauge symmetry in Maxwell theory for $d=3$ may sound surprising from a holographic point of view, since we certainly know examples of holographic CFTs in $1+1$ dimensions with $U(1)$ global symmetries.  In fact what happens in all such examples is that in the bulk we have not the pure Maxwell theory \eqref{maxwell}, but instead the Maxwell/Chern-Simons theory with action\footnote{Any solution of Maxwell-Chern Simons theory can be locally decomposed into $A=A_{flat}+\hat{A}$, with $A_{flat}$ a flat connection and $\hat{A}$ obeying $2\pi \star\hat{F}+kq^2 \hat{A}=0$, which is the equation for a vector boson with mass $\frac{|k|q^2}{2\pi}$.  In AdS the natural boundary conditions for Maxwell-Chern Simons theory are to set either the left-moving or right-moving part of the pullback of $A$ to the $AdS$ boundary to zero, and also to require the vanishing of the pullback of $\star F$ there \cite{Kraus:2006wn,Andrade:2011sx}. The latter condition keeps only the normalizable piece of $\hat{A}$, while the former chooses whether the current in the boundary CFT will be right-moving or left-moving.  To have a boundary current with both left- and right- moving parts, we need two gauge fields in the bulk \cite{Achucarro:1987vz,deBoer:1998kjm}.}
\be
S=-\int_M\left(\frac{1}{2q^2}F\wedge \star F+\frac{k}{4\pi}A\wedge F\right).
\ee
This theory \textit{does} have a long-range gauge symmetry: the logarithmic infrared divergence in the energy of a localized charge in Maxwell theory is regulated by the Chern-Simons term, allowing finite-energy states of nonzero asymptotic charge $\frac{k}{2\pi} \int_{\partial\Sigma}A$.  This example shows that at least in $d=3$, one can have a long-range $U(1)$ gauge symmetry without a massless photon.

\bfig
\includegraphics[height=4cm]{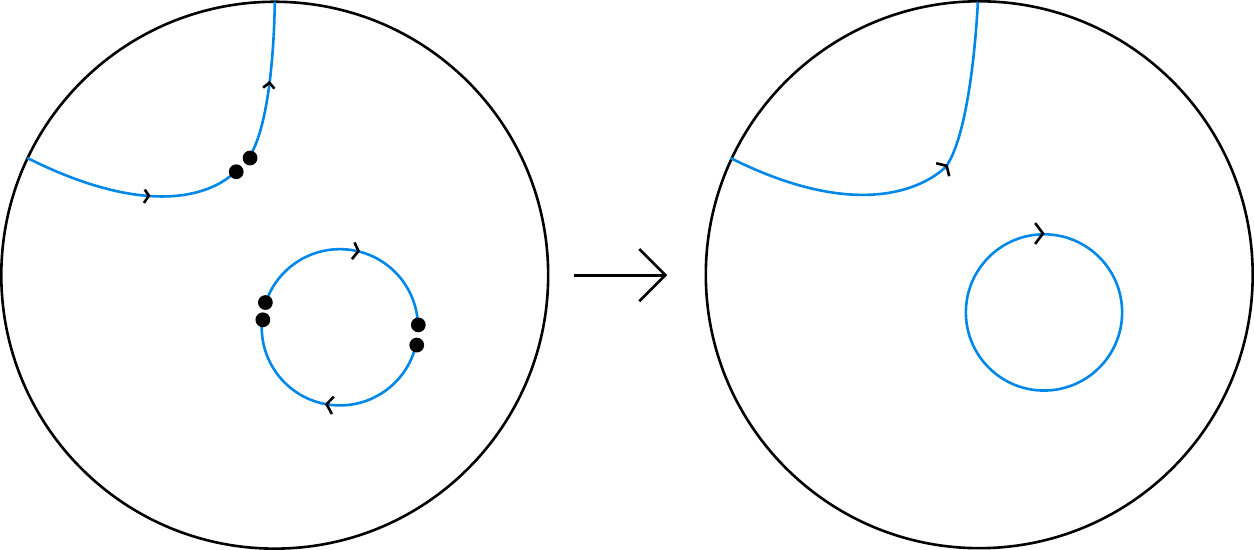}
\caption{Merging Wilson lines with interior endpoints to make boundary attached Wilson lines and Wilson loops.  See section three of \cite{Harlow:2015lma} for a quantitative illustration of how this merging works.  These gluing rules ensure that the ``meat'' of the lines and loops are all the same.} \label{wilsongluefig}
\efig
Our definition \ref{gaugedef} applies whether or not the gauge theory has ``dynamical charges'', which we define as follows:
\begin{mydef}\label{chargedef}
In a quantum field theory with a long-range gauge symmetry, we say the that there are \textit{dynamical charges in representation $\alpha$} if, in addition to the Wilson loops $W_\alpha$ and the boundary-attached  Wilson lines $W_{\alpha,ij}$, there are also Wilson lines labelled by $\alpha$  which have one or both endpoints on points in the interior of $\Sigma$; we call these interior endpoints ``charged operators in representation $\alpha$ ''.  These interior endpoints do not carry $G$ representation indices, but they do carry Lorentz indices which depend on the type of endpoint.\footnote{In Lagrangian gauge theories we can express these operators as the product with indices contracted of a gauge non-invariant Wilson line with an interior endpoint carrying an $\alpha$ represention index and a gauge non-invariant charged local operator at that endpoint carrying the conjugate index, hence our name for the interior endpoints, but we emphasize that it is only their combination which makes sense abstractly so that is what we define here.}  For Wilson lines with one endpoint in $\partial \Sigma$, we require that merging the interior endpoints of one such line and its conjugate gives a boundary-attached Wilson line $W_{\alpha,ij}$, while for Wilson lines with two endpoints in the interior of $\Sigma$ we require that  merging the conjugate endpoints of two such lines gives a Wilson loop in the same representation.  In both cases this merging requires a rescaling to get an operator with finite expectation value, see figure \ref{wilsongluefig} for an illustration and \cite{Harlow:2015lma} for more details on how the merging works.
\end{mydef}

The most obvious example of a theory with a long-range gauge symmetry with dynamical charges is obtained by adding some charged matter to the $d=4$ Maxwell theory \eqref{maxwell} in Minkowski space.\footnote{Strictly speaking this theory probably does not exist because of the Landau pole, but we can obtain it at low energies from some UV completion.} A more interesting example is quantum chromodynamics, which here we will define as an $SU(3)$ gauge theory with two massless Dirac fermions transforming in the fundamental representation of $SU(3)$, quantized in $AdS_4$ \cite{Aharony:2012jf}.  This theory has a dimensionless parameter, given by the strong coupling scale $\Lambda_{QCD}$ measured in units of the radius of curvature of $AdS_4$.  When this parameter is large the theory behaves as in Minkowski space: the quarks and gluons are confined into hadrons, and there are no finite energy states with nonzero color.\footnote{This still haven't been proven of course, but the conceptual, numerical, and experimental evidence is so overwhelming that we are happy to accept it as fact.}  There is therefore no long-range gauge symmetry.  As the parameter decreases however, eventually the quarks and gluons are liberated and the theory becomes perturbative \cite{Aharony:2012jf}.  Beyond this point the theory exhibits a long-range $SU(3)$ gauge symmetry with dynamical charges in the fundamental representation.  

This second example shows that a theory can have a long-range gauge symmetry in a background other than $\Rd$ even if it doesn't have it in $\Rd$. This may seem surprising, since we defined the existence of a global symmetry as a property of the theory on $\Rd$ which may or may not be preserved in other backgrounds.  The difference is that global symmetries have well-defined \textit{local} consequences: the local operators transform nontrivially and the stress tensor is invariant.  So ultimately we can study these on the simplest background, $\mathbb{R}^d$, and they are there or they aren't.  There is never a situation where a global symmetry is not present on $\mathbb{R}^d$ but is present somewhere else.   Long-range gauge symmetries, by contrast, are  properties of the \textit{phase} that the theory is in, via condition (3) in definition \ref{gaugedef}.  For example an observer of size $10^{-18}$ meters would look at QCD on $\mathbb{R}^4$ and see weakly coupled gluons, even though the theory is eventually confining and thus has no long-range gauge symmetry. Conversely an observer a theory with emergent gauge fields would look at short distances and see nothing resembling definition \ref{gaugedef}, even though at long distances there might be Wilson lines and massless gauge bosons.  That these two rather distinct notions are related via holographic duality, as we will see in more detail soon, is yet another manifestation of remarkable ``UV/IR connection'' \cite{Susskind:1998dq} of AdS/CFT .  

The reader may wonder why in condition (3) we have demanded that the ground state is invariant under the asymptotic symmetry, while in our definition \ref{globaldef} we took pains to include spontaneously broken global symmetries.  The reason is that unlike theories with spontaneously-broken global symmetries, gauge theories which in the Higgs phase do not really have any special properties which distinguish them abstractly from other quantum field theories. For example we will review in section \ref{phasesec} that in some cases there is not even a good distinction between a Higgs phase and a confining phase; they both are just some gapped system with no long-range gauge symmetry \cite{Fradkin:1978dv,Banks:1979fi}.  In AdS/CFT a bulk gauge theory in the Higgs phase is \textit{not} dual to a boundary theory with a spontaneously broken global symmetry: indeed the CFT is studied on spatial $\mathbb{S}^{d-1}$, so typically no spontaneous breaking of global symmetry is possible (there are certain topological exceptions, see footnote \ref{finiteVssb}).  

We will momentarily turn to the lattice to give a more systematic picture of definition \ref{gaugedef}, but first a technical aside.  We have found that our use in condition (3) of a temporal Wilson line to characterize the phase of QCD with fundamental quarks sometimes leads to confusion, since the more standard way of using a Wilson line to diagnose confinement, looking for an area-law scaling of the expectation value of the fundamental-representation Wilson loop, does not work when there are fundamental quarks \cite{Wilson:1974sk,Kogut:1974ag}.  The problem is that as we separate a pair of fundamental/anti-fundamental background color charges, the energy density in the color string between them will eventually pull a quark-antiquark pair out of the vacuum, which screens the charges and thus avoids the linear potential which would lead to an area law.  This problem also interferes with the recent ``unbroken one-form symmetry'' definition of confinement \cite{Gaiotto:2014kfa}, for basically the same reason.  It does \textit{not} however affect our condition (3), since by definition we are studying only states which transform nontrivially under $U(g,\partial \Sigma)$.  It is true that our temporal Wilson line might be screened by a dynamical charge nearby, but then there would need to be an unscreened dynamical charge elsewhere to ensure the state transforms correctly under $U(g,\partial\Sigma)$.  In a confining phase, the only way to avoid an infinite energy cost would be for this extra dynamical charge to be ``right at infinity''.  In Minkowski space we have excluded this by demanding that the energy density fall off at infinity in all states, while in AdS it is excluded automatically by the AdS potential, which assigns more and more energy to particles which are closer and closer to the boundary.  We illustrate this point in an exactly-soluble setting in subsection \ref{phasesec} below.

In the remainder of this section we will use lattice gauge theory to further motivate and analyze definition \ref{gaugedef}.  We will also make a few comments on the thorny question of the meaning of the topology of the gauge group, and briefly discuss a more general structure where global symmetries mix with long-range gauge symmetries.  Readers who are already satisfied with definition \ref{gaugedef}, and who feel no confusion about the difference between  $U(1)$ gauge theory and $\mathbb{R}$ gauge theory, or $SO(3)$ gauge theory and $SU(2)$ gauge theory, may wish to skip ahead to section \ref{symsec}.

\subsection{Hamiltonian lattice gauge theory for general compact groups}\label{latsec}
The details of definition \ref{gaugedef} may seem a bit arbitrary, so we now explain how they naturally arise in the framework of Hamiltonian lattice gauge theory \cite{Kogut:1974ag}.  Although this may seem like a detour, this framework has several very convenient features:
\bi
\item Lattice gauge theory may be defined for any compact Lie group $G$, discrete or continuous.  By contrast, many discrete gauge theories do not yet have simple continuum Lagrangian formulations.  Often the best one can do is start with a continuous gauge theory and then Higgs it to a discrete subgroup, but this includes a lot of extra machinery which is irrelevant for the discrete gauge theory. 
\item On the lattice, the topology of the gauge group is manifest from the beginning.  There can be no confusion between $SO(3)$ and $SU(2)$, or $U(1)$ and $\mathbb{R}$.  
\item The Hamiltonian formulation in particular is useful because it allows an explicit discussion of the Hilbert space and the structure of the operators which does not rely on knowing the Hamiltonian. Thus the operators we discuss should arise in any gauge theory, even if the Lagrangian has other terms (eg Chern-Simons or $\theta$ terms) beyond or instead of the standard Yang-Mills Lagrangian.
\item The phase structure of gauge theory is more clear on the lattice than in the continuum, and in particular in some limits it is exactly soluble.  This will enable us to illustrate the various possibilities in detail for the special case of gauge group $\mathbb{Z}_2$, where we will see explicitly that the phase boundary between allowing finite energy charges and not allowing them persists in the presence of dynamical charges.
\ei
We must however also acknowledge several shortcomings of the lattice approach:
\bi
\item It is not the continuum.  Although the structure we will see is consistent with our continuum expectations, and in particular with definition \ref{gaugedef}, in the end the lattice theory has a lot of extra ``short distance'' information which should all go off to infinite energy in the continuum limit.  We do not expect this to affect the phase structure, which is what we really care about, but ``expect'' and ``know'' are not the same thing.  
\item Our lattice presentation is still ultimately ``Lagrangian'': it makes reference to unphysical states, and uses a specific set of ``fundamental'' fields.  As we emphasized at the beginning of this section, different such presentations may flow to the same theory at long distances, and if we are not careful we might mislead ourselves about what to expect.  We will try to be careful.
\ei
With these comments out of the way, we now begin with the structure of Hamiltonian lattice gauge theory for arbitrary compact gauge group $G$.  

In mathematics the term ``lattice'' refers to a regular set of points in $\mathbb{R}^n$, but in lattice gauge theory it also includes a graph connecting those points. The vertices of this graph are called ``sites'', and each edge together with a choice of orientation is called a ``link''.  Links can be identified by a pair $(\vx,\vd)$, where $\vx$ is the starting point of the link and $\vd$ is the displacement vector to its endpoint.  The links $(\vx,\vd)$ and $(\vx+\vd,-\vd)$ describe the same edge with opposite orientations.  The basic idea of Hamiltonian lattice gauge theory is that each edge comes with a gauge field and each site comes with a gauge transformation which we quotient by, with any matter fields living on the sites.  The Hilbert space prior to imposing constraints is a tensor product
\be\label{bigH}
\mathcal{H}=\bigotimes_{e\in E}\mathcal{H}_e\bigotimes_{\vx\in X}\mathcal{H}_{\vx},
\ee
where $X$ is the set of sites, $E$ is the set of edges, each $\mathcal{H}_{\vx}$ is the Hilbert space of the matter fields at site $\vx$, and each $\mathcal{H}_e$ is a copy of the Hilbert space $\HG$ of a quantum-mechanical particle moving on the group manifold $G$.  $\HG$ is spanned by a set of states $|g\ran$, which are mutually orthogonal and normalized so that for any $g'$ we have
\be
\int dg \lan g'|g\ran=1,
\ee 
where $dg$ is the invariant Haar measure on $G$, normalized so that the volume of $G$ is one.  In particular if $G$ is discrete, then $\int dg$ is just a uniform average over group elements.  There are three natural families of operators on $\HG$:
\begin{align}\nonumber
W_{\alpha,ij}|g\ran&=D_{\alpha,ij}(g)|g\ran\\\nonumber
L_h|g\ran&=|hg\ran\\
R_h|g\ran&=|gh\ran.
\end{align}
Here $\alpha$ denotes an irreducible representation of $G$, $D_{\alpha,ij}(g)$ are the representation matrices of that representation, and $W_{\alpha,ij}$ is called the \textit{Wilson link in representation $\alpha$}.  $L_h$ and $R_h$ are called \textit{left and right multiplication operators}, if we view $U^\alpha_{ij}$ as analogous to the position operator in ordinary single-particle quantum mechanics then $L_h$ and $R_h$ are analogous to the momentum operator.  The hermiticity properties of these operators are 
\begin{align}\nonumber
W^\dagger_{\alpha,ij}&=W_{\alpha_*,ji}\\\nonumber
R_h^\dagger&=R_{h^{-1}}\\
L_h^\dagger&=L_{h^{-1}},
\end{align}
where as in definition \ref{gaugedef} we have taken $\dagger$ acting on $W_{\alpha,ij}$ to exchange $ij$ indices in addition to performing the Hilbert space adjoint. $\alpha_*$ is the conjugate representation of $\alpha$.  The algebra of these operators is determined by the following relations:
\begin{align}\nonumber
L_hL_{h'}&=L_{hh'}\\\nonumber
R_hR_{h'}&=R_{h'h}\\\nonumber
L_h R_{h'}&=R_{h'}L_h\\\nonumber
R_{h}^\dagger W_\alpha R_h&=W_\alpha D_\alpha(h)\\
L_h^\dagger W_\alpha L_h&=D_\alpha(h)W_\alpha,\label{gaugealg}
\end{align}
where in the last two equations we have suppressed representation indices using matrix multiplication.  This algebra is invariant under the exchange
\begin{align}\nonumber
L_h&\leftrightarrow R_{h^{-1}}\\
W_\alpha&\leftrightarrow W_\alpha^\dagger,\label{exchange}
\end{align}
and choosing a frame under \eqref{exchange} amounts to choosing an orientation for the edge.  To avoid confusion we will therefore always label Wilson links and left/right multiplication operators by links $(\vx,\vd)$ instead of edges, even though the operators for the two links corresponding to the same edge act on the same Hilbert space.

Gauge transformations are then defined to act at sites of the lattice, the action of a gauge transformation by a group element $g$ at site $\vx$ on the Hilbert space \eqref{bigH} is given by
\be\label{latticegauge}
U_{g}(\vx)\equiv \prod_{\vd} R_g^\dagger(\vx,\vd)V_g(\vx)=\prod_{\vd}L_g(\vx+\vd,-\vd)V_g(\vx),
\ee
where the product is over all $\delta$ such that the link exists and $V_g(x)$ is an additional unitary operator which implements the gauge transformation on any charged matter fields at site $\vx$.  Physical states are then required to be invariant under these transformations for arbitrary $g$ and $\vx$, with the possible exception of gauge transformations at boundary points as we discuss momentarily.  Under a general gauge transformation $\prod_{\vx}U_{g(\vx)}(\vx)$ the operators transform as
\begin{align}\nonumber
W_\alpha^\prime(\vx,\delta)&=D_\alpha(g(\vx+\vd))W_\alpha D_\alpha(g^{-1}(\vx))\\\nonumber
R^\prime_h(\vx,\vd)&=R_{g^{-1}(\vx)hg(\vx)}(\vx,\vd)\\\nonumber
L^\prime_h(\vx,\vd)&=L_{g^{-1}(\vx+\vd)hg(\vx+\vd)}(\vx,\vd)\\
\phi^\prime(\vx)&=D_{\alpha}(g(\vx))\phi(\vx),\label{latgaugetrans}
\end{align}  
where $\phi$ are matter fields transforming in representation $\alpha$ of $G$.  One obvious set of gauge-invariant operators are the Wilson loops 
\be
W_\alpha(C)\equiv\Tr\left(W_{\alpha}(\ell_N)\ldots W_{\alpha}(\ell_1)\right),
\ee
where $C$ is a closed curve consisting of the links $\ell_1,\ell_2,\ldots \ell_N$ in order.  If there are matter fields transforming in representation $\alpha$, then we also have gauge-invariant ``string'' operators
\be\label{stringop}
\phi_C(\vec{y},\vx)\equiv \phi^\dagger(\vec{y})W_\alpha(\ell_N)\ldots W_\alpha(\ell_1)\phi(\vx),
\ee
where now $C\equiv\{\ell_1,\ldots,\ell_N\}$ is a curve from point $\vx$ to point $\vec{y}$.

\bfig
\includegraphics[height=6cm]{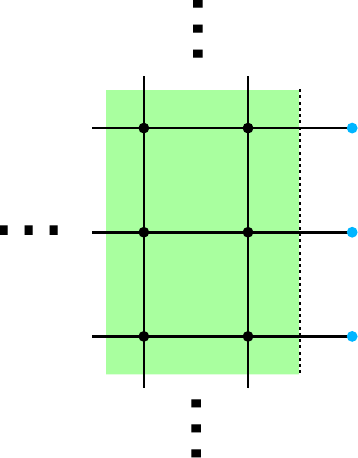}
\caption{Lattice points in the vicinity of a boundary: the blue dots are sites which are external to the green spatial region $R$, but which are endpoints of links which puncture its boundary $\partial R$.}\label{latticeboundaryfig}
\efig
We can also consider boundary conditions, in figure \ref{latticeboundaryfig} we illustrate a two-dimensional spatial lattice in the vicinity of a spatial boundary.  In constructing the Hilbert space, we need to decide whether or not we quotient by gauge transformations associated to the blue sites which are outside of the boundary but attached to links which pierce it.  If we do, then we are simply removing the degrees of freedom on these boundary-piercing links, so we are left with only the ``purely interior'' degrees of freedom.  In Maxwell theory this corresponds to setting $\star F$ to zero at the boundary, which is one way to satisfy the boundary term \eqref{maxvar} in the variation of the Maxwell action \eqref{maxwell}.  Alternatively if we do not quotient by the gauge transformations on the blue sites, in Maxwell theory this corresponds to setting the pullback of $A$ to the boundary to zero (note that there are no links connecting blue sites).  The latter boundary conditions are the natural ones in AdS/CFT, so we will adopt them here.   We then have three more interesting classes of gauge-invariant operators illustrated in figure \ref{bulkops}: 
\bi
\item \textit{Wilson lines}, defined by 
\be
W_{\alpha}(C)\equiv W_{\alpha}(\ell_N)\ldots W_{\alpha}(\ell_1),
\ee
where $C$ is a curve beginning with a link $\ell_1$ that pierces the boundary from the outside and ending with a link $\ell_N$ which pierces the boundary from the inside.
\item \textit{Wilson lines ending on charges}, which are defined similarly except that only one end pierces boundary;  the other is instead at a matter operator charged in either the same representation as the line or its conjugate representation, depending on the orientation.  For example if $\phi(\vec{x})$ is a scalar field in representation $\alpha$ at spatial point $\vx$, and $C\equiv\{\ell_1,\ldots \ell_N\}$ is a sequence of links connecting $\vx$ to the boundary, then
\be
\phi_C(\vx)\equiv W_\alpha(\ell_N)\ldots W_\alpha(\ell_1)\phi(\vx)
\ee
is a gauge-invariant operator.   
\item \textit{Localized asymptotic symmetries}, defined by
\be
U(g,R)\equiv \prod_{\ell\in R}L_{g}(\ell),
\ee
with $R$ a subset of the outward-pointing boundary-piercing links.  
\ei
\bfig
\includegraphics[height=5cm]{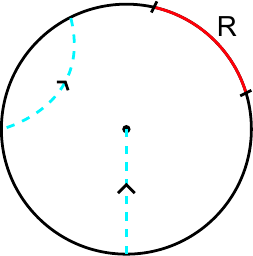}
\caption{Gauge invariant operators in the presence of a boundary.}\label{bulkops}
\efig
The reader can check using \eqref{gaugealg} that these operators have the properties described in conditions (1) and (2) of definition \ref{gaugedef}, and are also consistent with definition \ref{chargedef} if there is charged matter.  

To discuss condition (3) from definition \ref{gaugedef}, we need to introduce a Hamiltonian.  There is no unique choice of Hamiltonian, just as there is no unique choice of action, but one nice option is to take the Hamiltonian which is obtained from the standard Wilson action \cite{Wilson:1974sk} in the limit of continuous time \cite{Creutz:1976ch,Fradkin:1978th}.  In writing this Hamiltonian it will be convenient to allow Wilson lines in representations which are not irreducible: these are defined in the obvious way by direct sums of the Wilson lines in irreducible representations.  For convenience we will restrict to a cubic lattice, in which case there is a natural set of ``smallest loops'' called \textit{plaquettes}, and we will set the lattice spacing to unity.\footnote{More generally we can consider any lattice with the structure of a CW complex, see appendix \ref{stabilizerapp}.} The form of the Hamiltonian is different depending on whether the gauge group $G$ is discrete or continuous, for the continuous case we have the Kogut-Susskind Hamiltonian \cite{Kogut:1974ag}
\be\label{KSH}
H=\frac{g^2}{4}\sum_{\ell\in L}\sum_b P_b(\ell)P_b(\ell)-\frac{1}{g^2}\sum_{\gamma\in \Gamma}W_\alpha(\gamma).
\ee
Here $L$ is the set of (oriented) links, $\Gamma$ is the set of (oriented) plaquettes, $P_b$ is minus the Yang-Mills electric flux, defined by
\be
L_{e^{i\epsilon^b T_b}}\equiv e^{-i\epsilon^b P_b},
\ee
and $\alpha$ is a faithful but not necessarily irreducible representation of $G$.\footnote{We need to allow reducible representations because some compact groups do not have any faithful irreducible representations, two examples of such groups are $\mathbb{Z}_2\times \mathbb{Z}_2$ and $U(1)\times U(1)$.  By theorem \ref{liefaithfulthm}, a finite-dimensional faithful representation always exists for any compact Lie group.} The sum over plaquettes includes plaquettes which contain boundary-piercing links, in these plaquettes the Wilson line is defined to be the identity on links which are not part of the lattice.  We are here normalizing the Lie algebra generators in the representation $\alpha$ such that
\be
\Tr\left(T^{\{\alpha\}}_a T^{\{\alpha\}}_b\right)=\frac{1}{2}\delta_{ab},
\ee
so in the continuum limit this Hamiltonian matches onto the standard Yang-Mills Hamiltonian
\be
H=\int d^{d-1}x\left(\frac{g^2}{2} P_b^i P_b^i+\frac{1}{4g^2} F^b_{ij}F^{b,ij}\right),
\ee
with $P_b^i\equiv \frac{1}{g^2}F^{b,i0}$.  We note in passing that the Kogut-Susskind kinetic operator $\sum_a P_a P_a$ has a beautiful group-theoretic interpretation: for any compact Lie group, by Schur orthogonality and the Peter-Weyl theorem (see theorems \ref{schurthm} and \ref{pwthm}) the states
\be
|\alpha,ij\ran\equiv \frac{1}{\sqrt{d_\alpha}}\int dg D_{\alpha,ij}(g)|g\ran,
\ee
where $\alpha$ is any irreducible representation and $d_\alpha$ is its dimension, are an orthornomal basis for the Hilbert space $\mathcal{H}_G$ at each edge \cite{knapp2013lie}.  When $G$ is continuous, $\sum_a P_a P_a$ is then just the quadratic casimir of the Lie algebra representation associated to $\alpha$:
\be
\sum_aP_aP_a|\alpha,ij\ran=\sum_a T^{\{\alpha\}}_a T^{\{\alpha\}}_a |\alpha,ij\ran.
\ee

For discrete gauge groups, the continuous-time Wilson action instead leads to a Hamiltonian
\be\label{discreteham}
H=-\frac{g^2}{2}\sum_{\ell\in L} \sum_{h\in S}L_h(\ell)-\frac{1}{g^2}\sum_{\gamma\in \Gamma}W_\alpha(\gamma),
\ee
where $\alpha$ is again a faithful representation of $G$ and $S$ is the set of elements of $G$ which maximize the quantity $\Tr\left(D_\alpha(h)+D_\alpha(h^{-1})\right)$ as we vary over the set of group elements which are not the identity.  We describe how to obtain this somewhat unusual kinetic term in appendix \ref{gaugeapp}, we were unable to find it in the literature except in the special case $G=\mathbb{Z}_n$ \cite{Wegner:1984qt,Fradkin:1978th,Arakawa:2003ae}.

In either the discrete or continuous case, if we have scalar matter fields transforming in a representation $\beta$ of the gauge group then we should also add to the Hamiltonian a matter kinetic term 
\begin{align}\nonumber
H_{matter}=\sum_{\vec{x}}\Big(\pi(\vx)&\pi^\dagger(\vx)+m^2\phi^\dagger(\vx)\phi(\vx)\Big)\\\nonumber
-\frac{1}{2}\sum_{(\vx,\vec{\delta})\in L}\Big(&\phi^\dagger(\vx+\vec{\delta})W_\beta(\vx,\vec{\delta})\phi(\vx)+\phi^\dagger(\vx)W_\beta^\dagger(\vx,\vec{\delta})\phi(\vx+\vec{\delta})\\
&-\phi^\dagger(\vx+\vec{\delta})\phi(\vx+\vec{\delta})-\phi^\dagger(\vx)\phi(\vx)\Big).
\end{align}
Here the $\beta$ representation indices have been contracted in the obvious way.  If the matter fields themselves are also discrete, then a set of manipulations analogous to those for a discrete gauge field in appendix \ref{gaugeapp} will tell us what should replace $\pi\pi^\dagger$ in this Hamiltonian.  In fact the only example we will study in detail is an example of this type.

Finally we note that in this formalism we can introduce a temporal Wilson line in representation $\alpha$ which punctures our timeslice at site $\vx$, as required by condition (3) in definition \ref{gaugedef}, in the following manner.  We first extend the pre-constraint Hilbert space \eqref{bigH} by including a new tensor factor $\mathcal{H}_\alpha$ with Hilbert space dimensionality $d_\alpha$:
\be
\wt{\mathcal{H}}=\mathcal{H}\otimes \mathcal{H}_\alpha.
\ee
We then modify the gauge transformation \eqref{latticegauge} at site $\vx$ to be
\be\label{temporalw}
\wt{U}_g(\vx)\equiv U_g(\vx)D_\alpha(g),
\ee
where $D_\alpha(g)$ acts on our new tensor factor, and then instead of demanding physical states are invariant under $U_g(\vx)$ we instead demand that they are invariant under $\wt{U}_g(\vx)$.  The form of the Hamiltonian and the constraints away from $\vx$ are unmodified.  This illustrates clearly that temporal Wilson lines should \textit{not} be thought of as operators: they are modifications of the theory, and in particular introducing one changes the spectrum of Hamiltonian since different states become physical.

\subsection{Phases of gauge theory}\label{phasesec}
We now illustrate the notion of a long-range gauge symmetry in the simplest lattice gauge theory with charged matter: the $\mathbb{Z}_2$ gauge theory with a single discrete matter field $\wt{Z}=\pm 1$ transforming in the sign representation of $\mathbb{Z}_2$.  Since every element of $\mathbb{Z}_2$ is its own inverse, there is no meaning to the orientation of links.  It is therefore convenient to relabel the gauge field operators 
\begin{align}\nonumber
Z(e)&\equiv W_{sign}(\ell)=W_{sign}(-\ell)\\
X(e)&\equiv L_{-1}(\ell)=L_{-1}(-\ell),
\end{align}
so that we have the Pauli algebra $Z^2=X^2=1$, $ZX=-XZ$.  The matter fields are $\wt{Z}$ and its conjugate $\wt{X}$, which again obey the Pauli algebra.  Since we want the ground state to be invariant under gauge transformations, the natural boundary condition for the matter fields (analogous to $\phi=0$ in scalar electrodynamics) is to not include matter fields on the blue sites in figure \ref{latticeboundaryfig}.  The Hamiltonian is
\begin{align}\nonumber
H=&-g^2\sum_{e\in E}X(e)-\frac{1}{g^2}\sum_{\gamma\in\Gamma}Z(\gamma)\\
&-\lambda\sum_{\vx}\wt{X}(\vx)-\frac{1}{\lambda}\sum_{e\in E}\wt{Z}(e_+)Z(e)\wt{Z}(e_-),\label{Z2ham}
\end{align}
where $e_+$ and $e_-$ denote the two sites at the ends of $e$ and $Z(\gamma)=W_{sign}(\gamma)$, the sum over $\vx$ in the term proportional to $\lambda$ does not include the blue sites in figure \ref{latticeboundaryfig}, and the sum over $e$ in the term proportional to $1/\lambda$ does not include boundary-piercing links. The phase diagram of this model as a function of $\lambda$ and $g$ was studied in detail in \cite{Fradkin:1978dv} (see \cite{Banks:1979fi} for a similar analysis of the $U(1)$ case).  We here just review a few limits to illustrate the power of condition (3) in definition \ref{gaugedef} for characterizing this phase diagram.  In discussing the phase diagram it will sometimes be convenient go to the ``unitarity gauge'' $\wt{Z}=1$, after which the Hamiltonian can be expressed entirely in terms of the gauge degrees of freedom:
\be
H=-g^2\sum_{e\in E}X(e)-\frac{1}{g^2}\sum_{\gamma\in \Gamma}Z(\gamma)-\lambda \sum_{\vec{x}}\prod_{\delta}X(\vx,\delta)-\frac{1}{\lambda}\sum_{e\in E}Z(e).\label{Uham}
\ee

\bfig
\includegraphics[height=5cm]{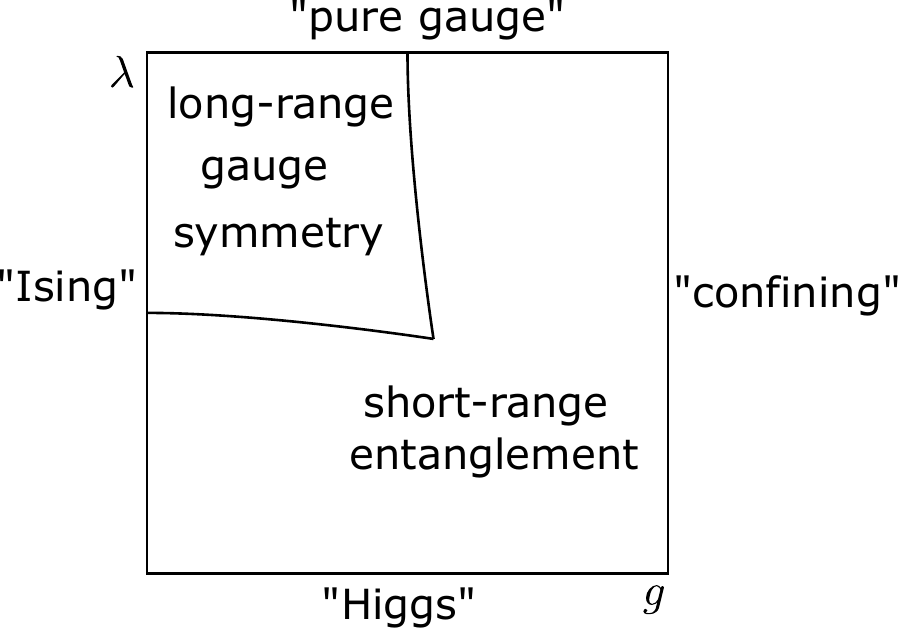}
\caption{The Fradkin-Shenker phase diagram of $\mathbb{Z}_2$ lattice gauge theory with a charged matter field for $d\geq 3$.  The presence of a long-range gauge symmetry is what distinguishes the ``topological'' or ``free-charge'' phase from the ``Higgs-confining'' phase, which has only short-range entanglement.  This phase boundary exists even though there is a matter field which is charged in the fundamental representation of the gauge group.}\label{z2phasefig}
\efig
We show the phase diagram for this model from \cite{Fradkin:1978dv} in figure \ref{z2phasefig}.  We can motivate it by considering a few limits:
\bi
\item \textbf{Large $g$, finite $\lambda$:} In this limit the Hamiltonian is dominated by $-g^2\sum_{e\in E} X(e)$.  The ground state therefore has $X=1$ on all links, which by the gauge constraint means that $\wt{X}=1$ on all sites.  As this is a product state, there is no long-range correlation.  In unitarity gauge we can reach all other eigenstates by acting with subsets of the $Z(e)$ on this state: each $Z(e)$ we act with creates a string with two charges at the endpoints, as in equation \eqref{stringop}, and the energy of any such eigenstate is just proportional to the length of all strings.  In this limit the theory is therefore in what we might call a ``confining phase'': a string which connects any finite point to infinity necessarily involves a linearly divergent energy, and without such a string we cannot have a state which is charged under $U(g,\partial \Sigma)$ unless we put a charge ``right next to the boundary'', but this is precisely what our insistence on restricting to states where $T_{00}$ decays at infinity (or just being in AdS) prevents.   In this limit we therefore have no long-range gauge symmetry, since we fail condition (3) of definition \ref{gaugedef}.

\item \textbf{Small $\lambda$, finite $g$:}   In this limit the unitarity-gauge Hamiltonian \eqref{Uham} is dominated by $-\frac{1}{\lambda}\sum_{e\in E}Z(e)$, so the ground state in unitarity gauge has $Z=1$ on all links except for the boundary-piercing ones.  This is again a product state, so there is no long-range correlation.  Excited states are produced by acting with $X(e)$, and the energy again scales with the number of $X(e)$ we act with.  Since the behavior of the boundary-piercing links differs from the rest of the space, the stress tensor does not go to zero at infinity and condition (3) of definition \ref{gaugedef} is violated.\footnote{This may seem artificial, what is really going on here is that in this limit it is more natural to instead choose boundary conditions where we have $\wt{Z}=1$ on the blue sites in figure \ref{latticeboundaryfig}, and where we then include the boundary-piercing links in the $1/\lambda$ term; we then just have $Z=1$ on all links in the ground state.  This state however is not invariant under the asymptotic symmetry, as we expect for the Higgs vacuum, so it still violates condition (3).}  This phase might be called the ``Higgs phase'', since $\lambda$ behaves like the inverse of the radius of the Higgs field in the Abelian Higgs model, but in fact one of the main points of \cite{Fradkin:1978dv} is that this phase is continuously connected to the previous one, after all the excitations are string-like in both cases, so calling one ``confining'' and the other ``Higgs'' is not really sensible: it is better to just say that both have short-range entanglement and no long-range gauge symmetry.

\item \textbf{Large $\lambda$, finite $g$:} In this limit the term $-\lambda \sum_{\vx}\wt{X}(\vx)$ just sets $\wt{X}=1$ everywhere, so the matter field drops out of the Gauss constraint and we are just left with pure $\mathbb{Z}_2$ lattice gauge theory. At large $g$ this is in the ``confining phase'' we discussed above, with $X=1$ on every link in the ground state. We discuss the small $g$ limit momentarily, but, for spacetime dimension $d\geq 3$, as we decrease $g$ one expects a phase transition at some finite value of the coupling \cite{Wegner:1984qt}.  

\item \textbf{Small $g$, finite $\lambda$:} In the strict $g=0$ limit, for $d\geq 3$ the plaquette term sets all $Z=1$ so the Hamiltonian \eqref{Z2ham} just becomes that of the quantum transverse field Ising model.  This again has a phase transition at some finite value of $\lambda$.  There is no gauge field left, so there is no long-range gauge symmetry.  This transition persists when $g$ is small but nonzero, at small $\lambda$ we should still be in the ``Higgs'' regime, but as $\lambda$ increases the Ising transition moves us to a different phase, which we now study.

\item \textbf{Small $g$, large $\lambda$:}  This is the fun regime. In unitarity gauge, the Hamiltonian becomes
\be\label{toric}
H=-\frac{1}{g^2}\sum_{\gamma\in \Gamma}Z(\gamma)-\lambda \sum_{\vec{x}}\prod_{\vec{\delta}}X(\vx,\vec{\delta}),
\ee
which is sometimes called the ``toric code'' Hamiltonian  \cite{Kitaev:1997wr}.  These terms couple different links together, so the ground state will not be a product state and there is the possibility of some kind of interesting long-range correlation.  In \cite{Kitaev:1997wr} it was pointed out that one way to characterize this long-range correlation is to study the theory on closed spatial manifolds with nontrivial topology.  On such manifolds, the hamiltonian \eqref{toric} has a nontrivial ground state degeneracy, which depends in an interesting way on the choice of manifold.  This certainly is not true for the trivial product ground states we found in the previous limits, which give a unique ground state on any manifold.  Indeed in this limit the space of zero energy states is precisely that of a nontrivial topological field theory, the pure $\mathbb{Z}_2$ gauge theory.  For our purposes however we are instead interested in the excited states of this theory in infinite volume, which are nontrivial even when the spatial topology is trivial.  To understand these excitations, we need to first understand the ground state.  As explained in \cite{Kitaev:1997wr}, the Hamiltonian \eqref{toric} is nicely understood using the machinery of stabilizer codes \cite{Gottesman:1997zz}.  We review this machinery briefly in appendix \ref{stabilizerapp}, where we use it to show that on a spatial cubic lattice with our choice of boundary conditions, the Hamiltonian \eqref{toric} has a unique ground state, on which $\prod_{\delta}X(\vx,\delta)$ and $Z(\gamma)$ both act as the identity for all $\gamma$ and $\vx$ (we also compute the ground state degeneracy for any lattice which discretizes a spatial $d-1$-manifold, with or without boundary, in terms of topological invariants of that manifold).  We may then ask how creating a charged excitation changes the energy.  For example we can act on this ground state with a line of $Z$ operators which extends from a boundary-piercing link to some finite point $\vx_0$ in the center of the lattice.  This operator clearly commutes with all $Z(\gamma)$, and in fact it commutes with almost all $\prod_{\vec{\delta}}X(\vx,\delta)$ as well.  The only term in the Hamiltonian \eqref{toric} it does not commute with is $\prod_{\vec{\delta}}X(\vx_0,\delta)$, which it anticommutes with instead.  Therefore acting with this operator on the ground state increases the energy by $2\lambda$, which obviously is finite even in infinite volume.  Thus this phase allows finite-energy charged excitations: in \cite{Fradkin:1978dv} it was called the ``free charge'' phase for this reason, we instead say that there is a $\mathbb{Z}_2$ long-range gauge symmetry.
\ei
Thus we see that condition (3) in definition \ref{gaugedef} is indeed sufficient to distinguish the two phases in diagram \ref{z2phasefig}, even though the Wilson loop has a perimeter scaling in both phases.\footnote{Note that we did not need to use a temporal Wilson line to check condition (3), since we could just directly use the dynamical charge $\wt{Z}$. The analysis would have been identical using a temporal Wilson line: given the modified constraint \eqref{temporalw}, we have a new set of gauge-invariant operators which are simply Wilson lines which connect the boundary to the location of the temporal Wilson line.  Their energetics work in the same way as Wilson lines which end on dynamical charges.}  On one side of the phase boundary there is a long-range gauge symmetry, while on the other side there isn't.

\subsection{Comments on the topology of the gauge group}\label{gaugetopsec}
In lattice gauge theory with no charged matter, the topology of the gauge group is explicitly included in the formulation of the theory.  This may at first seem to be in some tension with the fact that if $G$ and $G'$ are connected Lie groups with isomorphic Lie algebras, then for $d>2$ they have identical continuum Yang-Mills path integrals on $\mathbb{R}^d$.  In more detail, we can define the boundary conditions on the Yang-Mills field in $\mathbb{R}^d$ by conformally compactifying to $\mathbb{S}^d$.  $G$ and $G'$ have the same set of principal bundles over $\mathbb{S}^d$, as well as the same set of connections on those bundles, and therefore the sum over bundles and connections on those bundles is the same for $G$ and $G'$.\footnote{To see that the bundles are the same, note that $\mathbb{S}^d$ is constructed from the union of two balls, each of which is contractible and has boundary $\mathbb{S}^{d-1}$.  Principal $G$ bundles over $\mathbb{S}^d$ are therefore classified by  $\pi_{d-1}(G)$.  Since $G$ and $G'$ are connected and share a Lie algebra, they are each a quotient of the same connected simply-connected covering group $\wt{G}$ by some discrete central subgroup (see theorem \ref{liethm}).  Using basic properties of covering spaces we then have $\pi_{d-1}(G)=\pi_{d-1}(G')=\pi_{d-1}(\wt{G})$ for $d>2$ \cite{hatcher2002algebraic}.  Since connections on these bundles are Lie-algebra-valued one-forms, they will then clearly also coincide for $G$ and $G'$.}  The global information about the gauge group in the lattice theory is lost in the continuum limit because integrals over group variables on the edges of the lattice are dominated by group elements which are close to the identity.  But does this then mean that if $G$ and $G'$ have the same Lie Algebra, then pure Yang-Mills theory on $\mathbb{R}^d$ with gauge group $G$ is identical to pure Yang-Mills theory on $\mathbb{R}^d$ with gauge group $G'$?  This question was studied in detail in \cite{Aharony:2013hda}, where it was argued that in fact they are different.  We basically agree with their reasoning and their conclusion, but as our perspective is different in emphasis we now briefly present it.\footnote{If there are charged matter fields then the meaning of the topology of the gauge group is sometimes more obvious: for example an $SU(2)$ gauge theory with matter in the fundamental representation of $SU(2)$ cannot be viewed as an $SO(3)$ gauge theory, since the $SU(2)$ fundamental is not a representation of $SO(3)$.}  

The main point of \cite{Aharony:2013hda} was that, although the Yang-Mills path integral is identical on $\mathbb{R}^d$ for gauge group $G$ and gauge group $G'$, the set of line and surface operators is actually different.  What we want to emphasize here is that this statement is true \textit{despite} the fact that the Hilbert space and Hamiltonian of these theories on spatial $\Rdd$ are identical.  This may seem paradoxical: operators are just maps from Hilbert space to itself, so how can two theories with the same Hilbert space have different operators?  The resolution of this puzzle is that the operators exist either way, it is only their interpretation which is different.  This is possible because, as we reviewed in section \ref{notsec}, there is additional algebraic structure in quantum field theory beyond just the set  of all operators on Hilbert space.  Namely, for each spatial subregion $R$ we must have an associated subalgebra $\mathcal{A}[R]$ of the full set of operators.  Until we have decided which subalgebras are associated with which spatial regions, we have not fully specified a quantum field theory.  We now illustrate this for the simplest example: $G=\mathbb{R}$ and $G'=U(1)$.  

In fact we already discussed the difference between $\mathbb{R}$ and $U(1)$ gauge theory for $d=4$ in section \ref{anomsplit}: the $\mathbb{R}$ theory has more Wilson lines, since the representations of $\mathbb{R}$ are continuous, but the $U(1)$ theory has 't Hooft lines which the $\mathbb{R}$ theory lacks.  We now expand a bit more on this point.  On $\Rd$, neither $U(1)$ nor $\mathbb{R}$ have nontrivial bundles, so we may simply define the Hilbert space to be null space of the Gauss constraint in a Hilbert space spanned by a set of states labeled by spatial configurations of a one-form $A_\mu$. Acting on this Hilbert space we may consider the set of two-dimensional  operators
\be
W_\alpha(D)=e^{i\alpha \int_D F},
\ee
with $D$ a spatial disk, and the set of codimension-two operators
\be
T_\beta(B)=e^{\frac{2\pi i}{q^2}\beta\int_B \star F},
\ee
where $B$ is a $d-1$ dimensional spatial ball.  These operators are clearly gauge-invariant for any real $\alpha$ and $\beta$, and it would be silly to say that one or the other doesn't exist.  The nontrivial point however is that there are certain collections of special values of $\alpha$ and $\beta$ for which we can interpret the $W_\alpha$ as one-dimensional loop operators on $\partial D$ and the $T_\beta$ as $d-3$ dimensional closed surface operators on $\partial B$ without violating commutativity at spacelike separation: for $\alpha$ and $\beta$ in such a set, $W_\alpha$ and $T_\beta$ commute even if $\partial D$ and $\partial B$ are linked in space (see \cite{Buchholz:2015epa,Buchholz:2016yqp,Buchholz:2018npi} for related discussion).  The former are then referred to as Wilson lines and the latter as 't Hooft surfaces.  These sets are not all mutually compatible, so we need to make a definite choice which one to adopt.  The simplest such collection allows $\alpha$ to be an arbitrary number but requires $\beta$ to vanish: making this choice is equivalent to choosing the gauge group to be $\mathbb{R}$.  Another good choice is to take $\alpha$ and $\beta$ to both be integers: this is equivalent to choosing the gauge group to be $U(1)$ with coupling $q$.  More generally what we need is the Dirac quantization condition
\be
\alpha\beta\in \mathbb{Z}
\ee
for all allowed $\alpha$ and $\beta$: up to a rescaling of $q$ all other choices for the allowed set are equivalent to either $\beta=0$, $\alpha$ arbitrary or $\alpha,\beta\in \mathbb{Z}$.  

This discussion hopefully makes it clear that on $\Rd$ the distinction between $\mathbb{R}$ and $U(1)$, or more generally between $G$ and $G'$, is ``semantic''.  The reader may object that we should therefore instead just view the $G$ and $G'$ theories on $\Rd$ as being identical.  We disagree: as emphasized in \cite{Aharony:2013hda}, once we study these theories on other spacetime topologies they have different principal bundles and they really are different.  For example the spectrum of the Hamiltonian in the $\mathbb{R}$ gauge theory is continuous on a spatial torus, while in the $U(1)$ gauge theory it is discrete.  These distinctions arise because on more complicated topologies we can have loops and codimension-three closed surfaces which are not boundaries, so there can be Wilson loops and 't Hooft surfaces which cannot be realized as integrals of the field strength.  We view it is a major advantage of demanding the additional structure of a local net of operator algebras on $\Rd$ that it forces us to acknowledge the distinction between $U(1)$ and $\mathbb{R}$ \textit{without needing to go to other topologies}. 

For some readers this may still feel a bit abstract however: wouldn't it be better if we could just do an experiment?  For example in real quantum electrodynamics is the gauge group $U(1)$ or $\mathbb{R}$?  One possibility would be to argue that this question is semantic and therefore meaningless, but this is clearly false.  For example we might tomorrow observe a magnetic monopole, in which case we would immediately know the gauge group is $U(1)$.   Moreover if we are lucky, that monopole might have the minimal charge allowed by Dirac quantization (meaning $\beta=1$), in which case the set of allowed $\alpha$ and $\beta$ would be determined once and for all, as would the gauge group of electrodynamics.  Alternatively if we could convince ourselves we'd discovered a particle of charge $\sqrt{2}$, we would immediately know the gauge group is $\mathbb{R}$.\footnote{Convincing ourselves of this would probably be impossible, since we always measure charge with finite precision.  A version of this which is more practical would be discovering a heavy particle in the fundamental representation of $SU(3)$ color which was neutral under the electroweak $SU(2)\times U(1)$, which would immediately tell us that the gauge group of the standard model is $SU(3)\times SU(2)\times U(1)$ instead of $(SU(3)\times SU(2)\times U(1))/\mathbb{Z}_6$.}   Absent such discoveries, we are in a situation where indeed one might say that we do not know whether the gauge group of electrodynamics is $U(1)$ or $\mathbb{R}$.  As Bayesians however, it would be crazy to ignore the observational fact that the charges of the electron and proton are exact opposites to within one part in $10^{21}$ \cite{Patrignani:2016xqp}.  By far the most plausible explanation of this remarkable agreement is that the gauge group of electrodynamics is indeed $U(1)$, which presumably is why this is the terminology most people use.  

In fact one of the main goals of this paper is to argue that in quantum gravity dynamical objects exist carrying all charges allowed by the topology of the gauge group (conjecture \ref{allcharge}), which is precisely saying that in quantum gravity on $\Rd$ (or $AdS_d$) we will never be in the situation where the gauge group is ambiguous.  This is quite plausible also from the point of view that quantum gravity should include a sum over topologies, since on general topologies the gauge group is unambiguous.  Indeed our argument for conjecture \ref{allcharge} in AdS/CFT will be based on a refined version of this observation.

\subsection{Mixing of gauge and global symmetries}\label{mixsec}
There are interesting situations where global symmetries can combine with long-range gauge symmetries to make a more general kind of structure.\footnote{This section was inspired by a discussion with Thomas Dumitrescu.}  Rather then attempting a general discussion of this phenomenon, we will just give a simple example.  Namely, consider two free complex scalar fields in $d=4$ dimensions.  This theory has a $U(2)$ global symmetry.  We may then turn on a dynamical gauge field for the diagonal $U(1)$ subgroup.  What is global symmetry group of the resulting theory?  A first guess might be $SU(2)$, but this wrong because the central element 
\be 
g_c\equiv\begin{pmatrix}-1 & 0\\0 & -1\end{pmatrix}
\ee
is actually a long-range gauge transformation; it acts trivially on all local operators, violating condition (c) of definition \ref{globaldef}.  We might then guess $SU(2)/\mathbb{Z}_2$, but this group is not represented accurately on the full Hilbert space.   For example  the group element $\begin{pmatrix}i&0\\0&-i\end{pmatrix}$ squares to $g_c$, which is represented nontrivially on the Hilbert space as an element of the long-range gauge symmetry group, instead of squaring to the identity like it would in $SU(2)/\mathbb{Z}_2$.  One way of describing this situation is to say that the global symmetry group is indeed $SU(2)/\mathbb{Z}_2$, but that it is realized on the Hilbert space in the generalized kind of projective representation discussed in appendix \ref{projapp}, which allows the phase $\alpha$ from equation \eqref{projrep} to depend on the total electric charge.  This is one way to think about it, but we think a better description is to say that, rather than having a separate global symmetry and long-range gauge symmetry, the two are mixed together into a new kind of symmetry with symmetry group $U(2)$.  Clearly more could be said about this, but we leave it for future work.  We note now however that our argument against global symmetries in quantum gravity will rule out this possibility as well.

\section{Symmetries in holography}\label{symsec}
Having at last established our notions of global symmetry (definition \ref{globaldef}) and long-range gauge symmetry (definition \ref{gaugedef}) in 
quantum field theory, we are in a position to move on to quantum 
gravity and begin establishing conjectures \ref{nosym}-\ref{compact} in 
AdS/CFT.  Along the way we will also clarify the duality between global 
symmetries in the boundary theory and long-range gauge symmetries in the bulk.

\subsection{Global symmetries in perturbative quantum gravity}
To argue that there are no global symmetries in quantum gravity, we need to first acknowledge that our definition \ref{globaldef} of global symmetry, which is for quantum field theories, needs to be modified to deal with the following two issues:
\bi
\item General relativity has a long-range spacetime gauge symmetry, diffeomorphism invariance, which precludes the existence of any strictly local gauge-invariant operators.  Since condition (c) in definition \ref{globaldef} required global symmetries to act faithfully on the local operators, that definition becomes trivial. 
\item We do not yet have a complete bulk theory of quantum gravity, and our understanding based on effective field theory applies only in restricted situations.  Since we are trying to rule out \textit{exact} global symmetries, we need to say something about how they are defined in regimes which go beyond the validity of effective field theory.  
\ei
We postpone the second point to the next subsection, here we first address the question of how to define global symmetries in gravitational theories within the framework of effective field theory coupled perturbatively to gravity. 

We begin by recalling a few basic facts about the long-range diffeomorphism symmetry of gravity in asymptotically-AdS spacetime.  In any asymptotically-AdS spacetime, the geometry is required to approach the AdS metric\footnote{So far we have used $d$ to denote the spacetime dimension of whatever quantum field theory we are considering.  Since we now will be considering both the bulk gravity theory and its dual conformal field theory, we now adopt the standard convention that the boundary CFT has $d$ spacetime dimensions.}
\be
ds^2=-(r^2+1)dt^2+\frac{dr^2}{r^2+1}+r^2 d\Omega_{d-1}^2
\ee
at large $r$.  As in our discussion of $U(1)$ gauge theory below equation \eqref{maxvar}, we should only consider diffeomorphisms which preserve these boundary conditions, and moreover we should quotient only by those diffeomorphisms which become trivial at large $r$ \cite{Henneaux:1985tv}.  The diffeomorphisms which are nontrivial at large $r$ but nonetheless preserve the boundary conditions are precisely those which approach isometries of $AdS_{d+1}$, so the quotient of the set of diffeomorphisms which approach isometries by the set of diffeomorphims which become trivial is isomorphic to the group of $AdS_{d+1}$ isometries, $SO(d,2)$.\footnote{If there are fermions then this group is instead $Spin(d,2)$.  When $d=2$ the symmetry is enhanced to Virasoro symmetry, but we will not make use of this.}  Physical states and operators must both be invariant under diffeomorphisms which become trivial at infinity, but they will mostly transform in nontrivial representations of the quotient group $SO(d,2)$, which we will refer to as the \textit{asymptotic conformal symmetry}: it is a spacetime version of a long-range gauge symmetry.

It is clear that any strictly-local bulk operator will not be invariant under the set of diffeomorphisms which become trivial at infinity (unless it is topological, which is a situation we don't consider here).  To define a physical observable, we therefore need to introduce some gravitational analogue of the Wilson lines extending from the boundary to an interior point which we used to define operators carrying gauge charge in definition \ref{chargedef}.  In bulk effective field theory coupled perturbatively to gravity, we can construct such operators as ``gravitationally dressed'' versions of ordinary local operators.  The idea is to introduce a ``cutoff surface'' at some large but finite $r=r_c$, choose a point $x\equiv(r_c,t,\Omega)$ on this surface, fire a spatial geodesic into the bulk from $x$ of proper length 
\be
\ell\equiv\hat{\ell}+\log r_c,
\ee
and with tangent vector at $x$ of the form 
\be
\xi=-(r_c+\hat{\xi}^r/r_c)\partial_r+(\hat{\xi}^i/r_c^2) \partial_{i},
\ee
where the $i$ index runs over $t$ and $\Omega$, and then insert a local operator at the bulk endpoint $\mathring{x}$ of this geodesic.  In the limit $r_c\to\infty$ the quantities $\hat{\ell}$ and $\hat{\xi}^\mu$ are finite, and the choice of cutoff surface induces a residual conformal frame on the boundary.  If the operator we insert at the bulk endpoint is a scalar, then this construction defines a nonlocal operator which is invariant under diffeomorphisms which become trivial at infinity.  It is labelled by a boundary point $(t,\Omega)$, a renormalized tangent vector $\hat{\xi}^\mu$, and a renormalized geodesic distance $\hat{\ell}$.  We will refer to such an operator as a gravitationally-dressed scalar, and we illustrate one in the left diagram of figure \ref{conformalfig}.  If the local operator we insert in the bulk has tensor and/or spinor indices, then further dressing is necessary: the natural dressing, which we will adopt, is to pick the components of any such an operator in a frame which we parallel transport in from $x$ along the dressing geodesic. For example if $V^\mu$ is a vector field at the bulk endpoint $\mathring{x}(x,\hat{\ell},\hat{n})$ of a dressing geodesic, and $P_\mu^{\phantom{\mu}\nu}(\mathring{x},x)$ is the matrix which parallel transports a one-form along this geodesic from $x$ to $\mathring{x}$, then the operator
\be
\wt{V}^\mu(x,\hat{\ell},\hat{n})\equiv P_\nu^{\phantom{\nu}\mu}(\mathring{x},x)V^\nu(\mathring{x})
\ee
is invariant under diffeomorphisms which become trivial at the cutoff surface. Explicitly
\be\label{transport}
P_\nu^{\phantom{\nu}\mu}(\mathring{x},x)=\left(P \exp \left[\int_0^{\hat{\ell}+\log r_c}ds\frac{d\xi^\lambda}{ds}\Gamma_\lambda^T\right]\right)_\nu^{\phantom{\nu}\mu},
\ee
where $\Gamma_\lambda$ is the matrix with components $\left(\Gamma_\lambda\right)^\mu_{\phantom{\mu}\nu}\equiv \Gamma^\mu_{\lambda\nu}$, and $\xi^\mu$ is the tangent vector to our dressing geodesic, parameterized by proper length $s$,  so the resemblance to an ordinary Wilson line is quite clear. In particular under diffeomorphisms we have
\be\label{Ptrans}
P_{\phantom{\prime}\mu}^{\prime\phantom{\mu}\nu}(\mathring{x}^\prime,x^\prime)=\frac{\partial \mathring{x}^\alpha}{\partial \mathring{x}^{\prime\mu}}\frac{\partial x^{\prime \nu}}{\partial x^{\beta}}P_\alpha^{\phantom{\alpha}\beta}(\mathring{x},x),
\ee
so in defining $\wt{V}$ we have indeed traded in a ``bulk'' tensor index for a ``boundary'' one.  To all orders in perturbation theory around a fixed background, two operators constructed in this manner will commute if their dressing geodesics are spacelike-separated by a finite amount in that background.\footnote{The reader may consult \cite{Heemskerk:2012np,Kabat:2013wga,Almheiri:2014lwa,Donnelly:2015hta,Donnelly:2016rvo,Donnelly:2015taa,Giddings:2018umg} for more details on the algebra of these kinds of operators.}

\bfig
\includegraphics[height=6cm]{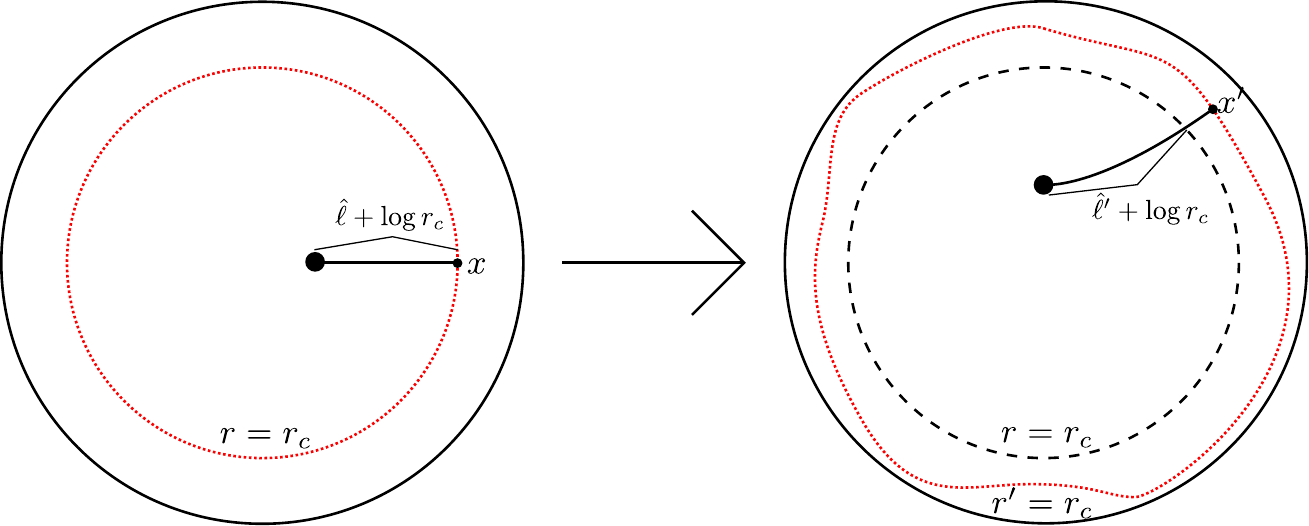}
\caption{The action of asymptotic conformal symmetry on a gravitationally-dressed local operator.  The transformation will in general change the cutoff surface to a new one, shown with red dots in the right diagram, so to define the transformed operator with respect to the old cutoff surface, shown with the black dashes, we need to change $\hat{\ell}$ and $\hat{\xi}$.  At finite $r_c$ there is also a change of $(t',\Omega')$ as we follow the geodesic in the right diagram from the new cutoff surface back to the old one, but this vanishes as $r_c\to\infty$.}\label{conformalfig}
\efig
It is instructive to consider the transformation properties of gravitationally-dressed local operators under the asymptotic conformal symmetry.  At first one might expect that this symmetry acts trivially on $\hat{\ell}$ and $\hat{\xi}$, since they are defined geometrically, but in fact it does not.  The reason is shown in figure \ref{conformalfig}: asymptotic conformal symmetries act nontrivially on the cutoff surface $r=r_c$, so acting on a dressed local operator with an asymptotic conformal symmetry sends it to an operator whose dressing geodesic is attached to a new cutoff surface.  We therefore need to change $\hat{\ell}$ and $\hat{\xi}$ to give the new location of the operator in terms of the old cutoff surface, since otherwise we would not be defining an action within a set of operators which are all defined in the same way.  We therefore have a transformation law
\be\label{conftrans}
\wt{\phi}_{a'}^\prime(t',\Omega',\hat{\ell}',\hat{\xi}')=D_{a'}^{\phantom{a'}a}\wt{\phi}_a(t,\Omega,\hat{\ell},\hat{\xi}),
\ee
where $a$ denotes a collection of Lorentz indices located at $x$,  $a'$ denotes the same collection at $x'$, and the matrix $D_{a'}^{\phantom{a]}a}$ is determined from the transformation \eqref{Ptrans} together with an additional parallel transport from the ``new'' cutoff surface back to the ``old'' one.\footnote{This business of rewriting things using the old cutoff surface is the holographic dual of the standard fact that in conformal field theory, each conformal transformation is a combination of a diffeomorphism with a Weyl transformation to return the metric to its original form (this is why for example a scalar can transform with a nontrivial conformal weight even though it is in a trivial Lorentz representation).}  Note that the transformation \eqref{conftrans} depends only on the geometry in the asymptotic region: as in electromagnetism, the identity component of the conformal group is generated by a set of local boundary integrals constructed by contracting the asymptotic Killing vectors with the \textit{boundary stress tensor} $T_{\mu\nu}$.  In AdS/CFT we can define $T_{\mu\nu}$ as simply being the CFT stress tensor, but it also has a bulk definition \cite{Balasubramanian:1999re} as the derivative of the bulk path integral with respect to the ``boundary'' metric
\be
\gamma^{\{boundary\}}_{\mu\nu}\equiv r_c^{-2} \gamma_{\mu\nu},
\ee  
where $\gamma_{\mu\nu}$ is the induced metric on the cutoff surface.

With these preliminaries out of the way, we can now give a definition of (internal) global symmetry with symmetry group $G$ in gravitational effective field theory in asymptotically-AdS space.   The basic idea is to define such a symmetry as a homomorphism from $G$ into the unitary operators on the Hilbert space which faithfully acts by conjugation on the set of gravitationally-dressed local operators, preserving the boundary point $x$, renormalized distance  $\hat{\ell}$, and renormalized tangent vector $\hat{\xi}$.  We moreover require that the symmetry operators commute with the boundary stress tensor $T_{\mu\nu}$, and therefore with the asymptotic conformal symmetry.  This definition however is not quite satisfactory, for two reasons.  First of all, in definition \ref{globaldef} we required global symmetries not just to act locally on local operators, but indeed to preserve the algebra $\mathcal{A}[R]$ of \textit{all} operators in any spatial region $R$.  In quantum field theories where all operators in $\mathcal{A}[R]$ are generated from local operators in $R$ this is automatic, but this not true in all quantum field theories; in fact we met several examples where it isn't in section \ref{globalsec}.  We can address this by requiring that global symmetries also act locally on ``gravitationally-dressed surface operators'', meaning operators where we insert a surface operator of arbitrary codimension onto a surface which is geometrically constructed starting from the end of a boundary-attached dressing geodesic.  ``Acting locally'' means that the operator is supported on the same surface before and after we act with the symmetry.  In particular this tells us that global symmetries must also act locally on operators which carry gauge charge, and are thus attached to the boundary by a dressing Wilson line. 

The other issue with the definition of the previous paragraph is that since we are now defining bulk global symmetries to act on gravitationally-dressed local operators, which are the same kind of objects which the asymptotic conformal symmetry acts on, we need to make sure that we have not accidentally included any of that symmetry as part of our definition of the global symmetry group.  Our requirements that global symmetries fix the boundary point $x$ to which any dressing geodesic is attached and commute with the boundary stress tensor dispense with most of the asymptotic conformal symmetry group.  But in fact there is a residual piece: in a theory with fermions, the $\mathbb{Z}_2$ fermion parity symmetry which acts as $+1$ on bosons and $-1$ on fermions is correctly understood as part of the asymptotic conformal symmetry group: it is a rotation by $2\pi$.  We therefore will include the following requirement for global symmetries in bulk effective field theory: given a global symmetry group $G$, for any nontrivial normal subgroup $H\subset G$ there must be two gravitationally dressed local operators which transform in the same representation of the asymptotic conformal group, but which transform in different representations of $H$. For example in the $\phi^4$ theory \eqref{phi4}, $\phi$ is a Lorentz scalar which is charged under the $\mathbb{Z}_2$ global symmetry while $\phi^2$ is a Lorentz scalar which is neutral. This requirement rules out the general possibility of a global symmetry for which the representation of any operator is determined by its Lorentz representation. Fermion parity is the only example of this that we know of, and any other would be very strongly constrained by locality.  But in any case it would not be independent of the asymptotic conformal symmetry, and so should be excluded.

\subsection{Global symmetries in non-perturbative quantum gravity}

We now turn to the question of how to define global symmetry in non-perturbative quantum gravity.  This is more difficult than  for the perturbative quantum gravity of the last section, since we need to come up with a precise property of a theory that we do not know how to describe in detail.  Once we move beyond bulk effective field theory, we are in the realm of operators which create black holes, modifications of the spatial topology, etc.  Clearly the less we need to assume about such operators the better.  On the other hand, in ruling out bulk global symmetries, which is our ultimate goal, we do not only want to discuss situations where the charged objects necessarily include low-energy effective field theory excitations of the vacuum.  For example what about a global symmetry under which the lightest charged states are black holes?   To rule out such a symmetry, we need to extend our notion of bulk local operator to include operators which create such states from the vacuum.

Let's first recall how ordinary gravitationally-dressed local operators in bulk effective field theory are embedded into the dual conformal field theory in AdS/CFT.  This subject has a long history \cite{Banks:1998dd,Polchinski:1999yd,Hamilton:2006az,Heemskerk:2012mn}, the modern understanding \cite{Almheiri:2014lwa}, recently reviewed in \cite{Harlow:2018fse}, is that bulk effective field theory operators should be viewed as logical operators on a protected subspace of the full CFT Hilbert space.  The details of this will not be important for us here, but the key point is that every bulk effective field theory operator has a limited domain of validity in the CFT, essentially consisting of those states where its dressing does not place it far behind the horizon of a black hole.  It is only in the limit where we pull such an operator all the way back to the boundary that this regime of validity extends to the full CFT Hilbert space.  We now generalize this idea to operators which create more complicated bulk objects via the following definition:\footnote{Readers who are only interested in ruling out global symmetries which act nontrivially on the fields in the low-energy effective action can skip definition \ref{quasilocaldef} and the ensuing subtleties.  In definition \ref{bulkglobaldef} they can replace ``quasilocal bulk operator'' by ``dressed local operator'', and the same contradiction still arises.}

\begin{mydef}\label{quasilocaldef}
A \textit{quasilocal bulk operator} in asymptotically-AdS quantum gravity, $\phi$, is an operator on the physical Hilbert space which has the property that there exists a maximal distance $L$ and a subspace $\mathcal{H}_{code}$ of the full non-perturbative Hilbert space such that:\footnote{This definition involves approximations defined using the Newton constant $G$, which is measured in AdS units.  For any fixed example of AdS/CFT this is just a number, and we have to live with the inherent imprecision of basing an approximation on the smallness of a finite number.  After all if it works for the fine structure constant, why shouldn't it work here?  Also, if there is a string scale which is parametrically lower than the Planck scale, then strictly speaking we should either use that scale in AdS units in our approximations or else upgrade effective field theory to effective string field theory.}
\bi
\item $\mathcal{H}_{code}$ contains the ground state.
\item The correlation functions of an $O(G^0)$ number of dressed low-energy bulk operators with renormalized distance $\hat{\ell}<L$ from the boundary and $O(G^0)$ time separation are well-described by low-energy bulk effective field theory for all states in $\mathcal{H}_{code}$ and to all orders in $G$.
\item Acting on the vacuum with $\phi$ an $O(G^0)$ number of times keeps us within $\mathcal{H}_{code}$, and there is a timeslice of the region attainable by operators with renormalized distance $\hat{\ell}<L$ on which the support of $\phi$ consists entirely of a gravitational Wilson line of the type defined in the previous subsection and a (possibly trivial) gauge Wilson line, lying on the same boundary-attached geodesic.  We sometimes say that $\phi$ is \textit{semiclassical} with respect to the operators in this region.  
\ei
\end{mydef}
This definition extends the idea of a dressed bulk local operator to an operator that affects a region of finite size in the bulk, up to the gravitational dressing which tells us where that region is and how the object created transforms under the asymptotic conformal symmetry, as well as now allowing a nontrivial gauge dressing.  The restriction to $\mathcal{H}_{code}$ ensures that we do not consider states where a huge central black hole reaches into the region $\hat{\ell}<L$.

\bfig
\includegraphics[height=7cm]{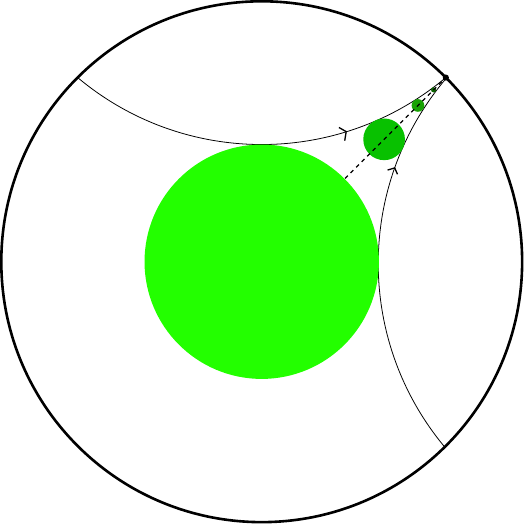}
\caption{Using an asymptotic conformal transformation to turn a quasilocal bulk operator into a boundary local operator.  The quasilocal bulk operator acts in a potentially complicated way in the bright green region in the center, and is connected to the boundary by the dashed gravitational/gauge Wilson line.  The appropriate one-parameter family of conformal transformations  ``focuses'' the operator towards the boundary endpoint of its dressing Wilson lines, and as it does so the region it affects, shown in progressively darker shades of green, gets smaller and smaller with respect to the boundary metric. States which are not in $\mathcal{H}_{code}$ for this operator get boosted off to infinite energy in the original conformal frame, so the final limiting operator is well-defined and local on the full CFT Hilbert space.}\label{zoomfig}
\efig
So far this is just a bulk quantum gravity definition, but we now make two assumptions about how bulk quasilocal operators fit into AdS/CFT:
\bi
\item[(1)] By acting with the asymptotic conformal symmetry on any bulk quasilocal operator $\phi$, and rescaling by a factor $r_c^\Delta$ for some $\Delta$, we can move all of its support to a point on the AdS boundary, in such a way that $\mathcal{H}_{code}$ can then be taken to be the full CFT Hilbert space and $\phi$ becomes a CFT local operator with conformal dimension $\Delta$. 
\item[(2)] Every CFT local operator of definite conformal dimension can be obtained from the limit of a bulk quasilocal operator in this way.
\ei
These are not assumptions which we can ``prove'' without a non-perturbative bulk description of quantum gravity, but they are quite plausible given the structure of AdS/CFT.\footnote{They \textit{can} be proven within non-perturbative models of the correspondence constructed using tensor networks, such as those of \cite{Pastawski:2015qua,Hayden:2016cfa}.}  The motivation for assumption (1) is shown in figure \ref{zoomfig}.  Assumption (2) is a kind of converse to assumption (1), roughly speaking it says that acting with any CFT local operator at boundary point $x$ creates a highly boosted bulk object which is localized near point $x$, even if that operator has very high conformal dimension.  We can justify this more carefully using the state-operator correspondence.  Indeed note that given any CFT local operator $\mO$ of definite scaling dimension, we can define a state of finite energy by inserting that operator at the south pole of the Euclidean path integral.  In the bulk this state describes an object of finite size, generically a black hole, sitting in the center of the spacetime.\footnote{There is an exception to this statement if the operator obeys some sort of differential equation in the boundary which causes the perturbation from the south pole to propagate up the side of the sphere instead of up into the center of the bulk.}  If we now act on this state with the conformal transformation shown in figure \ref{zoomfig}, the operator ``slides up'' the Euclidean sphere, as shown in figure \ref{sphereboostfig}, leading to a state which is produced by acting on the vacuum with the local operator $\mO$ at the equator.  We may then obtain the action of an associated quasilocal bulk operator on states other than the vacuum by defining it as the image of that local operator under the inverse of this conformal transformation, restricted to an appropriate $\mathcal{H}_{code}$ (strictly speaking we will also need to ``comb'' its gravitational and gauge dressing to all end at a single boundary point, but this can be done using only bulk low-energy effective field theory operators).
\bfig
\includegraphics[height=4cm]{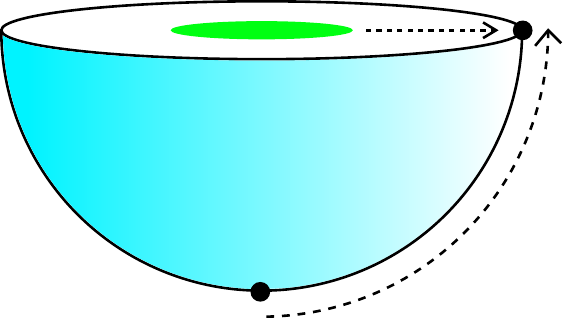}
\caption{The CFT dual of the conformal transformation in figure \ref{zoomfig}.  By the state-operator correspondence, any finite-energy state on the sphere is created by the insertion of a local operator at the bottom of the Euclidean path integral, and the conformal transformation in question just moves this operator up to the equator.}\label{sphereboostfig}
\efig

It is now at last time to give a definition of global symmetry in non-perturbative asymptotically-AdS quantum gravity.  Since we are ultimately trying to rule out the existence of such symmetries, our definition does not need to capture all features we might ideally like them to have: it is enough that those features it does capture already lead to a contradiction!  We therefore do not need to completely characterize the action of the global symmetry on all possible bulk operators, quasilocal bulk operators will basically be enough.  Here is our definition:  

\begin{mydef} \label{bulkglobaldef} A quantum gravity theory in asymptotically-$AdS$ space has a \textit{global symmetry with symmetry group $G$} if the following are true:
\bi
\item[(a)] There is a homomorphism $U(g,\partial\Sigma)$, not necessarily continuous, from $G$ into the set of unitary operators on the full diffeomorphism-invariant Hilbert space associated to any boundary time-slice $\partial \Sigma$.\footnote{In asymptotically-AdS quantum gravity, to get a Hilbert space we need to pick a boundary time slice.  A priori we are \textit{not} assuming that $U(g,\partial\Sigma)$ has support only at the boundary of the spacetime.}
\item[(b)] $U(g,\partial\Sigma)$ acts locally on the set of quasilocal bulk operators, meaning that if $\phi$ is a quasilocal bulk operator, then in the asymptotic region $\hat{\ell}<L$, $\phi$ and $U^\dagger (g,\partial\Sigma)\phi U(g,\partial\Sigma)$ both are dressed by the same gravitational Wilson line, and moreover if one is semiclassical with respect to all operators with $\hat{\ell}<L$, then so is the other with the same $L$.  
\item[(b')] $U(g,\partial \Sigma)$ acts within the algebra $\mathcal{A}[R]$ of operators in a boundary subregion $R\subset \partial \Sigma$, meaning that conjugating an element of $\mathcal{A}[R]$ by $U(g,\partial \Sigma)$ gives us another element of $\mathcal{A}[R]$.  Moreover it is continuous in the same sense as in condition (b) from definition \ref{globaldef}.
\item[(c)] $U(g,\partial \Sigma)$ acts faithfully on the set of quasilocal bulk operators which are gauge singlets, meaning that for all $g\in G$ there is a quasilocal bulk operator with no gauge Wilson line in the asymptotic region $\hat{\ell}<L$ which transforms nontrivially under $U(g,\partial \Sigma)$.
\item[(d)] For any normal subgroup $H\subset G$ containing at least two elements, there exist two gauge-singlet quasilocal bulk operators which transform in the same representation of the asymptotic conformal symmetry but different representations of $H$.

\item[(e)] $U(g,\partial \Sigma)$ commutes with the boundary stress tensor $T_{\mu\nu}$.
\ei
\end{mydef}
Note that conditions (a), (b'), and (e) apply throughout the CFT Hilbert space, while conditions (b), (c), (d) involve quasilocal bulk operators and thus only hold on the appropriate subspaces for those operators.  Conditions (a), (b'), and (e) are basically the AdS analogues of saying that the global symmetry preserves the (IR-safe version of the) S-matrix of quantum gravity in asymptotically-flat space, while (b), (c), and (d) say that the objects which carry the charge can live in the center of the bulk, not just at the boundary.  This definition essentially just upgrades that of the previous subsection, which applied to gravitationally-dressed local bulk operators in effective field theory, to one that applies to quasilocal bulk operators.  There are two notable points of departure however:
\bi
\item We have allowed quasilocal bulk operators to have nontrivial gauge dressing, since otherwise there would be local CFT operators which are not obtained as limits of quasilocal bulk operators.  In conditions (c) and (d) we then need to restrict to gauge-singlet quasilocal bulk operators, since these are the ones which become operators with compact support in the limit of vanishing gravitational coupling, and we want to recover definition \ref{globaldef} in that limit.  
\item Condition (b') may seem at first to follow from condition (b), and indeed for local operators in $R$ it does follow from the boundary limit of condition (b), together with an appropriate continuity assumption and also assumption (2) about quasilocal bulk operators. In general quantum field theories however there can be surface operators in the region $R$ which are not generated by the local operators in $R$, and we have not defined the ``quasilocal bulk surface operators'' of which these would be limits.  For example the closed Wilson loops in $\mathcal{N}=4$ Super Yang-Mills theory are limits of bulk operators which create closed strings.  To avoid the complexity of defining such operators, we have instead settled for condition (b'), which will already be enough to achieve a contradiction.  
\ei

We then have an immediate result:
\begin{thm}\label{boundarysym}
A global symmetry with symmetry group $G$ of a holographic asymptotically-$AdS$ quantum gravity theory is also a global symmetry with symmetry group $G$ of the dual conformal field theory.
\end{thm}
\begin{proof}
We need to show that definition \ref{bulkglobaldef} implies definition \ref{globaldef} in the boundary CFT.  Conditions (a), (b'), and (e) from definition \ref{bulkglobaldef}  imply conditions (a), (b), and (d) from definition \ref{globaldef}, while condition (c) of \eqref{bulkglobaldef}, together with assumption (1) about bulk quasilocal operators, implies condition (c) of \eqref{globaldef}
\end{proof}
This already suggests that something is wrong with the notion of a bulk global symmetry, since in AdS/CFT we usually think that a boundary global symmetry should be dual to a (long-range) gauge symmetry in the bulk.  In fact this tension can be sharpened into a real contradiction, leading to a proof of conjecture \ref{nosym}, as we now explain.

\subsection{No global symmetries in quantum gravity} 

\bfig
\includegraphics[height=5cm]{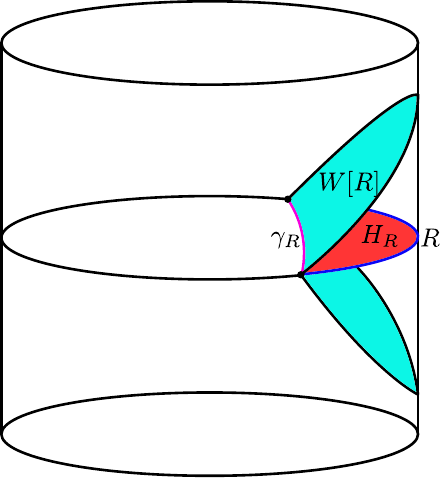}
\caption{The Hubeny-Rangamani-Takayanagi surface $\gamma_R$ is a bulk codimension-two surface of extremal area, obeying $\partial \gamma_R=\partial R$, and homologous to $R$ via a spatial surface $H_R$.  If there is more than one such surface, we pick the one of smallest area.  The entanglement wedge $W[R]$ is the bulk domain of dependence of $H_R$, here it is the spacetime region between the two codimension-one blue surfaces.  According to the leading-order Ryu-Takayangi formula, the von Neumann entropy of a CFT state on the subregion $R$ is equal to the area of $\gamma_R$ divided by $4G$.}\label{hrtfig}
\efig 
We will now argue that the existence of any global symmetry on the bulk side of AdS/CFT would be inconsistent with the local structure of the boundary conformal field theory.  The basic tool we will use is \textit{entanglement wedge reconstruction}, which is a recently-established property of the correspondence which says that there is a kind of ``sub-duality'' between any spatial subregion $R$ of the boundary CFT and a certain subregion of the bulk, the entanglement wedge of $R$ \cite{Czech:2012bh,Wall:2012uf,Headrick:2014cta,Dong:2016eik,Harlow:2016vwg}.  Giving a detailed explanation of this idea would take us too far afield, we refer the reader to \cite{Harlow:2018fse} for a recent overview, but the geometric definition of the entanglement wedge is given in figure \ref{hrtfig} (borrowed from \cite{Harlow:2018fse}).  What entanglement wedge reconstruction says is that on an appropriate code subspace, any bulk operator in $W[R]$ can be represented in the CFT by an operator with support only in $R$.  Therefore a boundary observer with access only to $R$ has complete information about what is going on in $W[R]$, but no information about what is going on in $W[R^c]$.  Just how small the code subspace needs to be for this statement to hold is a topic which is still being explored, see \cite{Harlow:2016vwg} for an optimistic outlook on this question, but at a minimum entanglement wedge reconstruction is expected to hold for any particular region $R$ in a code subspace where any black holes which are present are far outside of $W[R]$.  

We give two versions of our argument that there are no global symmetries.  The first assumes that global symmetries in conformal field theory on a spatial sphere are always splittable in the sense of definition \ref{splitdef}, while the second does not but instead requires us to consider more nontrivial bulk geometries which are under less control from the boundary point of view.  As explained in section \ref{splitsec}, splittability of global symmetries in conformal field theory on a spatial sphere follows from quite plausible axioms for quantum field theory, and intuitively is an expression of the local structure of the Hilbert space of quantum field theory on $\mathbb{R}^d$. 

\begin{thm}\label{noglobalthm}
No quantum gravity theory in asymptotically AdS space which has a global symmetry in the sense of definition \ref{bulkglobaldef} can be dual to a boundary conformal field theory.  
\end{thm}
\begin{proof}
Say that we had a bulk theory with a global symmetry group $G$.  By condition (d) in definition \ref{bulkglobaldef}, there are two quasilocal bulk operators which transform identically under asymptotic conformal symmetry, but which transform in different representations of $G$.  We will show that this is inconsistent with entanglement wedge reconstruction.  

Indeed note that by theorem \ref{boundarysym}, the symmetry operators $U(g,\partial \Sigma)$ also give a global symmetry of the boundary CFT provided that one exists.  Say that we decompose the boundary Cauchy slice $\partial \Sigma$ as the closure of a union of $n$ disjoint open regions $R_i$.  By splittability, we have that
\be\label{Ueq}
U(g,\partial \Sigma)=U(g,R_1)\ldots U(g,R_n)U_{edge},
\ee
where $U_{edge}$ is a unitary operator which ``fixes up'' the arbitrary choices which are made in defining the $U(g,R_i)$; it has support only in a small neighborhood of the union of the boundaries of the $R_i$. Now consider the action of these $U(g,R_i)$: by definition each one implements the symmetry on all operators in the domain of dependence of $R_i$, while it does nothing in the domain of dependence of its complement $R_i^c$.  By entanglement wedge reconstruction, in the bulk $U(g,R_i)$ implements the global symmetry on all operators which are supported only in the interior of $W[R_i]$, does nothing to operators which are supported only in the interior of $W[R_i^c]$, and acts in a potentially complicated manner in a neighborhood of the HRT surface $\gamma_{R_i}$.  

\bfig
\includegraphics[height=5cm]{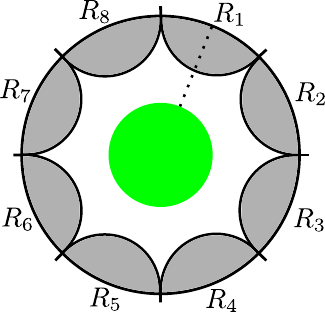}\caption{For any quasilocal bulk operator, we can always choose a large enough collection of small enough boundary regions that their entanglement wedges all lie in the ``semiclassical region'' of the code subspace for that operator.  Here we illustrate this for a bulk timeslice on which the gravitational dressing of the operator consists of a single gravitational Wilson line, indicated with the dotted line, and the entanglement wedges of the regions are shaded in grey.}\label{regions2fig}
\efig
The key point is that we can easily arrange for the two charged quasilocal bulk operators we are promised by condition (d) of definition \ref{bulkglobaldef} to be located such that their only support in the $W[R_i]$ is their gravitational Wilson lines.  The basic idea was already described in the introduction around figure \ref{regionsfig}, the precise version for quasilocal bulk operators is shown in figure \ref{regions2fig}.  But since via \eqref{Ueq} the charge is expressed entirely in terms of CFT operators with spatial support in regions whose entanglement wedges can access only the gravitational Wilson line parts of our quasilocal bulk operators, and since our two operators have identical gravitational Wilson lines since they transform in the same representation of the asymptotic conformal symmetry, there is no way for them to transform in different representations of our global symmetry.
\end{proof}
We emphasize that this contradiction arises already ``within the code subspace'', since to get into trouble we need only study quantities like
\be
\lan 0|\phi^\dagger U^\dagger(g,\partial \Sigma)\phi U(g,\partial\Sigma)|0\ran,
\ee
which involve only states obtained by acting in the vacuum with $\phi$, $U(g,\partial \Sigma)$, the $U(g,R_i)$, and $U_{edge}$.  $U(g,\partial \Sigma)$ should clearly preserve any reasonable code subspace, and since $U_{edge}$ has support only in a small neighborhood of $\cup_i \partial R_i$ we can take it to do so as well, at least in the vicinity of the time slice we consider in figure \ref{regions2fig}.  Arguing that the $U(g,R_i)$ individually can be taken to preserve the code subspace is a bit more subtle, but the idea, as already mentioned in the proof just given, is that since each one preserves all expectation values of operators supported in the interior of $D[R_i^c]$, and merely acts with the global symmetry on all expectation values of operators supported in the interior of $D[R_i]$, then it should preserve the semiclassical structure of the bulk everywhere away from a neighborhood of the HRT surface $\gamma_{R_i}$.  By smearing out the region of overlap near $R_i$, we can arrange for the energy created at the boundaries of the entanglement wedge to be finite: essentially we are just using entanglement wedge reconstruction to show that if they existed then bulk global symmetries would be splittable, at least if we take our bulk region to be an entanglement wedge.\footnote{When the symmetry group $G$ is continuous, it is not necessary to argue that $U_{edge}$ and $U(g,R_i)$ preserve the code subspace.  The reason is that we may then take the logarithm of \eqref{Ueq} to get an expression involving sums of charges, and then when we compute the commutator of the total charge with a quasilocal bulk operator $\phi$ we simply have a sum of commutators with boundary operators supported in regions whose entanglement wedges cannot reach the bulk endpoint of $\phi$, and which therefore must commute with it.  In quantum information theory this argument is called the ``Eastin-Knill theorem'' \cite{eastin2009restrictions}.  Without further assumptions it does not apply to discrete symmetry groups, which is why we have instead chosen to use special properties of holographic codes to argue that $U_{edge}$ and $U(g,R_i)$ can in fact be taken to preserve the code subspace without disrupting the semiclassical picture of the bulk away from the $\gamma_{R_i}$.}  

We note in passing that our proof of theorem \ref{noglobalthm} applies equally well to spontaneously-broken global symmetries in the bulk, since we did not assume anywhere that the vacuum was invariant.  It is amusing however to think about what such a global symmetry would have meant in the boundary CFT.  For simplicity consider the case of a spontaneously-broken $U(1)$ global symmetry in the bulk: there would be a massless Goldstone boson, which would be dual to a primary scalar operator of dimension $d$ in the boundary CFT.  The coefficient of this operator in the CFT action would set the symmetry-breaking expectation value for the Goldstone boson in the bulk, so the set of degenerate vacua would correspond to a continuous family of CFTs obtained by sourcing this operator with a finite coefficient: the operator would therefore need to be ``exactly marginal''.  Moreover the symmetry would ensure that in fact these CFTs were all isomorphic!  In more modern parlance, we would have a nontrivial conformal manifold on which all the CFTs were dual to each other.\footnote{\label{finiteVssb}This situation can also be described as spontaneous symmetry breaking in finite volume in the CFT.  This is often said to be impossible, but in fact there \textit{are} quantum field theories which exhibit spontaneous symmetry breaking in finite volume, at least in the sense of having exactly degenerate vacua related by the symmetry.  For example in $1+1$ electrodynamics with a $\theta$ term, 
\be
S=-\frac{1}{2q^2}\int F\wedge \star F-\frac{\theta}{2\pi}\int F,
\ee
at $\theta=\pi$ on a spatial circle the charge conjugation symmetry $F'=-F$ acts nontrivially on a pair of degenerate vacua \cite{Gaiotto:2017yup}.  We do not know of any examples in theories with non-topological local operators.}  We do not know of any examples of this, and find it rather implausible from the point of view of conformal perturbation theory, which is consistent with theorem \ref{noglobalthm}.

\bfig
\includegraphics[height=6cm]{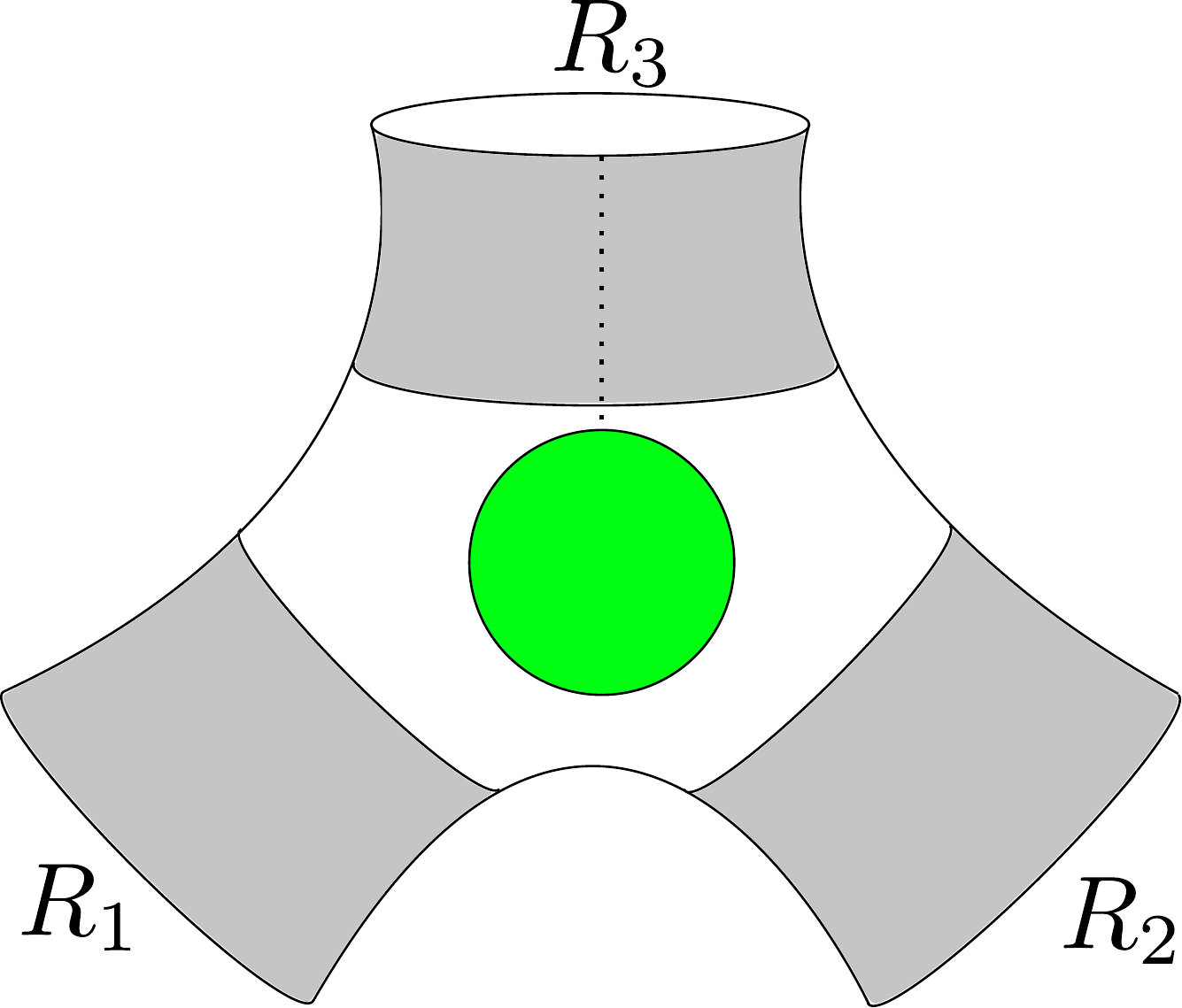}
\caption{A spatial slice of a three-exit wormhole for $d=2$.  The central region is not in the entanglement wedge of any one of the boundary components, but is in the entanglement wedge of any two.}\label{triple2fig}
\efig
Our second proof of theorem \ref{noglobalthm} proceeds on similar lines, except that instead of taking the $R_i$ to be $n$ disjoint subregions of a connected boundary as in figure \ref{regions2fig}, we instead take them to be connected components of a disconnected boundary.  Splittability of symmetries on these components is then automatic, since the Hilbert space of any quantum field theory on a disconnected space is always the tensor product of the Hilbert spaces of the connected components, so along the lines of theorem \ref{latticethm} any global symmetry in the boundary CFT can be decomposed as
\be
U(g,\partial\Sigma)=U(g,R_1)\ldots U(g,R_n),
\ee
without any need for a $U_{edge}$.  The idea is then to consider the action of this symmetry on states where the $n$ asymptotic regions are all connected in the bulk via a wormhole.  The AdS-Schwarzschild geometry is one such spacetime, which is dual to the thermofield double state
\be
|\psi_{tfd(\beta)}\ran\equiv \frac{1}{Z[\beta]^{1/2}}\sum_i e^{-E_i \beta/2}|i^\star\ran|i\ran
\ee
of the CFT on the disjoint union of two spheres for sufficiently small $\beta$ \cite{Maldacena:2001kr}, but for our purposes we need to consider geometries with $n\geq 3$.  There will then be an ``interior'' region which is not contained in the entanglement wedge of any one of the $R_i$, as shown for $n=3$ in figure \ref{triple2fig}, so we may again reach the same contradiction shown in figure \ref{regions2fig}.  This version of the argument has two appealing features: it dispenses with any assumption about splittability in the boundary CFT, and it makes the importance of black holes more apparent (black holes are implicitly present in any argument based on entanglement wedge reconstruction \cite{Almheiri:2014lwa}).  The main disadvantage however is that it is not immediately obvious that such configurations indeed exist as states in the Hilbert space of $n$ copies of the CFT on a spatial sphere, and it is also not immediately obvious that by taking $n$ to be large we can arrange for the interior region to be large enough to contain the object created by any particular quasilocal bulk operator.  Indeed no such construction has been worked out in complete detail, but in $d=2$ quite a lot is known and there is no sign of any obstruction.  Moreover there is no indication that any new obstruction will arise in higher dimensions.  We review the current status for $d=2$ in appendix \ref{wormholeapp}, and we suggest a region of moduli space which seems likely to satisfy all the necessary constraints.  

\subsection{Duality of gauge and global symmetries}\label{dualsec}
Having now established that global symmetries cannot exist in the bulk of AdS/CFT, one might then ask what a global symmetry of the boundary CFT is dual to in the bulk.  The traditional answer is a gauge symmetry \cite{Witten:1998qj}, but as we discussed in section \ref{gaugesec}, gauge symmetry in the conventional sense is too ambiguous of a notion to be dual to something as precise as a global symmetry.  We now argue that the correct statement is that a splittable global symmetry of the boundary CFT is dual to a long-range gauge symmetry in the bulk.  This proposal is clearly not subject to the contradiction of theorem \ref{noglobalthm}, since an operator which creates an object carrying gauge charge in the center of the bulk must have a Wilson line attaching it to the boundary, and this Wilson line will always enter the entanglement wedge of at least one of the $R_i$ in figure \ref{regions2fig} or figure \ref{triple2fig}.

We defined long-range gauge symmetries in quantum field theory via definition \ref{gaugedef}, to extend them to quantum gravity we just need to include gravitational dressing for the Wilson lines and loops and restrict them to appropriate code subspaces where that dressing does not place them far behind black hole horizons. Since the localized asymptotic symmetry operators $U(g,R)$ are supported only at the boundary, they will make sense on the full Hilbert space.  Moreover, as in assumption (b') from definition \ref{bulkglobaldef}, we will require the bulk long-range gauge symmetry $U(g,\partial \Sigma)$ to act within the local algebra $\mathcal{A}[R]$ for any boundary spatial region $R$; the motivation is again the idea that $\mathcal{A}[R]$ is generated by operators which are limits of quasilocal bulk operators, possibly also of the surface variety which we have not carefully defined, with any dressing Wilson lines ending in $R$.

We first argue that a long-range gauge symmetry in the bulk implies a splittable global symmetry in the boundary with the same symmetry group.  The obvious idea is to take the $U(g,R)$ of the bulk long-range gauge symmetry to be the $U(g,R)$ of a splittable boundary global symmetry.  We then need to establish that they obey conditions (b-d) of definition \ref{globaldef}, and also \eqref{split1}.  Condition (b) follows by the discussion at the end of the previous paragraph, and \eqref{split1} follows from the algebra \eqref{Walg} of the Wilson lines with the $U(g,R)$. Condition (d) follows because the boundary stress tensor $T_{\mu\nu}$ is the limit of the bulk metric, which is neutral under any (internal) long-range gauge symmetry (the metric would have to transform in a one-dimensional real unitary representation that preserves its signature, but there are no such representations).  The nontrivial step is to argue for condition (c), the faithfulness of the CFT global symmetry on the set of local operators.  Condition (3) in our definition \ref{gaugedef} of long-range global symmetry is clearly necessary for this to be possible, since a CFT operator transforming nontrivially under the global symmetry would be dual to a state of finite energy which is charged under the long-range gauge symmetry.  But just because charged states are allowed, this does not mean they exist.  In fact saying they do is basically the content of conjecture \ref{allcharge}!  Since establishing conjecture \ref{allcharge} is the main goal of the following section, we will here simply assume it, in which case by assumption there are charged states in all representations of the bulk gauge group, and therefore that group is represented faithfully on the set of local operators in the boundary CFT.

Conversely we now would also like to argue that a splittable global symmetry in the boundary CFT implies the existence of a long-range gauge symmetry in the bulk with the same symmetry group.  This argument is more difficult to make precise, since as part of it one would need to use special properties of the CFT which arise from it having a semiclassical holographic dual in the first place.  We have not had to deal with this so far because in proving theorem \ref{noglobalthm}, and also in the argument of the previous paragraph, we started in the bulk and went to the boundary.  What exactly the assumptions are on the CFT which lead to a semiclassical dual is not really a settled question, see \cite{Heemskerk:2009pn,Penedones:2010ue,Hartman:2014oaa,Maldacena:2015iua,Aharony:2016dwx} for a sampling of recent work and \cite{Harlow:2018fse} for a review of some aspects of the problem.  Here we will settle for arguing that \textit{if} a CFT has a semiclassical dual, then the $U(g,R)$ from a splittable global symmetry and the operators charged under that symmetry naturally give boundary conditions for reconstructing a bulk gauge field and bulk operators charged under it by solving the equations of motion derived from the assumed low-energy bulk Lagrangian radially inwards \cite{Heemskerk:2012np,Kabat:2012hp}.  

Indeed by the argument of theorem \ref{noglobalthm} the $U(g,R)$ operators must be localized on the boundary from the bulk point of view, and it is natural to identify them with the localized asymptotic global symmetry operators $U(g,R)$ from definition \ref{gaugedef}.  Their algebra with the charged boundary local operators whose existence is required by definition \ref{globaldef} is consistent with interpreting them as the boundary limits of  quasilocal bulk operators carrying gauge charge in the form of a boundary-attached Wilson line.  The existence of these charged boundary local operators also implies, via the state-operator correspondence, that in the bulk description there are states of finite energy which are charged under the long-range gauge symmetry, so condition (3) in definition \ref{gaugedef} is satisfied.  It is more nontrivial to evolve these boundary operators inwards to construct that the Wilson lines and Wilson loops with support in the bulk, how we do this depends on the low-energy bulk Lagrangian, and also on the topology of spacetime.  For example if the boundary global symmetry group is connected, we work near the vacuum, and the bulk effective action is dominated by the Yang-Mills term
\be
S=-\frac{1}{4q^2}\int d^{d+1}x\sqrt{-g}F_{\mu\nu}^a F^{\mu\nu}_a,
\ee
then at leading order in $q$, one can use the AdS/CFT dictionary to derive an expression of the form
\be\label{HKLLJ}
A_\mu^a(x)=\int dX K_{\mu\nu}^{ab}(x,X)J^\nu_b(X),
\ee
where $X$ is a boundary point, $x$ is a bulk point, $J^\nu_a$ is the Noether current of the boundary global symmetry, and $K_{\mu\nu}^{ab}$ is a c-number function.  This expression may then be systematically corrected to higher order in the interactions, producing a CFT representation of $A_\mu^a$ (in some gauge) which obeys the bulk equations of motion derived from the bulk effective Lagrangian to all orders in perturbation theory \cite{Kabat:2011rz,Kabat:2012av,Heemskerk:2012mn,Harlow:2018fse}.  A similar analysis should work in the presence of Chern-Simons terms, $\theta$ terms, etc.  Once we have $A_\mu^a$, we may then construct the desired Wilson lines and loops.  

The case where the gauge group is discrete is both simpler and more nontrivial: the equations of motion become easier to solve since at leading order the relevant line and surface operators are topological, but since we no longer have a Noether current there is no formula along the lines of \eqref{HKLLJ}.  What we need to do instead is reconstruct the charged matter fields, which do have representations similar to \eqref{HKLLJ}, and then use the fusing operation shown in figure \ref{wilsongluefig} to extract the Wilson lines and Wilson loops. It may seem surprising that the charged matter fields are necessary in the discrete case when they weren't in the continuous case, but we will momentarily see that, as first pointed out in \cite{Harlow:2015lma}, the charges are also necessary for reconstructing the bulk gauge field in the continuous case if the spacetime topology is nontrivial.\footnote{In situations where the charged operators in the boundary theory all have high scaling dimension, in the bulk we will need a version of the fusing of figure \ref{wilsongluefig} which makes sense for quasilocal bulk operators.  We will not attempt to say more about this, fortunately our arguments for conjectures \ref{allcharge}-\ref{compact} do not rely on this since we will only need the converse statement that a bulk long-range gauge symmetry implies a boundary global symmetry.}  

We close this section by noting that an alternative perspective on the relationship between the boundary global symmetries and bulk gauge symmetries is provided by the observation that by using the $U(g,R)$, together with the Noether current for the global symmetry in the continuous case, we can turn on a background gauge field in the CFT for the global symmetry as in section \ref{backgroundsec}.  This background gauge field is quite naturally interpreted as the fixed boundary value of a bulk gauge field \cite{Witten:2017hdv}, although to really see that this is correct we need to reconstruct the dynamical part of that gauge field, as just discussed.

\section{Completeness of gauge representations}\label{completenesssec}
We now turn to establishing conjecture \ref{allcharge}, which in AdS/CFT we can now state more precisely as claiming that whenever there is a long-range gauge symmetry in the bulk gravity theory, in the boundary CFT there are states in the Hilbert space on a spatial $\mathbb{S}^{d-1}$ which transform in all finite-dimensional irreducible representations of the global symmetry dual to that long-range gauge symmetry.  Before doing so, we need to first complete our argument from subsection \ref{dualsec} that a long-range gauge symmetry in the bulk indeed implies a global symmetry in the boundary with the same symmetry group:  in that argument we assumed that the asymptotic symmetry operators $U(g,\partial\Sigma)$ act faithfully on the set of boundary local operators rather than showing this.  We will show this in a moment, but first we point out that in fact establishing it is actually also sufficient to establish that there are states of the CFT on $\mathbb{S}^{d-1}$ transforming in all irreducible representations of the bulk gauge group.  This follows from two convenient facts about compact Lie groups (recall that we have defined long-range gauge symmetries to require the gauge group to be compact).  The first is theorem \ref{faithfulthm}, which says that any faithful unitary representation of a compact Lie group has a faithful subrepresentation which is finite-dimensional.  The second is theorem \ref{levythm}, which says that if $\rho$ is a finite-dimensional faithful representation of a compact Lie group $G$, then any finite-dimensional irreducible representation of $G$ appears in the direct sum decomposition of $\rho^{\otimes n}\otimes \rho^{*m}$ for some finite $n$ and $m$.  The idea is to apply these results to the action $D$ of $G$ on the set of local operators defined by equation \eqref{Dmap}.\footnote{In applying them we need to know that $D$ actually gives a good continuous representation of $G$.  Theorem \ref{Drepthm} tells us that this will be the case if the ground state on $\mathbb{S}^{d-1}$ is invariant, and the state operator correspondence tells us that it will be (the identity operator is always neutral).}  Indeed condition (c) of definition \ref{globaldef} and theorem \ref{faithfulthm} tell us that there is a finite subset of the local operators which transform in a faithful representation of $G$, and theorem \ref{levythm} then tells us that by acting with products of these operators and their hermitian conjugates on the vacuum, we can prepare states which transform in any irreducible representation of $G$. Thus to establish conjecture \ref{allcharge} in AdS/CFT, again invoking the state-operator correspondence,  we need only show that the long-range gauge symmetry acts faithfully on the Hilbert space of the CFT on $\mathbb{S}^{d-1}$.

\bfig
\includegraphics[height=5cm]{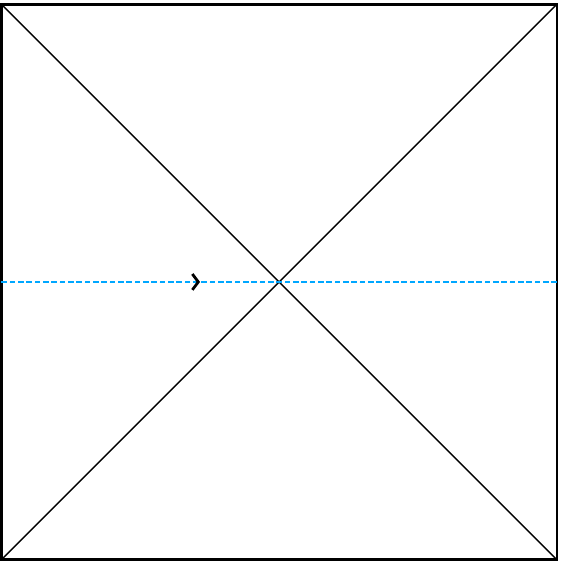}
\caption{A wormhole-threading Wilson line.}\label{threadfig}
\efig
The basic idea for establishing this faithful action appeared already in \cite{Harlow:2015lma} for the special case $G=U(1)$, we here extend it to arbitrary compact $G$.  We begin by noting that if we study a theory with a long-range gauge symmetry in the maximally extended Ads-Schwarzschild background, there are Wilson line operators which begin on one connected component of the spatial boundary and end on the other, threading the wormhole in between.  We illustrate such a Wilson line in figure \ref{threadfig}.  In any particular irreducible representation $\alpha$, the algebra of this Wilson line with the asymptotic symmetry operator on the ``right'' component of the spatial boundary, denoted $\Sigma_R$, is given by equation \eqref{Walg} to be
\be
U^\dagger(g,\Sigma_R)W_\alpha U(g,\Sigma_R)=D_\alpha (g)W_\alpha,
\ee 
where we have suppressed representation indices.  Using the conjugation properties of $W_\alpha$ given in definition \ref{gaugedef}, we then have
\be\label{Wilsontrans}
U^\dagger(g,\Sigma_R)W_\alpha U(g,\Sigma_R)W_\alpha^\dagger=D_\alpha (g).
\ee
Finally we note that in the dual CFT, $U(g,\Sigma_R)$ are nothing but the global symmetry operators $U(g,\mathbb{S}^{d-1})$ of the ``right'' CFT on $\mathbb{S}^{d-1}$, so we need only argue that $U(g,\Sigma_R)$ is nontrivial for all $g\in G$.  Indeed note that for any $g$ there is some irreducible representation $\alpha_g$ for which $D_{\alpha_g}(g)$ is nontrivial (see eg the proof of theorem \ref{liefaithfulthm}).  But then equation \eqref{Wilsontrans} with $\alpha=\alpha_g$ tells us that $U(g,\Sigma_R)$ must be nontrivial, since otherwise the Wilson lines on the left hand side would cancel each other and we would find $D_{\alpha_{g}}(g)$ to be the identity.  Therefore $U(g,\mathbb{S}^{d-1})$ faithfully represents the bulk gauge group, also establishing conjecture \ref{allcharge} by way of the argument in the previous paragraph.

Both this argument and our second argument for theorem \ref{noglobalthm} ultimately rest on the basic fact that the Hilbert space of any quantum field theory on a disconnected space tensor factorizes into a product over copies of the theory on each connected component: this is the ``UV information'' which AdS/CFT provides to us that goes beyond bulk effective field theory, as emphasized in \cite{Harlow:2015lma}.  Our first argument for theorem \ref{noglobalthm} also uses more or less the same idea, now couched in the notion that global symmetries should be always be splittable on a topologically trivial space.  

We now close this section by giving an alternative argument for conjecture \ref{allcharge} in the special case where the bulk gauge group $G$ is connected.  In this case the Lie algebra of $G$ is uniquely determined by the set of Noether currents $J^\mu_a$ in the boundary CFT, so the question is whether or not the boundary global symmetry group $G'$ differs from $G$ in its global topology (as discussed in section \ref{gaugetopsec} this difference \textit{is} physically meaningful).  More precisely, theorem \ref{liethm} tells us that $G$ and $G'$ are both quotients of the same connected simply-connected covering group $\wt{G}$ by discrete central subgroups $\Gamma$ and $\Gamma'$, and we would like to argue that $\Gamma=\Gamma'$. We should first recall what are the principles which define $\Gamma$ and $\Gamma'$: $\Gamma$ is identified by what set of topologically nontrivial gauge field configurations are summed over in the bulk, while $\Gamma'$ is identified by our requirement that boundary global symmetries act faithfully on the set of local operators.  The idea is then to note that $\Gamma$ also controls what kind of topologically nontrivial boundary conditions can be turned on for the bulk gauge field.  In the boundary theory these boundary conditions are just background gauge fields for $G'$, and which of these can be turned on is controlled by $\Gamma'$.  Therefore since these sets must coincide, we must have $\Gamma=\Gamma'$.    

To see this more concretely, we can study the boundary theory on spatial $\mathbb{S}^2\times X$, where $X$ is arbitrary.  We then consider possibly-nontrivial $G$ bundles on this space which are described by splitting $\mathbb{S}^2$ into hemispheres and gluing with a map $g:\mathbb{S}^1\to G$ at the equator, for example as in the Dirac/Wu-Yang monopole \eqref{wymonopole} for $G=U(1)$.  Such bundles are classified by $\pi_1(G)$, and studying the CFT in such a background is dual to studying the bulk in a sector of fixed nonzero magnetic charge.  Since $\wt{G}$ is simply-connected, all nontrivial elements of $\pi_1(G)$ lift to paths in $\wt{G}$ from the identity to a nontrivial element of $\Gamma$.  So clearly the larger $\Gamma$ is as a subgroup of $\wt{G}$, the more bundles are possible.  In the boundary CFT however there is a limit on how large $\Gamma$ can be: if we move a charged CFT operator around the equator of the $\mathbb{S}^2$, we want it to be single-valued in both its northern and southern representations (geometrically we want it to be a good section). This means that $\Gamma$ must lie in the kernel of $\wt{D}$, where $\wt{D}$ is the natural lift of the representation $D$ of $G'$ on the CFT local operators to a representation of $\wt{G}$ (any representation of $G'$ can be lifted in this manner).  Therefore we can get the largest set of background gauge fields by taking $\Gamma=\mathrm{Ker}(\wt{D})$, so we should identify $\wt{G}/\mathrm{Ker}(\wt{D})$ as the bulk gauge group.  But $\wt{G}/\mathrm{Ker}(\wt{D})$ is also precisely the quotient we would perform to obtain the group $G'$ which is represented faithfully on the set of CFT local operators, so we therefore have $\Gamma=\Gamma'$.  This argument is basically the CFT dual of Dirac quantization: the set of charged representations which exist in the boundary theory controls the set of which magnetic boundary conditions can be turned on.
 
\section{Compactness}\label{compsec}
We now turn to conjecture \ref{compact}, which we can now interpret more precisely as saying that all long-range gauge symmetries in quantum gravity are compact.  We are immediately confronted however with the inconvenient fact that in definition \ref{gaugedef} we \textit{defined} long-range gauge symmetries to be compact.  We did this for two reasons:
\bi
\item Finite-dimensional representations of compact Lie groups are always unitary (see theorem \ref{urepthm}), so the Wilson lines and loops have nice conjugation properties.
\item Our discussion of lattice gauge theory in section \ref{gaugesec} makes it clear that long-range gauge symmetry is possible with any compact gauge group, but for noncompact gauge groups this is far from clear.  For example the ordinary Yang-Mills kinetic term has negative modes if the Lie Algebra of the gauge group is not compact.
\ei
Rather then try to develop a general theory of what kinds of noncompact gauge groups are possible, we will instead proceed directly to the dual CFT.  Indeed we will argue any CFT which obeys a certain condition we introduce in a moment has the property that any noncompact global symmetry group must be a subgroup of a larger global symmetry group which is compact. The condition we will impose on CFTs is the following:
\begin{mydef}\label{finitegendef}
Let $S_0\equiv\{\mO_1,\mO_2,\ldots \mO_n\}$ be a finite subset of the primary operators in some conformal field theory, let $S_1$ denote the (usually infinite) set of primary operators such that for any element $\mO$ of $S_1$ there is a pair $\mO_i,\mO_j\in S_0$ such that $\mO$  appears with nonzero coefficient in their operator product expansion, let $S_2$ denote the set of operators which appear in the operator product expansion of some pair of operators in $S_1$, and so on.   We say that a conformal field theory is \textit{finitely generated} if 
\bi
\item For any $\Delta>0$ there is a finite number of primary operators with conformal dimension less than $\Delta$.
\item There exists a finite set $S_0$ of primary operators such that each primary operator of the theory appears in $S_N$ for some $N<\infty$.
\ei
\end{mydef}
Roughly speaking this condition formalizes the idea that there should be a finite number of fields in the path integral. For example free massless scalar field theory for $d>2$ is finitely generated since all of the primary operators are polynomials of $\phi$ and its derivatives.  From the bulk point of view, finite generation says that all objects can ultimately be built out of a finite number ingredients, which is quite plausible from the point of view that black hole entropy should be finite.  More carefully, say that we postulate that in a semiclassical bulk theory the types of bulk excitations should consist only of particle excitations, extended objects such as strings and $D$-branes, and black holes.  The spectrum of particle masses must be discrete with no accumulation points and bounded from above by the Planck mass, since if it were continuous or had accumulation points then renormalization would drive the strong coupling scale of gravity down to the AdS scale. The finiteness of the Bekenstein-Hawking entropy tells us that  black holes must also have a discrete spectrum with no accumulation points.  The extended objects are a little more subtle, but for $d>2$ the dynamics of AdS ensure that they also should have a discrete spectrum \cite{Seiberg:1999xz}.\footnote{We discuss the $d=2$ at the end of this section.} Therefore we expect that any holographic CFT with $d>2$ should be finitely generated.  In fact we can make the following conjecture, to which we are not aware of any counterexample:
\begin{conj}
Any conformal field theory in $d\geq 2$ with a discrete spectrum and a unique stress tensor is finitely generated, and any conformal field theory in $d>2$ with a unique stress tensor is finitely generated.  
\end{conj}
 In any event we can now give our argument for conjecture \ref{compact}, which we phrase as a theorem:
\begin{thm}\label{compactthm}
Let $G$ be a noncompact global symmetry of a finitely-generated conformal field theory.  Then there exists also a compact global symmetry $G'$ such that $G\subset G'$  
\end{thm}
\begin{proof}
Let $S_0=\{\mO_1,\ldots,\mO_n\}$ be the finite set of primary operators which generate all of the others.  There will always be some $\Delta$ such $\Delta_i<\Delta$ for all $i=1,\ldots,n$, and since the symmetry operators $U(g,\mathbb{S}^{d-1})$ commute with the stress tensor the $\mO_i$ must together be part of a finite-dimensional representation $\rho$ of $G$ (otherwise there would be infinitely many operators of dimension less than $\Delta$).  By theorem \ref{Drepthm} (generalized to unbounded operators as explained below the proof), the representation $\rho$ will be unitary. Since all local operators are generated by those in $S_0$, $\rho$ must also be faithful (by definition \ref{globaldef} the representation $D$ of $G$ on all local operators from equation \eqref{Dmap} is always faithful).  In particular $G$ is isomorphic to its image $\rho(G)$, which is a subgroup of $U(M)$ for some finite $M$.  The idea is then to notice that the closure of $\rho(G)$ in $U(M)$, $G'\equiv\ol{\rho(G)}$, is also a subgroup of $U(M)$.  In fact it is a closed subgroup, so since it is a closed subset of a compact space it is compact.  Moreover by theorem \ref{closedsgthm}, $G'$ is a Lie subgroup.  Now by finite generation any primary operator transforms in a representation of $G$ which appears in a finite tensor product of some copies of $\rho$ and its conjugate.\footnote{Note that if $\mO_3$ appears in the OPE of $\mO_1$ with $\mO_2$, then the three point function $\lan\mO_1\mO_2 \mO_3^\dagger\ran$ is nonzero.  This is only allowed by the global symmetry if the representation of $\mO_3$ appears in the direct sum decomposition of the tensor product of the representations of $\mO_1$ and $\mO_2$.}  Therefore by continuity it will also transform in a representation of $G'$, and the correlation functions of all local operators will obey the selection rules of $G'$ symmetry, not just those of $G$ symmetry.  Finally we note that $G'$ is by definition represented faithfully on the local operators, since distinct elements of $G'$ are automatically distinct in $U(M)$.
\end{proof}   
Since this argument is somewhat abstract, it is worthwhile discussing two simple examples.  The first example is a free scalar field with a noncompact target space in $d=2$: this has a noncompact global symmetry group, $\mathbb{R}$, but it is not finitely generated, both because $e^{i\alpha\phi}$ is a good primary operator with conformal dimension $\frac{\alpha^2}{4\pi}$ for any real $\alpha$, and because the three point function of such operators includes a delta function $\delta(\alpha_1+\alpha_2+\alpha_3)$.  The second example is two compact free scalars of equal radius, again in $d=2$.  This theory \textit{is} finitely generated, and the global symmetry group is $U(1)\times U(1)$, which is indeed compact.  We note however that it has an interesting noncompact subgroup consisting of the points $\theta_1=\lambda, \theta_2=\sqrt{2}\lambda$ in $U(1)\times U(1)$ for all real $\lambda$.  This subgroup is realized faithfully on the two-dimensional set of operators $(e^{i\phi_1},e^{i\phi_2})$, and its closure in $U(2)$ is indeed $U(1)\times U(1)$, consistent with theorem \ref{compactthm}.    

It is worth emphasizing that this second example illustrates the incompleteness of a certain argument that global symmetries must be compact which one sometimes hears.  This argument begins by requiring only the first point in definition \ref{finitegendef}, and then claiming that since there are no faithful finite-dimensional unitary  representations of noncompact groups, there cannot be a noncompact global symmetry.  This argument is correct for connected semisimple Lie groups, but it is wrong for general noncompact Lie groups. For example we just met a faithful finite-dimensional unitary representation of $\mathbb{R}$, given by $(e^{ix}, e^{i\sqrt{2}x})$.  Other noncompact groups also have faithful finite-dimensional unitary representations, for example there is a two-dimensional faithful unitary representation of $SL(2,\mathbb{Z})$.\footnote{This representation is generated by the diagonal matrix $(i,-i)$ and a matrix obtained by conjugating the diagonal matrix $(e^{i\pi/3},e^{-i\pi/3})$ by a generic element of $SU(2)$.  This is a representation of $SL(2,\mathbb{Z})$ because $SL(2,\mathbb{Z})$ is isomorphic to the free group on a generator $S$ of order four and a generator $ST$ of order six, with the identification $S^2=(ST)^3$, and the generic conjugation ensures there are no further relations.  We thank Yves de Cornulier for explaining this representation to us \cite{309050}.}  The correct general statements along these lines are theorems \ref{urepthm} and \ref{liefaithfulthm}, which say that all finite-dimensional representations of compact groups are unitary and that at least one of those is faithful.

Returning now to the $d=2$ case, there (and only there) it is possible for ``long strings'' near the boundary to lead to a bulk theory with a continuous spectrum \cite{Seiberg:1999xz,Maldacena:2000hw,Maldacena:2000kv,Maldacena:2001km}. The CFT dual of such a bulk theory therefore will not obey definition \ref{finitegendef}, since it will have a continuous spectrum of conformal dimensions, so theorem \ref{compactthm} does not apply.  In all known examples this happens because the boundary CFT includes massless scalar fields with a noncompact target space: in higher dimensions this does not lead to a continuous operator spectrum because the conformal curvature coupling $R\phi^2$ always lifts the flat direction due to the positive curvature of $\mathbb{S}^{d-1}$ for $d>2$.  We point out however that the first condition in definition \ref{finitegendef} was only used once in the proof of theorem \ref{compactthm}: to argue that the operators $S_0$ are part of a finite-dimensional representation of $G$.  If we replace this condition by simply \textit{requiring} that the operators in $S_0$ transform in a finite-dimensional representation of any global symmetry, then the proof of theorem \ref{compactthm} goes through as before and we get a version of theorem \ref{compactthm} which does not require a discrete spectrum of conformal dimensions with no accumulation points. For example in the boundary CFT dual to string theory on $AdS_3\times \mathbb{S}^3\times T^4$ with $NS$-$NS$ flux, long strings lead to a continuous spectrum but we expect that there is still a finite set of operators whose OPE recursively generates all of the other primaries.\footnote{It was shown in \cite{Maldacena:2001km} that the OPE of two short string operators generates long strings with winding number $w = 1$. For larger winding numbers, the selection rules proven in that paper show that the OPE of one short string operator and one long string operator with winding number $w$ can generate long strings with winding number at most $w+1$. Moreover, evidence has been given \cite{Ribault:2005ms,Giribet:2011xf} that such long strings with winding number are indeed generated. Therefore it seems reasonable to expect that all operators in the boundary CFT are generated iteratively from a finite set of the discrete short string operators.}  And indeed this theory has no noncompact global symmetries, and all bulk gauge fields are compact.  From this point of view, the culprit which allows the $d=2$ free noncompact scalar to have a noncompact global symmetry is not the continuous nature of the spectrum: it is the selection rule in the OPE which prevents us from obtaining all primaries starting from a finite set.    

\section{Spacetime symmetries}\label{bigdiffsec}

So far we have been primarily discussing internal global symmetries, which send the algebra of operators $\AR$ in any spacetime region $R$ into itself.  There are of course also spacetime global symmetries such as boosts and translations, which map $\AR$ to $\mathcal{A}[R']$ for some other region $R'$.  These are examples of the following general definition of global symmetry in quantum field theory:
\begin{mydef}\label{spacetimedef}
A quantum field theory on a spacetime $M$ with topology $\mathbb{R}\times \Sigma$ and metric $g_{\mu\nu}$ has a \textit{global symmetry with symmetry group $G$} if the following are true:
\bi
\item[(a)] There is a homomorphism $U(g,\Sigma)$ from $G$ into the set of unitary and antiunitary operators on the Hilbert space.  
\item[(b)] There is a smooth homomorphism $f_g$ from $G$ to the group of conformal isometries of $M$, meaning diffeomorphisms which preserve the metric $g_{\mu\nu}$ up to an overall position-dependent scalar factor (the group operation is composition, so we have $f_{g_1}\circ f_{g_2}=f_{g_1g_2})$, such that
\be
U^\dagger(g,\Sigma)\AR U(g,\Sigma)=\mathcal{A}[f_{g^{-1}}(R)].
\ee
As before, if $R$ is spatially bounded then this map is required to be continuous in the strong operator topology on any uniformly-bounded subset of $\AR$.
\item[(c)] For all $g$ other than the identity, there exists a local operator $\mO$ such that
\be
U^\dagger(g,\Sigma)\mO(x)U(g,\Sigma)\neq \mO(x).
\ee
\item[(d)] The stress tensor transforms as a conformal tensor, meaning that\footnote{The extra non-tensor factor in front here arises from the fact that the conformal transformations which are global symmetries are combinations of diffeomorphisms with Weyl transformations.  This is because we need to cancel the transformation of the metric; it is a background field and cannot transform under a global symmetry.   This factor is the identity for transformations which are genuine isometries, but for conformal transformations it is essential, for example to get the right scaling dimension for $T_{\mu\nu}$.}
\be
U(h,\Sigma)T_{\mu\nu}(x)U^\dagger(h,\Sigma)=\left(\det\partial f_{h} \sqrt{\frac{\det g(f_{h}(x))}{\det g(x)}}\right)^{\frac{d-2}{d}}\frac{\partial f_{h}^\alpha}{\partial x^\mu}\frac{\partial f_{h}^\beta}{\partial x^\nu}T_{\alpha \beta}(f_{h}(x)),
\ee
where we have used $h$ instead of $g$ for the element of $G$ to avoid confusion with the metric $g_{\mu\nu}$.   
\ei
\end{mydef}
These general global symmetry transformations act on local operators as
\be
U^\dagger(g,\Sigma)\mO_i(x)U(g,\Sigma)=\sum_j D_{ij}(g,x)\mO_j(f_{g^{-1}}(x)),
\ee
where $D$ obeys
\be
\sum_k D_{ij}(g_1,x)D_{jk}(g_2,f_{g_1^{-1}}(x))=D_{ik}(g_1g_2,x),
\ee
which can be thought of as an infinite-dimensional representation of $G$ with $x$ being another ``index''.  

Definition \ref{spacetimedef} reduces to our previous definition \ref{globaldef} of global symmetry if we take $M=\mathbb{R}^d$ with the usual flat metric and assume that all $f_g$ are the identity.  More generally we can always extract an ``internal subgroup'' from $G$ as follows:
\begin{mydef}
Given a global symmetry with symmetry group $G$, its \textit{internal part} is the global symmetry with symmetry group $G_I$ obtained by restricting to only those $g\in G$ such that $f_g$ is the identity.  
\end{mydef}
Since $G_I$ is the kernel of a homomorphism, it is always a closed normal subgroup of $G$.  Moreover if $M=\mathbb{R}^d$ then the internal part of any global symmetry will be a global symmetry of the theory in the sense of definition \ref{globaldef}. When definition \ref{spacetimedef} applies on a more general $M$ we can say that the symmetry is preserved on $M$ in the sense of definition \ref{symextend}.

At first it may seem that condition (d) in definition \ref{spacetimedef} is too strong, for example it implies that when $M=\mathbb{R}^d$ with flat metric, all elements of $G_I$ must commute with all translations, rotations, and boosts, as well as with dilations and special conformal transformations if the theory is conformally invariant.  In fact for elements of $G_I$ which are in the identity component of $G$, this follows from the Coleman-Mandula theorem and its various cousins, which basically say that if $G$ contains the Poincare group as a subgroup, then the Lie algebra of $G$ must be the direct sum of either the Poincare algebra or the conformal algebra with a finite-dimensional compact ``internal'' Lie algebra whose elements all commute with the Poincare/conformal generators \cite{Coleman:1967ad,Maldacena:2011jn,Alba:2015upa}.\footnote{We can also consider supersymmetries, which we have not included in definition \ref{spacetimedef}, which are constrained by an analogous theorem \cite{Haag:1974qh}.  Since supersymmetries are defined only at the level of the Lie algebra (we don't exponentiate them to get a group), the issues we discuss in this section do not arise.  Indeed the presence of the bulk gravitino ensures that any supersymmetry is always gauged in the bulk, so we will not discuss them further.}  Our next order of business in this section will be to extend this result from Lie algebras to Lie groups, establishing a kind of Coleman-Mandula theorem for disconnected groups, which we view as motivating (d) as the most general possibility.\footnote{In our argument we will assume that the internal symmetry group $G_I$ is compact, which in particular implies that the full symmetry group $G$ is finite-dimensional.  This excludes the Virasoro algebra and Kac-Moody current algebra in $d=2$.  These are natural to exclude, since in holography they work somewhat differently than the symmetries we study here.  For example the higher Virasoro and Kac-Moody currents do not give rise to new fields in the bulk, so the noncompact $G_I$ which arises is not dual to a long-range gauge symmetry with noncompact gauge group so there is no violation of conjecture \ref{compact}.}

We first review a few basic properties of the Poincare and conformal groups for $\mathbb{R}^d$, which we define to be isomorphic to $\mathbb{R}^d\rtimes OSpin(d-1,1)$ and $OSpin(d,2)$ respectively.  The former indicates a semidirect product of translations with the Lorentz group.  In both cases the ``$O$'' indicates that we have included both spatial and temporal reflections, and ``$Spin$'' indicates that fermion parity, defined as rotation by $2\pi$ about any axis, is represented nontrivially.  We can obtain the identity components by dropping the $O$'s, and if we quotient by fermion parity then ``$Spin$'' becomes ``$SO$''.  The Coleman-Mandula theorem and theorem \ref{liethm} then tell us that the identity component $G_0$ of $G$ must be a quotient of either $(\mathbb{R}^d\rtimes Spin(d-1,1))\times \wt{(G_I)_0}$ or $Spin(d,2)\times \wt{(G_I)_0}$ by a discrete central subgroup.  The only candidates for this subgroup are combinations of fermion parity with a discrete central subgroup of $\wt{(G_I)_0}$.  This combination does not need to be a product group, for example the theory of two free Dirac fermions with equal nonzero mass in $3+1$ dimensions has a $U(2)$ global symmetry mixing the fermions, but the product of fermion parity and the central element $\begin{pmatrix} -1 & 0 \\ 0 &-1\end{pmatrix}$ of $U(2)$ acts trivially on all states and thus should be quotiented by if we want a faithful representation.  

We can also consider elements of $G_I$ which are not in $G_0$.  We then have the following theorem:
\begin{thm}[Discrete Coleman-Mandula theorem] \label{DCMthm}Say that in a quantum field theory on $\Rd$ we have a global symmetry with a symmetry group $G$, which contains the identity component of the Poincare or conformal group, or one of their $\mathbb{Z}_2$ quotients by fermion parity, and say also that the internal subgroup $G_I$ of $G$ is compact and the Coleman-Mandula theorem applies.\footnote{In this theorem we do not impose condition (d) from definition \ref{spacetimedef}, since otherwise the result would be trivial.  The compactness of $G_I$ is motivated in the previous footnote.}  Then any element  of $G_I$ must commute with all elements of this identity component.  More prosaically, it must commute with translations, boosts, and rotations, as well as dilations and special conformal transformations if there are any.  
\end{thm}
\begin{proof}
Consider $h\in G_I$ which is also in $G_n$, the $n$th connected compoment of $G$, and let $g$ be an element of the identity component of the Poincare/conformal group or its $\mathbb{Z}_2$ quotient by fermion parity, which for brevity we will call $\hat{G}_0$.  Since by definition $g\in G_0$, by continuity we must have $g^{-1} h g\in G_n$.  Therefore we must have
\be
g^{-1}h g=\wt{g}_h(g)h,
\ee
with $\wt{g}_h(g)\in G_0$.  We will argue that $\wt{g}_h(g)$ is the identity.  We first note that since $G_I$ is a normal subgroup, we must have $\wt{g}_h(g)\in G_I\cap G_0$.  As we just discussed, the Coleman-Mandula theorem therefore says that $\wt{g}_h(g)$ commutes with any element of $\hat{G}_0$.  We therefore have
\begin{align}\nonumber
\wt{g}_h(g_1)\wt{g}_h(g_2)&=g_1^{-1}h g_1 h^{-1}g_2^{-1}hg_2h^{-1}\\\nonumber
&=(g_1g_2)^{-1}h(g_1g_2)h^{-1}\\
&=\wt{g}_h(g_1g_2),
\end{align}
so $\wt{g}$ defines a homomorphism from $\hat{G}_0$ to $G_I$.  Finally we note that since $G_I$ is compact, by theorem \ref{liefaithfulthm} it has a faithful finite-dimensional representation $\rho$.  Therefore the composition $\rho\circ \wt{g}$ gives a finite-dimensional unitary representation of $\hat{G}_0$.  Any such representation must be trivial however, in the Poincare case because $Spin(d-1,1)$ is noncompact and simple and translations do not commute with it, while in the conformal case just because $Spin(d,2)$ is noncompact and simple.  Finally since $\rho$ is faithful, it must be that $\wt{g}_h(g)$ is the identity for all $g$, $h$.    
\end{proof}
We view this theorem as motivating condition (d) in definition \ref{spacetimedef}.  It is worth emphasizing that it does \textit{not} say that elements of $G_I$ must commute with spatial and temporal reflections, since these are not in the identity component of the Poincare/conformal groups.  In general the best we can say is that every element $g$ of $G$ can be written as
\be\label{gdecomp}
g=\hat{g}_0h,
\ee
where $\hat{g}_0$ is in the identity component of the Poincare/conformal group (or its $\mathbb{Z}_2$ quotient by fermion parity), and $h$ has the property that $f_h$ is either the identity, a reflection of a particular spatial direction, a time reversal, or a product of the two.\footnote{In even dimensions we can replace the spatial reflection by a simultaneous reflection of all spatial directions, usually called parity, but in odd dimensions this is just a rotation.  Therefore when working in arbitrary dimensions it is safer to talk about reflections in a single spatial direction, for example the natural generalization of the $CPT$ theorem to arbitrary dimensions is the $CRT$ theorem.}  Acting on elements of $G_I$ by conjugation, $h$ can induce a nontrivial outer automorphism of $G_I$ even if it includes a spatial or temporal reflection.  One simple example of this arises in the theory of a single free Dirac fermion in $3+1$ dimensions, with Lagrangian
\be
\mathcal{L}=-i\overline{\psi}\gamma^\mu\partial_\mu \psi.
\ee
The internal symmetry group $G_I$ for this theory is the $U(2)$ that rotates the two independent left-handed Weyl spinors contained in $\Psi$ into each other.  In particular this $U(2)$ includes the chiral rotation
\be\label{diracchiral}
\psi'=e^{i\theta\gamma^5}\psi
\ee
as the diagonal subgroup generated by the identity, fermion number as the subgroup generated by $\sigma_z$, and charge conjugation as the $\mathbb{Z}_2$ that exchanges the two left-handed fermions.  
This theory is also invariant under the parity transformation
\begin{align}\nonumber
(t',\vec{x}')&=(t,-\vec{x})\\
\psi'(t,\vec{x})&=i\gamma^0\psi(t,-\vec{x}).\label{parity}
\end{align}
The point is that this parity transformation does not commute with the chiral symmetry transformation \eqref{diracchiral}: if $R(\theta)$ and $P$ are the unitary operators implementing chiral symmetry and parity on the Hilbert space, then we have
\be
P^{-1}R(\theta)P=R(-\theta),
\ee  
which is the algebra of the nonabelian group $O(2)$.\footnote{One might try to modify our definition \eqref{parity} of parity by including an element of the $U(2)$ internal symmetry in hopes of obtaining something that commutes with chiral symmetry.  This however is impossible: chiral symmetry is in the center of $U(2)$.}
More complicated examples of this phenomenon have been studied in the particle physics literature \cite{Feruglio:2012cw,Holthausen:2012dk}, and it is also discussed using somewhat different terminology in section 2.C of \cite{Weinberg:1995mt}. 

It is also worth emphasizing that neither definition \ref{spacetimedef} nor theorem \ref{DCMthm} \textit{require} the existence of elements $g$ of $G$ whose associated $f_g$ involves any particular spatial or temporal reflection.  For example in the standard model of particle physics there are no global symmetries which reflect only time or only space (the CPT theorem ensures that there will always be a symmetry which reflects both). And moreover even if such elements exist, they may act on the operators in a nonstandard way.  For example if we look at only the first two generations of leptons and quarks in the standard model, parity and charge conjugation as conventionally defined are not symmetries but their product is.  

Having introduced our general definition \ref{spacetimedef} of global symmetries, we may now ask if our theorem \ref{noglobalthm}, which rules out internal global symmetries in the bulk of AdS/CFT, applies also to global symmetries for which $f_g$ can be nontrivial.  At first this seems like a rather silly question: general relativity is a diffeomorphism-invariant theory, so shouldn't any spacetime symmetries obviously need to be gauged?  In fact the truth is a bit more subtle.  The right statement is that to remove negative-norm modes of the graviton, it is only necessary that the \textit{identity component} of the diffeomorphism group be gauged \cite{Weinberg:1995mt}.  After all the other connected components might not even be symmetries, as happens in the standard model, and then we surely had better not gauge them!  But then this leads to an interesting question: say that our bulk theory is indeed invariant under diffeomorphisms which change the orientation of time and/or space: could these be global symmetries rather than gauge symmetries?  From the bulk point of view it is fairly subtle to decide this: ultimately it comes down to whether or not the gravitational path integral includes temporally and/or spatially unoriented manifolds (it includes them if these symmetries are gauged, but it doesn't if they aren't).  From the point of view of conjecture \ref{nosym} however, it would be rather surprising if there were such global symmetries in quantum gravity.  In fact there are not, and a slight generalization of theorem \ref{noglobalthm} suffices to establish it.  

Indeed note that if we study the boundary CFT on $\mathbb{R}\times \mathbb{S}^{d-1}$ (which is conformally flat so the results of this section apply), any global spacetime symmetry in the bulk would imply the existence of a global spacetime symmetry of the boundary CFT by the same argument as for theorem \ref{boundarysym}.  From equation \eqref{gdecomp} we see that every element of that boundary global symmetry group is the product of a conformal transformation which is continuously connected to the identity and a group element $h$ such that $f_h$ is either the identity, a time reversal, an antipodal mapping of $\mathbb{S}^{d-1}$, or a time reversal and an antipodal mapping.  We want to show that these global symmetries cannot arise from global symmetries in the bulk.  Decoupling of negative-norm graviton modes tells us that the identity component conformal transformation must be gauged, so we are then just left with $h$.  If $f_h$ is the identity then theorem \ref{noglobalthm} already gives us the desired contradiction.  Moreover if $f_h$ is a time-reversal, the argument for theorem \ref{noglobalthm} still works provided that we take the boundary time slice in figure \ref{regions2fig} to be at $t=0$.  Finally if $f_h$ involves an antipodal mapping of the sphere, we can still basically use the argument of theorem \ref{noglobalthm}, the only difference is that in figure \ref{regions2fig} we should combine pairs of regions which are on opposite sides of the sphere.  As long as the regions are small enough, the entanglement wedge of their union will just be the union of their entanglement wedges, so the contradiction still follows.  In both cases where $f_h$ is nontrivial there is no need for a discussion of quasilocal bulk operators: the metric itself is already not invariant so we can just use it.

Conjectures \ref{allcharge} and \ref{compact} do not at first seem to have meaningful analogues for spacetime symmetries, since spacetime symmetry groups are noncompact, but actually there is a fairly trivial generalization based on restricting to just the rotation subgroup $SO(d)\subset SO(d,2)$.  This group is of course compact, and the obvious extension of conjecture \ref{allcharge} says that there should be states in the bulk transforming in all irreducible representations of $SO(d)$ (or $Spin(d)$ if there are fermions).  In other words, there should be objects of all possible spins.  In fact this conjecture does indeed follow from a simple generalization of the argument of section \ref{completenesssec}.  Namely we consider gravitational Wilson lines of spin $j$ threading the throat of the AdS-Schwarzschild geometry from one side to the other, localized at some point $x\in \mathbb{S}^{d-2}$.  Under one-sided rotations which preserve $x$, this Wilson line will transform in the spin-$j$ representation of $SO(d-1)$ (or $Spin(d-1)$).  Since for any element of $SO(d)$ (or $Spin(d-1)$) we can pick an $x$ and $j$ such that that element is represented nontrivially on this Wilson line, we see that $SO(d)$ (or $Spin(d)$) must be represented faithfully on the one-boundary Hilbert space.  From here we then would like to use theorems \ref{faithfulthm} and \ref{levythm} to conclude that there must be states of all spin, but we need to be a little careful since now rotations can move the operators around.  This problem however is easily solved: we can simply act with all (smeared) operators at the north pole of $\mathbb{S}^{d-1}$, and then classify their representations with respect to the $SO(d-1)$ (or $Spin(d-1)$) subgroup which fixes the north pole.  Since we can obtain all tensor products of the faithful representation in this way, and since this subgroup is sufficient to diagnose the representation of $SO(d)$ (or $Spin(d)$),  we may indeed use theorems \ref{faithfulthm} and \ref{levythm} to conclude that there are states of all spin (all integers for $SO(d)$ and all half-integers for $Spin(d)$).

\section{\textit{p}-form symmetries}\label{psec}
In the last few years it has been understood that there is a powerful generalization of the global symmetries we have been discussing so far.  These new symmetries are variously called higher symmetries, gauge-like symmetries, $p$-form symmetries, or generalized global symmetries \cite{nussinov2009symmetry,Kapustin:2013uxa,Kapustin:2014gua}, \cite{Gaiotto:2014kfa}.  We will call them $p$-form global symmetries, since this name gives the most information about the symmetry being discussed.  Understanding $p$-form global symmetries begins with the observation that the ordinary global symmetries we have been discussing so far can be thought of as global symmetries which act on local operators: indeed condition (c) in definition \ref{globaldef} tells us that we can diagnose the full symmetry group just by looking at how local operators transform.  $p$-form global symmetries are defined as global symmetries which act nontrivially only on surface operators of dimension at least $p$, and which act faithfully on surface operators of dimension exactly $p$.  In this language, the global symmetries we have been discussing so far become zero-form symmetries.  It is natural to ask to what extent conjectures \ref{nosym}-\ref{compact} have generalizations to $p>0$, and to what extent we can use AdS/CFT to give arguments for those generalizations.  Answering these questions is the goal of this section.  We begin by discussing $p$-form global symmetries in more detail.

\subsection{\textit{p}-form global symmetries} 
It is perhaps easiest to introduce $p$-form global symmetries by generalizing the ``path integral insertion'' perspective on ordinary global symmetries described in and around figure \ref{gluingfig} \cite{Gaiotto:2014kfa}.  Recall that in that language, a global symmetry corresponds to a family of codimension-one insertions $U(g,\Sigma)$, where $g$ is any element of $G$ and $\Sigma$ is any closed oriented codimension-one surface in spacetime.  One then requires that these surface insertions obey the group algebra $U(g_1,\Sigma)U(g_2,\Sigma)=U(g_1g_2,\Sigma)$, and also that they are topological in the sense that away from other path integral insertions, $\Sigma$ can be freely deformed without changing the result of the path integral.  Finally one requires $\Sigma$ can also be continuously deformed past a local insertion $\mO(x)$, but perhaps at the price of a representation of $G$ acting on that local insertion.  For example if $\Sigma'$ contains $x$ in its interior while $\Sigma$ does not,\footnote{Here which side of a surface we call its interior is determined by its orientation, and flipping this orientation is equivalent to inverting $g$.} then in the path integral we have
\be
\lan \ldots \mO_i(x)U(g,\Sigma')\ran=\sum_j D_{ij}(g)\lan \ldots \mO_j(x)U(g,\Sigma)\ran,
\ee
where here ``$\ldots$'' denotes other insertions which do not interfere with the deformation between $\Sigma$ to $\Sigma'$.  This is a path integral representation of equation \eqref{Dmap}, and the matrix $D$ is the same matrix appearing there; in particular it is required to be faithful in the sense of being nontrivial for all $g$ other than the identity.  

\bfig
\includegraphics[height=3cm]{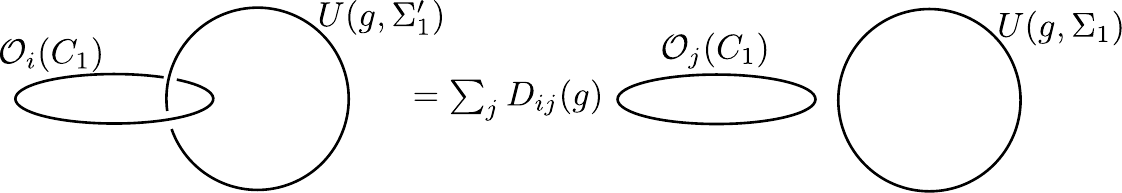}
\caption{A one-form global symmetry in $d=3$: linking a symmetry insertion $U(g,\Sigma_1')$ with a line insertion $\mO_i(C_1)$ acts on that line insertion with a representation of the (abelian) symmetry group $G$.}\label{pformfig}
\efig
$p$-form global symmetries are then defined analogously by requiring that there be a family of $(d-p-1)$-dimensional  insertions $U(g,\Sigma_{d-p-1})$, where again $g$ is any element of $G$ but now $\Sigma_{d-p-1}$ is any closed oriented $(d-p-1)$-dimensional surface in spacetime.  As before we demand the group algebra $U(g_1,\Sigma_{d-p-1})U(g_2,\Sigma_{d-p-1})=U(g_1g_2,\Sigma_{d-p-1})$ is satisfied, and also that $\Sigma_{d-p-1}$ can be freely deformed away from other path integral insertions.  When $p>0$, $\Sigma_{d-p-1}$ can always be deformed ``around'' any local operator without picking up a representation of $G$.  Moreover it can similarly be deformed around any surface operator of dimension less than $p$. This is not true however for a surface $C_p$ of dimension $p$, since it is possible for $C_p$ and $\Sigma_{d-p-1}$ to be linked nontrivially in spacetime.  One finally then requires that if $C_p$ and $\Sigma'_{d-p-1}$ are linked once (this counting includes the orientations of $C_p$ and $\Sigma_{d-p-1}$, and inverting $g$ is equivalent to flipping the orientation of $\Sigma'_{d-p-1}$), while  $\Sigma_{d-p-1}$ and $C_p$ are not linked, then in the path integral we have
\be
\lan \ldots \mO_i(C_p)U(g,\Sigma^\prime_{d-p-1})\ran=\sum_j D_{ij}(g)\lan \ldots \mO_j(x)U(g,\Sigma_{d-p-1})\ran,
\ee
where $\mO_i(C_p)$ is any surface operator on $C_p$ and $D_{ij}(g)$ is a representation of $G$.  We show an example for $p=1$ and $d=3$ in figure \ref{pformfig}.  As for zero-form symmetries, one requires that $D_{ij}(g)$ is nontrivial for all $g$ other than the identity.

One of the most fundamental distinctions between zero-form  global symmetries and $p$-form global symmetries with $p>0$ is that in the latter case the symmetry group $G$ must be abelian.  The reason is that if $\Sigma_{d-p-1}$ and $\Sigma'_{d-p-1}$ are two nearby surfaces of codimension $p+1$, they have no natural ordering.  Indeed in Lorentzian signature we can continuously deform them without intersection to exchange their time ordering.  In the limit where we bring the two surfaces together we must therefore have
\be
U(g_1,\Sigma_{d-p-1})U(g_2,\Sigma_{d-p-1})=U(g_2,\Sigma_{d-p-1})U(g_1,\Sigma_{d-p-1}).
\ee
Another important distinction is that in order for a $p$-form global symmetry to exist, there must be $p$-dimensional surface insertions which cannot be generated by insertions of lower dimensionality, since otherwise they would have to be neutral.  

Perhaps the most basic example of a theory with a $p$-form global symmetry with $p>0$ is free Maxwell theory, with gauge group $U(1)$.  This theory has a two-form conserved current 
\be
J_e\equiv \frac{1}{q^2}F,
\ee
which we can use to introduce the codimension-two symmetry insertions
\be\label{maxwelloneform}
U(e^{i\theta},\Sigma_{d-2})\equiv e^{i\theta \int_{\Sigma_{d-2}}\star J_e}.
\ee
These are nothing but the exponential of the integrated electric flux through $\Sigma_{d-2}$.  In section \ref{gaugesec} we studied these insertions at spatial infinity, where we used them to define long-range gauge symmetry, but the idea is now to consider them for arbitrary closed oriented $\Sigma_{d-2}$, and in particular to interpret them as the symmetry insertions for a one-form global symmetry with symmetry group $U(1)$.  They will be topological by the source-free Maxwell equation $d\star F=0$, but in order to make good on this interpretation we also need to say what are the line insertions which are charged under this one-form global symmetry.  Indeed the answer is obvious: they are the Wilson loops $W_n(C_1)$.  Since a Wilson line of charge $n$ represents the worldline of a background heavy particle of charge $n$, when $C_1$ is linked with $\Sigma_{d-2}$ the symmetry insertion $U(e^{i\theta},\Sigma_{d-2})$ will detect this charge and pick up a factor of $e^{in\theta}$ compared to when they are not linked. 

In fact free Maxwell theory with gauge group $U(1)$ also has another conserved current, the $(d-2)$-form current
\be
J_m\equiv \star F.
\ee
This leads to a second $p$-form global symmetry, this one with $p=d-3$ and symmetry insertions
\be
U(e^{i\theta},\Sigma_{2})\equiv e^{i\theta \int_{\Sigma_2}F}. 
\ee
The $(d-3)$-dimensional surface insertions charged under this symmetry are the 't Hooft surfaces defined by equation \eqref{thooftfree}.

Another example of a one-form symmetry arises in $SU(N)$ Yang-Mills theory with no matter fields.  This is the $\mathbb{Z}_N$ ``center symmetry'' of Polyakov and 't Hooft \cite{Polyakov:1978vu,tHooft:1977nqb}, whose symmetry insertion $U(e^{2\pi i n/N},\Sigma_{d-2})$ is defined to act as $e^{2\pi i n/N}$ on the Wilson loop in the fundamental representation of $SU(N)$.\footnote{This is not how center symmetry was originally described.  Instead one considered the set of gauge configurations of the theory in Euclidean signature with a temporal circle, and then considered the action on Wilson loops wrapping this circle of ``illegal gauge transformations'' which are not periodic around the loop.  This idea always seemed somewhat mysterious: why should we be allowed to consider gauge transformations which are not periodic?   And moreover, in defining a global symmetry why should we need to talk about gauge transformations at all?  Perhaps the main insight of \cite{Gaiotto:2014kfa} is that with the right definition, we don't!}  We can describe the symmetry insertions in this example more  concretely using the Hamiltonian lattice presentation of gauge theory we reviewed in section \ref{latsec}.  The basic idea for any gauge group $G$ is to consider operators of the form
\be\label{centersymops}
U(g,\Sigma_{d-2})\equiv\prod_{\ell\in \Sigma_{d-2}} L_g,
\ee
where the product is over the links which puncture any spatial $(d-2)$-dimensional surface $\Sigma_{d-2}$.  These operators will not however be invariant under the gauge transformations \eqref{latgaugetrans} unless $g$ is in the center $Z_G$ of $G$, so to get a good operator we need to restrict to $g\in Z_G$.  This is why the one-form global symmetry group of pure $SU(N)$ gauge theory is $\mathbb{Z}_N$, even though depending on the background there may be a full $SU(N)$ long-range gauge symmetry.  The latter is possible because we do not quotient by gauge transformations at spatial infinity, so the asymptotic symmetry operators do not need to be restricted to the center.

In our discussion of zero-form global symmetries in section \ref{globalsec}, we began with an algebraic definition, definition \ref{globaldef}, and from this we derived the path integral insertion point of view.  The reader may wonder why we have begun with the latter point of view here.  The reason is that if the spatial topology $\Sigma$ is simple, meaning that the homology group $H_p(\Sigma)$ is trivial, there can never be operators on the Hilbert space which are charged under a $p$-form global symmetry.  This is because on such a time-slice, any closed oriented $p$-dimensional surface $C_p$ will intersect any closed oriented  $(d-p-1)$-dimensional surface $\Sigma_{d-p-1}$ an equal number of times in opposite directions.  Nonetheless it will still be useful for us to give an algebraic definition of $p$-form global symmetries which generalizes definition \ref{globaldef} to $p>0$:
\begin{mydef}\label{globalpdef}
Let $\Sigma$ be a $(d-1)$-dimensional spatial manifold in which there is at least one closed oriented $p$-dimensional surface and one closed oriented $(d-p-1)$-dimensional surface which intersect each other exactly once.  We say that a quantum field theory on $M=\mathbb{R}\times \Sigma$ has a \textit{p-form global symmetry with (abelian) symmetry group G} if the following are true:
\bi
\item[(a)]  For any closed oriented $(d-p-1)$ surface $\Sigma_{d-p-1}\subset \Sigma$, there is a homomorphism $U(g,\Sigma_{d-p-1})$ from $G$ into the set of unitary operators on the Hilbert space of the theory quantized on $\Sigma$.  Moreover for any spatial region $R\subset \Sigma$ such that $\Sigma_{d-p-1}\subset R$, we have $U(g,\Sigma_{d-p-1})\subset \AR$.
\item[(b)]  For any such $\Sigma_{d-p-1}$, any $g\in G$, and any spatial region $R$, we have
\be
U^\dagger(g,\Sigma_{d-p-1})\AR U(g,\Sigma_{d-p-1})=\AR.
\ee
Moreover if $R$ is spatially bounded then the restriction of this map to any uniformly-bounded subset of $\AR$ is continuous in the strong operator topology.
\item[(c)] For any element $g$ of $G$ other than the identity, there is a $p$-dimensional surface operator $\mO$, a $p$-dimensional surface $C_p\subset \Sigma$, and a $(d-p-1)$-dimensional surface $\Sigma_{d-p-1}\subset \Sigma$ such that 
\be
U^\dagger(g,\Sigma_{d-p-1})\mO(C_p)U(g,\Sigma_{d-p-1})\neq \mO(C_p).
\ee 
\item[(d)]  For all $x\in \mathbb{R}\times\Sigma$, $g\in G$, and $\Sigma_{d-p-1}$, we have
\be
U^\dagger(g,\Sigma_{d-p-1})T_{\mu\nu}(x)U(g,\Sigma_{d-p-1})=T_{\mu\nu}(x).
\ee
\ei
\end{mydef}
Condition (d) implies that the symmetry operators $U(g,\Sigma_{d-p-1})$ are topological surface operators, and in fact it further implies that they commute with any $p'$-dimensional surface operator $\mO(C_{p'})$ with $p'<p$, since they can be continuously deformed around each other to change their time ordering.  Condition (d) also implies that the action of $G$ on the set of surface operators at $C_p$ defined by conjugation by $U(g,\Sigma_{d-p-1})$ is independent of small deformations of $C_p$ and $\Sigma_{d-p-1}$.  Indeed more is true: since their algebra is controlled entirely by the pieces of $U(g,\Sigma_{d-p-1})$ and $\mO(C_p)$ which are at the intersections of $\Sigma_{d-p-1}$ and $C_p$ (all other parts are spacelike separated), and since the symmetry group is abelian, we can pick a basis $\mO_i$ of surface operators such that their algebra with the $p$-form symmetry operators is given by
\be
U^\dagger(g,\Sigma_{d-p-1})\mO_i(C_p)U(g,\Sigma_{d-p-1})= D_i(g)^{n(C_p,\Sigma_{d-p-1})}\mO_i(C_p),\label{pDmap}
\ee
where $n(C_p,\Sigma_{d-p-1})$ is the intersection number of $C_p$ and $\Sigma_{d-p-1}$ and $D_i(g)$ is a homomorphism from $G$ into $U(1)$.

$p$-form global symmetries have many very interesting physical applications.  The basic idea is to use their existence, and whether or not they are spontaneously broken, in an extension of the Landau paradigm of characterizing the phases of many-body quantum systems by their symmetry structure \cite{Gaiotto:2014kfa},\cite{Yoshida:2015cia},\cite{Gaiotto:2017yup,Lake:2018dqm}.  One can also work out a transport theory of higher-form charges, for example leading to a new and much more satisfactory conceptual understanding of magnetohydrodynamics \cite{Grozdanov:2016tdf}.  Unfortunately describing these developments further here would take us too far afield.

\subsection{\textit{p}-form gauge symmetries}
Although $p$-form global symmetries were defined only recently, in an amusing twist of fate the $p$-form gauge symmetries which appear once we ``gauge'' them have been studied for decades \cite{Kalb:1974yc}.  The situation is especially simple when we gauge a $p$-form global symmetry which has symmetry group $\mathbb{R}$.  We then expect a $(p+1)$-form current $J_{p+1}$, for which we can first turn on a background $(p+1)$-form gauge field $A_{p+1}$ via a term 
\be
\delta S=\int_M A_{p+1}\wedge \star J_{p+1}
\ee
in the action.  We may then check if the partition function, possibly after some renormalization, is invariant under background gauge transformations
\be
A'_{p+1}=A_{p+1}+d\Lambda_{p},
\ee
where $\Lambda_p$ is an arbitrary $p$-form. If it is not invariant then we can say that the $p$-form global symmetry we started with has an 't Hooft anomaly, and we can proceed no further.  If it is invariant, then we are free to introduce a kinetic term and make $A_{p+1}$ a dynamical field, leading to a dynamical $(p+1)$-form gauge field.  A typical kinetic term one adds is
\be\label{pkinetic}
S=-\frac{1}{2q^2}\int_M F_{p+2}\wedge \star F_{p+2},
\ee
where $F_{p+2}=d A_{p+1}$, and one may also add various Chern-Simons and $\theta$-type terms.     We can also introduce a ``Wilson surface'' functional
\be\label{wsurface}
W_\alpha[\Sigma_{p+1}]=e^{i\alpha \int_{\Sigma_p} A_{p+1}},
\ee
which is gauge-invariant if $\partial \Sigma_{p+1}=0$ and otherwise transforms as
\be
W_\alpha'[\Sigma_{p+1}]=e^{i\int_{\partial \Sigma_{p+1}}\Lambda_p}W_\alpha[\Sigma_{p+1}].
\ee
There is also a gauge-invariant ``electric flux'' functional
\be
\Phi(\Sigma_{d-p-1})\equiv \frac{1}{q^2}\int_{\Sigma_{d-p-2}} \star F_{p+2},
\ee 
where in preparation for holograpy we have taken the spacetime dimension to be $d+1$.

The situation is not so simple for symmetry groups other than $\mathbb{R}$.  The reason is that it is then sometimes possible to turn on topologically-nontrivial background gauge field configurations which require more than one patch to describe.  In section \ref{backgroundsec} we reviewed how to do this for zero-form global symmetries using the idea of a connection on a principal bundle.  The generalization of this idea to $p$-form global symmetries is not straightforward: one immediately encounters the problem that the transition functions $g_{ij}:U_i\cap U_j\to G$ of an abelian principal bundle can be used to define a closed one-form $-i\partial_\mu g_{ij} g_{ij}^{-1}$ for use in the transformation of the gauge field, but there is no obvious way to use them to make a closed $(p+1)$-form for use in the transformation of $A_{p+1}$.  If one asks a mathematician how to solve this problem (we've asked several), one is usually told that the answer involves various types of abstract nonsense such as $n$-categories, stacks, and gerbes (see eg \cite{Baez:2010ya} for a relatively gentle introduction to this point of view, and also \cite{Johnson:2002tc}).  Although these ideas are indeed sometimes useful, a more plebeian approach is possible and we now say a little about how it works.  

For simplicity we will first describe the case where the $p$-form symmetry group is $U(1)$. The basic idea is that to describe a background gauge field for a $p$-form global symmetry, in addition to $p+1$-form gauge fields in each patch $U_i$ and $p$-form transition functions in each double overlap $U_i\cap U_j$, we need additional transition functions in higher multiple intersections which are differential forms of lower degree \cite{Alvarez:1984es}.  More concretely, in each $k$-tuple intersection $U_{i_1}\cap\ldots \cap U_{i_k}$ we require the existence of a $(p+2-k)$-form $A_{i_1\ldots i_k}$ such that all such forms are related by the following recursive formula:\footnote{Up to notational differences, this formula generalizes equations 4.3-4.5 of \cite{Alvarez:1984es} to arbitrary $k$ (and fixes some wrong signs in 4.5).  It is instructive to check the self-consistency of this formula under taking the the exterior derivative of both sides, in the cohomological language of \cite{Alvarez:1984es} this amounts to showing that the ``co-boundary operator'' $\delta$ is nilpotent.  We emphasize that the $i$ indices label patches, they are not tensor indices.}
\be\label{overlaprec}
dA_{i_1\ldots i_{k+1}}=\frac{1}{k!}\sum_{\pi\in S_{k+1}}s_\pi A_{i_{\pi(1)}\ldots i_{\pi(k)}}.
\ee
Here $S_{k+1}$ denotes the permutation group on $k+1$ elements, and $s_\pi$ is one if $\pi$ is even and minus one if $\pi$ is odd.  This formula is valid for $k=1,2,\ldots, p+1$, with $A_{i}$ being the $p+1$-form gauge field in each patch and all the others being transition functions.  $A_{i_1\ldots i_{p+2}}$ is a scalar, and thus can't be related to the exterior derivative of something, but we instead require that
\be\label{cocycle}
\frac{1}{(p+2)!}\sum_{\pi\in S_{k+2}} s_\pi A_{i_{\pi(1)}\ldots i_{\pi(p+2)}}=2\pi n \qquad \qquad (n\in \mathbb{Z}).
\ee
It is instructive to consider the case $p=0$, in which case this sequence of forms truncates at $k=2$ (double overlaps), and \eqref{overlaprec} and \eqref{cocycle} just give
\begin{align}\nonumber
A_i-A_j&=dA_{ij}\\
A_{ij}+A_{jk}+A_{ki}&=2\pi n.
\end{align}
If we define $g_{ij}\equiv e^{i A_{ij}}$, then these are precisely the transformation rules \eqref{gtrans}, \eqref{gtriple} for a connection on a $U(1)$ principal bundle.  We can consider also the $p=1$ case, where \eqref{overlaprec} and \eqref{cocycle} now give
\begin{align}\nonumber
A_i-A_j&=dA_{ij}\\\nonumber
A_{ij}+A_{jk}+A_{ki}&=dA_{ijk}\\
A_{ijk}-A_{ijl}-A_{jkl}-A_{kil}&=2\pi n.
\end{align}
We may again interpret the $A_{ijk}$ as arising from $U(1)$ group elements $g_{ijk}=e^{iA_{ijk}}$ obeying a quadruple intersection rule, and indeed we can give a similar interpretation to equation \eqref{cocycle} for any $p$.  

These additional transition functions are needed to generalize the Wilson surface functional \eqref{wsurface} to multiple patches.  We first remind the reader that some use of the transition functions is necessary even to define ordinary Wilson lines when  the curve on which they are supported intersects multiple patches, for example the Wilson loop $W_{\alpha}[C]$ of a closed curve $C$ which passes through patches $U_1,U_2, \ldots,U_n,U_1$, in this order and possibly with repetitions, is given by
\be\label{wilsonpatch}
W_\alpha(C)=\Tr \left(D_\alpha(g_{1n}(x_1))Pe^{i\int_{x_n}^{x_1} A_n^\alpha} \ldots D_\alpha(g_{32}(x_3))Pe^{i\int_{x_2}^{x_3}A_2^\alpha}D_\alpha(g_{21}(x_2))P e^{i\int_{x_1}^{x_2} A_1^\alpha}\right),
\ee
where $C$ has been broken up into a line segment from a point $x_1$ in $U_n\cap U_1$ to a point $x_2$ in $U_1\cap U_2$, a line segment from $x_2$ to a point $x_3$ in $U_2\cap U_3$, and so on. The insertions of $D_\alpha(g_{i+1,i}(x_{i+1}))$ are essential to get an answer which is invariant under  gauge transformations and does not depend on the choice of patches. For a $U(1)$ $p=1$ gauge field the formula analogous to this one is
\be
W_n[\Sigma]=\exp \left(in\sum_i\int_{\Sigma_i}A_i-in\sum_{<ij>}\int_{\Sigma_{ij}}A_{ij}-in\sum_{<ijk>}\int_{\Sigma_{ijk}}A_{ijk}\right),
\ee 
where we choose a triangulation $\Sigma_i$ of $\Sigma$ such that each $\Sigma_i$ is contained in $U_i$.  $\Sigma_{ij}$ is the shared boundary between $\Sigma_{i}$ and $\Sigma_{j}$, with the orientation chosen to point from $i$ to $j$, and $\Sigma_{ijk}$ is a shared point between $\Sigma_i$, $\Sigma_j$, and $\Sigma_k$ whose orientation is chosen so that $ijk$ go clockwise around.  It is straightforward, although a bit tedious, to see that the terms involving $A_{ijk}$, and also the condition \eqref{cocycle}, are necessary for this object to be independent of the choice of patches \cite{Alvarez:1984es}.

Generalizing these results to Abelian groups other than $U(1)$ is simplified by the fact that every compact Abelian Lie group is just a product of $U(1)$ and $\mathbb{Z}_n$ factors.  To describe the $\mathbb{Z}_n$ case, we may begin with the $U(1)$ construction and then restrict the $A_{i_1\ldots i_k}$ such that the parallel transport of any closed surface operator, implemented by a Wilson surface with two identical boundaries of opposite orientation, always just results in a multiplication of the surface operator by an element in the image of the $\mathbb{Z}_n$-representation of that operator.  For example this requires $dA_i=0$, and also that $e^{iA_{i_1\ldots i_{p+2}}}\in \mathbb{Z}_n$.

This discussion has been somewhat sketchy, so we note in passing that on the lattice there is a natural generalization of the Wilson formulation of ordinary gauge theory which defines dynamical $p+1$-form gauge fields with any abelian gauge group in a very elegant manner \cite{Villain:1974ir,Savit:1977fw,Orland:1981ku,Lipstein:2014vca,Johnston:2014ofa}. For simplicity we will describe the Euclidean version, the Hamiltonian version is constructed on similar lines.  The basic idea for a cubic lattice in Euclidean spacetime of arbitrary dimension\footnote{This definition generalizes immediately to an arbitrary CW-complex, where $f_{p}$, $f_{p+1}$, and $f_{p+2}$ below are $p$, $p+1$, and $p+2$ -cells respectively.} is to assign to each ``minimal'' face $f_{p+1}$ of dimension $p+1$ a group element $g(f_{p+1})$.  Gauge transformations are defined as assignments of group elements to each ``minimal'' face $f_{p}$ of dimension $p$, and they act on $g(f_{p+1})$ as
\be
g'(f_{p+1})=g(f_{p+1})\prod_{f_p\in \partial f_{p+1}}g(f_p),
\ee
with the orientations of the $f_p$ taken to be outward.  The Wilson surface functional in any irreducible representation $\alpha$ on any $(p+1)$-dimensional surface $\Sigma$ is defined as
\be
W_n[\Sigma]\equiv \prod_{f_{p+1}\in \Sigma}D_\alpha(g(f_{p+1})),
\ee
which is gauge-invariant if $\Sigma$ is closed.  Given a $(p+2)$-dimensional minimal face $f_{p+2}$ and a representation $\alpha$ of $G$, we can also define a gauge-invariant ``plaquette'' functional
\be
W_\alpha[f_{p+2}]=\prod_{f_{p+1}\in \partial f_{p+2}}D_\alpha(g(f_{p+1})),
\ee
in terms of which we can write the Euclidean action 
\be
S=-\frac{1}{2q^2}\sum_{f_{p+2}}W_\alpha(f_{p+2}).
\ee
Here both orientations of $f_{p+2}$ are included in the sum, and $\alpha$ is a faithful representation of $G$ (if $\alpha$ is not irreducible then we have to sum over its irreducible components).  When $G=U(1)$ this reproduces equation \eqref{pkinetic} in the continuum limit.  Note in particular that in the continuum limit we have
\be
W_n(f_{p+2})=e^{i\int_{f_{p+2}} F_{p+2}},
\ee
so the action is unchanged if we locally take $F_{p+2}\to F_{p+2}+2\pi n$.  This means that configurations with $\int F_{p+2}=2\pi n$ will survive in the continuum limit, and thus that topologically nontrivial $(p+1)$-form gauge field configurations of the type we just discussed will be included, with nary a gerbe in sight!

We can define a notion of ``long-range $p$-form gauge symmetry'' in a manner analogous to ordinary the ordinary long-range gauge symmetry of section \ref{gaugesec}.  In a $d+1$-dimensional spacetime with time slice $\Sigma$ and asymptotic spatial boundary $\partial \Sigma$, we can assign asymptotic symmetry operators $U(g,\Sigma_{d-p-1})$ to any closed $(d-p-1)$-surface in $\partial \Sigma$, 
which in the $U(1)$ case are defined by
\be
U(e^{i\theta},\Sigma_{d-p-1})=e^{\frac{i\theta}{q^2} \int_{\Sigma_{d-p-1}}\star F_{p+2}}.
\ee
More generally in the Hamiltonian lattice they are defined as
\be
U(g,\Sigma_{d-p-1})\equiv \prod_{f_{p+1}\perp \Sigma_{d-p-1}}L_g(f_{p+1}),
\ee
where the product is over spatial $(p+1)$-dimensional lattice faces which puncture $\Sigma_{d-p-1}$ at the spatial boundary.  In the natural the boundary conditions analogous to those of figure \ref{latticeboundaryfig}, which require gauge transformations to vanish at the spatial boundary, spatial Wilson surfaces may end at this boundary and their end-surfaces (which are $p$-dimensional surfaces) will transform under the asymptotic symmetry transformations just as in \eqref{pDmap}.  The objects which are charged under this long-range symmetry are $p$-branes, meaning objects with a $(p+1)$-dimensional world volume, and to have a long-range $p$-form gauge symmetry we further require that the theory allows such objects to exist with finite energy, provided that the have finite spatial volume.  We illustrate these ideas more concretely in the following subsection.

\subsection{\textit{p}-form symmetries and holography}\label{pformholsec}
We now discuss the analogues of conjectures \ref{nosym}-\ref{compact} for $p$-form symmetries.  The obvious generalizations turn out to be the correct ones: there are no $p$-form global symmetries in the bulk, for any long-range $p$-form gauge symmetry with gauge group $G$ there are objects ($p$-branes) which transform in all irreducible representations of $G$, and under plausible assumptions $G$ must be compact.  The basic idea of this section is to consider holographic CFTs on the spatial manifold $\mathbb{T}^p\times\mathbb{S}^{d-p-1}$, wrap objects which carry $p$-form symmetry charge on $\mathbb{T}^p$, and then dimensionally reduce on this $\mathbb{T}^p$.  We will then be able to apply the same arguments as before in the remaining dimensions, establishing the $p$-form generalizations of conjectures \ref{nosym}-\ref{compact}.  Our arguments will be less detailed than they were for zero-form symmetries, for example we will not explicitly discuss the issue of gravitational dressing.

The basic problem we need to solve is that ordinary asymptotically-AdS geometries have boundary $\mathbb{R}\times\mathbb{S}^{d-1}$, not $\mathbb{R}\times\mathbb{T}^p\times\mathbb{S}^{d-p-1}$, so we need to come up with new solutions of the Einstein equation with negative cosmological constant that \textit{do} have boundary $\mathbb{R}\times\mathbb{T}^p\times\mathbb{S}^{d-p-1}$.  We can consider an ansatz of the form
\be
ds^2=-\alpha(r)dt^2+\frac{dr^2}{\alpha(r)\beta(r)}+e^{\gamma(r)}dx_p^2+r^2 d\Omega_{d-p-1}^2,
\ee
where $dx_p^2$ is the flat metric on a square spatial torus and the asymptotic boundary is at $r\to\infty$.  There are two interesting classes of solutions of this type.  In the first class, the functions $\alpha$, $\beta$, and $e^\gamma$ are strictly positive for all $r\geq 0$, and the $\mathbb{S}^{d-p-1}$ contracts at $r=0$.  For sufficiently large $\mathbb{T}^p$, the ground state of a holographic CFT on spatial $\mathbb{T}^p\times \mathbb{S}^{d-p-1}$ should be dual to such a geometry.  In fact a unique such solution does exist, as we explain in appendix \ref{spheretorapp}, and we will refer to it as the \textit{vacuum solution}.   The spatial topology of the vacuum solution is $\mathbb{T}^p\times B^{d-p}$, where $B^{d-p}$ is the solid ball in $d-p$ dimensions.  In the second class of solutions we have $\alpha(r_s)=0$ for some $r_s>0$, with $\alpha$, $\beta$ and $e^\gamma$ strictly positive for $r>r_s$.  These types of solutions give a generalization of the AdS-Schwarzshild solution to a wormhole whose bifurcate horizon has topology $\mathbb{T}^p\times \mathbb{S}^{d-p-1}$, so we will refer to them as \textit{wormhole solutions}.  Wormhole solutions do indeed exist, with one for each value of $r_s$, and we describe them in more detail in appendix \ref{spheretorapp}. Their spatial topology is $\mathbb{T}^p\times\mathbb{R}\times \mathbb{S}^{d-p-1}$, and they should be dual to the thermofield double state of two copies of the CFT on $\mathbb{T}^p\times \mathbb{S}^{d-p-1}$ at sufficiently high temperature.  As we lower the temperature, there should be a Hawking-Page-like transition to two copies of the vacuum solution, with thermofield-double entanglement between the particles on the two copies.  Both the vacuum and the wormhole solutions were constructed in \cite{Copsey:2006br} for the special case $d=4$, $p=1$, so our analysis in appendix \ref{spheretorapp} can be viewed as generalizing those results to arbitrary $d$ and $p$.

Let's first argue that there are no $p$-form global symmetries in the bulk.  We will take the boundary theory to be on spatial $\mathbb{T}^p\times \mathbb{S}^{d-p-1}$, with large enough $\mathbb{T}^p$ that the ground state is described in the bulk by the vacuum solution.  Now say that there were a $p$-form global symmetry in the bulk.  This would mean that for any $(d-p)$-dimensional surface $\Sigma_{d-p}$ in the bulk, we could define symmetry operators $U(g,\Sigma_{d-p})$ under which surface operators $\mO(C_p)$, with $C_p$ a $p$-dimensional surface which intersects $\Sigma_{d-p}$ nontrivially, would transform.\footnote{One might worry that $U(g,\Sigma_{d-p})$ should only be well-defined on states where the bulk geometry has surfaces $C_p$ which are not contractible and surfaces $\Sigma_{d-p}$ which intersect them nontrivially.  Note however that in states where this is not the case, we may simply define $U(g,\Sigma_{d-p})$ to act as the identity.  These words may not seem like they should be precise nonperturbatively, where topology-changing amplitudes are possible, but if there were indeed an exact $p$-form global symmetry then it would have to set to zero any amplitudes which would change the topology in a way which violated the symmetry.}  Our goal is to reproduce the situation of figure \ref{regions2fig}, with an extra $\mathbb{T}^p$ coming along for the ride.  For the same reasons as discussed around definition \ref{bulkglobaldef}, in the boundary CFT we expect conjugation by $U(g,\Sigma_{d-p})$ to preserve $\AR$ for any boundary spatial region $R$.  The idea is then that we can therefore use splittability to write $U(g,\Sigma_{d-p})$ in the CFT as a product of an appropriate $U_{edge}$ with a set of operators $U(g,\mathbb{T}^p\times R_i)$, where the $R_i$ are a tiling of the boundary $\mathbb{S}^{d-p-1}$ and each $U(g,\mathbb{T}^p\times R_i)$ is a unitary element of $\mathcal{A}[\mathbb{T}^p\times R_i]$ whose action on elements of $\mathcal{A}[\mathbb{T}^p\times R_i]$ by conjugation is identical to that of $U(g,\Sigma_{d-p})$, just as in equation \eqref{Ueq}.\footnote{The reader may worry about our application of splittability here, since the boundary now contains a $\mathbb{T}^p$ on which unbreakable surface operators can wrap.  And even worse, our $p$-form global symmetry ensures there \textit{will} be such surfaces.  But in fact we are not doing any split on $\mathbb{T}^p$, we are splitting only on $\mathbb{S}^{d-p-1}$, which we should be able to split as long as $p<d-2$.  And even when $p=d-2$, we expect splittability can be restored by adding some heavy degrees of freedom to the boundary theory (at the cost of breaking the $p$-form symmetry in the UV).}  We can choose $\Sigma_{d-p}$ so that its intersection with the boundary is $\mathbb{S}^{d-p-1}$, in which case in the bulk $U(g,\Sigma_{d-p})$ acts on operators which create $p$-branes wrapping $\mathbb{T}^p$.  For example in the vacuum solution, we can take $\Sigma_{d-p}$ to be the set of points $t=x_p=0$, which is spanned by the radial direction and the coordinates on $\mathbb{S}^{d-p-1}$ and thus has topology $B^{d-p}$.  We therefore have all the ingredients of the setup of figure \ref{regions2fig}: if there were a $p$-form global symmetry, then there would be an operator which creates a charged $p$-brane wrapping $\mathbb{T}^p$ at a point in the center of the spatial $B^{d-p}$ in the vacuum solution, but the algebra of this operator with the $U(g,\mathbb{T}^p \times R_i)$ would have to be trivial by entanglement wedge reconstruction.  This contradicts the operator being charged under the $p$-form global symmetry in the first place, so there couldn't have been such a symmetry.  

The natural interpretation of this contradiction is that we should instead consider long-range $p$-form gauge symmetries in the bulk, since then an operator $\mO(C_p)$ which creates a charged brane wrapping $\mathbb{T}^p$ must be dressed by a Wilson surface $W_{\alpha}(C_{p+1})$ whose surface $C_{p+1}$ wraps $\mathbb{T}^p$ and also sweeps out a radial curve in $B^{d-p}$ from the location of the brane to the boundary $\mathbb{S}^{d-p-1}$.  The asymptotic $p$-form symmetry operators $U(g,\Sigma_{p-d-1})$ should then be interpreted as the symmetry operators of a \textit{boundary} $p$-form global symmetry a la definition \eqref{globalpdef}.  For the convenience of the reader we indicate the support of these various objects in the following table:
\begin{center}
\begin{tabular}{l|c c c}
& $r$ & $\mathbb{T}^p$ & $\mathbb{S}^{d-p-1}$\\\hline
$\mO(C_p)$ &  & x &  \\
$W_\alpha(C_{p+1})$ & x & x & \\
$U(g,\Sigma_{d-p-1})$ & & & x
\end{tabular}
\end{center}

We can use the same idea of dimensional reduction on $\mathbb{T}^p$ to also rerun the argument of section \ref{completenesssec} for the presence of states in all irreducible representations of a long-range $p$-form gauge symmetry with (compact) gauge group $G$ in the bulk.  Namely we may look at Wilson surfaces in the wormhole solution which wrap $\mathbb{T}^p$ and also sweep out a radial curve from one asymptotic boundary to the other, just as in figure \ref{threadfig}.  These Wilson surfaces are charged under the $p$-form asymptotic symmetry operators $U(g,\mathbb{S}^{d-p-1}_R)$, where $\mathbb{S}^{d-p-1}_R$ is the spatial sphere in the ``right'' asymptotic boundary, and by varying the representation of the Wilson surface we can again conclude that $U(g,\mathbb{S}^{d-p-1}_R)$ is nontrivial for all $g$ other than the identity.  We would now like to use theorems \ref{faithfulthm} and \ref{levythm} to show that this implies that there must be states transforming in all irreducible representations of $G$, but in order to be able to use the tensor product in the construction of theorem \ref{levythm} we need to make use of a generalization of the state-operator correspondence to surface operators (we can multiply two operators to get another operator transforming in the product of the representations of the first two, but we can't multiply two states and stay in the same Hilbert space!)  The idea is to use the Euclidean CFT path integral on $\mathbb{T}^p\times B^{d-p}$, with metric
\be\label{pstateop}
ds^2=dx_p^2+dr^2+r^2 d\Omega_{d-p-1}^2
\ee
and $r\in[0,R]$, to generate states of the CFT on $\mathbb{T}^p\times \mathbb{S}^{d-p-1}$.  If we do this with no insertions, we get a state which is neutral under conjugation by any $p$-form global symmetry operator $U(g,\mathbb{S}^{d-p-1})$.  If we insert a $p$-dimensional surface operator wrapping $\mathbb{T}^{p}$ at a definite point in $B^{d-p}$, then we get a state which transforms under $U(g,\mathbb{S}^{d-p-1})$ in the same representation as that surface operator does, while if we insert two of them at different points on $B^{d-p}$, then we get a state which transforms in the tensor product representation.  Conversely if we are given a state on $\mathbb{T}^p\times \mathbb{S}^{d-p-1}$, then by evolving it to small $r$ we can construct a $p$-dimensional surface insertion which gives that state when evolved back to $r=R$.\footnote{Note that unlike in the ordinary state-operator correspondence, evolution in $r$ is not part of the conformal symmetry group.  This means that the conformal transformation properties of the states and operators considered here will not be related in a nice way, which is why such a correspondence is usually not considered.
See \cite{Belin:2018jtf} for more discussion on this.  For our purposes this does not matter, since we only care about transformations under $p$-form global symmetries and these will be the same.}   Therefore the faithful action of $U(g,\mathbb{S}^{d-p-1})$ on the Hilbert space of the CFT on $\mathbb{T}^p\times \mathbb{S}^{d-p-1}$ does indeed imply a faithful action on the $p$-dimensional surface operators, which we may then multiply (at different points) to construct states in arbitrary finite-dimensional irreducible representations of the $p$-form global symmetry group $G$ using theorem \ref{levythm}.  These are also dual to $p$-dimensional surface operators, which can then be interpreted as creating the bulk $p$-branes which carry whichever finite-dimensional irreducible representation of $G$ we like.  

Finally we note that this state-operator correspondence for surface operators can also be used to establish a version of theorem \ref{compactthm} for $p$-form global symmetries: if we assume that the set of $p$-dimensional surface operators is finitely generated, meaning that the spectrum of the CFT on $\mathbb{T}^p\times\mathbb{S}^{d-p-1}$ is discrete and there is a finite set of surface operators at $r=0$ in the geometry \eqref{pstateop} whose operator product expansion recursively generates all the other ones, then any noncompact $p$-form global symmetry must be a subgroup of a compact one.  As before there is a subtlety for $d=p+2$, since there can be ``long branes'' near infinity which cause the spectrum on $\mathbb{T}^p\times\mathbb{S}^{d-p-1}$ to be continuous, in which case we need to additionally assume that the set of surface operators which generate the rest transform in a finite-dimensional representation of any $p$-form global symmetry.  This subtlety is not merely academic, in fact it potentially arises in all simple models of holography which are constructed from the near-horizon limit of a stack of BPS $(d-1)$-branes. For example $\mathcal{N}=4$ super Yang-Mills theory in $d=4$ on spatial $\mathbb{T}^2\times \mathbb{S}^{1}$ has a continuous spectrum due to D3 branes near the boundary which wrap $\mathbb{T}^2\times \mathbb{S}^{1}$ \cite{Maldacena:2001km}.  In this example there are no two-form global symmetries to discuss, but in other examples there might be. 

\subsection{Relationships between the conjectures?}
So far we have given independent arguments for conjectures \ref{nosym}-\ref{compact} (and their $p$-form generalizations), but in principle they might not be logically independent.  In fact in some cases there are simple relationships between them \cite{Banks:2010zn}, we here discuss these relationships and point out their limitations.  

The first potential relationship arises from the observation that for some gauge groups there is a close connection between the existence of a one-form global symmetry and the absence of matter fields charged under those gauge groups.  For example in $U(1)$ Maxwell theory with no dynamical electric charges, we have a $U(1)$ one-form global symmetry with symmetry operators \eqref{maxwelloneform}.  Similarly in $SU(N)$ gauge theory with no fundamental quarks, $\mathbb{Z}_N$ center symmetry is a one-form global symmetry.  One might hope to use these examples as motivation to give a general argument that a violation of conjecture \ref{allcharge} necessarily leads to the existence of a one-form global symmetry, and thus a violation of the one-form version of conjecture \ref{nosym}.  In other words one might argue that the one-form version of conjecture \ref{nosym} implies the zero-form version of conjecture \ref{allcharge}.  Unfortunately this idea does not work in general, these examples rely on special properties of the groups and representations involved.  Indeed consider an arbitrary gauge group $G$, under which matter fields transform in a representation $\Phi$.  We might like to use the kernel of $\Phi$ as a candidate for a one-form global symmetry, as we did in the above examples.  But in general this kernel will not lie in the center of $G$, and when it does not then we cannot use it to define a one-form symmetry (the candidate one-form symmetry operators \eqref{centersymops} would not be gauge-invariant).  We can realize a nontrivial one-form global symmetry only if the intersection of the kernel of $\Phi$ with the center of $G$ is nontrivial, but this will not always be the case.  One simple counterexample is a discrete gauge theory with gauge group $S_4$ (the permutation group on four elements), with a single matter field which transforms in the sign representation of $S_4$.  The kernel of this representation is the set of even permutations, but the center of $S_4$ is trivial so none of them can be used to create a one-form global symmetry.  

There is also an argument that in some cases conjectures \ref{nosym} and \ref{allcharge} together imply conjecture \ref{compact} \cite{Banks:2010zn}.  The idea is that if we had a noncompact gauge symmetry for which there were matter fields transforming in all irreducible representations, then there would also need to be a global symmetry.  For example say that there were a global symmetry with symmetry group $\mathbb{R}$. By conjecture \ref{allcharge} there would need to be a particle $a$ of charge one and a particle $b$ of charge $\sqrt{2}$.  But then any Lagrangian built out of polynomials of the fields for these charges would also have to be invariant under a global symmetry for which $a$ was neutral and $b$ had charge $\sqrt{2}$.  This argument is reminiscent of our proof of theorem \ref{compactthm}, for which it gave some inspiration, but it has several problems as stated.  The first is the explicit reference to a Lagrangian built out of polynomials of fundamental fields: it is far from clear that all quantum field theories can be constructed this way. Secondly, our arguments for conjecture \ref{allcharge} \textit{assume} the gauge group to be compact, without this there is no particular reason to expect all finite-dimensional irreducible representations to be realized.  Thirdly, it is not clear that this argument generalizes to noncompact groups other than $\mathbb{R}$.  And finally, even if we do consider $\mathbb{R}$, do accept the existence of the particles $a$ and $b$, and do accept the Lagrangian argument, it could be that the symmetry where $a$ is neutral and $b$ has charge $\sqrt{2}$ is \textit{also} gauged.  This is exactly what happened in our $U(1)\times U(1)$ example discussed below theorem \ref{compactthm}.  Our argument for theorem \ref{compactthm} avoids the first problem by using the operator product expansion instead of a Lagrangian, the second problem by using finite-generatedness instead of conjecture \ref{allcharge}, the third problem by working with arbitrary groups, and the fourth by showing not that a noncompact gauge symmetry in the bulk would lead to a bulk global symmetry, but instead that it must fit into a larger bulk gauge symmetry which is compact.

\section{Weak gravity from emergent gauge fields}\label{wgcsec}
There is a set of proposals, called \textit{weak gravity conjectures}, which attempt to generalize conjectures \ref{nosym}-\ref{allcharge}, the absence of global symmetries in quantum gravity and the presence of objects carrying all allowed long-range gauge charges, to some kind of lower bound on how weak (long-range) $U(1)$ gauge couplings can be \cite{ArkaniHamed:2006dz,Cheung:2014vva,Heidenreich:2015nta,Heidenreich:2016aqi,Ooguri:2016pdq,Cheung:2018cwt}.  These proposals typically involve asserting the existence of some object or objects whose $U(1)$ gauge charge $Q$ and mass $M$ obey (in $d\geq 4$ spacetime dimensions)
\be\label{wgc}
Q^2\geq \frac{8\pi(d-3)}{d-2} G M^2,
\ee
where $G$ is Newton's constant and the $O(1)$ constant comes from the charge-to-mass ratio of an extremal Reissner-Nordstrom black hole.  Often there are additional restrictions on the properties of the object(s), and rules about when saturation of the inequality counts as success.  

In \cite{Cheung:2014vva} a nontrivial proposal was given by Cheung and Remmen for a generalization of the inequality \eqref{wgc} to the case of multiple $U(1)$ gauge groups. First define
\be
C_d\equiv \sqrt{\frac{d-2}{8\pi(d-3)}}.
\ee
If there is a $U(1)^k$ long-range gauge symmetry, and if we label types of object by $i$, then for each $i$ we can define a vector in $\mathbb{R}^k$ by
\be
\vec{z}_i\equiv C_d\frac{\vec{Q}_{i}}{M_i\sqrt{G}},
\ee 
where $\vec{Q}_i$ is the vector which gives the charges of the $i$th type of object under $U(1)^k$.  The idea of \cite{Cheung:2014vva} is then that the right generalization of \eqref{wgc} is a requirement that the convex hull of all physically realized $z_i$ in $\mathbb{R}^k$ must contain the unit ball, again perhaps with further restrictions on which objects count and when saturation is acceptable.

The reason why even for $k=1$  there are many weak gravity conjectures is that there is no single nontrivial version of the conjecture for which there is a convincing general argument.  The closest one gets to a starting point for such an argument is a proposal for a principle that non-supersymmetric extremal black holes of any mass should not be stable \cite{ArkaniHamed:2006dz,Ooguri:2016pdq} (see \cite{Fisher:2017dbc,Cheung:2018cwt,Hamada:2018dde} for some other recent efforts). Unlike in our discussion of black holes and continuous global symmetries in the introduction however, here there is no known reason why such stability would be problematic.  Moreover it is not clear exactly what form of the conjecture should follow from this principle, for example are the objects obeying \eqref{wgc} or its convex hull generalization allowed to be black holes?  Here we also will not give a precise formulation (or proof) of a weak gravity conjecture.  We will instead just observe that one recent attempt \cite{Harlow:2015lma} to give a real quantum-gravity motivation for equation \eqref{wgc} also reproduces in a nice way the convex hull condition of \cite{Cheung:2014vva}.

The proposal of \cite{Harlow:2015lma} is to take seriously the factorization of the two-boundary gravitational system in AdS/CFT, specifically along the lines of arguing that any gauge constraints in the bulk must be emergent, and see what this emergence says about equation \eqref{wgc}.  This idea has not yet led to a general explanation of a weak gravity conjecture, but it does turn out that in simple models of an emergent $U(1)$ gauge field, a version of equation \eqref{wgc} is always satisfied \cite{Harlow:2015lma}.  In particular in the lattice version of the $\CPN$ nonlinear-$\sigma$ model with lattice spacing $1/\Lambda$, at large $N$ and for appropriate values of the coupling, there is an emergent $U(1)$ gauge field in the infrared, together with $N$ scalars of charge one and mass $m$, and for $d>4$ the low-energy gauge coupling is given by 
\be
\frac{1}{q^2}=N\Lambda^{d-4}.
\ee
Here the overall constant is non-universal, so we have just chosen it to be one.  For $d=4$, we instead have
\be
\frac{1}{q^2}=\frac{N}{12\pi^2}\log \left(\Lambda/m\right),
\ee
where the mass of the charged scalars cuts off an infrared divergence and the coefficient of the logarithm is universal.  The point is then that if we perturbatively couple this model to gravity, the charged scalars also generate an effective Newton constant
\be
\frac{1}{G}=N\Lambda^{d-2},
\ee
which we can use to test equation \eqref{wgc}.\footnote{There could also be a bare Newton's constant, but as long as it is positive then this only drives the overall Newton's constant to be smaller, making \eqref{wgc} easier to satisfy.  The primary consequence of the gauge field being emergent is that there is \textit{not} a large bare Maxwell term in the effective action at the cutoff scale.}  The idea is that in order for this analysis (presented in more detail in \cite{Harlow:2015lma}) to make sense, we need the mass of the scalars to be small in cutoff units:
\be
m^2\ll \Lambda^2,
\ee
so in particular we should have $m^2<C_d^2 \Lambda^2$.  
But we may then use our UV expressions for $1/q^2$ and $1/G$ to obtain (for simplicity working in $d>4$)
\be
m^2<C_d^2\Lambda^2=C_d^2\frac{q^2}{G},
\ee
which is precisely \eqref{wgc}.  

We will now extend this analysis to $k$ copies of the $\CPN$ model, each with its own value $N_i$ of $N$ and each with its own mass $m_i$ for the charged scalars.  There is an emergent $U(1)^k$ gauge symmetry in the infrared, and the gauge couplings are given (working in $d>4$ for simplicity) by
\be\label{qcpn}
\frac{1}{q_i^2}=N_i\Lambda^{d-4}.
\ee
Once we couple to gravity there is also an effective Newton constant
\be\label{Gcpn}
\frac{1}{G}=\sum_i N_i \Lambda^{d-2},  
\ee
with the sum appearing in \eqref{Gcpn} but not in \eqref{qcpn} because each set of $N_i$ scalars couples only to its own $U(1)$ gauge field but they all couple to gravity.  These equations can be combined to give 
\be
\Lambda^{-2}=G\sum_j \frac{1}{q_j^2},
\ee
and we now require that
\be\label{mineq}
m_i^2<C_d^2 \Lambda^2=\frac{C_d^2}{G\sum_j\frac{1}{q_j^2}}.
\ee
In this theory the Cheung-Remmen convex hull condition tells us that we need
\be
\sum_i\left(\lambda_i \frac{C_d q_i}{\sqrt{G}m_i}\right)^2\geq 1
\ee
for all $0\leq\lambda_i\leq 1$ obeying $\sum_{i=1}^k \lambda_i=1$.  We can use \eqref{mineq} term by term in this sum, leading to
\be\label{hullwgc}
\sum_{i,j}\left(\lambda_i\frac{q_i}{q_j}\right)^2\geq 1,
\ee 
which we claim is indeed true for all $\lambda_i$ for any collection of $q_i$.   The argument begins by defining 
\be
x_i\equiv \frac{q_i^{-2}}{\sum_j q_j^{-2}}
\ee
and
\be
f(\lambda)=\sum_i \frac{\lambda_i^2}{x_i},
\ee
in terms of which \eqref{hullwgc} becomes $f(\lambda)\geq 1$.  We may then observe that $f$ is a strictly convex function of the $\lambda_i$, since for all $\lambda_i\neq \lambda_i'$ and $s\in (0,1)$ we have
\begin{align}\nonumber
f(s\lambda+(1-s)\lambda')=&s f(\lambda)+(1-s)f(\lambda')-s(1-s)\sum_i (\lambda_i-\lambda_i')^2/x_i\\
&< s f(\lambda)+(1-s)f(\lambda').
\end{align}
The set of allowed $\lambda_i$ is convex, so any critical point of $f$ in this set will be a unique global minimum.  By taking the derivative with respect to $\lambda_i$, constrained by $\sum_i\lambda_i=1$, one easily sees that in fact there is a (unique) critical point at $\lambda_i=x_i$, where indeed we have $f=1$.  Thus the Cheung-Remmen convex hull condition holds in this many-parameter example of a set of emergent gauge fields coupled to gravity; we view this as evidence supporting the idea that the right motivation for whatever is the correct version of the weak gravity conjecture involves viewing the bulk gauge field as emergent.

\paragraph{Acknowledgments}
We thank Tom Banks, Thomas Dumitrescu, Zohar Komargodski, Nati Seiberg, and Sasha Zhiboedov for many useful discussions on the issues in this paper.  We also thank Nima Arkani-Hamed, Chris Beem, Mu-Chun Chen, Clay Cordova, Simeon Hellerman, Gary Horowitz, Ethan Lake, Hong Liu, Roberto Longo, Juan Maldacena, Greg Moore, Andy Strominger, Raman Sundrum, Wati Taylor, and Edward Witten for useful discussions.

We thank 
the Aspen Center for Physics, which is supported by
the National Science Foundation grant PHY-1607611, 
the Harvard Center for the Fundamental Laws of Nature, 
the Institute for Advanced Study,
the Kavli Institute for Theoretical Physics, 
the Okinawa Institute of Science and Technology Graduate School,
the Perimeter Institute, 
the Simons Center for Geometry and Physics, 
the Yukawa Institute of Fundamental Physics, 
for their hospitality during various stages of this work.
DH also thanks
the Kavli Institute for Physics and Mathematics of the Universe and the Maryland Center for Fundamental Physics for hospitality, and Alexander Huabo Yu Harlow for creating a stimulating environment while this work was being completed.

DH is supported by the US Department of Energy grants DE-SC0018944 and DE-SC0019127, the Simons foundation as a member of the
{\it It from Qubit} collaboration, and the MIT department of physics.
HO is supported in part by
U.S.\ Department of Energy grant DE-SC0011632,
by the World Premier International Research Center Initiative,
MEXT, Japan,
by JSPS Grant-in-Aid for Scientific Research C-26400240,
and by JSPS Grant-in-Aid for Scientific Research on Innovative Areas
15H05895.

\appendix

\section{Group theory}\label{groupapp}
In this appendix we briefly review some standard aspects of Lie group theory which are necessary for our work, but which may not be common knowledge for all physicists.  For many more details see eg \cite{lee2001introduction,knapp2013lie}, our discussion of representation theory largely follows \cite{knapp2013lie}.
\subsection{General structure of Lie groups}
A \textit{Lie group} is a group $G$ which is also a smooth manifold, and for which multiplication and inversion are smooth maps in that smooth structure. A vector field $X$ on $G$ is called \textit{left-invariant} if for any $h$ in $G$ it is preserved by the pushforward of the map $L_h:g\mapsto hg$.  The set of left-invariant vector fields forms a real vector space $\mathfrak{g}$, called the \textit{Lie algebra} of $G$, whose dimensionality equals that of the manifold, and which is closed under taking vector field commutators (abstractly a Lie algebra is a vector space with a bracket operation which is antisymmetric and obeys the Jacobi identity).  If $G$ has dimension zero as a manifold, then $\mathfrak{g}$ is empty.  There are then two classic results:
\begin{thm}[Closed subgroup theorem]\label{closedsgthm}
Let $G$ be a Lie group, and $H\subset G$ a subgroup of $G$ which is topologically closed.  Then $H$ is an embedded submanifold, and thus is itself a Lie group.  
\end{thm}
\begin{thm}[Lie group-Lie algebra correspondence]\label{liethm}
Let $\mathfrak{g}$ be an abstract real Lie algebra.  There exists a unique (up to isomorphism) connected simply-connected Lie group $\wt{G}$ whose Lie algebra is isomorphic to $\mathfrak{g}$.  Moreover any other connected Lie group $G$ whose Lie algebra is isomorphic to $\mathfrak{g}$ is itself isomorphic to a quotient of $\wt{G}$ by a discrete central subgroup $\Gamma\subset \wt{G}$.  More generally, any Lie group $G$ with a given Lie algebra is an extension of one of the connected ones by a discrete ``component'' group $C$,  meaning that there is a surjective homomorphism from $G$ to $C$ which sends each connected component of $G$ to a distinct element of $C$,\footnote{Mathematicians like to describe this situation by saying that there is a \textit{short exact sequence} $1\to G_0 \to G \to C \to 1$, where each arrow denotes a homomorphism and the kernel of each homomorphism is the image of the previous one.  In this sequence the other three homomorphisms are trivial inclusions and projections.} and that therefore $G_0\cong G/C$, where $G_0$ is the identity component of $G$.
\end{thm}
The proofs of these theorems use standard geometric techniques (vector flows, Frobenius's theorem, etc), they are nicely explained in \cite{lee2001introduction} (Ado's theorem is also needed, which is proven in \cite{knapp2013lie}).  We will give the proof of one further result which we will need below:
\begin{thm}\label{connectedgenthm}
Let $G$ be a connected Lie group, and $H\subset G$ be a subgroup which contains an open neighborhood $U$ of the identity in $G$.  Then $H=G$. 
\end{thm}
\begin{proof}
We will show that $H$ is both open and closed: since $G$ is connected, this implies $H=G$.  $H$ is open because for any $h\in H$, the set $hU$ is open in $G$, it contains $h$, and it is contained in $H$.  Therefore $H=\bigcup_{h\in H} (hU)$.  $H$ is closed because if $g\notin H$, then we also have $gU\cap H=\emptyset$.  Indeed if we had $gu=h$ for some $u\in U$ and $h\in H$, then we would have $g=hu^{-1}$, and thus $g\in H$.  Therefore we have $G-H=\bigcup_{g\notin H}gU$.
\end{proof}

\subsection{Representation theory of compact Lie groups}\label{repsubsec}
A \textit{representation} of a Lie group $G$ on a complex vector space $V$ is a homomorphism $\rho$ from $G$ into the set of linear operators on $V$, for which the map $\Phi_\rho:G\times V\to V$ defined by $\Phi_\rho(g,v)=\rho(g)v$ is jointly continuous.\footnote{In the main text we used ``physics'' notation where the components of the representation matrices for a representation $\rho$ in some specific basis for $V$ are denoted $D_{\rho,ij}(g)$. In this appendix we simplify things by just using $\rho(g)$ to refer to the abstract operators.}  If $\rho$ is injective then it is said to be \textit{faithful}.  $\rho$ is said to be \textit{unitary} if $V$ admits an inner product with respect to which $\rho(g)$ is unitary for any $g$, and is said to be \textit{finite-dimensional} if $V$ is finite-dimensional.   The \textit{kernel} of $\rho$, denoted $\mathrm{Ker}(\rho)$, is the set of elements of $G$ which are mapped to the identity operator on $V$.  $\mathrm{Ker}(\rho)$ is always a closed normal subgoup of $G$, and $\rho$ is faithful if and only if $\mathrm{Ker}(\rho)=\{e\}$.  A subspace $S\subset V$ is called \textit{invariant} if $\rho(G)S=S$, and $\rho$ is said to be \textit{irreducible} if the only closed invariant subspaces are $V$ itself and $0$.  By the closed subgroup theorem  any finite-dimensional representation of a Lie group $G$ is automatically smooth, which is why we only required $\rho$ to be continuous, and actually by theorem \ref{infinitedecthm} below the same is true for infinite-dimensional unitary representations if $G$ is compact.

The representations of a general Lie group $G$ can be quite sophisticated, but if $G$ is compact and $\rho$ is either unitary or finite-dimensional then there is a simple theory of all representations which can be derived from the existence of the invariant Haar measure $dg$ on $G$. Indeed there is a simple theorem relating these two conditions:
\begin{thm}\label{urepthm}
Let $G$ be a compact Lie group, and $\rho$ be a finite-dimensional representation of $G$.  Then $\rho$ is unitary.  
\end{thm}
\begin{proof}
Let $(v,v')_0$ be any inner product on $V$.\footnote{Recall that an inner product on a complex vector space $V$ is a map $(\,,\,):V\times V\to \mathbb{C}$ which is linear in the second argument, obeys $(v,v')^*=(v',v)$ for any $v,v'$, and for which $(v,v)\geq 0$ for any $v$, with equality only if $v=0$.  These conditions imply that an inner product is antilinear in the first argument.  Mathematicians usually instead take the first argument to be linear and the second to be antilinear, but our choice is closer to bra-ket notation.}  We may then define a new inner product
\be
(v,v')\equiv \int dg(\rho(g)v,\rho(g)v')_0,
\ee
which is easily shown to be an inner product with respect to which $\rho(g)$ is unitary for any $g$.  
\end{proof}
For a unitary representation, the orthogonal complement of an invariant subspace is also invariant.  Therefore theorem \ref{urepthm} shows that any finite-dimensional representation of a compact Lie group can be decomposed into a direct sum of irreducible representations.  We next establish a famous technical lemma, which we then use to prove perhaps the most remarkable feature of the representation theory of compact groups: the Schur orthogonality relations.
\begin{thm}[Schur's lemma] \label{schurlem} Let $\alpha$ and $\alpha'$ be finite-dimensional irreducible representations of a group $G$ on $V$ and $V'$ respectively (here $G$ can be an arbitrary group and we assume no continuity properties of $\alpha$ and $\alpha'$).  If $L:V\to V'$ is a linear map obeying $\alpha'(g)L=L\alpha(g)$ for all $g\in G$, then either $L$ is a bijection or $L=0$.  Moreover if $\alpha=\alpha'$ and $V=V'$, then $L$ is a multiple of the identity.
\end{thm}
\begin{proof}
It is easy to see that the kernel and image of $L$ are invariant subspaces of $V$ and $V'$ respectively. Irreducibility of $\alpha$ implies that the kernel of $L$ is either $0$ or $V$: if it is $V$ then $L=0$, while if it is $0$ then $L$ is injective. If $L$ is injective, then irreducibility of $\alpha'$ implies that its image must $V'$, so $L$ is surjective.  In the case $\alpha=\alpha'$ and $V=V'$, since $L$ is finite-dimensional if it is not equal to zero then it has a nonzero eigenvalue $\lambda$. We may then consider the operator $\hat{L}\equiv L-\lambda I$, which again is a linear map which commutes with $\alpha$.  But it is not injective so it must be zero.
\end{proof}
\begin{thm}[Schur orthogonality relations]\label{schurthm}
Let $\alpha$ and $\alpha'$ be irreducible finite-dimensional representations of a compact Lie group $G$ on the vector spaces $V$ and $V'$, which are inequivalent in the sense that there is no invertible linear map $L:V\to V'$ such that $L\alpha(g)=\alpha'(g)L$ for any $g\in G$.  Then in the inner products for which $\alpha$ and $\alpha'$ are unitary we have
\begin{align}\label{schur1}
\int dg (u',\alpha'(g)v')^*(u,\alpha(g)v)&=0 \qquad\qquad\qquad \forall u,v\in V, u',v'\in V'\\\label{schur2}
\int dg (u',\alpha(g)v')^*(u,\alpha(g)v)&=\frac{(u,u')(v,v')^*}{\mathrm{dim}(V)} \qquad \forall u,v,u',v'\in V.
\end{align}
Choosing orthonormal bases for $V$ and $V'$ and reverting to physics notation, we have
\begin{align}\label{schur1phys}
\int dg D^*_{\alpha' i' j'}(g)D_{\alpha,ij}(g)&=0\\\label{schur2phys}
\int dg D^*_{\alpha,i'j'}(g)D_{\alpha,ij}(g)&=d_\alpha^{-1}\delta_{ii'}\delta_{jj'}.
\end{align}
\end{thm}
\begin{proof}
Given any $u\in V$, $u'\in V'$, we can define a map $L_{u,u'}:V\to V'$ via 
\be
(v',L_{u,u'}v)\equiv \int dg (u',\alpha'(g)v')^*(u,\alpha(g)v).
\ee
It is straightforward to verify that $\alpha'L_{u,u'}=L_{u,u'}\alpha$ using the invariance of the Haar measure, so by theorem \ref{schurlem} $L_{u,u'}$ must either be a bijection or be zero.  Moreover it cannot be a bijection since $\alpha$ and $\alpha'$ are inequivalent, so it must be zero, establishing equation \eqref{schur1}.  To establish equation \eqref{schur2}, we can similarly define maps $L_{u,u'}:V\to V$ and $L_{v,v'}:V\to V'$ via
\be
(v',L_{u,u'}v)\equiv (u,L_{v,v'}u')\equiv\int dg (u',\alpha(g)v')^*(u,\alpha(g)v).
\ee
By the invariance of the Haar measure these maps both commute with $\alpha(g)$ for any $g$, so by theorem \ref{schurlem} they both must be multiples of the identity on $V$.  This establishes equation \eqref{schur2} up to an overall constant, which we may then fix by taking $u=u'$ and summing $u$ over an orthonormal basis for $V$ using the unitarity of $\alpha$.
\end{proof}
The Schur orthogonality relations immediately imply the orthogonality of the characters $\chi_\alpha(g)\equiv \Tr \alpha(g)$ of inequivalent finite-dimensional irreducible representations, as well as the fact that $\int dg \chi_\alpha^* (g)\chi_\rho(g)$ counts the number of times a finite-dimensional irreducible representation $\alpha$ appears in the direct-sum decomposition of an arbitrary finite-dimensional representation $\rho$.  They can be interpreted as saying that the rescaled set of matrix coefficients $\sqrt{d_\alpha}D_{\alpha,ij}(g)$ give a set of orthonormal states in the Hilbert space $L^2(G)$ of square-normalizable complex-valued functions on $G$.  In fact they are an orthonormal basis:
\begin{thm}[Peter-Weyl theorem]\label{pwthm}
Let $G$ be a compact Lie group.  Then the rescaled matrix coefficients $\sqrt{d_\alpha}D_{\alpha,ij}(g)$ for all finite-dimensional irreducible representations give an orthonormal basis for $L^2(G)$.
\end{thm}
The proof of this theorem is an exercise in functional analysis and can be found in \cite{knapp2013lie}, presenting it here would be too much of a digression.  

We emphasize that so far in this subsection all results have been essentially topological, and have not actually used the smoothness in the definition of the Lie group $G$.  Indeed theorems \ref{urepthm}-\ref{pwthm} are also true if multiplication and inversion are only taken to be continuous and the topology of $G$ is only required to be compact and Hausdorff, since these are sufficient for the existence of the Haar measure.  When we do assume that $G$ is a Lie group however we then have the following remarkable result:
\begin{thm}\label{liefaithfulthm}
Any compact Lie group $G$ has a faithful finite-dimensional unitary representation, and thus is isomorphic to a closed subgroup of $U(n)$ for some $n$.
\end{thm}
\begin{proof}
The proof begins with the observation that by the Peter-Weyl theorem \ref{pwthm}, for any $g\in G$ we can find a finite-dimensional irreducible representation $\alpha_g$ for which $\alpha_g(g)$ is not the identity (otherwise we could never approximate a function on $G$ which takes different values at $e$ and $g$).  We may first consider the case where $G$ is discrete, so its identity component $G_0$ consists of only the identity.  $G$ is therefore finite, and we can construct a faithful representation via $\oplus_{g\in G}\,\alpha_g$.  Alternatively say that there exists a $g_1\neq e$ in $G_0$: then $G_1\equiv\mathrm{ker}(\alpha_{g_1})$ is a closed  subgroup of $G$, so by the closed subgroup theorem \ref{closedsgthm} it is a Lie subgroup whose dimensionality is at most that of $G$.  In fact its dimensionality must be strictly less than that of $G$, since if they were equal then by theorem \ref{connectedgenthm} we would have $(G_1)_0=G_0$, which contradicts the fact that $\alpha_{g_1}(g_1)$ is not the identity.  Now say that $G_1$ is zero-dimensional: then as before we see that $\alpha_{g_1}\oplus_{g\in G_1}\alpha_g$ is a faithful finite-dimensional representation of $G$.  Alternatively if $G_1$ has positive dimension then we have $g_2\in (G_1)_0$ such that $g_2\neq e$, so we can take $G_2\equiv \mathrm{ker}(\alpha_{g_1}\oplus \alpha_{g_2})$, which again will be a closed subgroup of $G_1$ of dimension strictly less than that of $G_1$.  Continuing in this way we eventually reach a $G_n$ which is zero-dimensional, and we may then take $\alpha\equiv \alpha_{g_1}\oplus \ldots \oplus\alpha_{g_n}\oplus_{g\in G_n}\alpha_g$, which will be faithful.  It is unitary by theorem \ref{urepthm}.
\end{proof}
Thus we see that the structure theory of compact Lie groups and their finite-dimensional representations is quite well understood.  In fact their unitary infinite-dimensional representations are also understandable along similar lines, we now note two results in this direction.  
\begin{thm}\label{infinitedecthm}
Let $\rho$ be a unitary representation of a compact Lie group $G$ on a Hilbert space $V$.  Then $\rho$ is the direct sum of a set of finite-dimensional irreducible representations.  
\end{thm}
The proof of this theorem uses the Peter-Weyl theorem to show that there cannot be any elements of $V$ which are orthogonal to the direct sum of all finite-dimensional invariant subspaces, see \cite{knapp2013lie} for the proof.  We then also have
\begin{thm}\label{faithfulthm}
Let $\rho$ be a faithful unitary representation of a compact Lie group $G$ on a Hilbert space $V$.  Then there is a finite-dimensional invariant subspace of $V$ on which $\rho$ also acts faithfully,  so $\rho$ has a finite-dimensional subrepresentation which is also faithful.
\end{thm}
\begin{proof}
The faithfulness of $\rho$ ensures that for any element $g\in G$ there is a finite-dimensional irreducible representation $\alpha_g$ appearing in the direct sum decomposition promised by theorem \ref{infinitedecthm} for which $\alpha_g(g)$ is not the identity.  The remainder of the proof is identical to that of theorem \ref{liefaithfulthm}.
\end{proof}
The last result we will need relates arbitrary irreducible representations of a compact group to any particular faithful finite-dimensional one \cite{Levy:2003my}:
\begin{thm}\label{levythm}
Let $G$ be a compact Lie group, $\rho$ be a faithful finite-dimensional representation of $G$, and $\rho^*$ be its conjugate representation.  Then for any finite-dimensional irreducible representation $\alpha$ of $G$ there exist nonnegative integers $n$ and $m$ such that $\alpha$ appears in the direct sum decomposition of the tensor-product $\rho^{\otimes n} \otimes \rho^{*\otimes m}$.
\end{thm}
\begin{proof}
Consider the representation 
\be
\rho_n\equiv\left(1\oplus \rho\oplus \rho^*\oplus\rho\otimes\rho^*\right)^{\otimes n}.
\ee
It has character
\be
\chi_n(g)\equiv \Tr\rho_n(g)=|1+\chi_\rho(g)|^{2n},
\ee
where  $\chi_\rho(g)\equiv \Tr \rho(g)$ is the character of $\rho$.  By Schur orthogonality we can count the number of times any irreducible representation $\alpha$ appears in the direct sum decomposition of $\rho_n$ by
\be
\int_G dg \chi_\alpha(g)|1+\chi_\rho(g)|^{2n}.
\ee
The quantity $|1+\chi_\rho(g)|$ obeys
\be
0\leq |1+\chi_\rho(g)|\leq 1+d_\rho,
\ee
with the maximum attained only when $g=e$ since $\rho$ is faithful.  But then we have
\be
\lim_{n\to\infty}\frac{\int_G dg\chi_\alpha|1+\chi_\rho(g)|^{2n}}{\int_G dg |1+\chi_\rho(g)|^{2n}}=d_\alpha,
\ee
so at some sufficiently large $n$ we must have
\be
\int_G dg \chi_\alpha(g)|1+\chi_\rho(g)|^{2n}>0.
\ee
\end{proof}

If $G$ is connected, much more is known about its representation theory, and indeed both the connected compact Lie groups and their finite-dimensional irreducible representations have been classified long ago using semisimple theory.  In this paper however we have striven to treat discrete and continuous groups on equal footing, so we will stop our review here.

\section{Projective representations}\label{projapp}
In this appendix we discuss the possibility of extending our definition of global symmetry to include projective representations of the symmetry on Hilbert space, where the multiplication rule \eqref{repeq} would be generalized to include a phase
\be \label{projrep}
U(g,\Sigma)U(g',\Sigma)=e^{i\alpha(g,g')} U(gg',\Sigma).  
\ee
We now argue that in quantum field theory on $\mathbb{R}^d$, any such phase can be removed by a redefinition of the $U(g,\Sigma)$.  We first consider the situation where the symmetry is unbroken: then there is an invariant vacuum state, on which the symmetry can at most act with a phase
\be
U(g,\Sigma)|0\ran=e^{if(g)}|0\ran.  
\ee  
But if we act on this state with $U(g,\Sigma)U(g',\Sigma)$, we immediately discover that we must have
\be
\alpha(g,g')=f(g)+f(g')-f(gg') \qquad \mathrm{mod} \, 2\pi.
\ee
We may then define ``improved'' symmetry operators
\be
\wt{U}(g,\Sigma)\equiv e^{-if(g)}U(g,\Sigma),
\ee
which act in the same way on the local operators but now obey \eqref{projrep} with $\alpha=0$.  Thus in any quantum mechanical system, nontrivial projective representations are only possible if there is no invariant state: in other words the symmetry must be spontaneously broken.  There are indeed quantum mechanical systems where a spontaneously broken global symmetry is represented projectively in a nontrivial way, see appendix D of \cite{Gaiotto:2017yup} for an example, but we now argue that in quantum field theory this is impossible.  

The reason is that in quantum field theory on $\mathbb{R}^d$, spontaneously broken global symmetries (as we have defined them) always lead to the superselection structure described around equation \eqref{ssstates}.  Since the operators always transform in non-projective representations of the symmetry (the phase $\alpha$ cancels when act on operators by conjugation), and since we can get to all states by acting with operators that do not change the superselection sector on the degenerate vacua, any projectiveness on the states can arise only from phases in the action of the symmetry on the degenerate vacuum states:
\be\label{ssbproj}
U(g)|b\ran=e^{if(g,b)}|gb\ran.
\ee  
Strictly speaking to have a genuine projective representation we should not allow $f$ to depend on $b$, but we have allowed this since in any case it will not help: such phases can again be removed by the redefinition
\be
\wt{U}(g)\equiv U(g) e^{-if(g,B_i)},
\ee
where $B_i$ are the operators which diagnose which superselection sector a state is in.  Since the $B_i$ commute with all local operators, this modification has no effect on the action of the symmetry on local operators.  Thus the $\wt{U}(g)$ give a non-projective representation of the symmetry on the Hilbert space.\footnote{More precisely since we have defined representations to be continuous, it gives a homomorphism from G to the unitary operators on the Hilbert space which may or may not be continuous.}

In equation \eqref{ssbproj} we considered a kind of generalized projective representation, where instead of respecting the group multiplication law up to a $c$-number phase we respect it up to a nontrivial unitary operator which commutes with all of the local operators. One might ask if there are other examples of this kind of thing, where the unitary operator depends on something other than degenerate vacuum data.  In a quantum field theory where all states can be obtained by acting on a single ground state with local operators, there can be no nontrivial unitary operator which commutes with all of the local operators.  There are two ways we could try to relax this assumption in the hopes of getting something interesting.  The first is to have multiple ground states, each of which has on top of it a superselection sector built by acting with local operators.  This is the case we just considered, and we saw that allowing the unitary operators to depend on the superselection sector data did not lead to anything worthwhile.  The second possibility is to consider theories where not all states can be obtained by acting on the ground state(s) with local operators.  The only possibility we are aware of is to have a theory with a ``long range gauge symmetry with dynamical charges'', a notion we define in section \ref{gaugesec}.  It basically means that there is a weakly-interacting gauge field and operators charged under the associated gauge symmetry, which must be attached to infinity by Wilson lines to be gauge-invariant.  The gauge symmetry is then represented nontrivially on the Hilbert space, in what is sometimes called an asymptotic symmetry group, and since this can be understood as being realized by a surface operator at infinity it will give a set of nontrivial unitary operators that commute with all local operators but act nontrivially on the endpoints of Wilson lines. We could therefore imagine trying to use these long-range gauge symmetries as generalizations of the phases $e^{i\alpha(g,g')}$ in a projective representation of the global symmetry.  Indeed in section \ref{mixsec} we give a concrete example of a theory that realizes this phenomenon, and in a way in which the unitary cannot be removed by redefining the symmetry operators.  One might then wish to say that this is a genuine projective representation of the global symmetry, but as we explain in section \eqref{mixsec} we find it more natural to instead say that it is a mixing of the global symmetry with a long-range gauge symmetry.  Therefore we are not aware of any situation in quantum field theory where the most natural description of the symmetry structure is to say that a global symmetry is represented projectively on the Hilbert space.

\section{Continuity of symmetry operators}\label{contapp}
In this appendix we discuss the continuity of the action of global symmetries in quantum field theory, both on the Hilbert space and on the algebra $\AR$ of bounded operators in a bounded spatial region $R$.  

First some definitions.  Let $V$ be a Hilbert space, which we will always endow with the standard topology induced by the Hilbert space norm
\be
||v||\equiv \sqrt{(v,v)}.
\ee
 One says that a linear operator $\mO$ on $V$ is \textit{bounded} if there exists a real constant $C$ such that $||\mO v||<C||v||$ for all $v\in V$, and we will denote by $\BV$ the set of bounded operators on $V$.  We will say that a subset $M\subset \BV$ is \textit{uniformly bounded} if there exists a single real constant $C$ such that $||\mO v||<C||v||$ for all $v\in V$ and $\mO\in M$.  The \textit{operator norm} $||\mO||$ of any bounded operator $\mO$ is the smallest real constant $C$ such that $||\mO v||\leq C||v||$ for all $v\in V$.  

To discuss the continuity of maps to and from $\BV$, we need to give it a topology.  There are several possibilities.  One obvious one is the \textit{norm topology}, which has as a basis the set of balls 
\be
B_\epsilon(\mO_0)\equiv \{\mO\in \BV \,\Big|\,||\mO-\mO_0||<\epsilon\},
\ee
with $\mO_0\in \BV$ and $\epsilon>0$.  This topology however is much too strong for our purposes.   For example in the norm topology the $U(1)$ global symmetry $\phi'=e^{i\theta }\phi$ of a free complex scalar field has symmetry operators $U(g,\Sigma)$ which are not continuous, since there are states of arbitrarily large charge in the Hilbert space.  A topology which is better suited is the \textit{strong operator topology}, which has as a basis the set of finite intersections of balls of the form
\be
B_\epsilon(\mO_0,v_0)\equiv \{\mO\in \BV \,\Big|\,||(\mO-\mO_0)v_0||<\epsilon\},
\ee
with $\mO_0\in \BV$, $v_0\in V$, and $\epsilon>0$.  This topology is sometimes also called the topology of pointwise convergence, since a sequence $\mO_n$ of operators converges to an operator $\mO$ in the strong operator topology if and only if $\mO_nv\to\mO v$ for any $v\in V$.  Similarly, if $X$ is a topological space then a map $f:X\to \BV$ is continuous in the strong operator topology if and only if the map $f_v:X\to V$ defined by $f_v(x)=f(x)v$ is continuous for any fixed $v\in V$.

In discussing the continuity of symmetries, there are two maps whose continuity properties we are interested in.  The symmetry operators $U(g,\Sigma)$ directly define a map
\be
U:G\to \BV,
\ee
and also induce an associated map
\be
f_U:G\times \AR \to \AR
\ee
for $R$ any spatial region, defined by $f_U(g,\mO)=U^\dagger(g)\mO U(g)$. As a warmup, we first establish the following theorem
\begin{thm}
Let $V$ be a Hilbert space, $G$ a Lie group, and $U$ a map from $G$ to $\BV$ for which $U(g)$ is unitary for all $g\in G$.  Then the map $\Phi_U:G\times V \to V$  defined by $\Phi_U(g,v)=U(g)v$ is jointly continuous if and only if $U$ is continuous in the strong operator topology.  In particular, if $U$ is a homomorphism which is strongly continuous then it is a representation of $G$ on $V$ in the sense of subsection \ref{repsubsec}.
\end{thm}
\begin{proof}
If $\Phi_U$ is jointly continuous, then strong continuity of $U$ follows immediately from fixing the second argument.  To establish the converse, we need to show that for any ball 
\be
B_\epsilon(v_0)\equiv \{v\in V \,\Big|\, ||v-v_0||<\epsilon\}
\ee
in $V$, $\Phi_U^{-1}(B_\epsilon(v_0))$ is open in $G\times V$.  We can do this by showing that any point $(g,v)$ in $\Phi_U^{-1}(B_\epsilon(v_0))$ is contained in an open set $S\times B_\delta(v)$, with $S$ open in $G$, which is itself contained in $\Phi_U^{-1}(B_\epsilon(v_0))$.  We therefore want to show that
\be
||U(g')v'-v_0||<\epsilon \qquad \qquad \forall g'\in S, v'\in B_\delta(v).
\ee
This follows because by the triangle inequality and the unitary invariance of the Hilbert space norm we have
\begin{align}\nonumber
||U(g')v'-v_0||&\leq ||U(g')(v'-v)||+||(U(g')-U(g))v||+||U(g)v-v_0||\\
&=||(v'-v)||+||(U(g')-U(g))v||+||U(g)v-v_0||.
\end{align}
The third term on the second line is less than $\epsilon$ since $(g,v)$ is in $\Phi_U^{-1}(B_\epsilon(v_0))$, and using our freedom to choose $\delta$ and $S$ and the strong continuity of $U$ we can make the first and second terms as small we like.  Therefore we can arrange for the sum of all three to be less than $\epsilon$.    
\end{proof}
This theorem tells us that in quantum field theory $U(g,\Sigma)$ will be strongly continuous if and only if its action on the Hilbert space gives a continuous representation of $G$.  We saw in the beginning of section \ref{globalsec} that if $G$ is continuous as a Lie group, meaning its dimension as a manifold is greater than zero, then if it is spontaneously broken the $U(g,\Sigma)$ defined by equation \eqref{ssbUdef} may \textit{not} be strongly continuous, since elements of $g$ which are arbitrarily close to the identity still send one ground state to another which is orthogonal.  If the symmetry is unbroken however, then we take it as a natural postulate that $U$ will indeed be strongly continuous.  For example in the free complex scalar example, any particular normalizable state will be acted on continuously even though there are states with arbitrary large charge.  More generally the idea is that if the vacuum is invariant, then any particular excited state should only differ from the vacuum in a finite region and by a finite amount of excitation so it should only transform in a representation of limited complexity.    We now use the idea that $U$ should be strongly continuous for unbroken symmetries to motivate the continuity clause in condition (b) of our definition \ref{globaldef} of global symmetry.  
\begin{thm}\label{fcontthm}
Let $V$ be a Hilbert space, $G$ a Lie group, and $U$ a strongly continuous map from $G$ to the unitary operators on $V$.  Then the restriction to any uniformly bounded subset $M$ of $\BV$ of the map $f_U:G\times \BV \to \BV$ defined by $f_U(g,\mO)=U^\dagger(g)\mO U(g)$ is strongly continuous.
\end{thm}
\begin{proof}
We will show that for any ball $B_\epsilon(\mO_0,v_0)$ in $\BV$, $f_U^{-1}(B_\epsilon(\mO_0,v_0))\cap (G\times M)$ is open in $G\times M$.  We can do this by showing that for any $(g,\mO)\in f_U^{-1}(B_\epsilon(\mO_0,v_0))\cap (G\times M)$, there is an open set $S\subset G$ containing $g$ and a ball $B_\delta(\mO,\hat{v})$ such that $S\times (B_\delta(\mO,\hat{v})\cap M)\subset f_U^{-1}(B_\epsilon(\mO_0,v_0))\cap (G\times M)$.  In other words for any $\epsilon$, $\mO_0$, and $v_0$, we want to pick $S$, $\delta$, and $\hat{v}$ such that
\be
||\left(U^\dagger(g')\mO' U(g')-\mO_0\right)v_0||<\epsilon \qquad \forall g'\in S, \mO'\in B_\delta(\mO,\hat{v})\cap M.
\ee
By the triangle inequality and the unitary invariance of the Hilbert space norm we have
\begin{align}\nonumber
||\left(U^\dagger(g')\mO' U(g')-\mO_0\right)v_0||\leq &||U^\dagger(g')\mO'(U(g')-U(g))v_0||+||U^\dagger(g')(\mO'-\mO)U(g)v_0||\\\nonumber
&+||(U^\dagger(g')-U^\dagger(g))\mO U(g)v_0||+||(U^\dagger(g)\mO U(g)-\mO_0)v_0||\\\nonumber
=&||\mO'(U(g')-U(g))v_0||+||(\mO'-\mO)U(g)v_0||\\
&+||(U^\dagger(g')-U^\dagger(g))\mO U(g)v_0||+||(U^\dagger(g)\mO U(g)-\mO_0)v_0||.\label{fcontreq}
\end{align}
The fourth term on the right hand side is less than $\epsilon$ since $(g,\mO)$ is in $f_U^{-1}(B_\epsilon(\mO_0,v_0))$, the third term can be made as small as we like using the strong continuity of $U$ and the boundedness of $\mO$, the second term can be made as small as we like by choosing $\hat{v}=U(g)v_0$ and taking $\delta$ to be small, and the first term can be taken to be arbitrarily small by using the strong continuity of $U$ together with the uniform boundedness of $M$.  Therefore for small enough $S$ and $\delta$ we can arrange for the whole right hand side to be less than $\epsilon$.  
\end{proof}
Thus we see that strong continuity on any uniformly bounded subset of $\AR$ is the right continuity requirement on $f_U$ for an unbroken global symmetry.  In fact we claim that if the region $R$ is bounded in size, then this should also be the right requirement even if the symmetry is spontaneously broken, since this should not affect the transformation of operators in a finite region, hence our inclusion of it in condition (b) of definition \ref{globaldef}.  It is worth emphasizing that without the restriction to uniformly bounded subsets the theorem would not apply, since the first term in the right hand side of equation \eqref{fcontreq} would not be bounded since there are elements $\mO'$ of any open ball $B_\delta(\mO,\hat{v})$ with arbitrarily large norm.

We can also consider what strong continuity of $f_U$ on uniformly bounded subsets implies in the converse direction about the continuity of $U$.  In general it does not imply anything, which is good since for spontaneously broken symmetries we sometimes do not want $U$ to be continuous.  But if we \textit{assume} that the symmetry is unbroken, by which we mean that there is an invariant ground state $\Omega\in V$, then we have the following theorem:
\begin{thm}\label{Ucontthm}
Let V be a Hilbert space, $G$ a Lie group, $\AR$ a subalgebra of $\BV$, and $U$ a map from $G$ to the unitary operators on $V$ such that the restriction to any uniformly bounded subset $M$ of $\AR$ of the map $f_U:G\times \BV\to\BV$ defined by $f_U(g,\mO)=U^\dagger(g)\mO U(g)$ is strongly continuous.  Moreover let there exist a state $\Omega\in V$ which is cyclic with respect to $\AR$,\footnote{A state $\Omega\in V$ is \textit{cyclic} with respect to a subalgebra $\AR\subset \BV$ if the set of states $\mO \Omega$, with $\mO\in \AR$, are dense in $V$. It is \textit{separating} if there is no $\mO\in \AR$ such that $\mO\Omega=0$.  In quantum field theory the Reeh-Schlieder theorem tells us that both of these properties hold for the ground state when $\AR$ is the algebra of operators in a bounded region (see eg \cite{Witten:2018zxz}).} and which is also invariant in the sense that $U(g)\Omega=\Omega$ for all $g\in G$.  Then $U$ is strongly continuous.
\end{thm} 
\begin{proof}
We want to show that for any $\epsilon>0$, $v_0\in V$, $\mO_0\in \BV$, we have that $U^{-1}(B_\epsilon(\mO_0,v_0))$ is open in $G$.  We do this by showing that for any $g$ such that $U(g)\in B_{\epsilon}(\mO_0,v_0)$, there is a neighborhood $S$ of $g$ in $G$ such that $U(S)$ is also contained in $B_{\epsilon}(\mO_0,v_0)$.  In other words we want
\be
||(U(g')-\mO_0)v_0||<\epsilon \qquad \forall g'\in S.
\ee
We first note that by the cyclicity of $\Omega$, we have
\be
v_0=\wt{O}\Omega+\wt{v}
\ee
for some $\wt{O}\in \AR$, with the norm of $\wt{v}$ being as small as we like.  From the triangle inequality and the invariance of $\Omega$ we then have
\begin{align}\nonumber
||(U(g')-\mO_0)v_0||\leq&||\left(U^\dagger(g'^{-1})\wt{\mO}U(g'^{-1})-U^\dagger(g^{-1})\wt{\mO}U(g^{-1})\right)\Omega||+||U(g')\wt{v}||\\
&+||U(g)\wt{v}||+||(U(g)-\mO_0)v_0||.
\end{align}
The fourth term will be less than $\epsilon$ since $U(g)$ is in $B_{\epsilon}(\mO_0,v_0)$, by cyclicity we can take $||U(g)\wt{v}||=||U(g')\wt{v}||=||\wt{v}||$ as small we like, and since $\wt{O}$ will always be part of some uniformly-bounded subset of $\AR$ the first term can be made arbitrarily small using the joint strong continuity of $f_U$ on uniformly-bounded subsets.  Therefore the sum of all three terms can be taken to be less than $\epsilon$.
\end{proof}
Thus we can be reassured that our continuity requirement in condition (b) of definition \ref{globaldef} is not too weak.  

Finally we point out that if we do have an invariant ground state which is both cyclic and separating with respect to $\AR$, then actually there is a different topology in which the situation is even nicer.  This topology is defined by noting that we can actually use the state $\Omega$ to define an inner product on $\AR$ via
\be
(\mO_1,\mO_2)_\Omega \equiv (\mO_1\Omega,\mO_2\Omega),
\ee
which gives $\AR$ the structure of a Hilbert space. Here $(\cdot,\cdot)$ is the usual Hilbert space inner product on $V$, and $(\cdot,\cdot)_\Omega$ is a good inner product on $\AR$ because $(\mO,\mO)_\Omega\geq 0$, with equality only when $\mO=0$ due to the fact that $\Omega$ is separating with respect to $\AR$.  We may then use this inner product to define an alternative topology on $\AR$, which we call the \textit{vacuum topology}, using as a basis the balls $B_\epsilon(\mO_0,\Omega)$.  Since these are a subset of the balls used in defining the strong operator topology, this topology is weaker than the strong operator topology.  We then have the following theorem:
\begin{thm}\label{Drepthm}
Let V be a Hilbert space, $G$ a Lie group, $\AR$ a subalgebra of $\BV$, and $U$ a map from $G$ to the unitary operators on $V$ such that the restriction to any uniformly bounded subset $M$ of $\AR$ of the map $f_U:G\times \BV\to\BV$ defined by $f_U(g,\mO)=U^\dagger(g)\mO U(g)$ is strongly continuous.  Moreover let there exist a state $\Omega\in V$ which is cyclic and separating with respect to $\AR$, and which is also invariant in the sense that $U(g)\Omega=\Omega$ for all $g\in G$.  Then the restriction to $\AR$ of $f_U$ is jointly continuous in vacuum topology on $\AR$, without any uniform-boundedness requirement, and in particular if $U$ is a homomorphism then $f_U$ gives a representation of $G$ on the Hilbert space $\AR$ with inner product $(\cdot,\cdot)_\Omega$.  Moreover this representation is unitary.
\end{thm}
\begin{proof}
We can first invoke theorem \ref{Ucontthm} to learn that $U$ is strongly continuous.  We may then imitate the proof of theorem \ref{fcontthm}, noting however that now we only need the inequality \eqref{fcontreq} to hold when $v_0=\Omega$.  But then the first term on the righthand side is automatically zero since $(U(g')-U(g))\Omega=0$, so we have no need of a uniform boundedness requirement.  Finally to see that the representation of $G$ on $\AR$ furnished by $f_U$ is unitary, we simply note that
\begin{align}\nonumber
(U^\dagger(g)\mO_1U(g),U^\dagger(g)\mO_2U(g))_\Omega&=(U^\dagger(g) \mO_1\Omega,U^\dagger(g)\mO_2\Omega)=(\mO_1\Omega,\mO_2\Omega)\\
&=(\mO_1,\mO_2)_\Omega.
\end{align}
\end{proof}
In particular this theorem tells us that if a global symmetry is unbroken, then the map $D$ defined by equation \eqref{Dmap} gives a unitary representation of $G$.  And in particular if $G$ is compact, then by theorem \ref{infinitedecthm} $D$ should decompose into a direct sum of finite-dimensional unitary representations.  Moreover not only did we not need a uniform-boundedness requirement in the proof of theorem \ref{Drepthm}, in fact we did not even need to assume that the elements of $\AR$ are bounded!  As long as we restrict to operators whose domain includes the invariant state $\Omega$, we still may use $\Omega$ to define an inner product on these operators in terms of which the action of $f_U$ is unitary and continuous, and thus gives a unitary representation.  

It is interesting to note that if we drop the assumption that the symmetry is unbroken, there are easy examples where the action $f_U$ of $G$ on local operators is not unitary. For example in a free scalar field theory in $d>2$, there is a spontaneously-broken global symmetry which acts on the scalar $\phi$ and the identity $1$ as
\be\label{shiftphi}
\begin{pmatrix}
\phi' \\ 1'
\end{pmatrix}
=\begin{pmatrix} 1 & a \\ 0 & 1\end{pmatrix}\begin{pmatrix}\phi \\ 1\end{pmatrix},
\ee
which is a non-unitary representation of the symmetry group $\mathbb{R}$.  In this kind of situation it is sometimes said that the symmetry ``acts non-linearly'' on $\phi$, but in fact $f_U$ always gives a linear action of $G$ on the set of local operators, and this is manifest in \eqref{shiftphi}.

\section{Building symmetry insertions on general closed submanifolds} \label{closedsubapp}

Consider a $(d-1)$-dimensional compact connected oriented manifold $\Sigma$ embedded in $\mathbb{R}^d$.  Since $H_{d-1}(\mathbb{R}^d)$ is trivial, there is a $d$-dimensional compact connected oriented submanifold $M$ in $\mathbb{R}^d$ such that $\Sigma = \partial M$. In this appendix we show that the insertion of a symmetry operator on $\Sigma$ into the path integral can always be understood in operator language as conjugating all operators in $M$ by $U(g,\mathbb{R}^{d-1})$, as shown in figure \ref{gluingfig} for the special case of $d=3$ and $\Sigma=\mathbb{T}^2$.  

\bfig
\includegraphics[height=5cm]{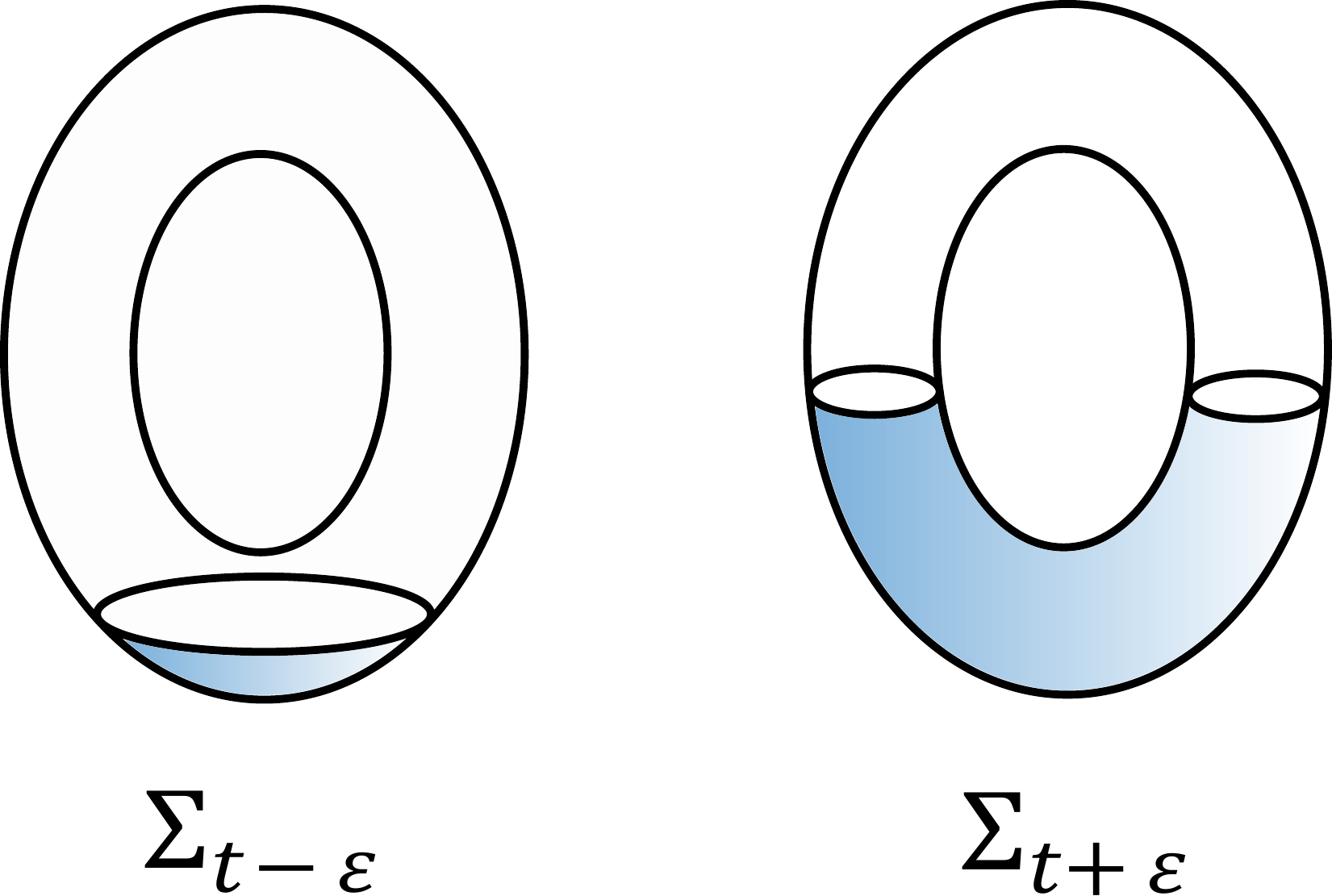}
\caption{Illustrations of $
\Sigma_t = f^{-1}((-\infty, t])=
\{ x \in \Sigma: ~ f(x) \leq t \} $ at two different values of
$t$ when $\Sigma$ is a torus.}\label{morseone}
\efig
Indeed by generically choosing a ``time'' direction in 
$\mathbb{R}^d$, with a linear coordinate $t$, we can define a Morse function $f$ on $\Sigma$ such that $f(p)=t$ at $p \in \Sigma$ (a Morse function is a smooth map from a manifold $\Sigma$ to $\mathbb{R}$ which has no degenerate critical points; such functions are dense in the set of smooth maps from $\Sigma$ to $\mathbb{R}$, so a generic orientation of the time direction will give us one).
For each $t$, define,
\be 
\Sigma_t = f^{-1}((-\infty, t]) = \{ p \in \Sigma: ~ f(p) \leq t \} .
\ee
See figure \ref{morseone} for its illustration. 
We also define,
\be
\overline{M}_t = \mathbb{R}^{d-1}_t \backslash M_t,
\ee
where
$\mathbb{R}^{d-1}_t$ and $M_t$ are sections of $\mathbb{R}^d$
and $M$ at $t$. Let us glue $\Sigma_t$ with $\overline{M}_t$
at their common boundaries $f^{-1}(t)$, to get a surface we call $C_t$.
In the following, we will use Morse theory to study how 
$U(g, {\cal C}_t)$ behaves as we increase
$t$ from $-\infty$ to $+ \infty$.

The Morse function $f$ has isolated non-degenerate
critical points on $\Sigma$. 
The fundamental theorems (Theorems 3.1 and 3.2 in \cite{Milner})
 of the Morse theory say:

\begin{thm}\label{morsetheoremone}
Suppose $t_1 < t_2$
and $f^{-1}([t_1, t_2])$ is compact and contains no critical points of $f$.
Then $\Sigma_{t_1}$ is diffeomorphic
to $\Sigma_{t_2}$ and the inclusion map
$\Sigma_{t_1} \rightarrow \Sigma_{t_2}$ is a homotopy equivalence.
\end{thm}

\bfig
\includegraphics[height=5cm]{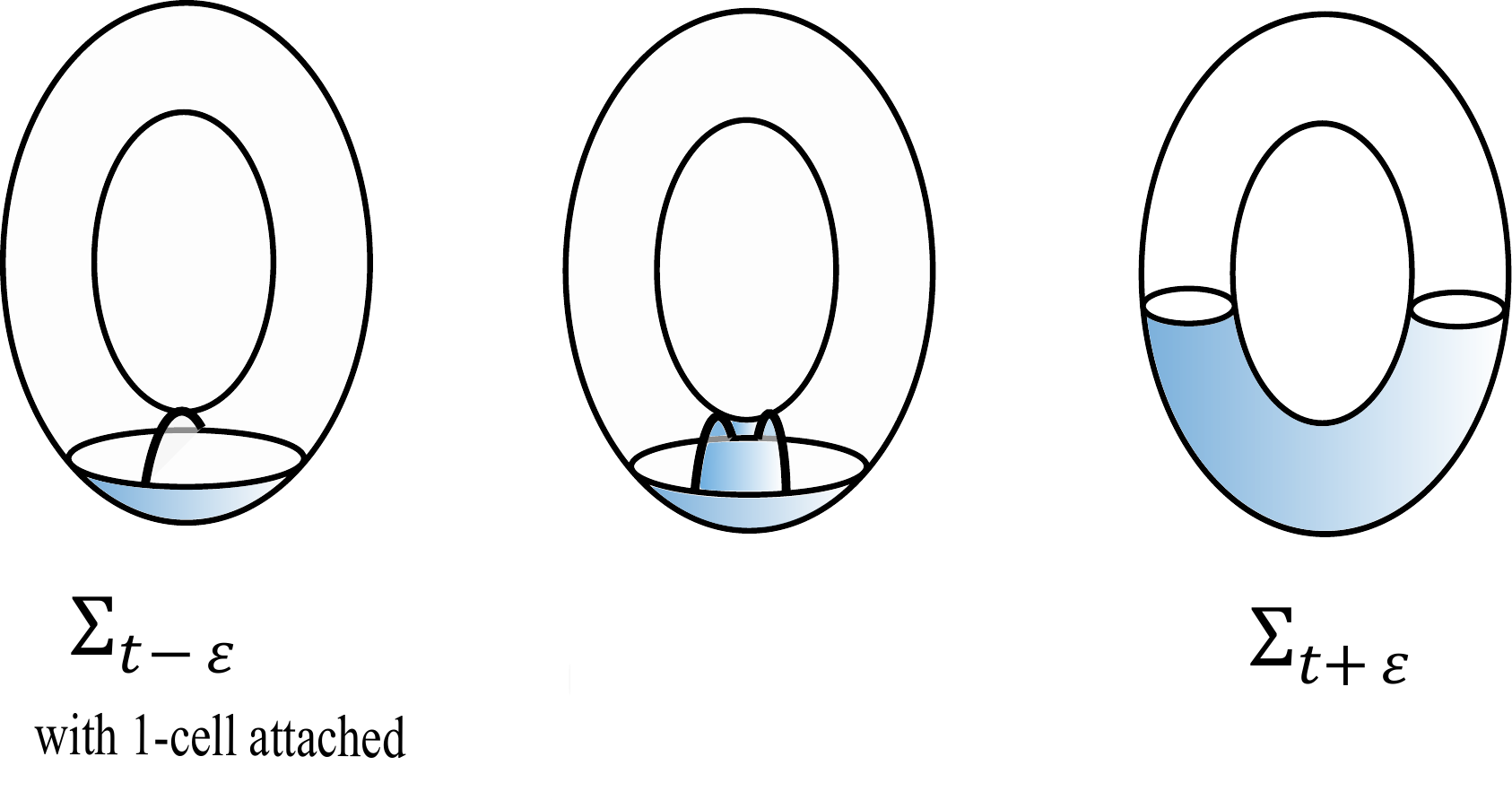}
\caption{When $p$ is a critial point of $f$ with index $n$,
$\Sigma_{f(p)+\epsilon}$
is homotopic to $\Sigma_{f(p) - \epsilon}$
with an $n$-cell attached, provided we choose $\epsilon>0$ to be sufficiently
small.}\label{morsetwo}
\efig
The second fundamental theorem tells us what happens at 
critical points. Before stating the theorem, let us note that
according to Morse's lemma, each
critical point $p$ of $f$ is characterized 
by its index $n$, which means that we can choose coordinates
$(x_1, \cdots , x_{d-1})$ around $p$ such that $p$ is at $x=0$ and,
\be 
f (x) = f(p) - x_1^2 - \cdots - x^2_n + x_{n+1}^2
+ \cdots + x_{d-1}^2,
\ee
holds throughout the coordinate patch (these coordinates are obtained by diagonalizing the Hessian matrix at $p$).
We can choose $\epsilon> 0$ sufficiently small
so that $f$ has no other critical point in $[t- \epsilon,
t+ \epsilon]$, where $t=f(p)$.

\begin{thm}\label{morsetheoremtwo}
If $p$ is a critical point of $f$ with $f(p)=t$ and index $n$, 
and if there is no other critical point in $f^{-1}([t-\epsilon, t+\epsilon])$ for
some $\epsilon > 0$, 
 $\Sigma_{t+\epsilon}$
is homotopic to $\Sigma_{t - \epsilon}$
with an $n$-cell attached.\footnote{See appendix \ref{stabilizerapp} for a brief discussion of CW complexes and the definition of an $n$-cell.} See figure \ref{morsetwo}  for illustration. 
\end{thm}

Since $\Sigma$ is compact,
there is $t_0$ such that $\Sigma_t$
is empty for  $t < t_0$. For such $t$, 
${\cal C}_t = \mathbb{R}^{d-1}_t$ and
$U(g, {\cal C}_t)$ is the symmetry generator. Let us choose $t_0$ to be
the largest such $t_0$.
Increasing $t$ continuously, we reach $t=t_0$
where $\mathbb{R}_{t_0}^{d-1}$ touches $\Sigma$.
Clearly, $\Sigma_{t_0+\epsilon}$ is homotopic to 
$\Sigma_{t_0 - \epsilon}$ (which is empty) with a $0$-cell (the 
point of the first contact) attached, as expected from Theorem \ref{morsetheoremtwo}. 
We can then continously deform
${\cal C}_{t_0 - \epsilon} =\mathbb{R}^{d-1}_{t-\epsilon}$ to 
${\cal C}_{t_0 + \epsilon}$ and 
$U(g, {\cal C}_{t_0 + \epsilon})$ is still a symmetry
generator.

\bfig
\includegraphics[height=5cm]{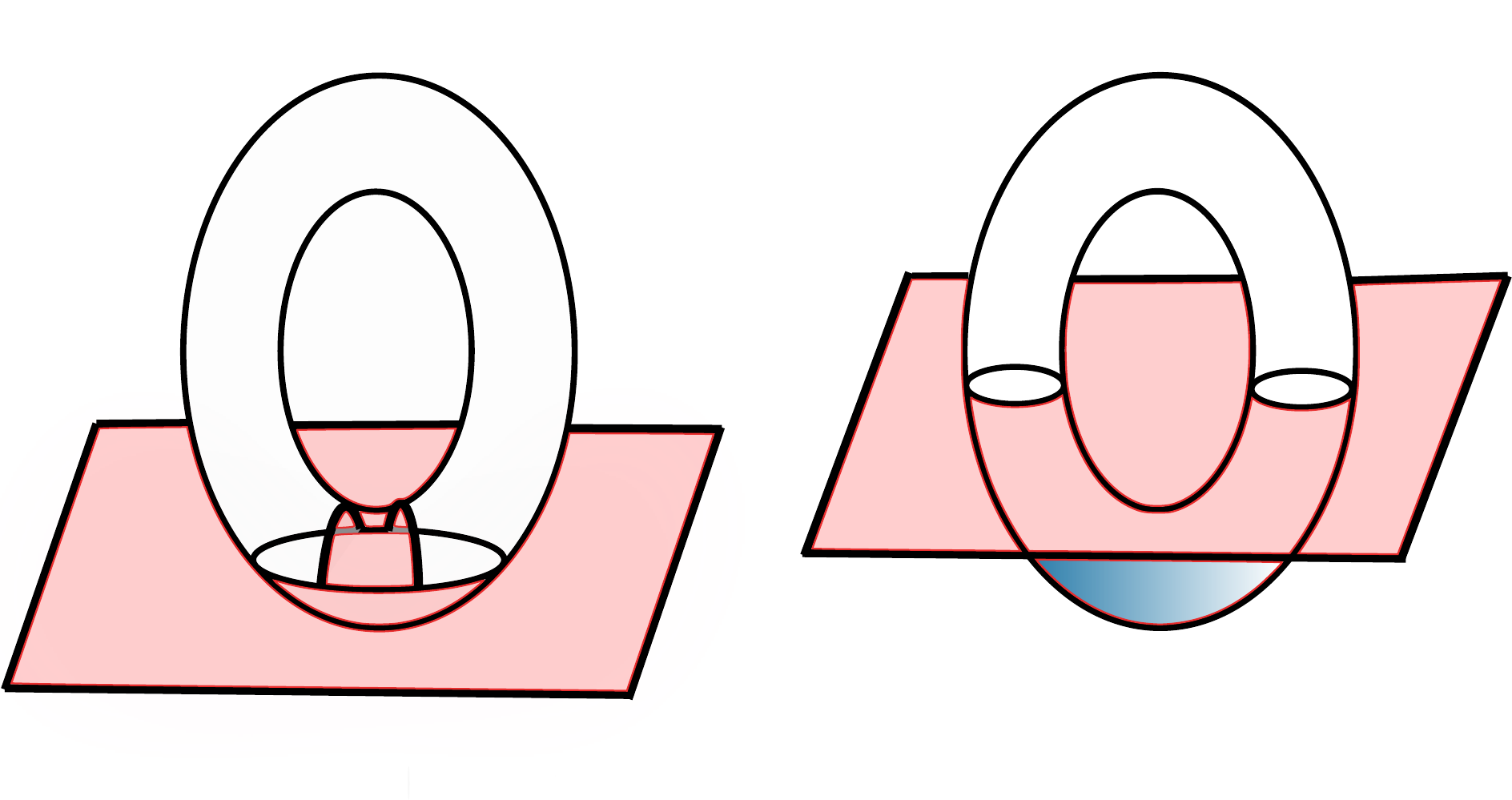}
\caption{
Since symmetry insertions on the same cell with
opposite orientations (in red and blue in the figure) cancel, $U(g, {\cal C}_{t-\epsilon})$
can be continuously deformed to $U(g, {\cal C}_{t+\epsilon})$.}\label{morsethree}
\efig
As we increase $t$ further, we will inevitably encounter
a critical point with non-zero index $n$ at some $t$.
According to Theorem \ref{morsetheoremtwo}, we can homotopically
deform $\Sigma_{t - \epsilon}$ to $\Sigma_{t + \epsilon}$ 
by attaching an $n$-cell. We can also deform 
$\overline{M}_{t-\epsilon}$ to $\overline{M}_{t+\epsilon}$
by attaching the same $n$-cell with opposite orientation. 
Since symmetry insertions on the pair of $n$-cells  with
opposite orientations has no effect, $U(g, {\cal C}_{t-\epsilon})$
can be continuously deformed to $U(g, {\cal C}_{t+\epsilon})$.
See  figure \ref{morsethree}  for illustration.

Since $\Sigma$ is compact, there is $t_1$ suth that 
$\Sigma_{t} = \Sigma$ for  $t > t_1$.
Choosing $t_1$ to be 
the smallest such $t_1$,
${\cal C}_{t_1} = \Sigma \cup \mathbb{R}^{d-1}_{t_1}$. 

We conclude
that the symmetry generator $U(g, {\cal C}_t) = U(g, \mathbb{R}^{d-1}_t)$ for $t < t_0$
can be deformed to $U(g, {\cal C}_{t_1}) =U(g, \Sigma \cup \mathbb{R}^{d-1}_{t_1})$
at  $t = t_1$. Since
$U(g, \Sigma) = U(g, \Sigma \cup \mathbb{R}^{d-1}_{t_1}) U(g, \mathbb{R}^{d-1}_{t_1})^\dagger$, this is what we wanted to show.

\section{Lattice splittability theorem}\label{splitapp}
In this appendix we give a proof of theorem \ref{latticethm}, which says that a unitary which acts locally on each tensor factor of a tensor product Hilbert space must itself be a tensor product of local unitaries.  \footnote{This proof uses a few basic facts about purifications.  These follow easily from the Schmidt decomposition of any pure state in a bipartite system, which says that for any state $|\psi\ran_{AB}$, there are orthonormal states $|a\ran_A$ and $|a\ran_B$ such that $|\psi\ran=\sum_a \sqrt{p_a} |a\ran_A|a\ran_B$, with $0\leq p_a\leq 1$ and $\sum_a p_a=1$.  For a brief overview of the Schmidt decomposition see, eg, appendix C of \cite{Harlow:2014yka}.} 

\begin{proof}
We first note that it is enough to establish the theorem for the case of two tensor factors, $\mathcal{H}=\HA\otimes \HB$, with a unitary $U_{AB}$ which send operators on $A$ to operators on $A$, and operators on $B$ to operators on $B$, since we can then iterate the argument to obtain the desired result for any finite number of tensor factors.  We thus just need to show that $U_{AB}=(U_A\otimes I_B)(I_A\otimes U_B)$.  

The basic idea is to double the size of the system, introducing copies $\mathcal{H}_{\hat{A}}$ and $\mathcal{H}_{\hat{B}}$ of $\HA$ and $\HB$, and then consider the state
\be\label{phieq}
|\phi\ran\equiv \frac{1}{\sqrt{|A||B|}}\sum_{ab}|a\ran_{\hat{A}}|b\ran_{\hat{B}}U_{AB}|ab\ran_{AB}.
\ee
Here $|a\ran_A$, $|a\ran_{\hat{A}}$ are orthonormal bases for $\HA$ and $\mathcal{H}_{\hat{A}}$, and similarly for $|b\ran_B$, $|b\ran_{\hat{B}}$.  Noting that $U_{AB}^\dagger(I_A\otimes\mO_B)U_{AB}=(I_A\otimes\mO'_B)$ for any $\mO_B$, and that any operator $\mO_{B\hat{B}}$ can be expanded as a sum of tensor products of operators on $\HB$ and $\mathcal{H}_{\hat{B}}$, a simple calculation shows that for any operators $\mO_{\hat{A}}$ and $\mO_{B\hat{B}}$ on $\mathcal{H}_{\hat{A}}$ and $\HB\otimes \mathcal{H}_{\hat{B}}$ respectively, we must have
\begin{align}\nonumber
\lan\phi|\mO_{\hat{A}}\mO_{B\hat{B}}|\phi\ran&=\lan\phi|U_{AB}\mO'_A\mO'_{B\hat{B}}U_{AB}^\dagger|\phi\ran\\\nonumber
&=\lan\phi|U_{AB}\mO_{\hat{A}}U_{AB}^\dagger|\phi\ran\lan\phi|U_{AB}\mO'_{B\hat{B}}U_{AB}^\dagger|\phi\ran\\
&=\lan\phi|\mO_{\hat{A}}|\phi\ran\lan\phi|\mO_{B\hat{B}}|\phi\ran.
\end{align}
In other words there is no correlation between $\hat{A}$ and $B\hat{B}$, so the partial trace of $|\phi\ran\lan\phi|$ over $A$ factorizes:
\be
\rho_{\hat{A}B\hat{B}}(\phi)\equiv \Tr_{A}|\phi\ran\lan\phi|=\rho_{\hat{A}}(\phi)\otimes\rho_{B\hat{B}}(\phi).
\ee
Moreover from \eqref{phieq} we have
\be
\rho_{\hat{A}}(\phi)=\frac{I_{\hat{A}}}{|A|},
\ee
where $|A|$ denotes the dimensionality of $\HA$.  

Now the key point is that the state $\rho_{\hat{A}B\hat{B}}(\phi)$ must be purified into $|\phi\ran$ by adding back the $A$ system, which means that its rank can be at most $|A|$.  But since the rank of $\rho_{\hat{A}}(\phi)$ is already $|A|$, this means that $\rho_{B\hat{B}}(\phi)$ must have unit rank, or in other words must be a pure state $|\chi\ran\lan\chi|_{B\hat{B}}$.  We may then observe that since any two purifications of a mixed state onto a given system differ at most by a unitary transformation on that system, and since the state
\be
|\psi\ran=\frac{1}{\sqrt{|A|}}\sum_{a}|a\ran_{\hat{A}}|a\ran_A|\chi\ran_{B\hat{B}}
\ee
is a purification of $\hat{A}B\hat{B}$ onto $A$, it must be that $|\phi\ran$, which is another such purification, is given by
\be
|\phi\ran=U_A|\psi\ran=\frac{1}{\sqrt{|A|}}\sum_{a}|a\ran_{\hat{A}}U_A|a\ran_A |\chi\ran_{B\hat{B}}
\ee
for some $U_A$. Moreover since again from \eqref{phieq} we have $\rho_{\hat{B}}(\phi)=\frac{I_{\hat{B}}}{|B|}$, by the same argument we must have
\be
|\chi\ran_{B\hat{B}}=\frac{1}{\sqrt{|B|}}\sum_b |b\ran_{\hat{B}}U_B|b\ran_B
\ee
for some $U_B$.  We then finally have that
\be
|\phi\ran=\frac{1}{\sqrt{|A||B|}}\sum_{ab}|a\ran_{\hat{A}}|b\ran_{\hat{B}}U_A|a\ran_A U_B|b\ran_B,
\ee
which is compatible with \eqref{phieq} if only if $U_{AB}=U_A\otimes U_B$.  
\end{proof}

\section{Hamiltonian for lattice gauge theory with discrete gauge group}\label{gaugeapp}
In this appendix we sketch how to derive the lattice gauge theory Hamiltonians \eqref{KSH}, \eqref{discreteham} from the continuous-time limit of the Wilson action.  The Euclidean Wilson action on a spacetime cubic lattice with lattice spacing $a$ is \cite{Wilson:1974sk}
\be
S_E=-\frac{a^{d-4}}{g^2}\sum_{\gamma\in \hat{\Gamma}} W_\alpha(\gamma),
\ee
where $\hat{\Gamma}$ is the set of (oriented) plaquettes in Euclidean spacetime and $\alpha$ is a faithful representation of $G$.  This action makes sense for any gauge group $G$, discrete or continuous.  To extract a Hamiltonian, we need to take the lattice spacing in the time direction, which we'll denote as $a_0$, to be much smaller than the lattice spacing in the space directions, which we'll continue to call $a$.  In this case the Wilson action becomes
\begin{align}\nonumber
S_E&=-\frac{a^{d-4}}{g^2}\left(\frac{a}{a_0}\sum_{\gamma\in \hat{\Gamma}_0}W_\alpha(\gamma)+\frac{a_0}{a}\sum_{\gamma\in\hat{\Gamma}_s}W_{\alpha}(\gamma)\right)\\
&\equiv -A \sum_{\gamma\in \hat{\Gamma}_0}W_\alpha(\gamma)-B\sum_{\gamma\in\hat{\Gamma}_s}W_{\alpha}(\gamma),
\end{align}
where $\hat{\Gamma}_0$ denotes the set of plaquettes which have a time component and $\hat{\Gamma}_s$ denotes the set of plaquettes with no time component.

We now study the thermal partition function 
\be
Z(\beta)\equiv\int \mathcal{D}g e^{-S_E},
\ee
where we are integrating over an element of $g$ assigned to each edge of a cubic Euclidean spacetime lattice with periodic time.  We can use gauge transformations to set the temporal edges all to the identity except for at one time, and the integral over the temporal edges at that time simply imposes a projection onto gauge-invariant states.  The thermal partition function then has the form \cite{Creutz:1976ch}
\be
Z(\beta)=\Tr(T^N),
\ee
where the trace is over only gauge-invariant states and $T$ is called the transfer matrix; it is given by
\be
\lan g'|T|g\ran=\exp\left(A \sum_{e\in E}\Tr\Big(D_{\alpha}(g_eg_e^{\prime -1})+D_\alpha(g_e'g_e^{-1})\Big)+B\sum_{\gamma\in\Gamma}W_\alpha(\gamma)\right).
\ee
Here $|g\ran$ and $|g'\ran$ are elements of gauge-field part of the Hilbert space \eqref{bigH}.  As in the main text, $E$ denotes the set of edges in a time slice and $\Gamma$ denotes the set of plaquettes in a timeslice.  Note that $\Gamma$ is \textit{not} equal to $\hat{\Gamma}_s$, which is the set of spatial plaquettes at all times.  We may re-express $T$ using our lattice gauge theory operators:
\be
T=\prod_{e\in E}\left(\int dh e^{A\Tr\left(D_\alpha(h)+D_\alpha(h^{-1})\right)}L_h(e)\right)e^{B\sum_{\gamma\in\Gamma}W_\alpha(\gamma)},
\ee
where we have written $L_h(e)$ instead of $L_h(\ell)$ since this  expression does not care which way we orient the link $\ell$ on edge $e$. Finally to extract the Hamiltonian we take the limit $a_0\to 0$, identifying the Hamiltonian via 
\be
T=e^{-a_0 H}.
\ee 
To proceed, we now need to decide if $G$ is continuous or discrete.  If it is continuous, in the limit $a_0$ the integral over $h$ will be dominated by the region near the identity.  We may then use a Gaussian approximation to evaluate it, which directly  leads to the Kogut-Susskind Hamiltonian \eqref{KSH} up to an additive $c$-number renormalization \cite{Creutz:1976ch}.  When $G$ is discrete things are a little more subtle, to obtain an interesting theory we need to forget the expressions for $A$ and $B$ in terms of $a$, $a_0$, and $g$, which after all came from trying to reproduce the Yang-Mills action in the continuum, and instead simply view $A$ and $B$ as parameters to vary as we like.  For $G$ continuous we took $A$ to infinity and $B$ to zero such that their product was finite, but for $G$ discrete the right limit is instead to take $A$ to infinity and $B$ to zero such that $Be^A$ is finite:  it is only in this limit that (after another $c$-number renormalization) we have that $T\approx 1-\epsilon H$ with $\epsilon$ small and $H$ a Hamiltonian with both ``electric'' and ``magnetic'' terms \cite{Fradkin:1978th}.  In this limit the identity contribution to the sum over $h$ is set to one by the $c$-number renormalization, which replaces $\Tr\left(D_\alpha(h)+D_\alpha(h^{-1})\right)$ by $\Tr\left(D_\alpha(h)+D_\alpha(h^{-1})\right)-2d_\alpha$ for each edge, and the other terms in the sum over $h$ which survive in the continuous-time limit are those which maximize $\Tr\left(D_\alpha(h)+D_\alpha(h^{-1})\right)$. This finally leads to the Hamiltonian \eqref{discreteham}, with the normalization of the new gauge coupling $g$ being chosen in a somewhat arbitrary manner.

\section{Stabilizer formalism for the \texorpdfstring{$\mathbb{Z}_2$}{Z2} gauge theory}\label{stabilizerapp}
The stabilizer formalism is a useful technique for defining nontrivial subspaces of the Hilbert space of $n$ qubits \cite{Gottesman:1997zz}.  In this appendix we explain how it may be used to compute the ground state degeneracy of the $\mathbb{Z}_2$ lattice gauge theory with charged matter in the limit of small $g$ and large $\lambda$, with Hamiltonian \eqref{toric}.  In fact in these ground states the charges are never excited, so our result also gives the ground state degeneracy of the pure $\mathbb{Z}_2$ gauge theory, which is one of the simplest topological quantum field theories.   In the main text we are primarily interested in cubic lattices which discretize the $d-1$-dimensional ball $B^{d-1}$, but, mostly for fun, we will use a few tools from algebraic topology to compute the ground state degeneracy for any spatial lattice with the structure of a $d-1$-dimensional CW complex.\footnote{CW complexes are discrete versions of manifolds, which are constructed recursively by starting with a collection of points, called zero-cells, attaching a set of intervals, called one-cells, such that the boundary of each one-cell consists of some subset of zero-cells, attaching a set of discs, called two-cells, such that the boundary of each two-cell consists of the zero-cells and one-cells, and so on up to $(d-1)$-cells if the complex is $(d-1)$-dimensional \cite{hatcher2002algebraic}. In our lattice gauge theory parlance, the zero-cells are the sites, the one-cells which are not in the boundary are the edges, and the two-cells are the plaquettes.}  In the continuum limit, this will give a formula for the Hilbert space dimension of the $\mathbb{Z}_2$ gauge theory on any spatial $d-1$-manifold, with or without boundary.  In particular we will show that the Hamiltonian \eqref{toric} has a unique ground state on any lattice whose CW complex is homeomorphic to $B^{d-1}$, on which the operators  $Z(\gamma)$ and $\prod_{\vec{\delta}}X(\vec{x},\vec{\delta})$ act as the identity for any plaquette $\gamma$ and site $\vec{x}$, while more generally the ground state degeneracies for any connected CW complex (or connected manifold) are given by \eqref{noBdeg} if there is no boundary and \eqref{Bdeg} if there is a boundary.

The basic idea of the stabilizer formalism is to consider the $+1$ eigenspace of an abelian subgroup $S$ of the $n$-qubit Pauli group $P_n$.  $P_n$ is the multiplicative group of operators on the Hilbert space of $n$ qubits which is generated by all single-qubit Pauli operators together with $i I$, where $I$ is the identity operator and $i=\sqrt{-1}$.  The stabilizer formalism then rests on the following theorem:
\begin{thm}\label{stabilizerthm}
Let $S$ be a abelian subgroup of $P_n$, not containing $-I$, which is generated by $m$ independent generators $\{g_1, \ldots, g_m\}$.  Then the subspace of states on which all elements of $S$ act as the identity has dimension $2^{n-m}$.
\end{thm}
\noindent
We refer the reader to \cite{nielsen2010quantum} for a proof, but the basic idea is that the projection onto the $+1$ subspace of each generator decreases the dimensionality of the subspace by a factor of two.  

We can apply this theorem to the lattice $\mathbb{Z}_2$ gauge theory with charged matter by noting that in unitarity gauge the Hilbert space is just the tensor product of a qubit on each edge of the lattice.  The set of plaquettes $Z(\gamma)$ and ``stars'' $\prod_{\vec{\delta}}X(\vec{x},\delta)$ generate an abelian subgroup $S$ of the Pauli group on this Hilbert space, and it is easy to see that 
 no product of plaquettes and stars can give $-I$. In fact, below we will classify all the relations among plaquettes and stars.  Hermitian elements of the Pauli group can only have eigenvalues $\pm 1$, so states where all plaquettes and stars act as the identity will necessarily be ground states of the Hamiltonian \eqref{toric}.  We may thus apply theorem \ref{stabilizerthm} to identify the dimensionality of the ground state subspace.  To show that the ground state is unique, we need to show that the number of independent generators of $S$ is equal to the number of edges in the lattice.  

\bfig
\includegraphics[height=8cm]{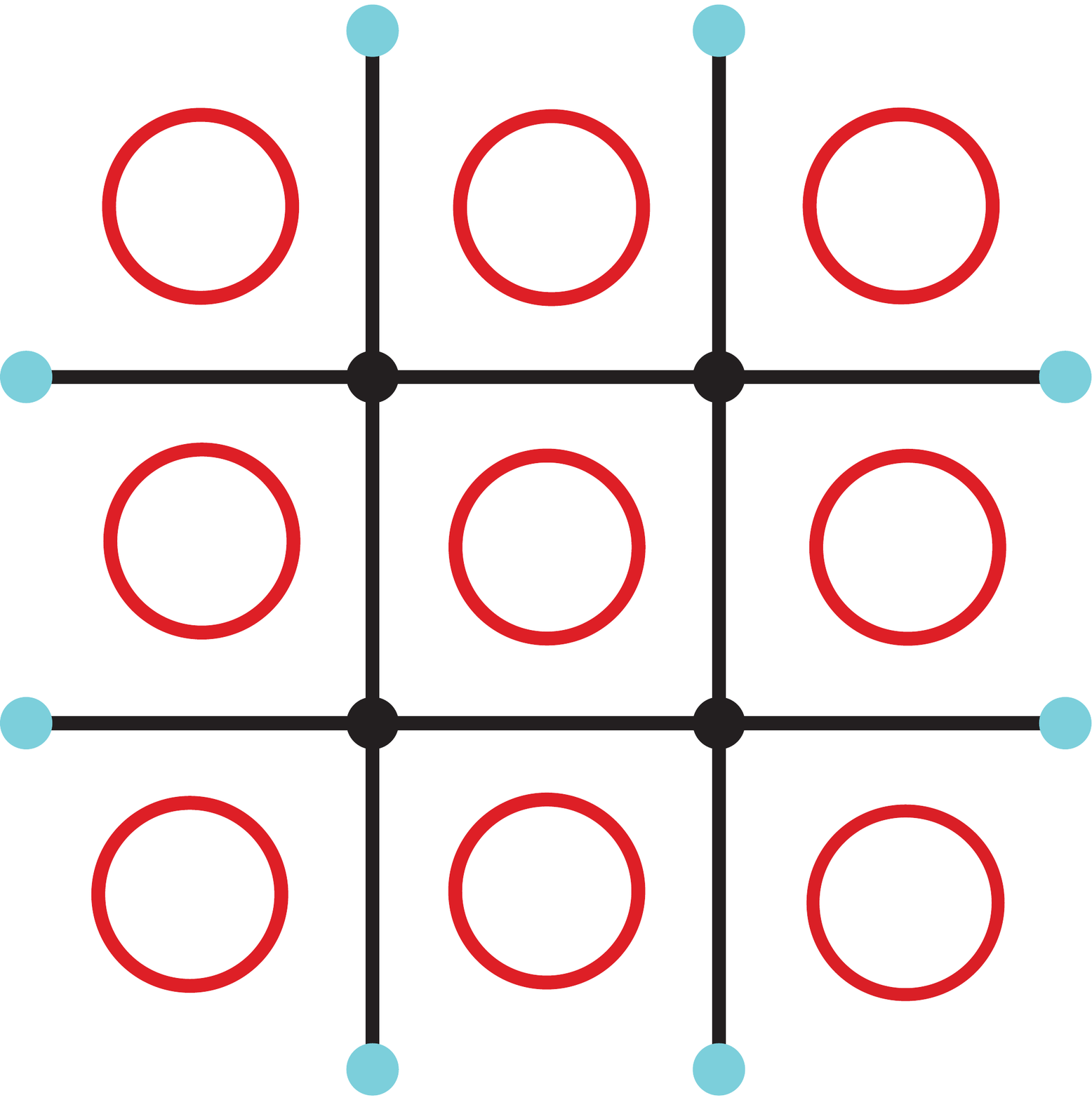}
\caption{Stabilizer generators for a cubic lattice with two spatial dimensions.  The nine red circles indicate plaquettes,
 and the four star constraints live at the black dots.  
 Since the product of all the plaquettes is the identity, 
 the number of independent generators is $8 + 4 = 12$, which 
 agrees with the number of edges. By 
 Theorem~\ref{stabilizerthm}, the ground state is unique.}\label{stabilizerfig}
\efig

Counting the number of independent generators of $S$ is nontrivial because there are relations among stars and plaquettes. 
 For example consider the situation in figure \ref{stabilizerfig}.   Since stars and plaquettes commute with each other, and since the only relations among Pauli generators that reduce their
numbers are $X^2 = Z^2 = 1$, any relation among stars and plaquettes can be expressed a product of a relation among stars
only and 
a relation among plaquettes only. Thus, it is sufficient to treat stars and plaquettes separately when counting their relations. There are no relations among the four stars, since it is not possible to cancel the $X(e)$ on boundary-piercing edges, but the product of the nine red plaquettes is equal to the identity. 
Therefore, the number of independent generators (plaquettes and stars) is equal to twelve, 
which indeed equals the number of edges. 
  It is easy to see that this counting works out more generally for a two-dimensional rectangular square lattice with some numbers of rows and columns. We now explain how to generalize this counting to arbitrary dimension and topology. 
	
For simplicity we first discuss the case where the lattice has no boundary, for example it could be a discretization of a Riemann surface.  We will refer to the CW complex associated to the lattice as $X$, and we will denote by $N_n(X)$ the number of $n$-cells in $X$.  We will take $X$ to be connected, since in the disconnected case the ground state subspace just tensor factorizes component by component.  The number of stars is $N_0(X)$, the number of edges is $N_1(X)$, and the number of plaquettes is $N_2(X)$.  There is however one relation between the stars: the product of all of them is the identity.  There can be no further relations, as can be seen by the following argument.  Any relation between the stars can be expressed by saying that the product of some subset of them is equal to the identity.  To get a nontrivial relation, at least one star must be included.  Consider any loop of edges which includes an edge attached to that star.  Each edge of the loop must appear in either zero or two stars in the relation in order for it to be equal to the identity, and moreover they must all appear in zero or all appear in two.  Since one of them appears in two, they all must.  But since this true for any loop containing that edge, to get a nontrivial relation we need to include all the stars. Thus we have
\be
\#(\mathrm{independent\,\,stars})=N_0(X)-1.
\ee 
Counting the relations between the plaquettes is more nontrivial, we claim that
\begin{align}\nonumber
\#(\mathrm{independent\,\, plaquettes})=&N_2(X)-(N_3(X)+b_2(X))+(N_4(X)+b_3(X))-\ldots \\
&+(-1)^{d-1}(N_{d-1}(X)+b_{d-2}(X))-(-1)^{d-1} b_{d-1}(X),
\end{align} 
where $b_m(X)$ is the dimensionality of the homology group $H_m(X,\mathbb{Z}_2)$.  The idea of this is as follows: the product of any set of plaquettes living on a two-cycle in $\mathbb{Z}_2$ homology is the identity, and so gives a relation between the plaquettes.  The set of two-cycles which are boundaries of three-chains is generated by products of three-cells, of which there are $N_3(X)$.  We also need to include one representative of each nontrivial homology class of two-cycles, hence our subtraction of $(N_3(X)+b_2(X))$.  But there aren't actually $N_3(X)$ independent homologically-trivial two-cycles, since those collections of three-cells which form three-cycles have trivial boundary and thus do not generate two-cycles.  So we need to add back the number of three-cycles, which is given by $(N_4(X)+b_3(X))$, except then some collections of the four cells are five cycles, which we need to resubtract, and so on.  In the last step we need to add or subtract the number of $d-1$-cycles, which are clearly never boundaries of $d$ cycles, so we are left with only $b_{d-1}$.	In stabilizer parlance, we have
\be
n=\#(\mathrm{edges})=N_1(X)
\ee
qubits and 
\be
m=\#(\mathrm{independent\,\,stars})+\#(\mathrm{independent\,\, plaquettes})
\ee
generators of $\mathcal{S}$, so the groundstate degeneracy is
\be
2^{n-m}=2^{b_1(X)},\label{noBdeg}
\ee
where we have used the expressions
\be\label{euler}
\chi(X)\equiv \sum_{n=0}^d (-1)^nN_n(X)=\sum_{n=0}^d (-1)^n b_n(X).
\ee
for the Euler characteristic of $X$, and also that $b_0(X)=1$ since $X$ is connected.  The expression \eqref{noBdeg} has a natural interpretation: the ground state subspace is labeled by the eigenvalues of the Wilson lines on the topologically distinct one-cycles of $X$ \cite{Kitaev:1997wr}.  

We now turn to lattices where $\partial X$ is nontrivial.  In order to allow a nontrivial long-range gauge symmetry, we had to choose boundary conditions on our gauge theory with matter fields as in figures \ref{latticeboundaryfig}, \ref{stabilizerfig}, where boundary edges are not included since we do not have degrees of freedom there and there are no star constraints on boundary sites.  For $X$ to be a CW complex however, we need to include these boundary edges as one-cells and boundary sites as zero-cells, since otherwise the boundaries of plaquettes which are adjacent to the boundary will not be part of the set of zero-cells and one-cells. Similarly $X$ needs to include all higher cells in $\partial X$ as well.  The number of edges which carry qubits is thus now given by
\be\label{boundaryedges}
\#(\mathrm{edges})=N_1(X)-N_1(\partial X).
\ee
There are no longer any relations between the star constraints, since given any edge in a star involved in such a relation we can construct a path to the boundary on which all edges would need to appear in two stars, but this is impossible for boundary-piercing edges since there are no star constraints on boundary sites.  Therefore we have
\be\label{boundarystars}
\#(\mathrm{independent \,\, stars})=N_0(X)-N_0(\partial X).
\ee
Counting the number of independent plaquettes is again more difficult, we claim that
\begin{align}\nonumber
\#(\mathrm{independent\,\, plaquettes})=&N_2(X)-N_2(\partial X)\\\nonumber
&-\left((N_3(X)-N_3(\partial X)+b_2(X)-b_2^{NT}(\partial X)+b_1^T(\partial X)\right)\\\nonumber
&+\left(N_4(X)-N_4(\partial X)+b_3(X)-b_3^{NT}(\partial X)+b_2^T(\partial X)\right)\\\nonumber
&-\ldots\\\nonumber
&+(-1)^{d-1}\left(N_{d-1}(X)+b_{d-2}(X)-b_{d-2}^{NT}(\partial X)+b_{d-3}^T(\partial X)\right)\\
&-(-1)^{d-1}\left(b_{d-1}(X)+b^T_{d-2}(\partial X)\right).\label{boundaryp}
\end{align}
In this formula we use a notation where we have split the $n$-cycles in $\partial X$ which are not boundaries in $\partial X$ into a set which \textit{are} boundaries in $X$, which have $b_n^{T}(\partial X)$ independent representatives, and a set which \textit{aren't} boundaries in $X$, which have $b_n^{NT}(\partial X)$ independent representatives.  By definition, we have
\be
b_n(\partial X)=b_n^T(\partial X)+b_n^{NT}(\partial X).  
\ee
To understand equation \eqref{boundaryp}, we begin as before: there are $N_2(X)-N_2(\partial X)$ plaquettes, but the product of plaquettes on any two-cycle in $\mathbb{Z}_2$ homology vanishes identically. This again imposes relations on the plaquettes.  The set of two-cycles which are boundaries is generated by the three-cells, of which there are $N_3(X)$, but the three cells which lie in the boundary are automatically trivial, so we should subtract $N_3(\partial X)$.  In counting two-cycles we should include a representative of each nontrivial class in $H_2(X,\mathbb{Z}_2)$, hence adding $b_2(X)$, but now we need to account for the fact that nontrivial two-cycles in $X$ which are homologous to nontrivial two-cycles in the boundary can still be generated by the three-cells, so we should subtract $b_2^{NT}(\partial X)$.  Finally in addition to the two-cycles, there are also relations from two-chains whose boundaries lie in $\partial X$, since these again are the identity.  When the boundary of such a two-chain is a boundary in $\partial X$, then the relation associated to it is equivalent to one from a two-cycle in $X$ which contains some boundary two-cells, so we only get new relations from those two-chains in $X$ whose boundary is in $\partial X$ but is not a boundary there.  These are counted precisely by $b_1^T(\partial X)$, hence we add this to our list of relations, finally subtracting the whole set as the second line of \eqref{boundaryp}.  We then observe that collections of three-cells which generate three-cycles or three-chains whose boundary is in $\partial X$ do not actually define two-cycles, and so we need to add back the third line of \eqref{boundaryp}.  And so on.  Combining \eqref{boundaryedges}, \eqref{boundarystars}, and \eqref{boundaryp}, and again using \eqref{euler} and $b_0(X)=1$, we at last have a ground state degeneracy
\be\label{Bdeg}
2^{n-m}=2^{b_0(\partial X)-1+b_1(X)-b_1^{NT}(\partial X)}.
\ee
This formula again has an elegant interpretation in terms of Wilson lines: 
\be\label{lines}
b_0(\partial X)-1
\ee
counts the number of independent Wilson lines stretching from one component of $\partial X$ to another, while
\be\label{loops}
b_1(X)-b_1^{NT}(\partial X)
\ee
counts the number of independent homologically-nontrivial Wilson loops which are not homologous to boundary one-cycles, since those which are must be trivial by the boundary conditions.  In particular if $X$ is homeomorphic to $B^{d-1}$, then \eqref{lines} and \eqref{loops} both vanish ($\partial B^{d-1}=\mathbb{S}^{d-2}$ is connected and there are no nontrivial one-cycles in $B^{d-1}$), so the ground state is unique.

\section{Multiboundary wormholes in three spacetime dimensions}\label{wormholeapp}
In this appendix we review some of what is known about multiboundary wormholes in $AdS_3/CFT_2$, focusing on the feasibility of constructing geometries which can be used in our second proof of theorem \ref{noglobalthm}.  The great advantage of $d=2$ is that there are no gravitational waves, so all solutions of the Einstein equation with negative cosmological constant and no matter are locally isometric to $AdS_3$.  More precisely, they are quotients of $AdS_3$ by a discrete subgroup $\Gamma$ of its isometry group $SO(2,2)$.  In $AdS_3/CFT_2$ such states can often be prepared by cutting the path integral of the CFT on a Riemann surface \cite{Krasnov:2000zq,Skenderis:2009ju,Balasubramanian:2014hda,Maxfield:2016mwh}, we now review this construction.  

\bfig
\includegraphics[height=5cm]{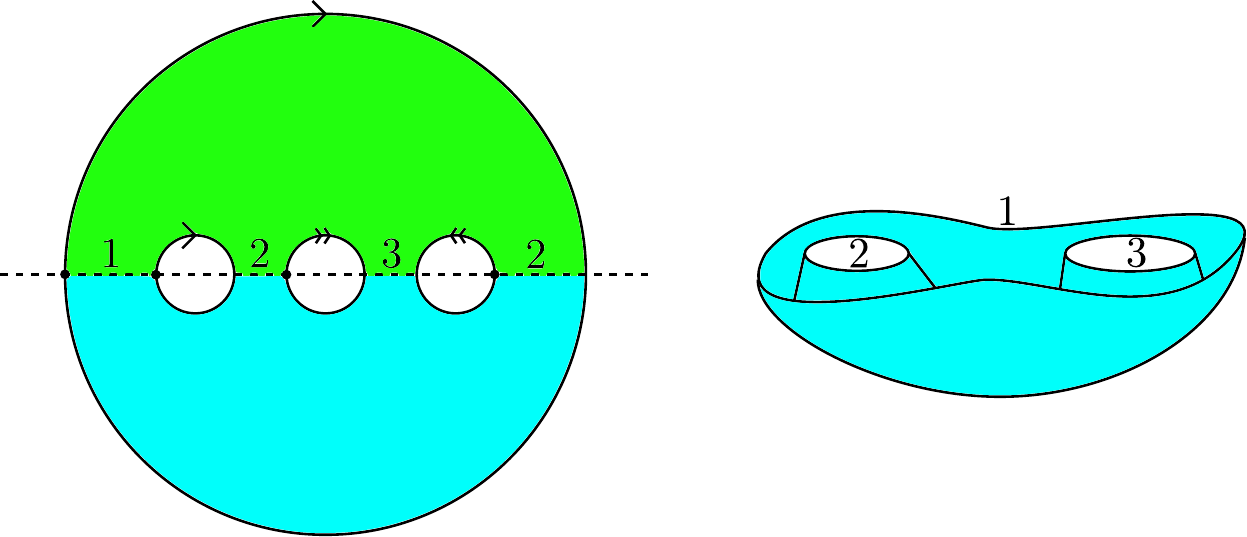}
\caption{A genus two Riemann surface constructed using four Schottky discs.  On the left the surface is the union of the green and blue regions, with the indicated identifications and the marked points  identified.  Performing the CFT path integral over just the blue region below the cut prepares a state in the Hilbert space of the CFT on three circles, labeled $1,2,3$.  On the right we show a heuristic picture of the cut geometry embedded into $\mathbb{R}^3$.}\label{genus2fig}
\efig
We begin by recalling the Schottky construction of an arbitrary Riemann surface.  Viewing the complex plane as the Riemann sphere, we place an even number of non-intersecting discs and then identify their boundaries in pairs with opposite orientation: the Riemann surface is the region to the exterior of all the discs.  Each identified pair can be viewed as adding a handle to the Riemann sphere, so if we place $2g$ discs we get a genus $g$ Riemann surface.  The moduli of the Riemann surface arise from the locations and sizes of the discs, as well as a possible twist in each identification.  By an $SL(2,\mathbb{C})$ transformation we can always choose one of the discs to be centered at infinity, and if we restrict to geometries which are time-reversal invariant then we can take all discs to be centered on the real axis with no twists.  A $g=2$ example is shown in figure \ref{genus2fig}, where we cut to get a state of the CFT on three circles. More generally by cutting a genus $g$ surface we can produce a pure state in the Hilbert space of the CFT on $g+1$ spatial circles.    

\bfig
\includegraphics[height=3.5cm]{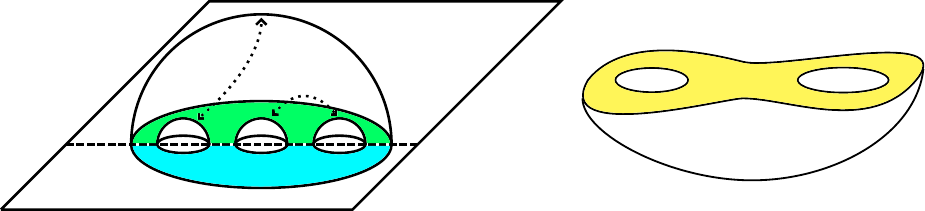}
\caption{A time-symmetric genus two handlebody.  In the left diagram, the handlebody lies above the small hemispheres and below the large hemisphere, with the indicated identifications of hemispheres.  The dashed boundary-time slice is extended straight up to give a symmetric timeslice of the bulk geometry.  In the right diagram this bulk timeslice is shaded yellow as a cut through the heuristic representation of the genus two handlebody embedded in $\mathbb{R}^3$.  Note that the three asymptotic boundaries are connected through a wormhole, as in figure \ref{triple2fig}.}\label{genus2bulkfig}
\efig
In order to find the bulk geometry of a state constructed in this manner, one needs to minimize the Euclidean Einstein-Hilbert action with negative cosmological constant over all solutions whose asymptotic boundary is the Riemann surface in question.  Assuming that this minimum has a time-symmetric slice whose boundary lies in the real axis of the Schottky construction (if not then the bulk interpretation of the state is unclear), one then takes that slice as initial data for the Lorentzian Einstein equation to construct the real-time bulk geometry.  The full set of these Euclidean solutions is rather complex, but there is an especially simple subset referred to as the \textit{handlebodies}, which are obtained by ``filling in'' the Riemann surface embedded in $\mathbb{R}^3$.  Given a Schottky presentation of a Riemann surface, there is a natural way to do this by viewing the complex plane in the Schottky construction as the boundary of the three-dimensional upper half plane, with metric
\be
ds^2=\frac{dx^2+dy^2+dz^2}{z^2}
\ee
and $z>0$, and then contracting the boundary of each disc using a hemisphere in the bulk. We illustrate this for genus two in figure \ref{genus2bulkfig}.  It is important to emphasize however that there can be different Schottky presentations of the same Riemann surface, which differ by acting with an element of the mapping class group of ``large'' diffeomorphisms that exchange the various cycles, eg $PSL(2,\mathbb{Z})$ for genus one, and these different presentations lead to different handlebodies in the bulk since different cycles are contracted.  Moreover in general the Schottky presentation in which the time-symmetric slice is the real axis is not the Schottky presentation from which the handlebody is constructed, unlike in figure \ref{genus2bulkfig} where it is.  At genus one there are only two time-symmetric handlebodies, the ``Euclidean BTZ'' and ``thermal AdS'' solutions, which differ by which of the two cycles is contracted in the bulk, and it is the Euclidean BTZ solution which is constructed as in figure \ref{genus2bulkfig}.

\bfig
\includegraphics[height=4cm]{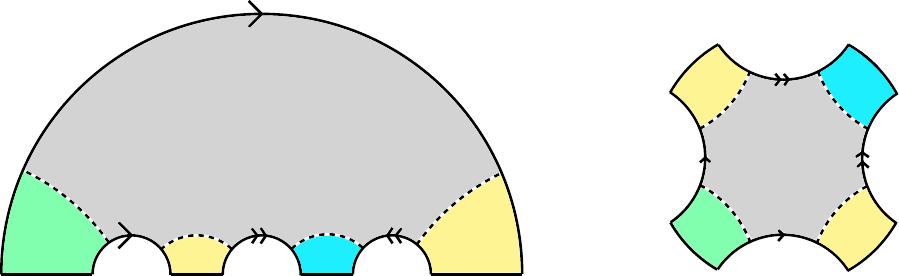}
\caption{The time-symmetric bulk slice of a three-boundary wormhole.  On the left we give the upper-half-plane presentation, while on the right we give the Poincare-disc presentation.  The ``interior'' region is shaded grey, while the three ``exterior'' regions are shaded green, blue, and yellow.  The dashed lines are the minimal length curves between the identification semicircles, and in Lorentzian signature they are the bifurcate horizons.}\label{genus2slicefig}
\efig
In fact at any genus we are especially interested in the particular handlebody where the Schottky presentation with time-symmetry about the real axis \textit{does} coincide with the Schottky presentation where the disc boundaries are contracted in the bulk, as shown in figure \ref{genus2bulkfig}.  The reason is that this is the only handlebody for which the time-symmetric bulk slice is connected, so in Lorentzian signature it is the one that describes a wormhole connecting all of the asymptotic boundaries.  For example at genus one the bulk timeslice of the ``thermal AdS'' handlebody is two disconnected discs.  We can understand better the structure of this wormhole by looking in more detail at the geometry of the time-symmetric slice, obtained by cutting through the geometry in the left diagram of figure \ref{genus2bulkfig}  directly above the dashed boundary cut.  This slice has the geometry of a quotient of the upper-half plane by a discrete subgroup, and in fact for this particular handlebody it is the Fuchsian presentation of the same cut Riemann surface on which the CFT path integral was evaluated to prepare the state.  Moreover the intersection of the bifurcate horizons in the Lorentzian solution with this timeslice are given precisely by the minimal length curves between the identification semicircles \cite{Aminneborg:1997pz}, which gives an elegant way of splitting the time-symmetric slice into ``interior'' and ``exterior'' regions.  We illustrate this for genus two in figure \ref{genus2slicefig}.  In general whenever this spatial slice connects $n$ asymptotic boundaries without any additional interior handles we can compute its volume using the Gauss-Bonnet theorem: it is an $n$-punctured sphere with a metric of constant negative curvature $R=-2$, and whose punctures are bounded by geodesics with $K=0$, so (in units where $\ell_{ads}=1$) we just have \cite{Marolf:2015vma}
\be\label{voln}
\mathrm{Interior \,\, spatial \,\, volume}=2\pi(n-2),
\ee 
which is independent of the moduli.  Notice that indeed for $n>2$ (and therefore $g>1$) we have a nontrivial interior which grows in size as we increase $n$.  Moreover it will not be in the entanglement wedge of any one of the boundaries, which is the key property for our wormhole-based proof of theorem \ref{noglobalthm}.  

\bfig
\includegraphics[height=6cm]{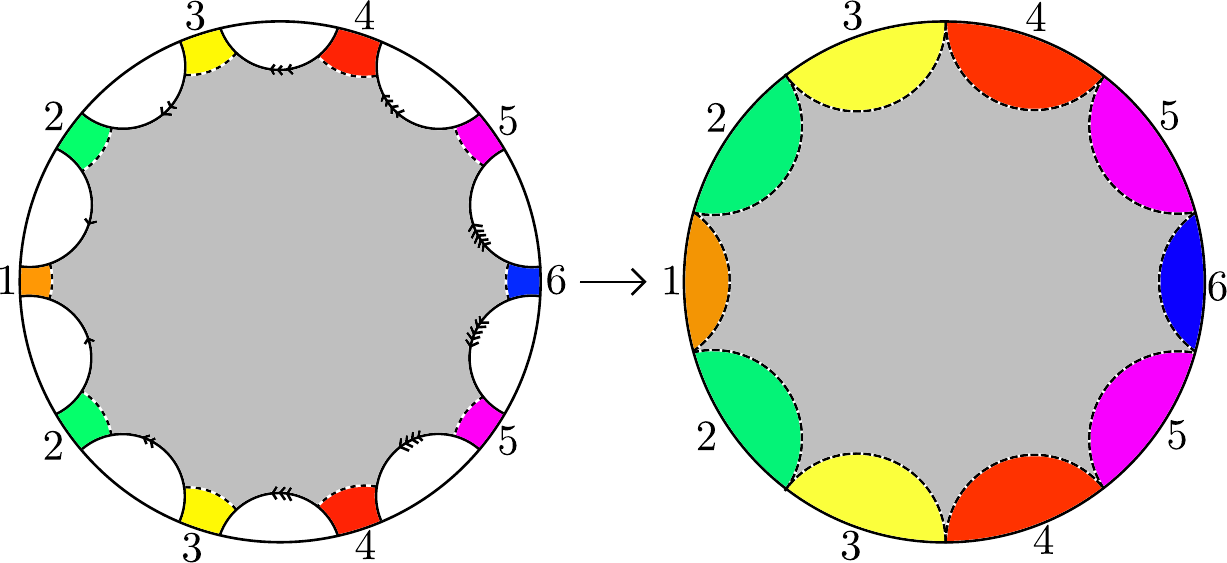}
\caption{The time-symmetric bulk slice of a genus five wormhole with six exterior regions.  The interior is shaded grey, while the exteriors are shaded in various colors.  The horizons are the dashed lines, and on this line in moduli space the horizon lengths are equal for boundaries 2,3,4, and 5, each of which is twice the length of the horizons for boundaries 1 and 6.  On the left we show a geometry where these length are all finite, while on the right we show the limiting configuration as the lengths go to infinity.}\label{genus5fig}
\efig
In order for that proof to be valid however, we need to check that these connected-wormhole handlebodies do actually dominate the Euclidean path integral, at least somewhere in moduli space.  For genus one the handlebodies are all the solutions, and we know that at high temperature the Euclidean BTZ geometry is dominant.  For $g\geq 2$ they are not: the others are usually called \textit{non-handlebodies}, and they are less well-understood.  Fortunately there is some evidence that non-handlebodies are always subleading to at least one handlebody in the Euclidean path integral \cite{Yin:2007at,Maxfield:2016mwh}, and in what follows we will assume this to be the case.  We are then left with the following question: at any particular point in moduli space, which choice of handlebody minimizes the Euclidean action?  Unfortunately even this question has not been systematically addressed, since evaluating the Euclidean action of a handlebody amounts to computing the classical action of a solution of the Liouville equation on the boundary Riemann surface \cite{Krasnov:2000zq}, which so far is only possible analytically in very restricted cases.\footnote{In fact the connection to the Liouville equation holds if we work in a conformal frame where the boundary metric has constant negative curvature for $g\geq 2$.  It might well be that it is easier to compute the action in some other conformal frame, but we won't pursue this here.}  Recently a numerical algorithm has been developed for computing the Liouville action on arbitrary Riemann surfaces \cite{Maxfield:2016mwh}, specifically with the goal of clarifying which handlebodies dominate the Euclidean gravitational path integral with a boundary Riemann surface in various regions of moduli space, but so far it has only been applied in a few special cases.  We also will not solve this problem, but will instead just suggest a limit in moduli space where we find it plausible that the connected wormhole should be the dominant handlebody.  

Our proposal is most natural in the Poincare disk representation of the bulk time-slice, shown for genus two as the right diagram in figure \ref{genus2slicefig}.  The idea is to introduce $2g$ equally-spaced and equally-sized semicircles around the edge of the Poincare disk, oriented such that there is a reflection symmetry across the real axis, and then identify the semicircles which are related by this reflection.  We leave the size of the semicircles as a free parameter, which means we are looking at a one-dimensional slice through the moduli space.  We illustrate this construction for genus five in figure \ref{genus5fig}, notice in particular the increased size of the interior region compared to figure \ref{genus2slicefig}, which is consistent with \eqref{voln}.  Our conjecture is then that as we take the radii of the identification semicircles to zero, shown in the right diagram of figure \ref{genus5fig}, this handlebody will be the dominant solution in the Euclidean gravity path integral.  Our conjecture is based on the observation that the Euclidean action is essentially the renormalized volume of spacetime, indeed evaluated on any solution which is a quotient of the hyperbolic three-plane we have we have
\be
S_E=\frac{1}{4\pi G}\left(\int_M d^3x \sqrt{g}-\frac{1}{2}\int_{\partial M} d^2 x\sqrt{\gamma}(K-1)\right).
\ee 
Given a choice of which boundary cycles to contract in the bulk, it is natural to expect that this action will tend to want to contract the smallest cycles, since most likely this can be done at the cost of the least volume in the bulk.  For the family of handlebodies we have constructed, in the limit of small identification semicircles, and therefore large horizon length, the cycles in the boundary which correspond to spatial circles in the time-symmetric slice become parametrically larger than their dual cycles, which are the cycles which appear as the boundaries of the Schottky discs.  At genus one and genus two we can confirm that this is indeed the case: the transition from thermal AdS to Euclidean BTZ indeed happens right when the thermal circle becomes smaller than the spatial one, and the numerical results of \cite{Maxfield:2016mwh} confirm that our limiting family of Riemann surfaces, which corresponds to the line $\ell_3=2\ell_{12}$ in their figure 7, dominates over the other possible handlebodies (and also one non-handlebody they were able to check analytically) in the limit of large horizon length.  Assuming this conjecture is also correct at higher genus, the connected wormhole will always dominate at sufficiently large horizon length, and any quasilocal bulk operator can fit into the interior region for sufficiently high genus.\footnote{Henry Maxfield has suggested a related set of surfaces constructed by taking $n$ copies of the complex plane and gluing them together using two pairs of branch points on each copy.  In the dual CFT this amounts to computing the four-point function of $\mathbb{Z}_n$ twist operators in the symmetric orbifold of $n$ copies of the CFT.  For this set of surfaces there is a natural guess for where the transition from ``totally connected'' to ``totally disconnected'' takes place: at the crossing-symmetric configuration of the four twist operators.  The argument that there is a totally connected phase for sufficiently large cross ratio is the same as for our surfaces: eventually the smallest cycles should all contract in the bulk.}  We are then able to run our second proof of theorem \ref{noglobalthm}.

\section{Sphere/torus solutions of Einstein's equation}\label{spheretorapp}
In this appendix we discuss in more detail the solutions of Einstein's equation with negative cosmological constant used in section \ref{pformholsec}, with metric of the form
\be\label{pformmetric}
ds^2=-\alpha(r)dt^2+\frac{dr^2}{\alpha(r)\beta(r)}+e^{\gamma(r)}dx_p^2+r^2 d\Omega_{d-p-1}^2.
\ee
The time, planar, radial, and spherical components of Einstein's equations with negative cosmological constant for metrics of the form \eqref{pformmetric} are given respectively by\footnote{The reader can compare these equations to those in \cite{Copsey:2006br} in the special case $d=4$, $p=1$.}
\begin{align}\nonumber
r(\alpha \beta'&+\alpha' \beta)\left(2(d-p-1)+pr \gamma'\right)+2(d-p-1)\alpha\beta\left(d-p-2+pr \gamma'\right)\\\nonumber
&+pr^2\alpha\beta\left(\frac{p+1}{2}\gamma'^2+2\gamma''\right)\\
&=2\big((d-p-2)(d-p-1)+d(d-1)r^2\big)\\\nonumber
r\beta'(r\alpha'&+\alpha(2(d-p-1)+(p-1)r\gamma'))+2\beta\Big((d-p-2)(d-p-1)\alpha+2(d-p-1)r\alpha'\\\nonumber
&+r^2\alpha''\Big)+2(p-1)\beta \Big(r\gamma'\big(d-p-1+r\alpha'+\frac{p}{4}r \gamma'\big)+r\gamma''\Big)\\
&=2\big((d-p-2)(d-p-1)+d(d-1)r^2\big)\\\nonumber
p(p-1)r^2\alpha\beta\gamma'^2&+2pr\beta(2(d-p-1)+r\alpha')\gamma'+4\beta(d-p-1)((d-p-2)\alpha+r\alpha')\\
&=4\big((d-p-2)(d-p-1)+d(d-1)r^2\big)\\\nonumber
r^2\alpha'\beta'&+r(2\alpha'\beta+\alpha \beta')(2(d-p-2)+pr\gamma')+2r^2\beta\alpha''+2(d-p-3)(d-p-2)\alpha \beta\\\nonumber
&+2p(d-p-2)r\alpha \beta \gamma'+pr^2 \alpha \beta \left(\frac{p+1}{2}\gamma'^2+2\gamma''\right)\\
&=2\big((d-p-3)(d-p-2)+d(d-1)r^2\big).
\end{align}

We first consider the vacuum solution, where it is the sphere that contracts in the bulk.  We can then assume a further symmetry between the time and planar directions, setting 
\be
\gamma=\log \alpha.
\ee
The first two equations of motion become redundant, and the third simplifies so that we can solve for $\beta$:
\be\label{betaeq}
\beta=\frac{4\alpha\left((d-p-2)(d-p-1)+d(d-1)r^2\right)}{4(d-p-2)(d-p-1)\alpha^2+4(d-p-1)(p+1)r\alpha\alpha'+p(p+1)r^2\alpha'^2}.
\ee
After this substitution, the first, third, and fourth equations of motion each give the same second order ordinary differential equation for $\alpha$.  

To find the right boundary conditions, we can expand $\alpha$ in a power series near $r=0$ and then substitute into this differential equation.  The result is that if we want $\alpha(0)>0$ then we must have
\be
\alpha(r)=\alpha(0)\left(1+\frac{1}{d-p}r^2+O(r^3)\right).  
\ee
This then tells us that we must impose $\alpha'(0)=0$, which from \eqref{betaeq} then implies that $\beta(0)\alpha(0)=1$, as needed to avoid a singularity at $r=0$.  The overall scale of $\alpha$ can be absorbed into a rescaling of the time coordinate, so we thus have a unique vacuum solution, as found by Horowitz and Copsey for $d=4$ and $p=1$.  

\bfig
\includegraphics[height=4cm]{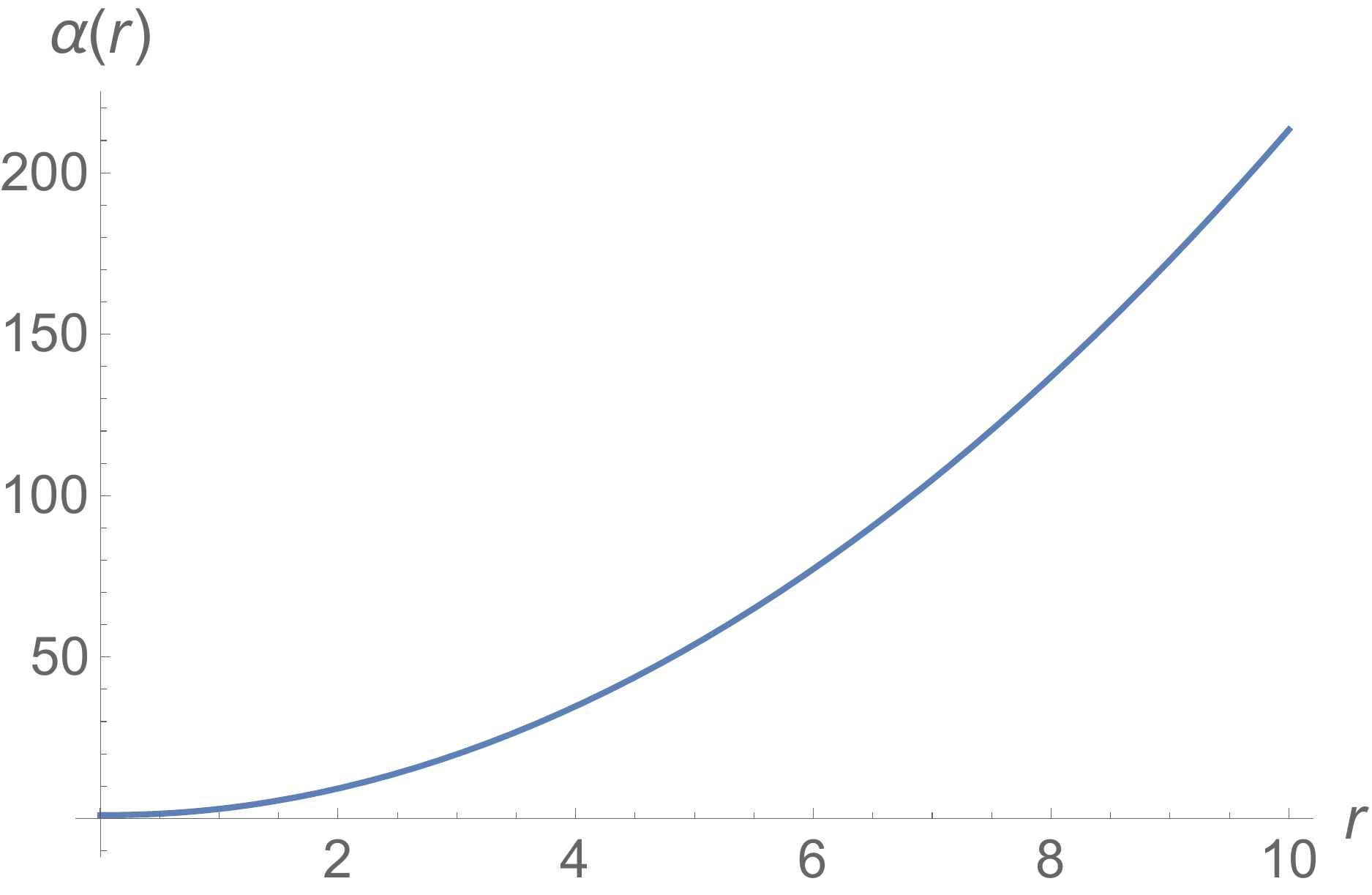} \hspace{2cm}
\includegraphics[height=4cm]{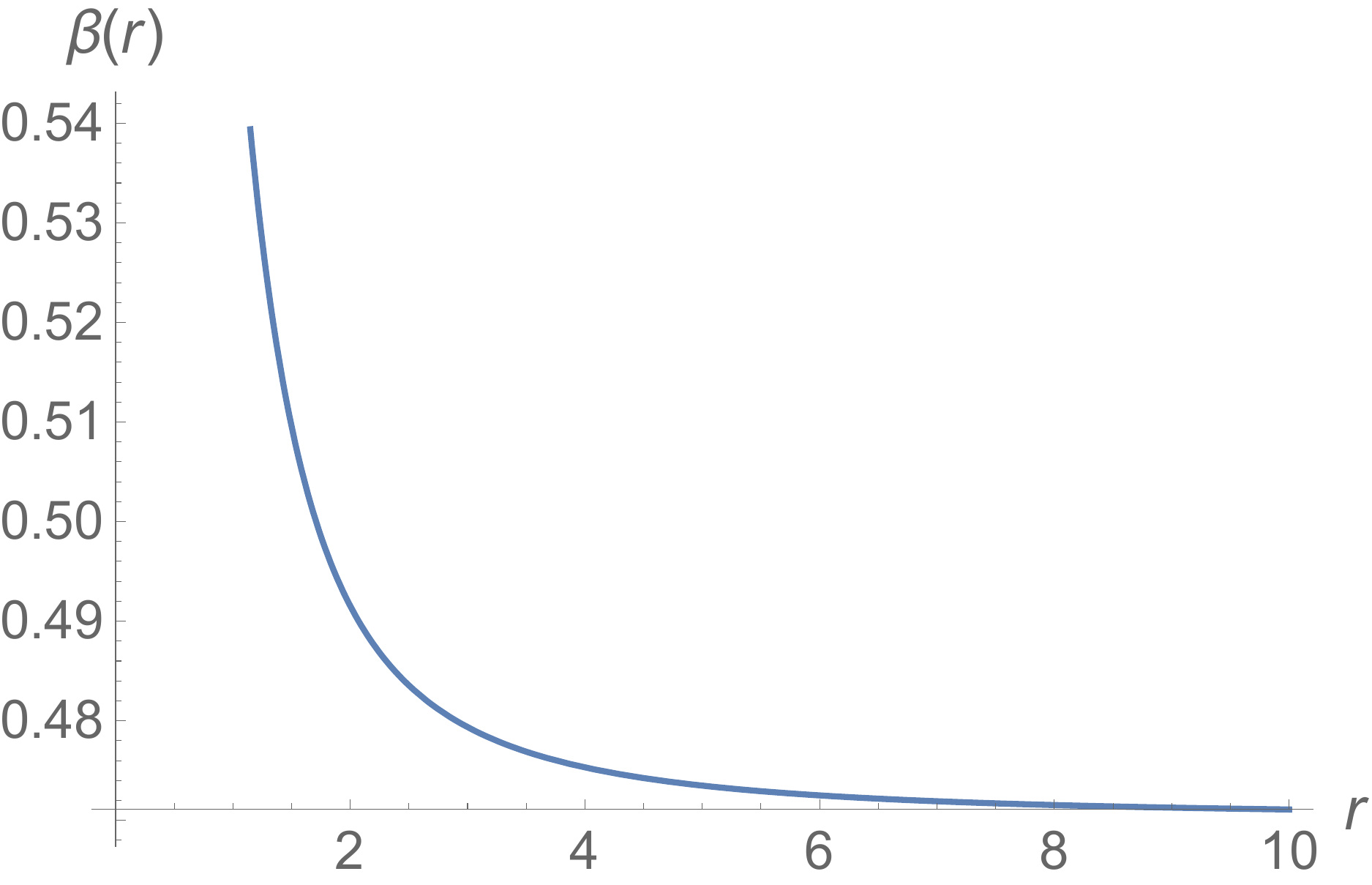}
\caption{Numerical plots of the vacuum solution for $p=2$ and $d=5$.}\label{vacuumsol}
\efig
The differential equation for $\alpha$ can only be solved numerically, which we've written a mathematica file (included in the arxiv submission) to do. We've checked for a variety of $d$ and $p$ that, with these boundary conditions, the solutions for $\alpha$ and $\beta$ are positive, and behave as $\alpha\beta=r^2+o(r^2)$ at large $r$, as required for the geometry to be asymptotically AdS.  We plot a typical example in figure \ref{vacuumsol}.

We now consider the wormhole solutions, where $\alpha$ vanishes at some $r_s>0$.  In this case we cannot assume symmetry between $t$ and $x$, so we must treat $\alpha$, $\beta$, and $\gamma$ independently.  We first observe that the third equation of motion is quadratic in $\gamma'$, and can be solved to give an expression for $\gamma'$ in terms of $\alpha$, $\alpha'$, and $\beta$:
\begin{align}\nonumber
\gamma'=&\frac{1}{p(p-1)r^2\alpha \beta}\Bigg(-pr\beta(r\alpha'+2(d-p-1)\alpha)\\\nonumber
&+\Big[p^2r^2\beta^2(r\alpha'+2(d-p-1)\alpha)^2\\
&+4p(p-1)r^2\alpha\beta\left((d-p-2)(d-p-1)+d(d-1)r^2-(d-p-1)\beta(r\alpha'+(d-p-2)\alpha)\right)\Big]^{1/2}\Bigg)
\end{align}
This expression then may be substituted into the other equations, to produce a pair of independent differential equations which are second order in $\alpha$ and first order in $\beta$.  One nice simplification occurs if we take the difference of the first and fourth equations, which tells us that
\begin{align}\nonumber
-2r^2\beta\alpha''+2r\alpha\beta-r^2\alpha'\beta'+2\alpha\beta(2(d-p-2)+pr\gamma')-r\alpha'\beta(2(d-p-3)+pr\gamma')=4(d-p-2).
\end{align}
We can pair this equation with, say, the first equation, and then solve them numerically.  We now need three boundary conditions: one is provided by $\alpha(r_s)=0$, and another can be fixed by rescaling time so that $\alpha'(r_s)$ takes any value we choose.  Finally by inspecting the form of the equations at a point where $\alpha=0$, we can see that we must have
\be
\beta(r_s)=\frac{d-p-2+dr_s^2}{r_s\alpha'(r_s)}.
\ee 
The parameter $r_s$ is physical, and sets the temperature parameter in the thermofield double state. 

\bfig
\includegraphics[height=4cm]{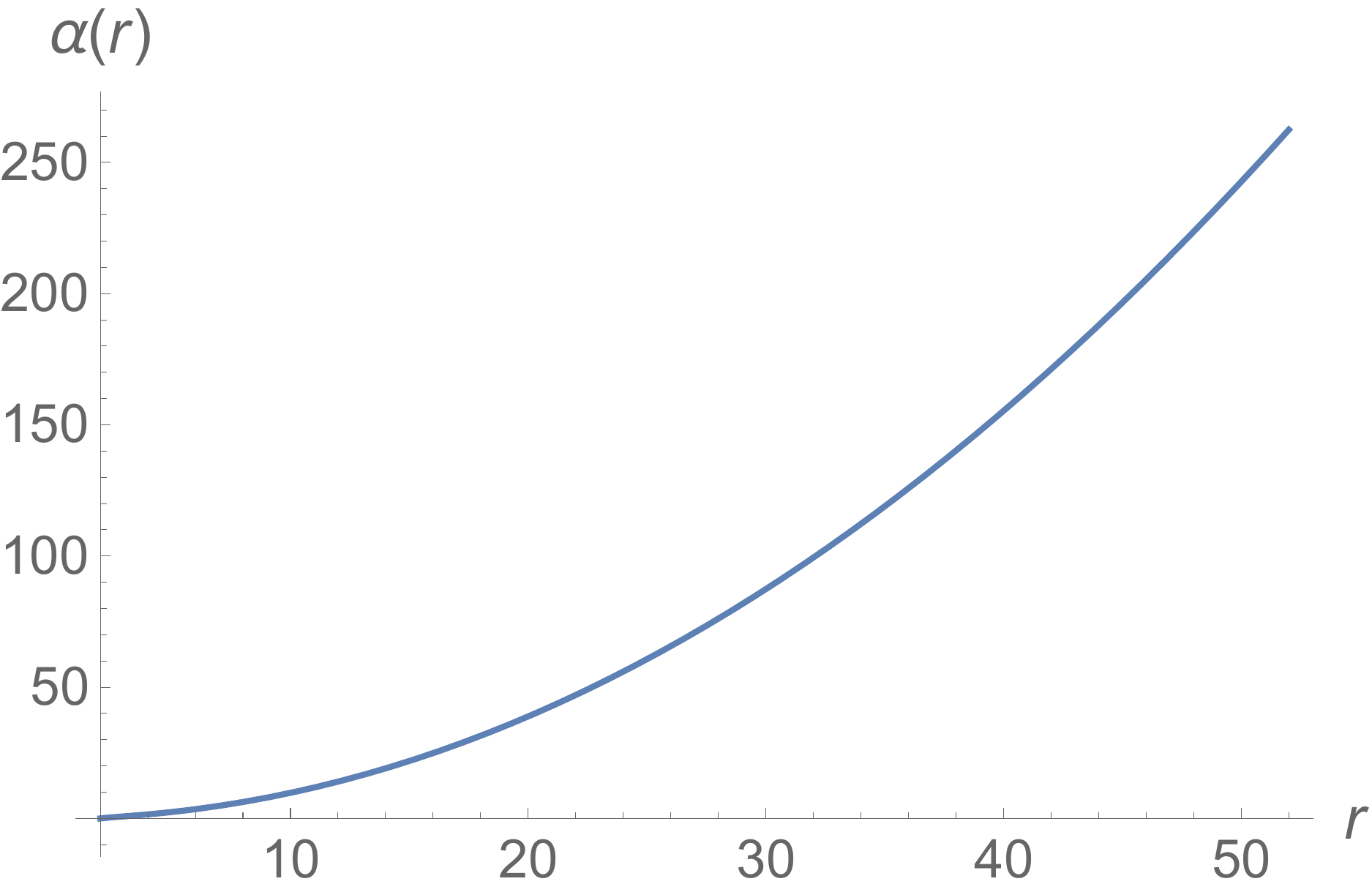}\hspace{1cm}
\includegraphics[height=4cm]{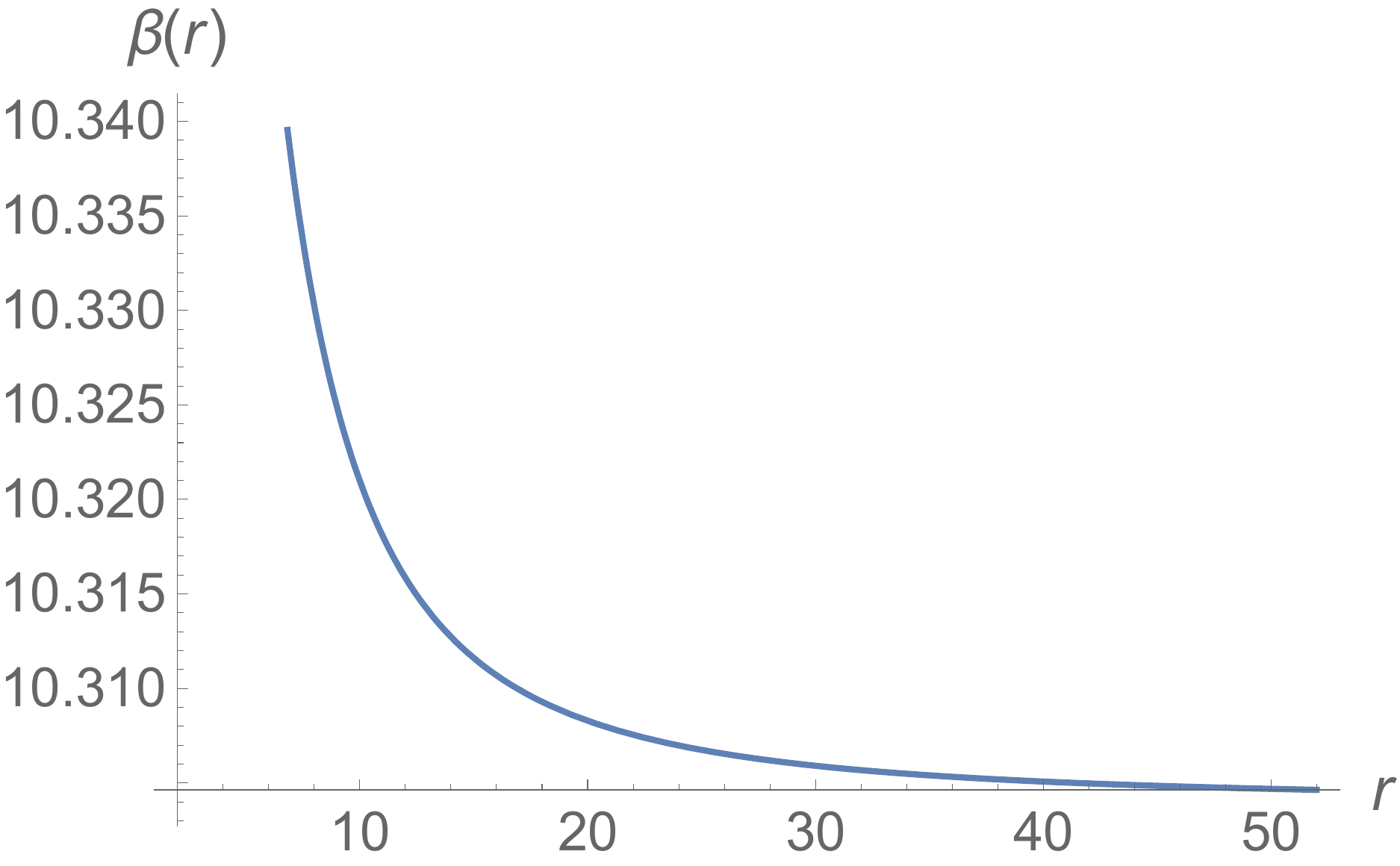}\hspace{1cm}
\includegraphics[height=4cm]{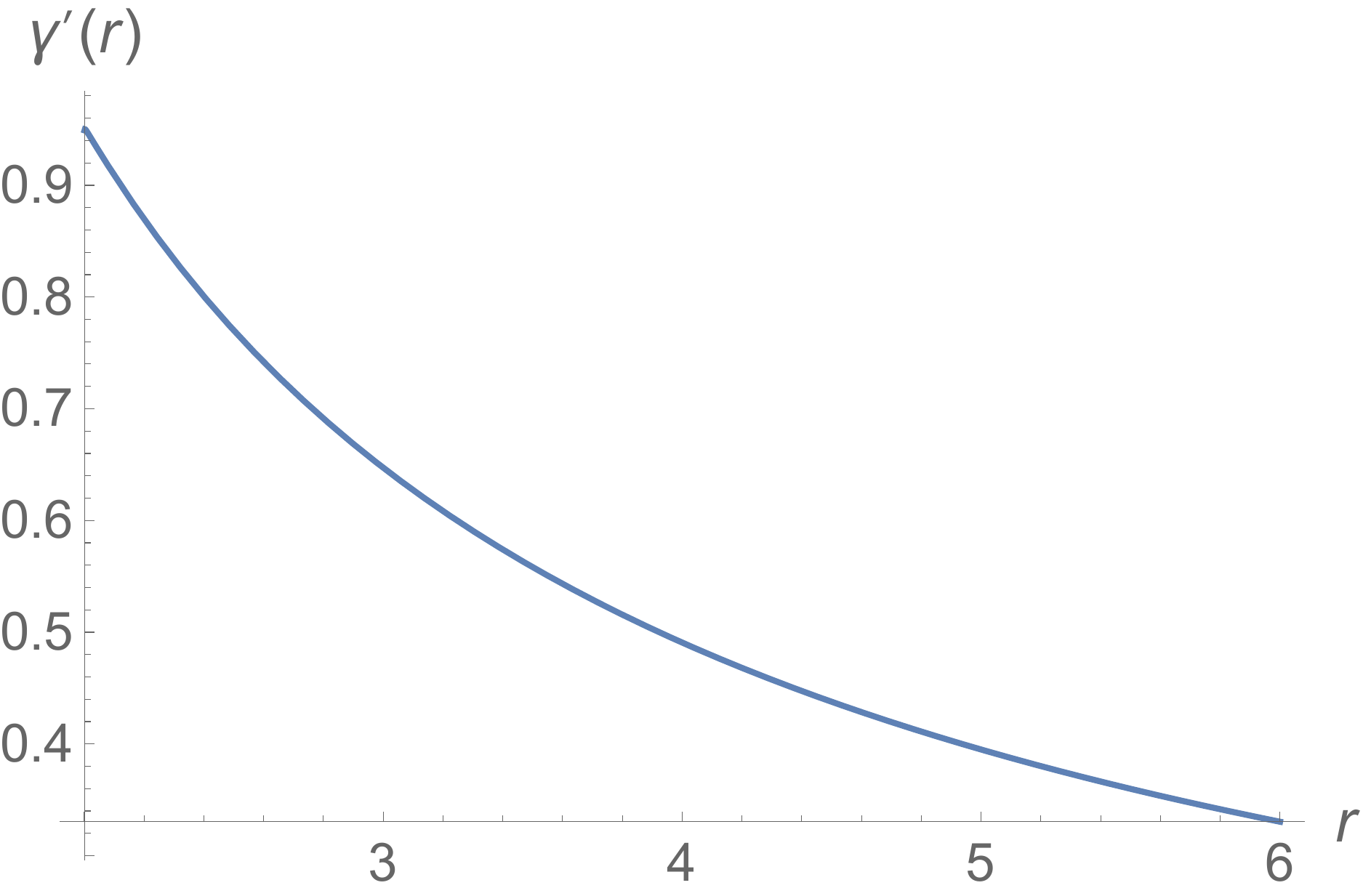}
\caption{Numerical plots of the wormhole solution, for $r_s=2$, $d=5$, and $p=2$.}\label{pformwormholeplot}
\efig
We've again written mathematica code (included in the arxiv submission) to solve these equations numerically, and again confirmed for a variety of $d$, $p$, and $r_s$ that $\alpha$ and $\beta$ are positive, and they have the right large-$r$ asymptotics. We plot an example in figure \ref{pformwormholeplot}.
\bibliographystyle{jhep}
\bibliography{bibliography}
\end{document}